\documentclass[mnsc,nonblindrev]{informs3_nojournal}
\usepackage{fix-cm}
\OneAndAHalfSpacedXI
\usepackage[utf8]{inputenc}
\usepackage{algorithmic}
\usepackage{algorithm}
\usepackage{subcaption}
\usepackage{csquotes}
\usepackage[section]{placeins}

\usepackage{enumitem}
\usepackage{fontenc}    % use 8-bit T1 fonts
\usepackage{url}            % simple URL typesetting
\usepackage{booktabs}       % professional-quality tables
\usepackage{amsfonts}       % blackboard math symbols
\usepackage{nicefrac}       % compact symbols for 1/2, etc.

\newcommand{\E} { \mathbb{E} }
\newcommand{\calR}{\mathcal{R}}
\newcommand{\calS}{\mathcal{S}}
\newcommand{\calW}{\mathcal{W}}

\newcommand{\vx}{\mathbf{x}}
\newcommand{\va}{\mathbf{a}}
\newcommand{\vb}{\mathbf{b}}

\usepackage{natbib}
 \bibpunct[, ]{(}{)}{,}{a}{}{,}%
 \def\bibfont{\small}%
\allowdisplaybreaks

\TheoremsNumberedThrough     % Preferred (Theorem 1, Lemma 1, Theorem 2)
\ECRepeatTheorems

\EquationsNumberedThrough    % Default: (1), (2), ...
\MANUSCRIPTNO{}

\begin{document}
%%%%%%%%%%%%%%%%

% Outcomment only when entries are known. Otherwise leave as is and
%   default values will be used.
%\setcounter{page}{1}
%\VOLUME{00}%
%\NO{0}%
%\MONTH{Xxxxx}% (month or a similar seasonal id)
%\YEAR{0000}% e.g., 2005
%\FIRSTPAGE{000}%
%\LASTPAGE{000}%
%\SHORTYEAR{00}% shortened year (two-digit)
%\ISSUE{0000} %
%\LONGFIRSTPAGE{0001} %
%\DOI{10.1287/xxxx.0000.0000}%

% Author's names for the running heads
% Sample depending on the number of authors;
% \RUNAUTHOR{Jones}
% \RUNAUTHOR{Jones and Wilson}
% \RUNAUTHOR{Jones, Miller, and Wilson}
% \RUNAUTHOR{Jones et al.} % for four or more authors
% Enter authors following the given pattern:
\RUNAUTHOR{Elmachtoub, Kim, and Tan}

% Title or shortened title suitable for running heads. Sample:
% \RUNTITLE{Bundling Information Goods of Decreasing Value}
% Enter the (shortened) title:
\RUNTITLE{Learning Fair Demand Models}

% Full title. Sample:
% \TITLE{Bundling Information Goods of Decreasing Value}
% Enter the full title:
\TITLE{Learning Fair Demand Models}

% Block of authors and their affiliations starts here:
% NOTE: Authors with same affiliation, if the order of authors allows,
%   should be entered in ONE field, separated by a comma.
%   \EMAIL field can be repeated if more than one author
\ARTICLEAUTHORS{%
\AUTHOR{Adam N. Elmachtoub}
\AFF{Department of Industrial Engineering and Operations Research and Data Science Institute, Columbia University, New York, NY 10027, \EMAIL{adam@ieor.columbia.edu}}
\AUTHOR{Hyemi Kim}
\AFF{Department of Industrial Engineering and Operations Research, Columbia University, New York, NY 10027, \EMAIL{hk3181@columbia.edu}}
\AUTHOR{Jonathan Y. Tan}
\AFF{Department of Industrial Engineering and Operations Research, Columbia University, New York, NY 10027, \EMAIL{jyt2123@columbia.edu}}
% Enter all authors
} % end of the block

\ABSTRACT{Data-driven pricing is increasingly prevalent in sectors such as airlines, lending, insurance, and retail. By learning demand models from customer features and setting prices accordingly, these systems may generate discriminatory outcomes that raise fairness concerns. This leads to fundamental questions -- how and where should systems incorporate fairness considerations in the pricing pipeline, and how does it ultimately affect societal outcomes? To answer these, we study a stylized model where a seller has a two-stage decision pipeline comprising linear demand model estimation followed by price optimization. The seller considers fairness notions in training loss, price, and demand, under both parity-wise and Rawlsian perspectives. 

We show that equalizing training loss across consumer groups leads to multiple solutions, which in turn can result in undesirable outcomes despite being a standard approach in fair machine learning. Focusing instead on fairness applied directly to prices or demand, we compare two strategies that enforce fairness in either the demand estimation stage or the  price optimization stage. For parity-wise fairness, we characterize when each strategy yields higher social welfare under small fairness levels. We show that when market sizes and prices in the dataset are similar, imposing price fairness in the estimation stage is more beneficial to consumers, whereas imposing demand fairness in the optimization stage yields better consumer outcomes. For Rawlsian fairness, the two strategies coincide exactly. Lastly, we extend our model to alternate demand functions and conduct a case study using real-world vaccine pricing data.
}

% Fill in data. If unknown, outcomment the field
%\HISTORY{This paper was submitted on XXXX, 2017.}

\maketitle
% \vspace{-25pt}

% Text of your paper here
\section{Introduction}
\label{sec:intro}
Data-driven pricing is now pervasive in practice across sectors including airlines, loans, insurance, transportation services, and retail. These systems estimate heterogeneous demand and set prices using customer features or segment-level attributes that incorporate behavioral, geographic, and demographic information, which raise broader concerns about transparency, consumer agency, and fairness.\footnote{https://www.nytimes.com/2025/11/28/opinion/dynamic-pricing-algorithms.html} For instance, Delta Air Lines recently faced criticism from U.S. lawmakers over fears that its planned use of artificial intelligence (AI) pricing tools could set individualized fares up to each consumer’s ``pain point''.\footnote{https://www.reuters.com/business/delta-air-assures-us-lawmakers-it-will-not-personalize-fares-using-ai-2025-08-01/}. Likewise, field experiments on the Instacart platform show that prices for  products from the same store at the same time can vary across consumers, breaking the commonly held notion that prices are uniform.\footnote{https://www.nytimes.com/2025/12/09/business/instacart-algorithmic-pricing.html} Beyond price disparities, fairness concerns also arise through market access. Survey evidence on vaccines shows willingness-to-pay differs strongly by income and socioeconomic status, implying that pricing policies can translate directly into access gaps \citep{wong2024cost}. Together, these examples highlight a fundamental tension between profit-oriented algorithmic pricing and fairness in both price and demand.

In this paper, we investigate how to learn demand models to ensure fairness and what the resulting consequences are on social outcomes.
A substantial body of research has focused on defining and enforcing fairness in machine learning \citep{mehrabi2021survey}. However, a lot of existing approaches treat fairness as a property of the predictive model instead of downstream decisions such as pricing. \citet{wang2024against} point out that systems that leverage machine learning to guide real-world decisions rely on proxy targets and produce disparate outcomes across groups. Similarly, \citet{scantamburlo2025prediction} emphasize the need for a clear conceptual separation between prediction and decision-making. This motivates us to examine how fairness can be incorporated across the pricing pipeline.

In particular, we consider a pricing process that involves two stages. First, in the estimation stage, a demand model is learned from observed data. Second, in the decision-making, an optimization problem is solved to find the profit-maximizing price. Several complementary questions arise which we address in this paper: \emph{What are the appropriate notions of fairness in this setting? Where should fairness be imposed -- the estimation stage or the decision-making stage? What are the corresponding impacts on societal outcomes?}

%Specifically, can fairness enforced at the estimation stage ensure equitable outcomes even if the subsequent decision-making process ignores fairness? Or, conversely, must fairness be imposed directly at the decision stage to achieve equitable results? Furthermore, if one stage proves more effective, under what conditions does it yield superior outcomes in terms of revenue, consumer surplus, and social welfare?

To begin, we consider prediction loss fairness \citep{khalili2023loss}, which is a common fairness notion in machine learning.  We next study two fairness notions that directly align with downstream objectives, originally proposed by \citet{cohen2022price}: \emph{price fairness}, which considers disparities in prices across groups, and \emph{demand fairness}, which considers disparities in market share (or access) across groups. We also consider two normative perspectives on fairness: parity-wise fairness and Rawlsian fairness. Parity-wise fairness focuses on reducing disparities in price and demand across groups. In many cases, fairness concerns arise at the group level, where pricing policies should not produce systematic disparities across protected attributes. For instance, New York has prohibited the so-called ``pink tax'' -- i.e., charging different prices for similar products based on gender -- since $2020$.\footnote{https://dos.ny.gov/news/former-governor-cuomo-reminds-new-yorkers-pink-tax-ban-goes-effect-today} Rawlsian fairness, in contrast, emphasizes improving outcomes for the least advantaged individual. In markets tied to essential goods and services, this concern is reflected in policy safeguards such as price caps that protect economically vulnerable consumers, including statutory caps on insulin out-of-pocket costs\footnote{S.954 -- Affordable Insulin Now Act of 2023: https://www.congress.gov/bill/118th-congress/senate-bill/954} and regulations that limit allowable rent increases\footnote{Rent Stabilization in NYC: https://www.nyc.gov/site/mayorspeu/programs/rent-stabilization.page}.

To formalize our analysis, we employ a stylized model with two customer groups and linear demand functions. Under this stylized model, we compare two frameworks. The first, which we term \emph{Fair-Estimate-then-Optimize} (FEO), incorporates fairness during the estimation of the demand function and subsequently uses the estimated model for pricing. The second, \emph{Estimate-then-Fair-Optimize} (EFO), applies fairness constraints in the pricing phase only similar to \citet{cohen2022price}. The two frameworks offer distinct advantages. FEO leaves the optimization stage untouched, allowing sales teams to focus solely on profit and commissions. In contrast, EFO produces an unbiased demand model. To provide a deeper comparison, we analytically compare the EFO and FEO frameworks under the proposed fairness criteria in terms of profit, consumer surplus, and social welfare. Our contributions can be summarized as follows:
\begin{enumerate}
%\item We propose two frameworks, Estimate-then-Fair-Optimize (EFO) and Fair-Estimate-then-Optimize (FEO), to model how fairness can be incorporated at two different stages, estimation and decision-making, of the data-driven pricing pipeline.
\item We demonstrate that under a stylized model, imposing parity-wise prediction loss fairness yields two potential solutions that induce opposite effects on consumer surplus and social welfare (Proposition~\ref{prop:ML_fairness_each}). When one of these solutions is chosen uniformly at random, both expected consumer surplus and social welfare decrease initially as the fairness level increases (Corollary~\ref{prop:ML_fairness_E}).
\item We analyze the impact of applying a little parity-wise fairness in the two frameworks. There exists a unique threshold, determined by the market sizes and price distributions, that characterizes whether EFO or FEO performs better in terms of consumer surplus and social welfare (Proposition~\ref{prop:parity_fairness}). We show that when the two groups have approximately the same market sizes and the same average prices in the training data, FEO is more beneficial to customers under parity-wise price fairness, whereas EFO is more beneficial under parity-wise demand fairness (Corollary~\ref{corollary2}).
\item We compare the seller's profit for FEO and EFO when imposing a little parity-wise fairness, and provide sufficient conditions under which FEO acts as a regularizer, improves estimation accuracy, and achieves higher profit than EFO (Proposition~\ref{prop4}).
\item Under Rawlsian fairness, EFO and FEO lead to the same outcome. Under certain conditions, fairness can increase profit, consumer surplus, and social welfare, challenging the conventional view that fairness constraints necessarily come at the cost of profit (Proposition~\ref{prop:rawlsian_fairness}).
\item We extend the two frameworks to linear and logistic demand models with personalized features. The corresponding problems are all convex optimization problems (Proposition~\ref{prop:convexity_linear} and \ref{prop:convexity_logistic}). 
% Furthermore, the training loss is convex with respect to the fairness level when imposing a little fairness (Proposition~\ref{prop:convexity_wrt_error_linear} and Corollary~\ref{prop:convexity_wrt_error_logistic})
\item We evaluate our frameworks on synthetic datasets and through a real-world case study of vaccination against tick-borne encephalitis (TBE). We define two groups based on income level, and estimate a logistic demand model using customer features such as socioeconomic status. Under parity-wise demand fairness, EFO strictly dominates FEO by delivering higher profit, consumer surplus, and social welfare. Other criteria offer no clear winner, instead requiring a policy choice to navigate the trade-off between profit and consumer surplus.%While most results align with our theoretical findings for the stylized model, feature-based demand models exhibit richer phenomena.
\end{enumerate}

\subsection{Related Literature}
Our work builds on and bridges three key areas of related literature: fairness in pricing, fairness in machine learning, and fairness in decision-making.

\paragraph{Fairness in pricing.}
There has been growing interest in developing fair pricing policies. \citet{cohen2022price} introduce formal fairness definitions in pricing and evaluate their implications in both infinite supply and monopolistic market settings. In a related vein, \citet{elmachtoub2024fair} examine fairness considerations under finite supply conditions, with a particular focus on vehicle-sharing systems. \citet{kallus2021fairness} provide a comprehensive taxonomy of fairness and welfare principles for personalized pricing across various domains, including lending, consumer goods, and public services, highlighting the inherent trade-offs among revenue, accessibility, and social equity. \citet{xu2022regulatory} propose two regulatory policy instruments: the $\epsilon$-difference constraint and the $\gamma$-ratio constraint, designed to balance consumer and producer surplus in the context of personalized pricing. Several papers investigate dynamic pricing and learning algorithms under fairness constraints \citep{cohen2025dynamic, chen2025fairness, chen2025utility}. %explore the integration of fairness into personalized contextual pricing, analyzing the trade-offs between fairness requirements and optimal revenue generation in both static and dynamic environments. 
\citet{gillis2019big} highlight that complex machine learning algorithms in pricing limit the applicability of existing law. \citet{quanexpress} provide field-experimental evidence that greater consumer data availability intensifies data-driven price discrimination, while regulatory intervention effectively alleviate its extent.

\paragraph{Fairness in Machine Learning.}
In the machine learning literature, fairness has been extensively studied at the level of predictive models. Numerous formal definitions have been proposed, including statistical parity or group fairness \citep{dwork2012fairness}, equalized loss \citep{khalili2023loss}, equal opportunity \citep{chouldechova2017fair}, equalized odds \citep{hardt2016equality}, accuracy equality \citep{berk2021fairness}, calibration \citep{kleinberg2016inherent}, and counterfactual fairness \citep{kusner2017counterfactual}.
These criteria are designed to mitigate disparate impact or disparate treatment by ensuring that predictive models do not systematically favor or disadvantage specific subpopulations. \citet{kallus2022assessing} provide bounds on disparity when protected class membership is not observed in the data.

\paragraph{Fairness in Decision-Making.}
In addition to the pricing context, a growing body of work -- particularly in operations research and operations management -- studies fairness at the level of decision-making, emphasizing fairness criteria defined over the outcomes induced by decisions. \citet{xinying2023guide} provide a comprehensive survey of how equity and fairness are formulated within optimization models, covering a wide range of decision-level criteria that integrate efficiency and fairness considerations.
In this line of work, \citet{bertsimas2011price,bertsimas2012efficiency} study resource allocation problems under a broad class of fairness objectives, including max–min and proportional fairness, characterizing the efficiency–fairness trade-off by quantifying the ``price of fairness.'' Beyond foundational formulations, related work also develops algorithms for decision-level fairness in specific operational contexts. Several works  provide fair algorithms for assortment planning \citep{chen2022fair, housni2025fairness, lu2023simple} and online matching \citep{ma2020group, ma2023fairness}. \citet{freund2023group} develop algorithms to access group fairness in dynamic refugee assignment problems. In particular, \citet{tang2023learning} develop fair policies for online resource allocation from collected data, showing how fairness can be embedded into data-driven decision pipelines at the allocation stage. \citet{tsang2026unified} develop a unified framework for convex fairness measures, providing an axiomatic characterization and a dual representation, and analyzing how different fairness measures affect the optimal value and solution in fairness-promoting optimization problems. 

A growing literature studies settings in which predictions are used to inform decisions, and examines where fairness constraints should be imposed within this prediction–decision pipeline. \citet{kleinberg2018human} distinguish prediction from decision-making, showing that improvements in predictive accuracy do not necessarily lead to better decisions or outcomes. \citet{mclaughlin2022fairness} show that excluding information about protected groups from the prediction may fail to reduce disparities in machine-assisted human decisions. \citet{corbett2017algorithmic} formulate fairness in pretrial decision-making as a constrained optimization problem that seeks to maximize utility while satisfying fairness constraints. \citet{liu2018delayed} demonstrate that focusing solely on fairness in prediction, particularly in the context of lending, can adversely affect customer utility. \citet{chohlas2024learning} propose a consequentialist framework that learns optimal decision policies by maximizing stakeholder-defined utility, thereby balancing fairness and practical outcomes through the use of contextual bandits and linear programming.

\section{Modeling Framework}
In Section~\ref{sub:demand_model}, we introduce a stylized model to analyze how fairness constraints influence pricing decisions and measures such as profit, consumer surplus, and social welfare. In Section~\ref{sub:pricing_pipeline}, we  describe the two-stage pricing pipeline comprising of demand estimation and pricing optimization. In Section~\ref{sub:fairness_criteria}, we present several fairness criteria -- prediction loss, price, and demand -- under two normative perspectives: parity-wise and Rawlsian fairness. Finally in Section~\ref{sub:FEO_vs_EFO}, we distinguish between two implementation frameworks, Fair-Estimate-then-Optimize (FEO) and Estimate-then-Fair-Optimize (EFO), which impose fairness at different stages of the pipeline.

\subsection{Demand Model}
\label{sub:demand_model}
We consider a stylized model in which a seller offers a single product, with a cost $c \geq 0$, to two distinct customer groups. The seller must determine a price for each group to maximize profit. Customers are distinguished based on an observable binary attribute $g \in \{0,1\}$, which allows the seller to set group-specific prices $p_g$. For example, $g$ may correspond to attributes such as operating system, income level, or geographic location. We assume a linear demand function,
\begin{equation*}
% d(g):=\left(1-g\right)a_0+ga_1+\left(\left(1-g\right)b_0+g b_1\right)p+\left(1-g\right)\epsilon_0+g\epsilon_1,
a_g+b_gp+\epsilon_g,
\end{equation*}
where $a_g$ is the market size of group $g$, $b_g$ is the price sensitivity of group $g$, and $\epsilon_g$ is some random noise for group $g$. We assume $\mathbb{E}[\epsilon_g]=0$ for $g\in\{0,1\}$. The variance of $\epsilon_g$, denoted by $\sigma_g^2:=\E[\epsilon_g^2]$, differs across groups. This reflects heterogeneous demand uncertainty across customer segments, which commonly arises in practice from varying measurement noise or data quality across groups. %as is common in practice, for instance due to differences in measurement noise or data quality across groups. 
We assume $b_0 < 0$ and $b_1 < 0$, implying that demand is downward sloping with respect to price.

To assess the implications of fairness criteria, we evaluate their impact using canonical economic measures: \textit{profit}, \textit{consumer surplus}, and \textit{social welfare}. The \textit{profit} with prices $(p_0,p_1)$ is
\begin{equation}
    \calR(p_0,p_1) := \sum_{g \in \{0,1\}} \left(p_g - c\right) \max\left(0,a_g+b_g p_g\right).
    \label{eq:profit}
\end{equation}
The \textit{consumer surplus} for group $g$ is then $\calS_g(p_g) = \frac{1}{2}(a_g+b_g p_g)\max\left(-\frac{a_g}{b_g}-p_g, 0\right)$,
and the \textit{total surplus} is $\calS(p_0,p_1) =  \calS_0(p_0) +  \calS_1(p_1)$. Finally, we define \textit{social welfare} as the sum of profit and consumer surplus, i.e., $\calW(p_0,p_1) := \calR(p_0,p_1) + \calS(p_0,p_1)$, which shall quantifies the overall societal impact of fairness constraints.

\subsection{Pricing Pipeline}
\label{sub:pricing_pipeline}
The pricing process typically involves two stages: learning (or estimating) a demand model, followed by pricing decisions (or optimization). Machine learning models primarily address the first stage while the outputs of these models are subsequently used in downstream decision-making processes, such as profit maximization. Our work assumes that the downstream decision-maker seeks to maximize profit using a pretrained demand model. We formulate these two stages below. 
\paragraph{Estimation Stage.} 
Suppose the machine learning team is given a dataset $\mathcal{D} := \{(g^{(i)}, p^{(i)}, d^{(i)})\}_{i=1}^{n}$, where $g^{(i)} \in \{0,1\}$ represents the group label, $p^{(i)} \in \mathbb{R}_{\geq0}$ denotes the price, and $d^{(i)} \in \mathbb{R}_{\geq0}$ indicates the demand of customer $i$ at price $p^{(i)}$. Here, we assume that the market sizes, $a_0$ and $a_1$, are known (as in \citet{qiang2016dynamic}). The machine learning model aims to estimate the demand parameter $b_g$ from the data. In this work, we adopt the classical mean squared error as the loss function. Thus, the optimal estimate ignoring fairness is obtained by
\begin{equation}
\begin{aligned}
    \left(\hat{b}_0^{\mathrm{LS}},\hat{b}_1^{\mathrm{LS}}\right)  := \arg\min_{\hat{b}_0, \hat{b}_1\leq0} \sum_{g\in\{0,1\}} \ell_g(\hat{b}_g),
\end{aligned}
\label{eq:unconstrained}
\end{equation}
where $\ell_g(\hat{b}_g):=\frac{1}{n_g}\sum_{\{i|g^{(i)}=g\}}\left(a_g+\hat{b}_gp^{(i)}-d^{(i)}\right)^2$ and $n_g:=\left|\left\{i|g^{(i)}=g\right\}\right|$ denotes the number of samples in group $g$.
\paragraph{Optimization Stage.} 
Given an estimated demand model with estimated parameters $\hat{b}_0$ and $ \hat{b}_1$, the optimization team aims to maximize the expected profit based on the demand model. More specifically, we denote by $\hat{\mathcal{R}}$ the profit function computed by the least squares estimators, i.e., 
\begin{equation*}
    \hat\calR(p_0,p_1;\hat b_0, \hat b_1) := \sum_{g \in \{0,1\}} \left(p_g - c\right) \max\left(0,a_g+\hat b_g p_g\right).
\end{equation*}
Therefore, the optimal price is set by
\begin{equation}
\begin{aligned}
    \left(p_0^*(\hat{b}_0,\hat{b}_1), p_1^*(\hat{b}_0,\hat{b}_1)\right):= \arg\max_{p_0, p_1\geq 0} \hat\calR(p_0,p_1;\hat b_0, \hat b_1).
\end{aligned}
\label{eq:price}
\end{equation}
Since the company cannot make any profit when setting a price less or equal than $c$, we have the following assumption:
\begin{assumption}
    \label{assumption:1}
    The least squares estimator satisfies $a_g+\hat{b}_g^{\mathrm{LS}}c>0$ for $g\in\{0,1\}$.
\end{assumption}

For simplicity, we use $p_g^{\mathrm{LS}}$ and $d_g^{\mathrm{LS}}$ to denote the optimal prices and corresponding estimated demand of least squares estimators, i.e., $p_g^{\mathrm{LS}}:=p_g^*\left(\hat{b}_0^{\mathrm{LS}},\hat{b}_1^{\mathrm{LS}}\right)$ and $d_g^{\mathrm{LS}}:=a_g+\hat{b}_g^{\mathrm{LS}}p_g^{\mathrm{LS}}$. Under Assumption \ref{assumption:1}, they can be written as $p_g^{\mathrm{LS}}=-\tfrac{a_g}{2\hat{b}_g^{\mathrm{LS}}}+\tfrac{c}{2}$ and $d_g^{\mathrm{LS}}=\tfrac{a_g}{2}+\tfrac{c\hat{b}_g^{\mathrm{LS}}}{2}$. %We also denote $\tilde{d}_g^{\mathrm{LS}}:=\tfrac{d_g^{\mathrm{LS}}}{a_g}$ as the normalized estimated demand for each group under LS.

\subsection{Fairness Criteria}
\label{sub:fairness_criteria}
Fairness concerns in pricing systems can arise at multiple stages of the pricing pipeline, from learning demand models to setting prices and determining market outcomes. As a result, what it means for a pricing system to be ``fair'' depends on which aspect of the system is being evaluated. In this section, we first clarify the objects of fairness considered in this paper. % before formalizing the corresponding fairness notions.}
We consider three objects of fairness:  prediction loss, price, and demand. Each captures a distinct dimension at which unfairness may arise.

First, we consider fairness in  prediction loss (or estimation error), which captures how accurately each group's demand model is learned. Because demand learning underlies subsequent pricing decisions, disparities in learning loss across groups indicate that the algorithm is systematically less accurate for some groups than for others. These learning-stage disparities shape the information on which pricing decisions are based and may therefore affect downstream pricing outcomes even when pricing policies appear similar across groups. Moreover, considering fairness in prediction loss across groups is a standard notion of fairness in machine learning.

Second, we consider price fairness. Because price reflects the pricing policy itself, differences in prices across groups are directly interpretable as differences in how the algorithm treats those groups. This interpretability makes price-based fairness particularly salient in practice and closely aligned with regulatory approaches, such as bans on particular types of  price discrimination.

Finally, we consider demand fairness. Because demand captures who participates in the market and at what level, disparities in demand across groups indicate differences in access or exclusion resulting from pricing decisions. These outcome-level disparities matter because they capture the ultimate consequences of pricing algorithms. This is particularly important in setting such as education and healthcare. 

For each type of fairness, we consider two normative notions: parity-wise fairness and Rawlsian fairness. Parity-wise fairness is a group fairness criterion that aims to equalize the corresponding metric across groups. For example, parity-wise price fairness requires prices to be (near) equal across groups. In contrast, Rawlsian fairness, also known as max-min fairness \citep{bertsimas2011price}, is an individual fairness notion that prioritizes protecting the worst-off individuals in terms of price or demand. For instance, Rawlsian demand fairness focuses on improving the demand of the group with the lowest demand.

Mathematically, learning demand models with parity-wise loss fairness can be formulated as the following constrained optimization problem:
\begin{equation*}
\begin{aligned}
    \min_{\hat{b}_0,\hat{b}_1\leq0} \quad & \ell_0(\hat{b}_0)+\ell_1(\hat{b}_1) \\
    \text{s.t.} \quad & \left|\,\ell_0(\hat{b}_0)-\ell_1(\hat{b}_1) \,\right| \leq (1 - \alpha)\left|
    \,\ell_0(\hat{b}_0^{\mathrm{LS}})-\ell_1(\hat{b}_1^{\mathrm{LS}})\, \right|.
    \label{eq:fair_opt}
\end{aligned}
\end{equation*}
The unit-less parameter $\alpha \in [0,1]$ controls the level of fairness. When $\alpha = 0$, the fairness constraint is inactive, corresponding to the status quo \eqref{eq:unconstrained} in which the optimization is fairness-agnostic. In contrast, $\alpha = 1$ enforces perfect fairness by requiring equal loss (i.e., equal error) across the two groups. Note that Rawlsian loss fairness is vacuous in our stylized model because losses across groups are uncoupled. Specifically, each group’s loss depends only on its own group parameter, and the unconstrained optimization already minimizes each group’s loss independently, leaving no scope for further improvement under a Rawlsian objective.

%The remaining fairness criteria share the same objective as in~\eqref{eq:fair_opt}, 
Table~\ref{tab:fairness} summarizes the corresponding constraints for price and demand fairness under parity-wise and Rawlsian notions. These constraints can appear either in the estimation stage or the optimization stage, as we shall discuss in the next subsection. The parameters $\bar{p}$, $\underline{p}$, $\bar{d}$, and $\underline{d}$ in Table~\ref{tab:fairness} can be chosen manually depending on the desired maximum or minimum price or demand levels. 
%For instance, in a stylized model, $\bar{p}$ (or $\bar{d}$) is the maximum optimal price (or maximum normalized demand) of two groups, while $\underline{p}$ (or $\underline{d}$) is the minimum optimal price (or minimum normalized demand) of two groups. 
For one specific fairness criterion, denote $\mathcal{F}(\alpha)$ as the set of feasible prices that satisfy this fairness criterion. For instance, under parity-wise price fairness, $\mathcal{F}(\alpha)=\{(p_0,p_1):\left|p_0-p_1\right|\leq\left(1-\alpha\right)\left|p_0^{\mathrm{LS}}-p_1^{\mathrm{LS}}\right|\}$.

\begin{table}[htbp]
    \centering
    \caption{Price and demand fairness under Rawlsian and parity-wise perspectives}
    \label{tab:fairness}
    {\begin{tabular}{c|cc}
        \toprule
        &Parity-wise & Rawlsian \\
        \midrule
        % Loss 
        % & $|\ell_0(\hat{b}_0)-\ell_1(\hat{b}_0) | \leq (1 - \alpha)\ell_0(\hat{b}_0^{\mathrm{LS}})-\ell_1(\hat{b}_0^{\mathrm{LS}})|$
        % & $-$\\[3pt]
        
        Price 
        & $\left|\;p_0-p_1\right|\leq\left(1-\alpha\right)\left|p_0^{\mathrm{LS}}-p_1^{\mathrm{LS}}\right|$
        & $p_g\leq \left(1-\alpha\right)\bar{p}+\alpha\underline{p}$\\[3pt]

        Demand
        &$|\nicefrac{d_0}{a_0}-\nicefrac{d_1}{a_1}| \leq (1 - \alpha)|\nicefrac{d_0^{\mathrm{LS}}}{a_0}-\nicefrac{d_1^{\mathrm{LS}}}{a_1}|$
        &$\nicefrac{d_g}{a_g}\geq \left(1-\alpha\right)\underline{d}+\alpha\bar{d}$\\
        \bottomrule
\end{tabular}}
    \begin{flushleft}
    \footnotesize
    \begin{center}
        \vspace{0.7em}
        \textit{Note.} $d_g$ represents the ex-ante demand for group $g$, derived from the firm’s pricing decision $p_g$, i.e., $a_g + \hat{b}_gp_g$.
    \end{center}
    \end{flushleft}
\end{table}

\subsection{Fair-Estimate-then-Optimize (FEO) vs. Estimate-then-Fair-Optimize (EFO)}
\label{sub:FEO_vs_EFO}
Equalized loss fairness is naturally imposed at the estimation stage -- where model predictions directly inform downstream decisions. In contrast, Rawlsian and parity-wise fairness can be applied either during demand estimation or at the decision stage. We formalize this distinction through two frameworks: \emph{Fair-Estimate-then-Optimize (FEO)}, which imposes fairness during estimation, and \emph{Estimate-then-Fair-Optimize (EFO)}, which enforces fairness during decision making. Mathematically, for a certain fairness criterion, FEO can be written as:
\begin{equation}
    \label{FEO}
    \begin{aligned}
        \min_{\hat{b}_0, \hat{b}_1\leq0} \quad&\sum_{g\in\{0,1\}}\ell_g(\hat{b}_g)\qquad \text{s.t.}~ \left(p_0^*(\hat{b}_0,\hat{b}_1), p_1^*(\hat{b}_0,\hat{b}_1)\right)\in \mathcal{F}(\alpha).
        %\text{s.t.}\quad& \left(p_0^*(\hat{b}_0,\hat{b}_1), p_1^*(\hat{b}_0,\hat{b}_1)\right)\in \mathcal{F}(\alpha).
    \end{aligned}
\end{equation} 
We replace the unconstrained problem \eqref{eq:unconstrained} with the fairness-constrained formulation \eqref{FEO}, which can be interpreted as a bilevel optimization problem. In this formulation, the upper-level problem selects the estimators $(\hat{b}_0, \hat{b}_1)$ by minimizing the estimation loss $\ell_g(\hat{b}_g)$, while the lower-level problem determines the induced optimal prices $p_g^*$ as a function of $(\hat{b}_0, \hat{b}_1)$ by maximizing the profit in~\eqref{eq:price}. This internalizes fairness constraints within the estimation stage, ensuring that the resulting decisions satisfy the desired fairness criteria even when the downstream decision stage is purely profit-maximizing. We denote $\left(\hat{b}_0^{\mathrm{FEO}},\hat{b}_1^{\mathrm{FEO}}\right)$ as the optimal solution of \eqref{FEO} and $p_g^{\mathrm{FEO}}:=p_g^*(\hat{b}_0^{\mathrm{FEO}},\hat{b}_1^{\mathrm{FEO}})$.
%Here, the decision will satisfy a certain level of fairness even if the downstream decision-maker does not consider fairness and only maximizes the expected profit. We denote $\left(\hat{b}_0^{\mathrm{FEO}},\hat{b}_1^{\mathrm{FEO}}\right)$ as the optimal solution of \eqref{FEO} and $p_g^{\mathrm{FEO}}:=p_g^*(\hat{b}_0^{\mathrm{FEO}},\hat{b}_1^{\mathrm{FEO}})$. 

On the other hand, EFO can be written as optimization problem \eqref{EFO}.
\begin{equation}
    \label{EFO}
    \begin{aligned}
        \max_{p_0, p_1\geq 0} \quad&\hat{\mathcal{R}}(p_0,p_1;\hat b_0^{\mathrm{LS}}, \hat b_1^{\mathrm{LS}})\qquad \text{s.t.}~ \left(p_0, p_1\right)\in \mathcal{F}(\alpha).
        % \max_{p_0, p_1\geq 0} \quad&\hat{\mathcal{R}}(p_0,p_1;\hat b_0^{\mathrm{LS}}, \hat b_1^{\mathrm{LS}})\\
        % \text{s.t.}\quad& \left(p_0, p_1\right)\in \mathcal{F}(\alpha).
    \end{aligned}
\end{equation}
Here, we utilize the estimates $(\hat{b}_0^{\text{LS}}, \hat{b}_1^{\text{LS}})$ to solve a profit-maximization problem under fairness constraints. This ensures the resulting prices are optimal given the estimated parameters while satisfying the fairness criterion. We denote $\left(p_0^{\mathrm{EFO}}, p_1^{\mathrm{EFO}}\right)$ as the optimal solution of \eqref{EFO}. 

The detailed comparison between the two frameworks, FEO and EFO, will be discussed in Section \ref{sec:theory}. It is worth noting that, from a practical standpoint, there are clear distinctions between the two frameworks in terms of flexibility and the scope of fairness enforcement. Under FEO, fairness is incorporated at the estimation stage, leaving the downstream optimization stage untouched. This is particularly convenient when the optimization engine may be a complex black box or when sales agents are only driven by profits. Once the demand model is estimated under a given fairness constraint, the decision maker can reuse it without having to further account for fairness. In contrast, EFO enforces fairness only at the optimization (decision) stage, allowing the decision maker greater flexibility to adjust fairness levels to specific contexts or objectives without retraining the demand model. Moreover, the demand model is unbiased and may yield higher profit. Beyond these differences, we will theoretically study the impact on how these strategies affect consumer surplus and social welfare in Section \ref{sec:theory}.

\section{Theoretical Analysis}
\label{sec:theory}
In this section, we analytically compare two frameworks, Fair-Estimate-then-Optimize (FEO) and Estimate-then-Fair-Optimize (EFO), through the lens of Rawlsian and parity-wise perspectives.

\subsection{Parity-wise Fairness}
\paragraph{Prediction Loss Fairness.} Within parity-wise fairness, we first consider loss fairness. Since the mean squared error (MSE) is a quadratic function, enforcing parity-wise loss fairness with a quadratic objective may result in two optimal solutions. Proposition \ref{prop:ML_fairness_each} shows that the choice between these two optimal solutions has a distinct impact on downstream outcomes such as total surplus and social welfare, even with infinite data. Notably, one may be considered more effective than the other. All proofs are provided in Appendix~\ref{appendix_proof}.

\begin{proposition}
    Without loss of generality, assume $\sigma_1^2 >\sigma_0^2$ and $n_g \rightarrow \infty$ for all $g \in \{0,1\}$. For $\alpha \leq \alpha^{\star}:=\left(b_0\right)^2\tfrac{\E[p^2 \mid g=0]}{\sigma_1^2 - \sigma_0^2}$, there are two optimal solutions: $\big(b_0^{-}(\alpha), b_1\big)$ and $\big(b_0^{+}(\alpha), b_1\big)$.
    \begin{itemize}
    \item Under $\big(b_0^{-}(\alpha), b_1\big)$, both consumer surplus and social welfare increase monotonically with $\alpha$.  
    \item Under $\big(b_0^{+}(\alpha), b_1\big)$, both consumer surplus and social welfare decrease monotonically with $\alpha$.  %for $\alpha \leq \alpha^{\star} \!\left( \tfrac{a_0 + b_0c}{2a_0 + b_0 c} \right)^2\!$, consumer surplus and social welfare decreases. For the remaining range of $\alpha$, consumer surplus and social welfare are zero.
    \item Profit with $\big(b_0^{+}(\alpha), b_1\big)$ decreases more sharply than under $\big(b_0^{-}(\alpha), b_1\big)$ for all $\alpha$ up to a certain threshold $\bar{\alpha}$. In the remaining range, the profit under at least one of the cases, $b_0^{+}$ or $b_0^{-}$, becomes zero.
    \end{itemize}
    For $\alpha > \alpha^{\star}$, the unique solution is $\big(b_0^{-}(\alpha), b_1\big)$, and the changes in profit, consumer surplus, and social welfare match those for $\alpha < \alpha^{\star}$.
    \label{prop:ML_fairness_each}
\end{proposition}

% The first plot in Figure~\ref{fig:ML_fairness} 
Figure~\ref{fig:EL_fairness_1} illustrates the statements of Proposition~\ref{prop:ML_fairness_each}. With respect to metrics such as profit, consumer surplus, and social welfare, the solution $b^{-}(\alpha)$ is preferable to $b^{+}(\alpha)$. In practice, however, when an optimization problem admits multiple optimal solutions, solvers -- such as Gurobi \citep{gurobi} or Pyomo \citep{bynum2021pyomo} -- return a single optimal solution, selected according to internal heuristics and default solver behavior. In the absence of specific user configurations (e.g., enabling solution pools, setting deterministic options, or applying post-processing), no attempt is made to identify or enumerate alternative optimal solutions.

\begin{figure}[htbp]
\FIGURE{
\begin{minipage}{\textwidth}
\centering
\captionsetup{justification=centering}
\begin{subfigure}{0.3\textwidth}
    \centering
    \includegraphics[width=\textwidth]{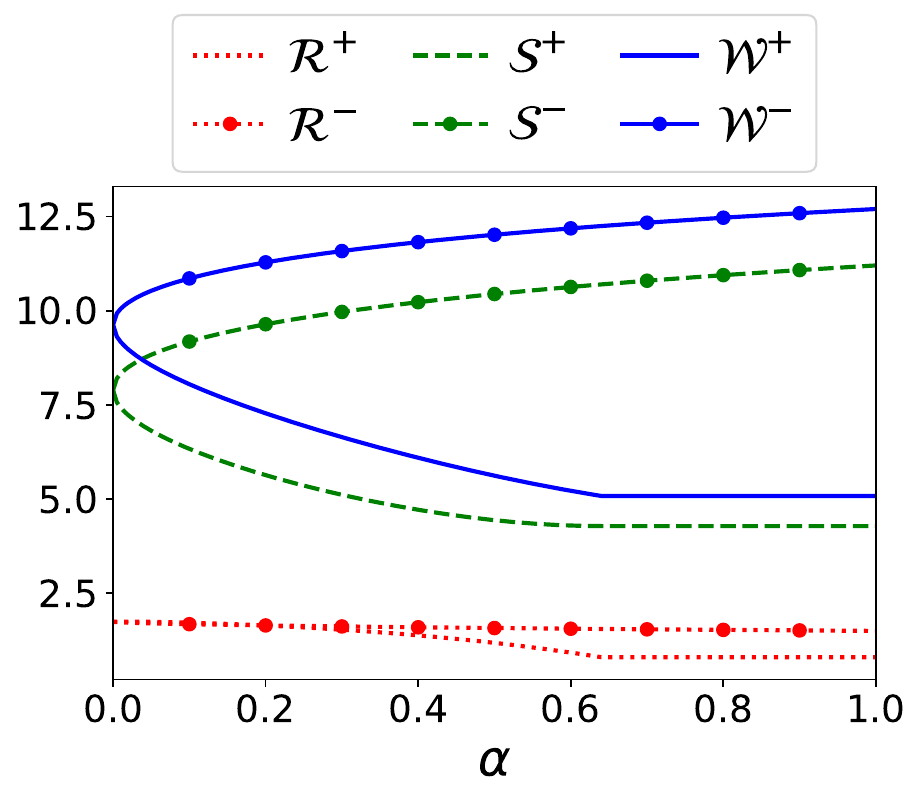}
    \caption{FEO}
    \label{fig:EL_fairness_1}
\end{subfigure}
\begin{subfigure}{0.3\textwidth}
    \centering
    \includegraphics[width=\textwidth]{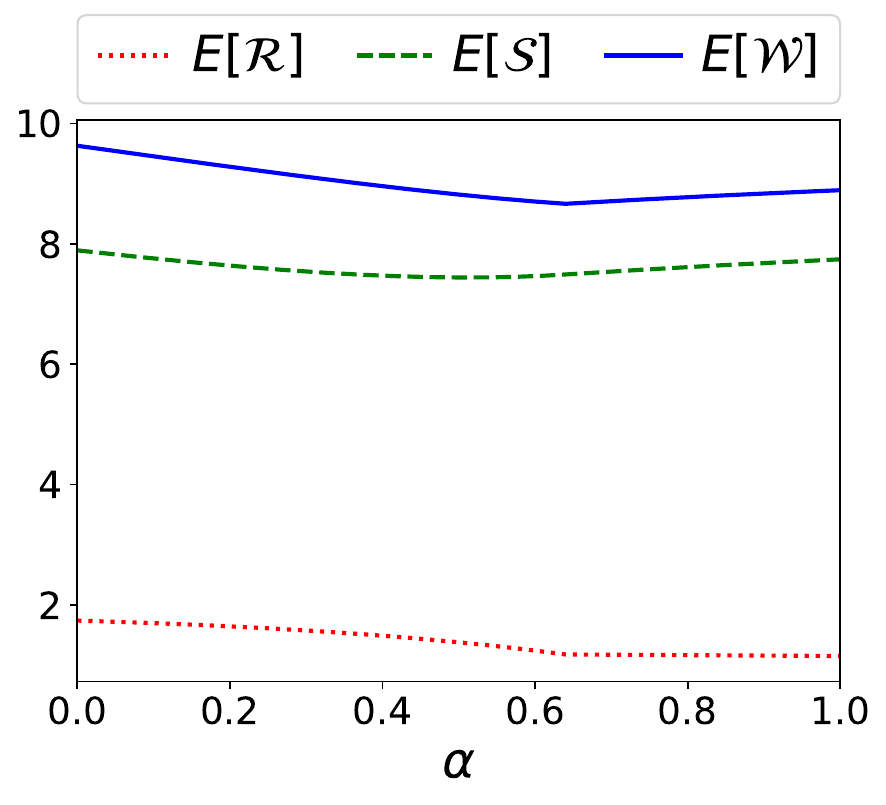}
    \caption{EFO}
    \label{fig:EL_fairness_2}
\end{subfigure}
\end{minipage}
}
{Results under Parity-wise Loss Fairness \label{fig:ML_fairness} \vspace{.3cm}}
{Figures \ref{fig:EL_fairness_1} and \ref{fig:EL_fairness_2} correspond to Proposition~\ref{prop:ML_fairness_each} 
and Corollary~\ref{prop:ML_fairness_E}, respectively. 
% The third plot illustrates the price gap, and the fourth plot depicts the demand gap. 
The parameters are 
$a_0 = 3.8$, $a_1 = 4$, 
$b_0 = -3.2$, $b_1 = -4$, 
$\sigma_0 = 0.1$, $\sigma_1 = 0.5$, 
and $c = 0.1$.}
\end{figure}

Reflecting this scenario, we now consider the case in which the solver selects among optimal solutions uniformly at random. Corollary \ref{prop:ML_fairness_E} demonstrates that, under these conditions, the expected value of relevant metrics -- consumer surplus and social welfare -- decline up to a certain point.

\begin{corollary}
\label{prop:ML_fairness_E}
Consider the case where $b^{-}(\alpha)$ and $b^+(\alpha)$ each are returned with probability $\nicefrac{1}{2}$. The expected consumer surplus and social welfare decrease with $\alpha$ when $\alpha\le\bar\alpha$. %is less than threshold $\bar\alpha$.
\end{corollary}

Corollary~\ref{prop:ML_fairness_E} illustrates that ignoring how an estimator affects downstream metrics can lead to undesirable outcomes. The second plot in Figure~\ref{fig:EL_fairness_2} demonstrates this result: up to a certain threshold $\bar\alpha$, the expected profit, consumer surplus, and social welfare all decrease as $\alpha$ increases.

\paragraph{Price and Demand Fairness.} Unlike loss fairness, fairness in prices and demand admits the two frameworks: FEO and EFO. We focus on the non-trivial case where $p_0^{\mathrm{LS}} \neq p_1^{\mathrm{LS}}$. If the unconstrained prices were already equal, adding price or demand fairness constraints would be unnecessary. 
%Without loss of generality, we have $p_0^{\mathrm{LS}} \neq p_1^{\mathrm{LS}}$, or otherwise there is no need to add price or demand fairness constraints. 
We restrict our attention to the setting in which only a small amount of fairness is imposed, meaning the fairness level $\alpha$ increases from $0$ to a small positive value $\epsilon$. 
This focus is motivated by the fact that $\alpha = 0$ corresponds to the prevailing practice in which firms typically do not incorporate fairness considerations. 
In other words, $\alpha = 0$ represents the status quo. Therefore, our main interest lies in understanding how introducing a small degree of fairness -- increasing $\alpha$ slightly from zero -- affects the resulting outcomes under this realistic baseline setting. Moreover, focusing on the neighborhood of $\alpha=0$ allows tractable analytical characterization with the status quo regime. Since closed-form comparisons are generally difficult to obtain for arbitrary values of $\alpha$, this difficulty motivates our local comparative analysis around $\alpha=0$. By Assumption \ref{assumption:1}, we can also assume that when there is a small amount of fairness applied, the optimal prices of FEO and EFO are still above $c$.

Proposition~\ref{prop3} shows that, under a little fairness, either FEO or EFO can yield higher consumer surplus and social welfare, depending on the specific fairness criterion and a threshold determined by the ratio of the estimated price distribution.
\begin{proposition}
\label{prop3}
Suppose Assumption~\ref{assumption:1} holds. Define
\begin{equation*}\label{tau}
\tau_p := \left(\frac{a_0^2 \, \frac{1}{n_1}\sum_{\{i:g^{(i)}=1\}} p^{(i)^2}} {a_1^2 \, \frac{1}{n_0}\sum_{\{i:g^{(i)}=0\}} p^{(i)^2}}\right)^{1/3},\quad \tau_a :=\frac{a_0}{a_1},\quad \tau_b :=\frac{\hat b_0^{\mathrm{LS}}}{\hat  b_1^{\mathrm{LS}}}.
\end{equation*}
Then there exists $\xi>0$ such that the following statements hold for all $\alpha\in(0,\xi)$.
\begin{enumerate}
\item Consider parity-wise price fairness.
\begin{enumerate}
    \item If $(\tau_p-\tau_b)(\tau_a-\tau_b)>0$, then $p_0^{\mathrm{FEO}}<p_0^{\mathrm{EFO}}$ and $p_1^{\mathrm{FEO}}<p_1^{\mathrm{EFO}}$. Consequently, $\mathcal S^{\mathrm{FEO}}> \mathcal S^{\mathrm{EFO}}$ and $\mathcal W^{\mathrm{FEO}}> \mathcal W^{\mathrm{EFO}}$. \label{prop2:case1a}
    \item If $(\tau_p-\tau_b)(\tau_a-\tau_b)\leq 0$, then  $p_0^{\mathrm{FEO}}>p_0^{\mathrm{EFO}}$ and $p_1^{\mathrm{FEO}}>p_1^{\mathrm{EFO}}$. Consequently, $\mathcal S^{\mathrm{FEO}}< \mathcal S^{\mathrm{EFO}}$ and $\mathcal W^{\mathrm{FEO}}< \mathcal W^{\mathrm{EFO}}$.
    \label{prop2:case1b}
\end{enumerate}
\item Consider parity-wise demand fairness and assume $c>0$.
\begin{enumerate}
    \item If $(\tau_p-\tau_b)(\tau_a-\tau_b)\geq0$, then $p_0^{\mathrm{FEO}}>p_0^{\mathrm{EFO}}$ and $p_1^{\mathrm{FEO}}>p_1^{\mathrm{EFO}}$. Consequently, $\mathcal S^{\mathrm{FEO}}< \mathcal S^{\mathrm{EFO}}$ and $\mathcal W^{\mathrm{FEO}}< \mathcal W^{\mathrm{EFO}}$. \label{prop2:case2a}
    \item If $(\tau_p-\tau_b)(\tau_a-\tau_b)<0$, then  $p_0^{\mathrm{FEO}}<p_0^{\mathrm{EFO}}$ and $p_1^{\mathrm{FEO}}<p_1^{\mathrm{EFO}}$. Consequently, $\mathcal S^{\mathrm{FEO}}> \mathcal S^{\mathrm{EFO}}$ and $\mathcal W^{\mathrm{FEO}}> \mathcal W^{\mathrm{EFO}}$. \label{prop2:case2b}
\end{enumerate}
\end{enumerate}
\label{prop:parity_fairness}
\end{proposition}

\begin{figure}[htbp]
\FIGURE{
\begin{minipage}{\textwidth}
\centering
\captionsetup{justification=centering}
\begin{subfigure}{0.24\textwidth}
    \centering
    \includegraphics[width=\textwidth]{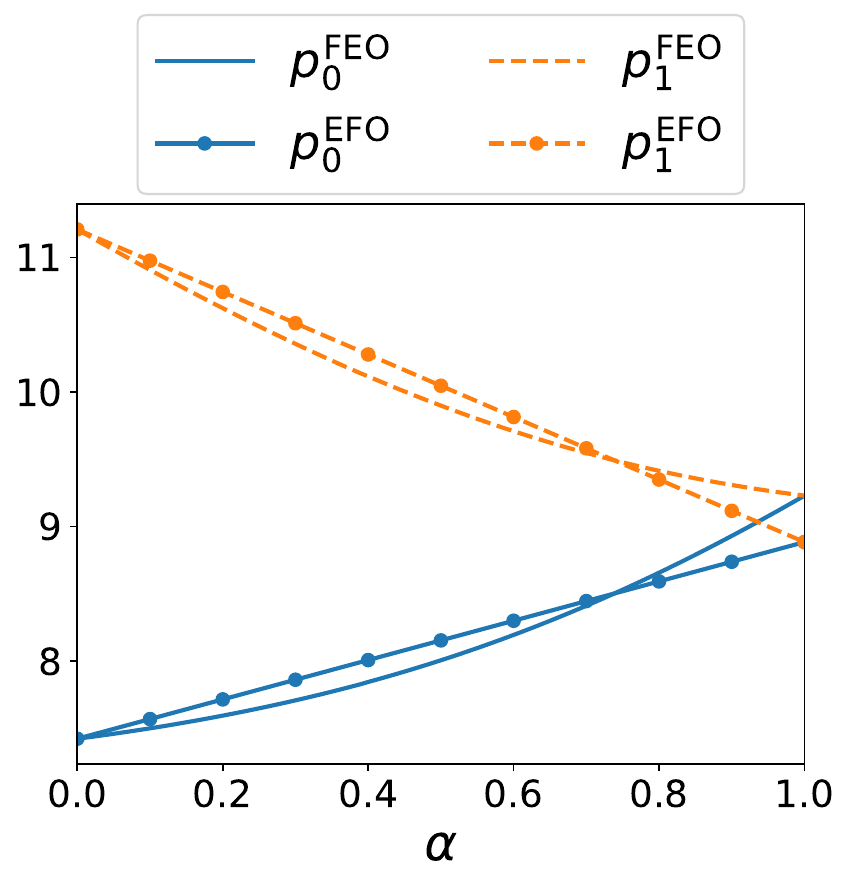}
    \caption*{Case \eqref{prop2:case1a}: Price}
    \label{fig:case1a_prices}
\end{subfigure}
\begin{subfigure}{0.24\textwidth}
    \centering
    \includegraphics[width=\textwidth]{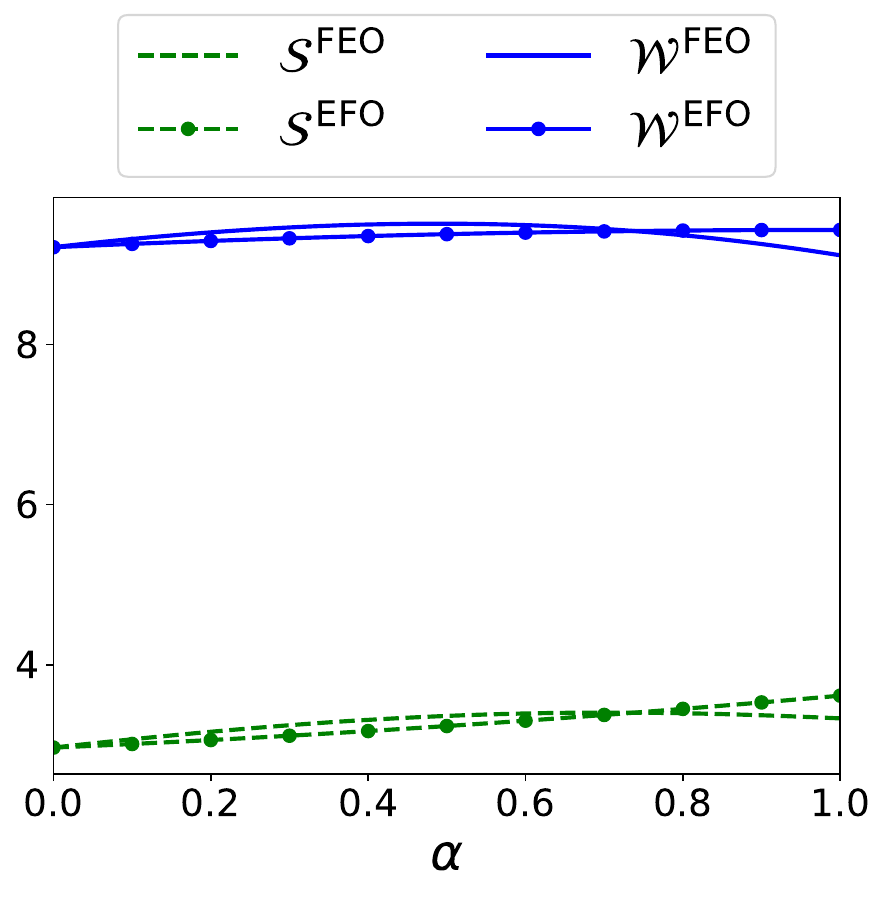}
    \caption*{Case \eqref{prop2:case1a}: Outcomes}
    \label{fig:case1a_outcomes}
\end{subfigure}
\begin{subfigure}{0.24\textwidth}
    \centering
    \includegraphics[width=\textwidth]{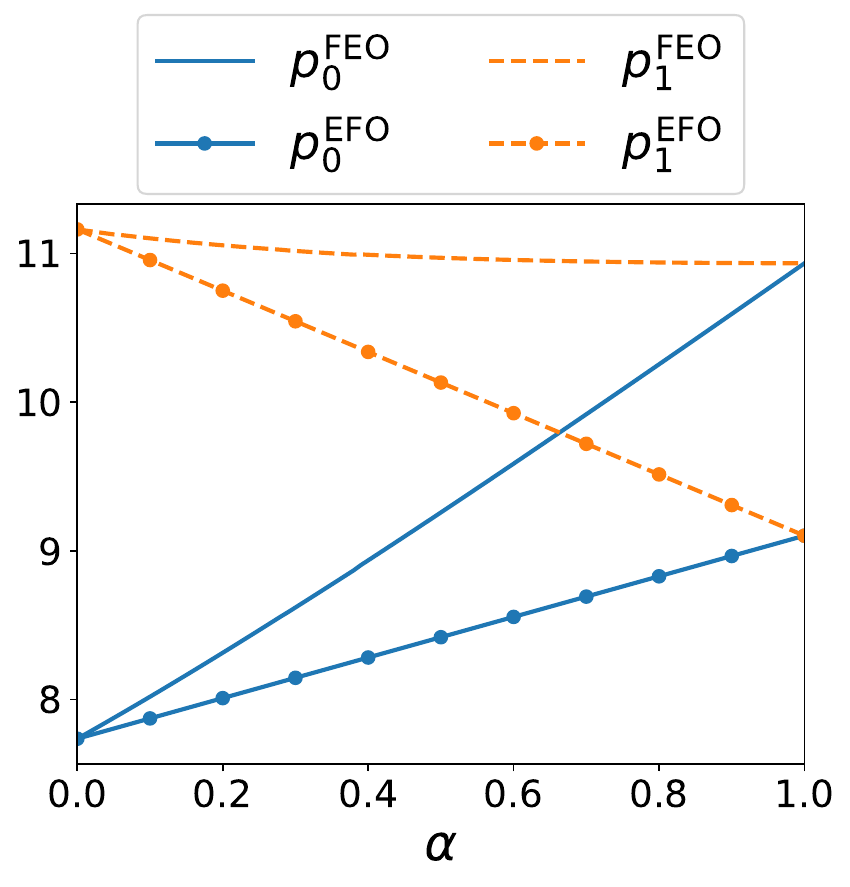}
    \caption*{Case \eqref{prop2:case1b}: Price}
    \label{fig:case1b_prices}
\end{subfigure}
\begin{subfigure}{0.24\textwidth}
    \centering
    \includegraphics[width=\textwidth]{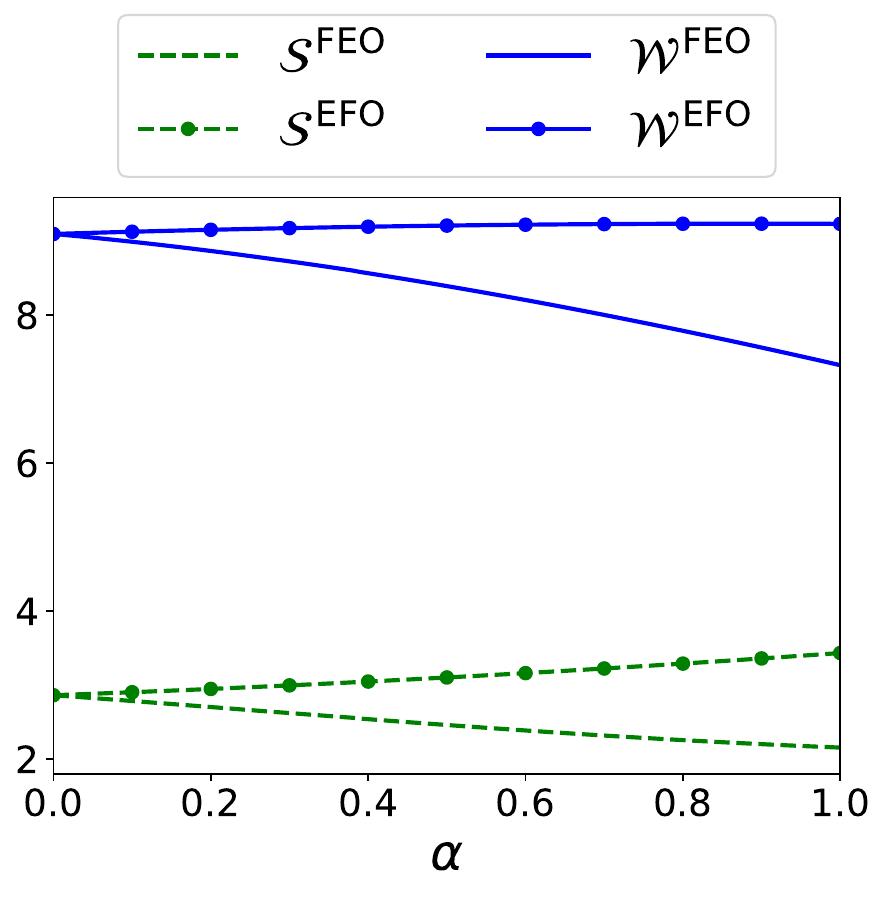}
    \caption*{Case \eqref{prop2:case1b}: Outcomes}
    \label{fig:case1b_outcomes}
\end{subfigure}
\end{minipage}
}
{Different cases in Proposition \ref{prop3}: parity-wise price fairness \label{fig:different_cases_price}\vspace{.3cm}}
{The parameters are $a_0=1$, $a_1=1$, $b_0=-0.08$, $b_1=-0.05$ and $c=2$. For Case \eqref{prop2:case1a}, prices of group $0$ are generated from $\mathrm{Unif}[4,8]$ and prices of group $0$ are generated from $\mathrm{Unif}[4,10]$; for Case \eqref{prop2:case1b}, prices of group $0$ are generated from $\mathrm{Unif}[2,3]$ and prices of group $0$ are generated from $\mathrm{Unif}[11,12]$.}
\end{figure}

\begin{figure}[htbp]
\FIGURE{
\begin{minipage}{\textwidth}
\centering
\captionsetup{justification=centering}
\begin{subfigure}{0.24\textwidth}
    \centering
    \includegraphics[width=\textwidth]{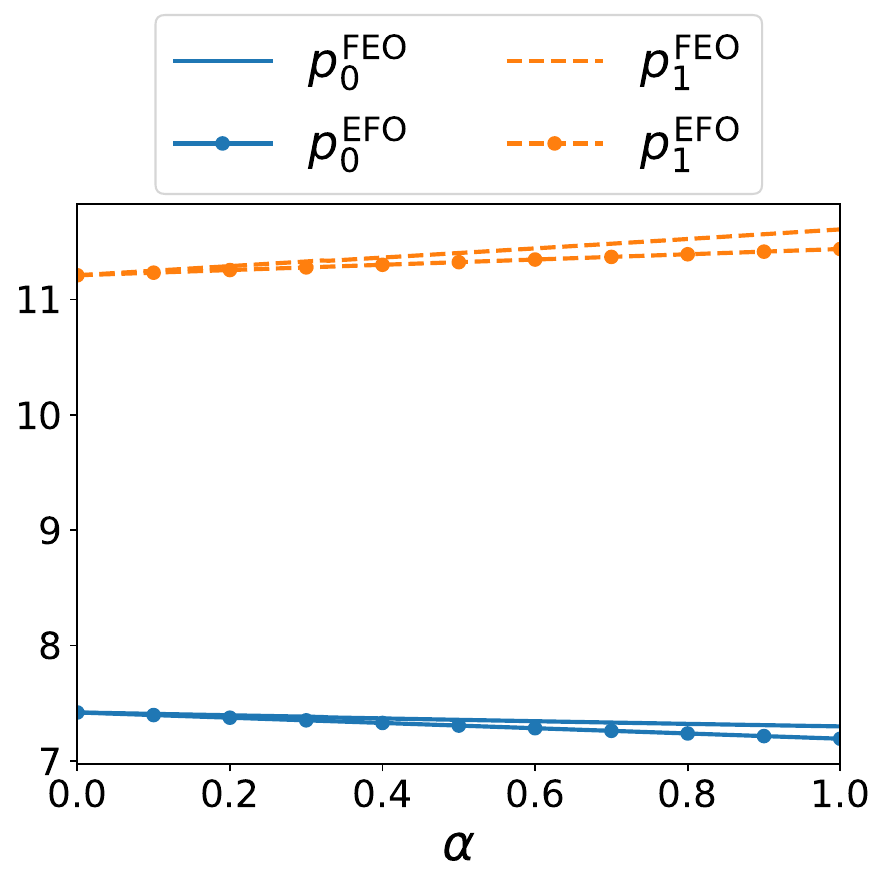}
    \caption*{Case \eqref{prop2:case2a}: Price}
    \label{fig:case2a_prices}
\end{subfigure}
\begin{subfigure}{0.24\textwidth}
    \centering
    \includegraphics[width=\textwidth]{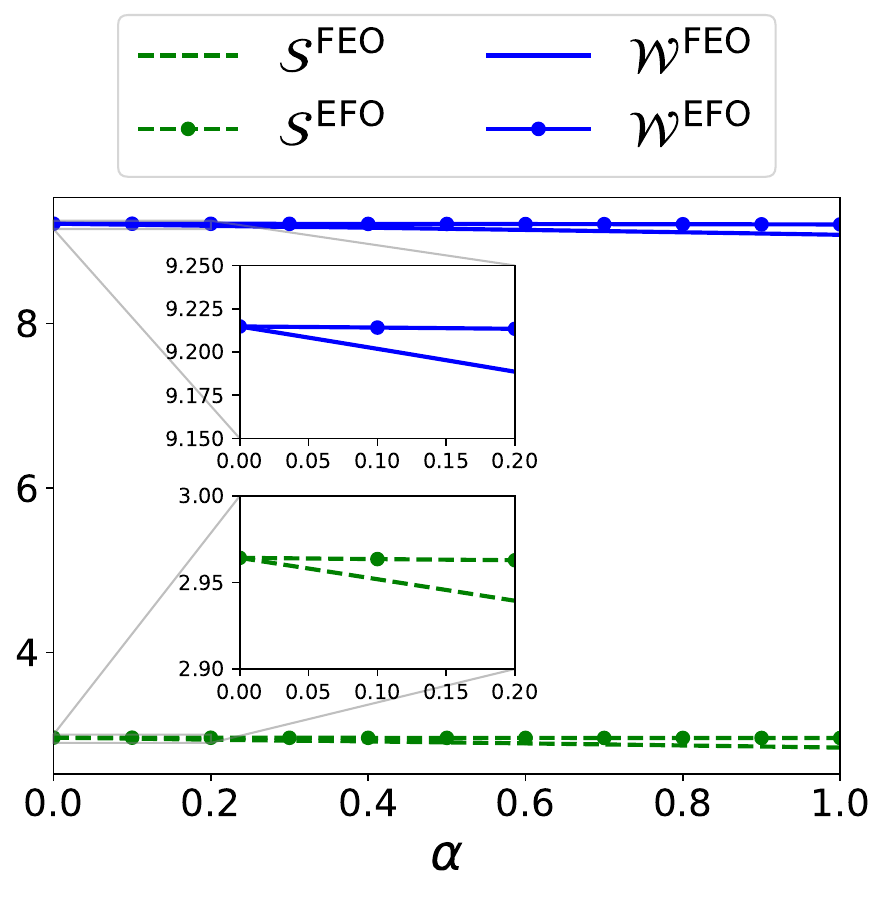}
    \caption*{Case \eqref{prop2:case2a}: Outcomes}
    \label{fig:case2a_outcomes}
\end{subfigure}
\begin{subfigure}{0.24\textwidth}
    \centering
    \includegraphics[width=\textwidth]{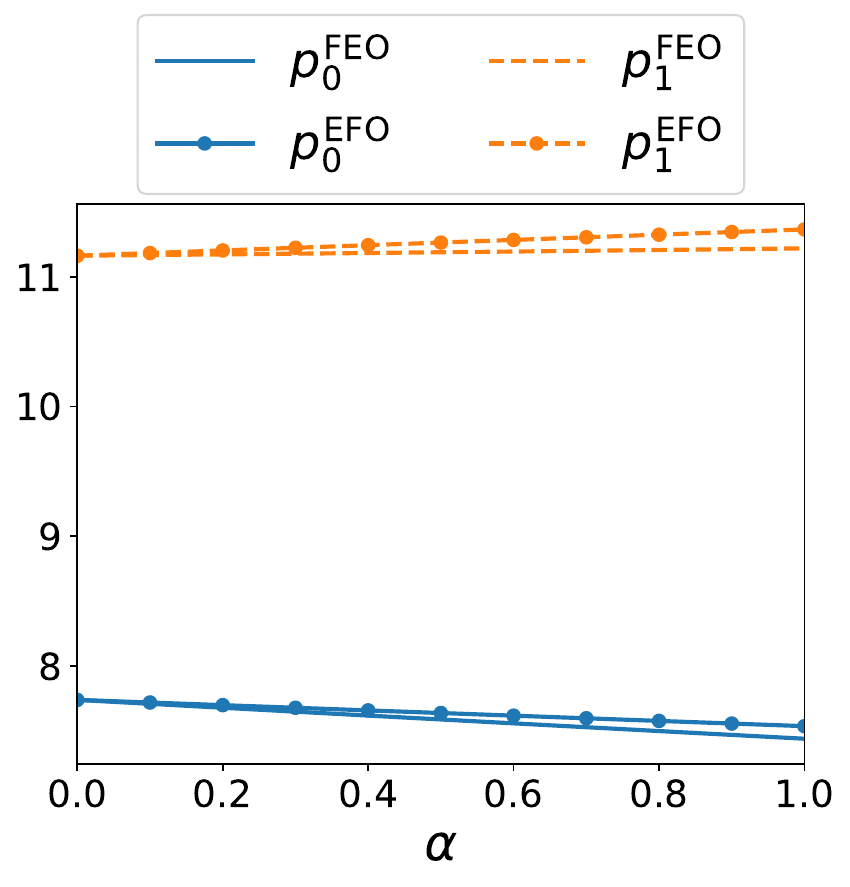}
    \caption*{Case \eqref{prop2:case2b}: Price}
    \label{fig:case2b_prices}
\end{subfigure}
\begin{subfigure}{0.24\textwidth}
    \centering
    \includegraphics[width=\textwidth]{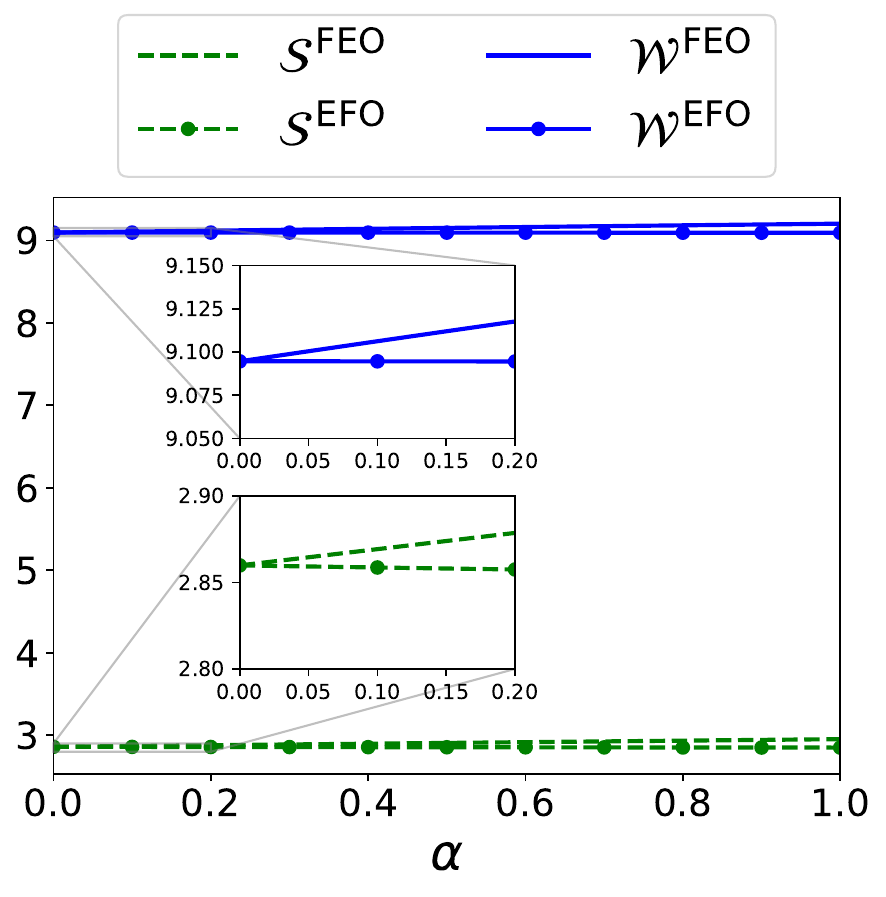}
    \caption*{Case \eqref{prop2:case2b}: Outcomes}
    \label{fig:case2b_outcomes}
\end{subfigure}
\end{minipage}
}
{Different cases in Proposition \ref{prop3}: parity-wise demand fairness \label{fig:different_cases_demand}\vspace{.3cm}}
{The parameters are the same as Figure~\ref{fig:different_cases_price}. Training data used for Case \eqref{prop2:case2a} and Case \eqref{prop2:case2b} are the same as Case \eqref{prop2:case1a} and Case \eqref{prop2:case1b}, respectively.}
\end{figure}

Figure~\ref{fig:different_cases_price} and Figure~\ref{fig:different_cases_demand} visualize the results for each case. To compare consumer surplus and social welfare between FEO and EFO, we only have to compare the price decisions. Lower prices always leave more surplus to consumers, and it always brings in more consumers. Since each consumer contributes positively to social welfare, lower prices also lead to higher overall welfare. 

$\tau_p$ can be interpreted as a balance index of the historical pricing data. It is determined by the average squared prices of each group in the training dataset and the market sizes $a_g$. To illustrate its dependence on these two components, we consider two simple cases. First, assuming equal market sizes ($a_0 = a_1$), the loss function simplifies to
$$\ell_0(\hat b_0)+\ell_1( \hat b_1) \propto (\hat b_0-\hat b_0^{\mathrm{LS}})^2 +  \tau_p^3(\hat b_1-\hat b_1^{\mathrm{LS}})^2 +C,$$
where $C$ is a constant. This formulation reveals that $\tau_p$ dictates the relative weight of the two groups: a smaller $\tau_p$ assigns greater penalty weight to group $0$'s fit. Consequently, under price fairness, a small $\tau_p$ restricts FEO from increasing $p_0$ as rapidly as EFO (corresponding to Case \ref{prop2:case1a}). Since the model preserves a price gap, $p_1$ is similarly lower under FEO, which drives up overall consumer surplus under FEO.%As a result, under price fairness, when $\tau_p$ is small (corresponding to Case~\ref{prop2:case1a}), FEO discourages increases in $p_0$. In other words, $p_0$ under FEO increases more slowly than under EFO. The same reasoning applies to $p_1$, and thus FEO yields higher consumer surplus. 
%This interpretation suggests that $\tau_p$ captures the relative weighting of the two groups in the loss function. 

Second, suppose the average squared prices in the training data are equal across both groups, i.e., $\frac{1}{n_1}\sum_{\{i:g^{(i)}=1\}} p^{(i)^2}=\frac{1}{n_0}\sum_{\{i:g^{(i)}=0\}} p^{(i)^2}$. The index reduces to $\tau_p=\tau_a^{2/3}$, implying that $\tau_p$ strictly increases with the market size ratio $\tau_a$. Therefore, the terms $(\tau_p-\tau_b)$ and $(\tau_a-\tau_b)$ are more likely to share the same sign, making the surplus-enhancing outcomes of Cases \eqref{prop2:case1a} and \eqref{prop2:case2a} in Proposition \ref{prop3} much more likely to occur.
%Second, suppose $\frac{1}{n_1}\sum_{\{i:g^{(i)}=1\}} p^{(i)^2}=\frac{1}{n_0}\sum_{\{i:g^{(i)}=0\}} p^{(i)^2}$. Then $\tau_p=\tau_a^{2/3}$, implying that $\tau_p$ varies in the same direction as $\tau_a$.

One practical scenario is that two different groups have similar market sizes, firms charge similar prices to all customers, and data are collected randomly across groups. Corollary \ref{corollary2} shows that if the market sizes and average squared prices in the training dataset are the same between two groups, Cases \eqref{prop2:case1a} and \eqref{prop2:case2a} in Proposition \ref{prop3} apply. In these cases, FEO yields higher consumer surplus and social welfare than EFO under parity-wise price fairness, while the opposite holds under parity-wise demand fairness.

\begin{corollary}
    \label{corollary2}
    Suppose Assumption~\ref{assumption:1} holds. If $\frac{1}{n_0}\sum_{\{i|g^{(i)}=0\}}p^{(i)^2}=\frac{1}{n_1}\sum_{\{i|g^{(i)}=1\}}p^{(i)^2}$ and $a_0=a_1$, then for parity-wise price fairness, $\calS^{\mathrm{FEO}}>\calS^{\mathrm{EFO}}$ and $\calW^{\mathrm{FEO}}>\calW^{\mathrm{EFO}}$.
    For parity-wise demand fairness, $\calS^{\mathrm{FEO}}<\calS^{\mathrm{EFO}}$ and $\calW^{\mathrm{FEO}}<\calW^{\mathrm{EFO}}.$
\end{corollary}

Notice that, under the same condition, parity-wise price fairness and parity-wise demand fairness lead to opposite results. This is because price fairness constraints and demand fairness constraints are pushing prices to different directions. For the group with higher price, adding price fairness constraints decreases the price of this group. However, at the same time, it has larger demand, so adding demand fairness constraints increases the price of this group to decrease the demand.

Comparing the profit of FEO and EFO also depends on different cases and fairness criteria. The objective of EFO is to maximize the estimated profit under certain fairness constraints, but there is a gap between the estimated profit and the real expected profit, since we can only estimate unknown parameters $(b_0, b_1)$ from historical data. Therefore, we are able to determine several cases where the profit of FEO is larger than that of EFO in Proposition \ref{prop4}.

\begin{proposition}
    \label{prop4}
    Suppose Assumption~\ref{assumption:1} holds. Define $\tau_p$, $ \tau_a$ and $\tau_b$ the same as in Proposition~\ref{prop3}. There exists $\xi>0$ such that the following statements hold for all $\alpha\in (0,\xi)$.
    \begin{enumerate}
        \item Consider parity-wise price fairness.
        \begin{enumerate}
            \item If $(\tau_p-\tau_b)(\tau_a-\tau_b)>0$, then $\mathcal{R}^{\mathrm{FEO}}>\mathcal{R}^{\mathrm{EFO}}$ when $a_0\Big(1-\tfrac{b_0}{\hat{b}_0^{\mathrm{LS}}}\Big)+a_1\Big(1-\tfrac{b_1}{\hat{b}_1^{\mathrm{LS}}}\Big)<0$ and $\mathcal{R}^{\mathrm{FEO}}<\mathcal{R}^{\mathrm{EFO}}$ when $a_0\Big(1-\tfrac{b_0}{\hat{b}_0^{\mathrm{LS}}}\Big)+a_1\Big(1-\tfrac{b_1}{\hat{b}_1^{\mathrm{LS}}}\Big)>0$.
            % \begin{equation*}
            %     a_0\left(1-\frac{b_0}{\hat{b}_0^{\mathrm{LS}}}\right)+a_1\left(1-\frac{b_1}{\hat{b}_1^{\mathrm{LS}}}\right)<0.
            % \end{equation*}
            \item If $(\tau_p-\tau_b)(\tau_a-\tau_b)\leq0$, then $\mathcal{R}^{\mathrm{FEO}}>\mathcal{R}^{\mathrm{EFO}}$ when $a_0\Big(1-\tfrac{b_0}{\hat{b}_0^{\mathrm{LS}}}\Big)+a_1\Big(1-\tfrac{b_1}{\hat{b}_1^{\mathrm{LS}}}\Big)>0$ and $\mathcal{R}^{\mathrm{FEO}}<\mathcal{R}^{\mathrm{EFO}}$ when $a_0\Big(1-\tfrac{b_0}{\hat{b}_0^{\mathrm{LS}}}\Big)+a_1\Big(1-\tfrac{b_1}{\hat{b}_1^{\mathrm{LS}}}\Big)<0$.
        \end{enumerate}
        \item Consider parity-wise price fairness.
        \begin{enumerate}
            \item If $(\tau_p-\tau_b)(\tau_a-\tau_b)\geq0$, then $\mathcal{R}^{\mathrm{FEO}}>\mathcal{R}^{\mathrm{EFO}}$ when $a_0\Big(1-\tfrac{b_0}{\hat{b}_0^{\mathrm{LS}}}\Big)+a_1\Big(1-\tfrac{b_1}{\hat{b}_1^{\mathrm{LS}}}\Big)>0$ and $\mathcal{R}^{\mathrm{FEO}}<\mathcal{R}^{\mathrm{EFO}}$ when $a_0\Big(1-\tfrac{b_0}{\hat{b}_0^{\mathrm{LS}}}\Big)+a_1\Big(1-\tfrac{b_1}{\hat{b}_1^{\mathrm{LS}}}\Big)<0$.
            \item If $(\tau_p-\tau_b)(\tau_a-\tau_b)<0$, then $\mathcal{R}^{\mathrm{FEO}}>\mathcal{R}^{\mathrm{EFO}}$ when $a_0\Big(1-\tfrac{b_0}{\hat{b}_0^{\mathrm{LS}}}\Big)+a_1\Big(1-\tfrac{b_1}{\hat{b}_1^{\mathrm{LS}}}\Big)<0$ and $\mathcal{R}^{\mathrm{FEO}}<\mathcal{R}^{\mathrm{EFO}}$ when $a_0\Big(1-\tfrac{b_0}{\hat{b}_0^{\mathrm{LS}}}\Big)+a_1\Big(1-\tfrac{b_1}{\hat{b}_1^{\mathrm{LS}}}\Big)>0$.
        \end{enumerate}
    \end{enumerate}
\end{proposition}

% \begin{proposition}
%     \label{prop4}
%     % Assumption 1?
%     Suppose Assumption~\ref{assumption:1} holds. Without loss of generality, suppose $p_0^{\mathrm{LS}}<p_1^{\mathrm{LS}}$. If
%     \begin{equation}
%     \label{eq:4}
%         \left[a_0\left(1-\frac{b_0}{\hat{b}_0^{\mathrm{LS}}}\right)+a_1\left(1-\frac{b_1}{\hat{b}_1^{\mathrm{LS}}}\right)\right]\left[\hat{b}_0^{\mathrm{LS}}-\tau\hat{b}_1^{\mathrm{LS}}\right]>0,
%     \end{equation}
%     then there exists $\xi>0$ such that, whenever the level of price fairness constraints $\alpha$ satisfies $0<\alpha<\xi$, we have $\calR^{FEO}>\calR^{EFO}$.
%     If
%     \begin{equation}
%     \label{eq:5}
%         \left[a_0\left(1-\frac{b_0}{\hat{b}_0^{\mathrm{LS}}}\right)+a_1\left(1-\frac{b_1}{\hat{b}_1^{\mathrm{LS}}}\right)\right]\left[\hat{b}_0^{\mathrm{LS}}-\tau\hat{b}_1^{\mathrm{LS}}\right]<0,
%     \end{equation}
%     then there exists $\xi>0$ such that, when the level of demand fairness constraints $\alpha$ satisfies $0<\alpha<\xi$, we have $\calR^{FEO}>\calR^{EFO}$, where $\tau$ is defined in~\eqref{tau}.
% \end{proposition}

Proposition \ref{prop4} shows that fairness constraints in the estimation stage can sometimes act as a regularizer. When the unknown parameters are accurately estimated, i.e., $\hat{b}_0^{\mathrm{LS}}=b_0$ and $\hat{b}_1^{\mathrm{LS}}=b_1$, there is no gap between the estimated profit and the real expected profit, and it is easy to show that EFO will give larger profit than FEO. The benefit of fairness constraints therefore arises in the presence of estimation error, where they can improve parameter accuracy and, consequently, increase profit.

\subsection{Rawlsian Fairness}
Rawlsian fairness seeks to protect the least advantaged individuals. More specifically, for price fairness it minimizes the maximum price charged to customers, while in demand fairness, it maximizes the minimum demand faced by customers.
For Rawlsian price fairness and Rawlsian demand fairness, Proposition \ref{prop:rawlsian_fairness} shows that FEO and EFO lead to the same exact prices. Thus, one can choose FEO or EFO based on discussion in Section~\ref{sub:FEO_vs_EFO}.
\begin{proposition}
    \label{prop:rawlsian_fairness}
    Suppose Assumption~\ref{assumption:1} holds. %Without loss of generality, suppose $p_0^{\mathrm{LS}}<p_1^{\mathrm{LS}}$. 
    Then, under either Rawlsian price or demand fairness with fairness level $\alpha$,
    \begin{equation*}
        p_0^*\left(\hat{b}_0^{\mathrm{FEO}},\hat{b}_1^{\mathrm{FEO}}\right)=p_0^{\mathrm{EFO}}\quad\text{and}\quad p_1^*\left(\hat{b}_0^{\mathrm{FEO}},\hat{b}_1^{\mathrm{FEO}}\right)=p_1^{\mathrm{EFO}}.%=(1-\alpha)p_1^{\mathrm{LS}}+\alpha p_0^{\mathrm{LS}}.
    \end{equation*}
    % Under Rawlsian demand fairness,
    % \begin{equation*}
    %     p_0^*\left(\hat{b}_0^{  \mathrm{FEO}},\hat{b}_1^{\mathrm{FEO}}\right)=p_0^{\mathrm{EFO}}=p_0^{\mathrm{LS}},\quad p_1^*\left(\hat{b}_0^{\mathrm{FEO}},\hat{b}_1^{\mathrm{FEO}}\right)=(1-\alpha)p_1^{\mathrm{LS}}+\alpha\frac{a_1\hat{b}_0^{\mathrm{LS}}}{a_0\hat{b}_1^{\mathrm{LS}}}p_0^{\mathrm{LS}}.
    % \end{equation*}
If $p_0^\ast(\hat b_0^{\mathrm{LS}}, \hat b_1^{\mathrm{LS}})<p_1^\ast(b_0, b_1) <p_1^\ast(\hat b_0^{\mathrm{LS}}, \hat b_1^{\mathrm{LS}})$, then both profit and total surplus can increase simultaneously up to a certain threshold, thereby improving overall social welfare.
\end{proposition}

% This result shows that, under \emph{Rawlsian price fairness}, if the unconstrained optimal price derived from 
% $(\hat{b}_0^{\mathrm{LS}}, \hat{b}_1^{\mathrm{LS}})$ exceeds the cap 
% $\tilde{p} := (1 - \alpha)p_1^{\mathrm{LS}} + \alpha p_0^{\mathrm{LS}}$, 
% then the chosen price is \emph{shaded down} to $\tilde{p}$ (see $p_1$ in Proposition~\ref{prop:rawlsian_fairness} and Figure~\ref{fig:rawls_price}). 
% Conversely, when the unconstrained optimal price lies below the cap, 
% i.e. when $\frac{a_g}{-2\hat{b}_g^{\mathrm{LS}}} + \frac{c}{2} \le \tilde{p}$, the solution coincides with the unconstrained optimum (see $p_0$ in Proposition~\ref{prop:rawlsian_fairness} and  and Figure~\ref{fig:rawls_price}).
% Under \emph{Rawlsian demand fairness}, a symmetric logic applies. If a group’s demand falls below a certain threshold, Rawlsian demand fairness seeks to increase that group’s demand, which is achieved by lowering the corresponding price.

%One interesting implication of Proposition~\ref{prop:rawlsian_fairness} is that both profit and total surplus can increase simultaneously under certain conditions. 
%For example, suppose 
The condition $p_0^\ast(\hat b_0^{\mathrm{LS}}, \hat b_1^{\mathrm{LS}})<p_1^\ast(b_0, b_1) <p_1^\ast(\hat b_0^{\mathrm{LS}}, \hat b_1^{\mathrm{LS}})$ in Proposition~\ref{prop:rawlsian_fairness} implies that group~$1$ is overcharged relative to its true optimal price. This case corresponds to the scenario shown in Figure~\ref{fig:prop4}. Under Rawlsian price fairness, increasing $\alpha$ -- which lowers the price cap -- brings the price of group~$1$ closer to its true optimal price, $p_1^*(b_0,b_1)$ (see Figure~\ref{fig:rawls_price}) and can lead to higher profit (see Figure~\ref{fig:rawls_price_outcomes}). Since prices decrease, consumer surplus increases, implying that both the firm and consumers benefit. A similar mechanism operates under Rawlsian demand fairness, where lowering the price for group~$1$ increases demand and yields higher profit and consumer surplus (see Figure~\ref{fig:rawls_demand} and \ref{fig:rawls_demand_outcomes}). This challenges the conventional view that fairness constraints inherently involve a trade-off with efficiency, and instead suggests that, under certain conditions, fairness can improve outcomes for both the firm and consumers.

\begin{figure}[htbp]
\FIGURE{
\begin{minipage}{\textwidth}
\centering
\captionsetup{justification=centering}
\begin{subfigure}{0.24\textwidth}
    \centering
    \includegraphics[width=\textwidth]{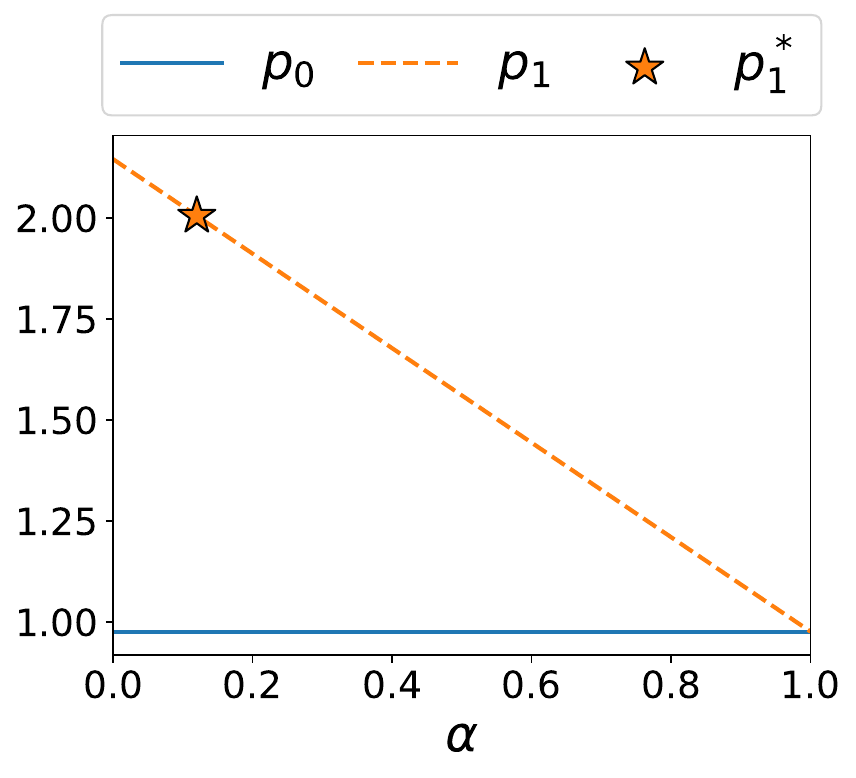}
    \caption{Price}
    \label{fig:rawls_price}
\end{subfigure}
\begin{subfigure}{0.24\textwidth}
    \centering
    \includegraphics[width=\textwidth]{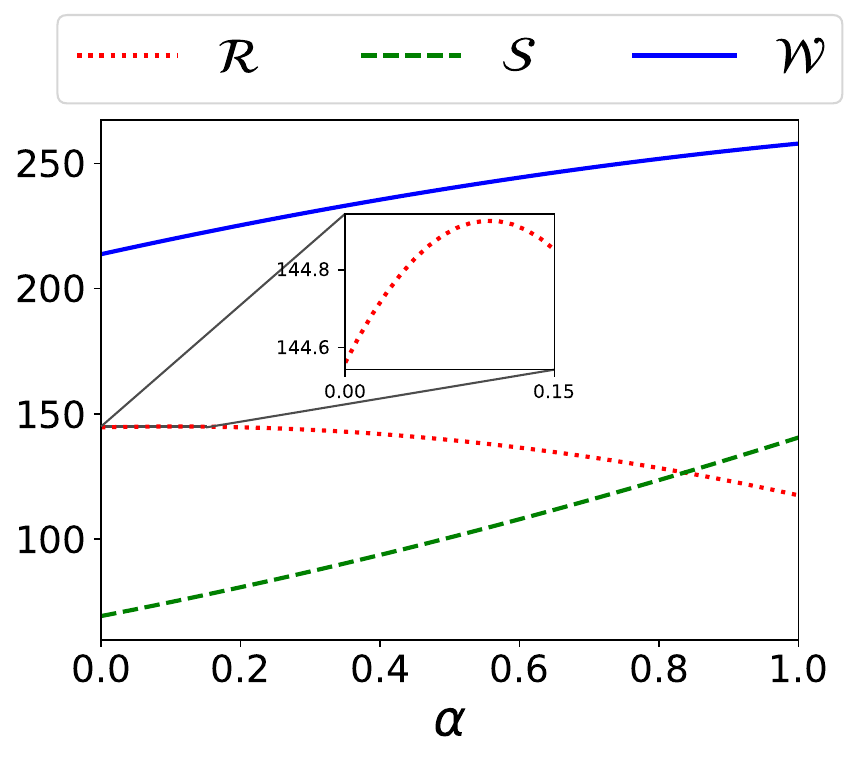}
    \caption{Outcomes}
    \label{fig:rawls_price_outcomes}
\end{subfigure}
\begin{subfigure}{0.24\textwidth}
    \centering
    \includegraphics[width=\textwidth]{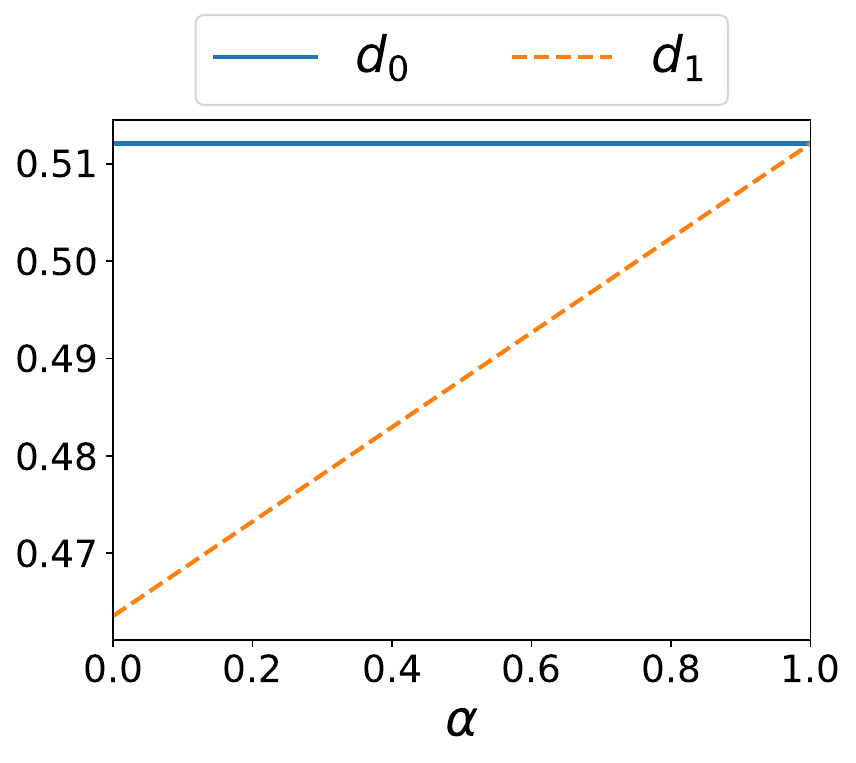}
    \caption{Demand}
    \label{fig:rawls_demand}
\end{subfigure}
\begin{subfigure}{0.24\textwidth}
    \centering
    \includegraphics[width=\textwidth]{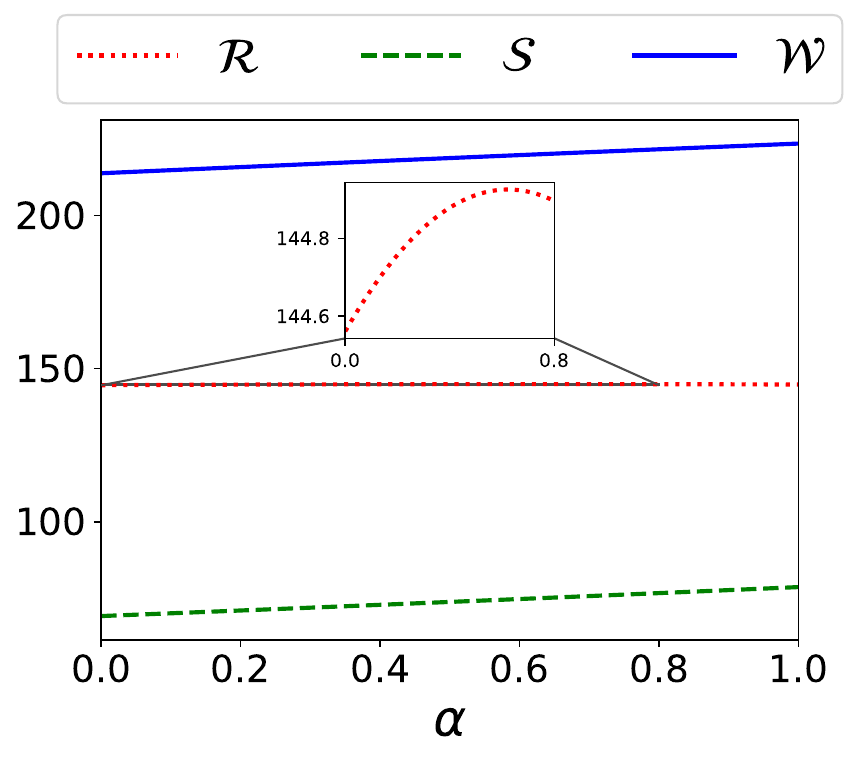}
    \caption{Outcomes}
    \label{fig:rawls_demand_outcomes}
\end{subfigure}
\end{minipage}
}
{Rawlsian Price/Demand Fairness \label{fig:prop4}\vspace{.3cm}}
{Figures~\ref{fig:rawls_price}--\ref{fig:rawls_price_outcomes} (resp. \ref{fig:rawls_demand}--\ref{fig:rawls_demand_outcomes}) present outcomes under Rawlsian price (resp. demand) fairness. 
The parameters are 
$a_0 = 1$, $a_1 = 1$, 
$b_0 = -0.5$, $b_1 = -0.25$, $\hat b^{LS}_0 = -0.5124$, $\hat b^{LS}_1 = -0.2320$, and $c = 0$. Note that $p_1^\ast(b_0,b_1)=2$.}
\end{figure}

\section{Extensions}
\label{sec:extensions}
The analysis in the previous section develops intuition in a stylized setting. We now turn to more general environments. More specifically, we extend the stylized model to include scenarios with unknown market sizes, i.e., unknown $(a_0, a_1)$ (Section~\ref{sec:unknown_market_sizes}), a linear model with features (Section~\ref{sec:linear_model}), and a logistic demand model with features (Section~\ref{sec:logistic_model}).

\subsection{Unknown Market Sizes}
\label{sec:unknown_market_sizes}
In this section, we explore the case where both the market size and price elasticity are unknown in the demand estimation stage. Similar to Assumption \ref{assumption:1}, we also assume that there exist prices above $c$ for the seller to get positive estimated demand with the least square estimator.
\begin{assumption}
    \label{assumption:3}
    The least square estimator satisfy $\hat{a}_g^{\mathrm{LS}}+\hat{b}_g^{\mathrm{LS}}c>0$ for $g\in\{0,1\}$.
\end{assumption}

First, for parity-wise loss fairness, similar to the case with known market size, the optimization problem admits multiple solutions that satisfy the fairness constraint and attain the same objective value, resulting in model multiplicity. Importantly, different solutions can lead to either increases or decreases in consumer surplus.

For parity-wise price and demand fairness, we generalize Proposition \ref{prop3} to Proposition \ref{prop7} which shows that there is also a threshold that determines whether FEO or EFO provides higher consumer surplus and social welfare under a certain fairness level $\alpha$. 

\begin{proposition}
\label{prop7}
Suppose Assumption~\ref{assumption:3} holds. Define
\begin{equation*}
         M_g :=
         \begin{bmatrix}
             1 & \frac{1}{n_g}\sum_{\{i|g^{(i)}=g\}}p^{(i)}\\
             \frac{1}{n_g}\sum_{\{i|g^{(i)}=g\}}p^{(i)} & \frac{1}{n_g}\sum_{\{i|g^{(i)}=g\}}p^{(i)^2}
         \end{bmatrix}
         ,\quad 
         v_g :=
         \begin{bmatrix}
             -\frac{\hat{b}_g^{\mathrm{LS}}}{2}\\[0.5ex]
             \frac{\hat{a}_g^{\mathrm{LS}}}{2}
         \end{bmatrix}
         ,
       \end{equation*}
and
\begin{equation*}
    \tau_p: =\left(\frac{v_0^\top M_0^{-1}v_0}{v_1^\top M_1^{-1}v_1}\right)^{1/3}, \quad \tau_a :=\frac{\hat a_0^{\mathrm{LS}}}{\hat a_1^{\mathrm{LS}}},\quad \tau_b :=\frac{\hat b_0^{\mathrm{LS}}}{\hat b_1^{\mathrm{LS}}}.
\end{equation*}
Then there exists $\xi>0$ such that the following statements hold for all $\alpha\in(0,\xi)$.
\begin{enumerate}
\item Consider parity-wise price fairness.
\begin{enumerate}
    \item If $(\tau_p-\tau_b)(\tau_a-\tau_b)>0$, then $p_0^{\mathrm{FEO}}<p_0^{\mathrm{EFO}}$ and $p_1^{\mathrm{FEO}}<p_1^{\mathrm{EFO}}$. Consequently, $\mathcal S^{\mathrm{FEO}}> \mathcal S^{\mathrm{EFO}}$ and $\mathcal W^{\mathrm{FEO}}> \mathcal W^{\mathrm{EFO}}$. \label{prop5:case1a}
    \item If $(\tau_p-\tau_b)(\tau_a-\tau_b)<0$, then $p_0^{\mathrm{FEO}}>p_0^{\mathrm{EFO}}$ and $p_1^{\mathrm{FEO}}>p_1^{\mathrm{EFO}}$. Consequently, $\mathcal S^{\mathrm{FEO}}< \mathcal S^{\mathrm{EFO}}$ and $\mathcal W^{\mathrm{FEO}}< \mathcal W^{\mathrm{EFO}}$.
    \label{prop5:case1b}
\end{enumerate}
\item Consider parity-wise demand fairness and assume $c>0$.
\begin{enumerate}
    \item If $(\tau_p-\tau_b)(\tau_a-\tau_b)>0$, then $p_0^{\mathrm{FEO}}>p_0^{\mathrm{EFO}}$ and $p_1^{\mathrm{FEO}}>p_1^{\mathrm{EFO}}$. Consequently, $\mathcal S^{\mathrm{FEO}}< \mathcal S^{\mathrm{EFO}}$ and $\mathcal W^{\mathrm{FEO}}< \mathcal W^{\mathrm{EFO}}$. \label{prop5:case2a}
    \item If $(\tau_p-\tau_b)(\tau_a-\tau_b)<0$, then $p_0^{\mathrm{FEO}}<p_0^{\mathrm{EFO}}$ and $p_1^{\mathrm{FEO}}<p_1^{\mathrm{EFO}}$. Consequently, $\mathcal S^{\mathrm{FEO}}> \mathcal S^{\mathrm{EFO}}$ and $\mathcal W^{\mathrm{FEO}}> \mathcal W^{\mathrm{EFO}}$. \label{prop5:case2b}
\end{enumerate}
\end{enumerate}
\end{proposition}

For Rawlsian fairness, we obtain the same result as in Proposition~\ref{prop:rawlsian_fairness}. The key reason is that, in our demand model, each group is governed by an independent parameter pair $(a_g, b_g)$. To illustrate, consider Rawlsian price fairness. As the price cap decreases, only the group whose price exceeds the cap needs to adjust its parameter to satisfy the constraint. Under FEO, this group optimally sets its price exactly equal to the cap, rather than below it, since any further price reduction would strictly increase loss. Similarly, under EFO, only the group with a price above the cap reduces its price, again stopping exactly at the cap to avoid unnecessary profit loss. As a result, the optimal solutions under FEO and EFO coincide under Rawlsian price fairness. An analogous argument applies to Rawlsian demand fairness.

%%%%%%%%%%%%%%%%%%%%%%%%%%%%%%%%%%%%%%%% 
\subsection{Feature-based Linear Demand Models}
\label{sec:linear_model}
In this section, we extend the framework introduced so far to the setting with features, which can be interpreted as the case of personalized pricing. We describe how to solve this problem, discuss the properties that arise in this setting, and examine whether the results observed in the stylized model continue to hold.

We consider a feature vector $\vx^{(i)} = [1, x^{(i)}_1, x^{(i)}_2, \ldots, x^{(i)}_m] \in \mathbb{R}^{m+1}$ that encodes the characteristics of customer~$i\in[n]$ (or the product). The first component of $\vx^{(i)}$ is fixed at~$1$ to capture the intercept. This formulation enables personalized pricing and provides a more accurate characterization of the demand function. Under the linear demand function,
$$d^{\mathrm{LS}}(\vx,p;a,b)=a^\top \vx + (b^\top \vx)p.$$
Note that when there are only two feature vectors, $\vx=[1,0]$ and $\vx=[1,1]$, the model reduces to the stylized setting in Section~\ref{sec:theory}. Here, we interpret the feature as one that may embed information correlated with $g$.

Similar to Section~\ref{sec:theory}, the optimal estimates of the demand parameters $(\hat{a}^{\mathrm{LS}},\hat{b}^{\mathrm{LS}})$ can be obtained by minimizing the following loss function,
% $$\ell^{\mathrm{LS}}(\hat{a},\hat{b})=\sum_{g \in \{0,1\}}\sum_{\{i \mid g^{(i)} = g\}}\dfrac{1}{n_g} 
%   \Bigl( \hat a^{\top} x^{(i)} + \left(\hat b^{\top} x^{(i)}\right)\,p^{(i)} - d^{(i)} \Bigr)^{2}.$$
$$\ell^{\mathrm{LS}}(\hat{a},\hat{b})
=\sum_{g \in \{0,1\}}\sum_{\{i \mid g^{(i)} = g\}}\dfrac{1}{n_g} 
  \left( d^{\mathrm{LS}}\left(\vx^{(i)},p^{(i)};a,b\right) - d^{(i)} \right)^{2}.$$
Given a feature vector $\vx$ and the estimated parameters $(\hat a, \hat b)$, the optimal price $p^*$ is defined as
\begin{equation*}
\begin{aligned}
p^*(\vx; \hat{a}, \hat{b}) := \arg\max_{p \geq 0} \, (p - c) \max\left(0, d(\vx, p; \hat{a}, \hat{b})\right).
\end{aligned}
\end{equation*}
Building on this setup, we introduce the definitions of parity-wise fairness and Rawlsian fairness within this feature-based framework.

\paragraph{Parity-wise Fairness.}
We formulate the $\alpha$-parity-wise price fairness criterion under the FEO framework as follows.
% \begin{equation}
% \begin{aligned}
%     \min_{\hat{a}, \hat{b}\leq0}\quad &\ell^{\mathrm{LS}}(\hat{a},\hat{b})\\
%     \text{s.t.} \quad &\left|\frac{1}{n_0}\sum_{\{i|g^{(i)}=0\}}\frac{\hat{a}^\top x^{(i)}}{2\hat{b}^\top x^{(i)}}-\frac{1}{n_1}\sum_{\{i|g^{(i)}=1\}}\frac{\hat{a}^\top x^{(i)}}{2\hat{b}^\top x^{(i)}}\right|\le (1-\alpha)\Delta_p,
% \end{aligned}
% \label{eq:FEO_price_parity_feature_linear}
% \end{equation}
\begin{equation}
\begin{aligned}
    \min_{\hat{a}, \hat{b}\leq0}\quad &\ell^{\mathrm{LS}}(\hat{a},\hat{b})\\
    \text{s.t.} \quad &\Big|\frac{1}{n_0}\sum_{\{i|g^{(i)}=0\}}p^*(\vx^{(i)};\hat a, \hat b)-\frac{1}{n_1}\sum_{\{i|g^{(i)}=1\}}p^*(\vx^{(i)};\hat a, \hat b)\Big|\le (1-\alpha)\Delta_p,
\end{aligned}
\label{eq:FEO_price_parity_feature_linear}
\end{equation}
which enforces that the disparity in predicted prices between the two groups is bounded. Here, $\Delta_p$ denotes the original average price gap between the two groups before applying the fairness criterion. Furthermore, under parity-wise price fairness, we additionally impose an individual-level price cap, such as
$$p^*(\vx^{(i)};\hat a, \hat b) \le \tilde{p}\quad\forall i\in[n],$$
to prevent the decision maker from setting excessively high prices. This is because constraint \eqref{eq:FEO_price_parity_feature_linear} only restricts the average price disparity across groups, and without an explicit cap, the optimization problem could satisfy the parity constraint by assigning extreme prices to a single individual.

Similarly, the formulation of demand fairness shares the same objective structure as that in~\eqref{eq:FEO_price_parity_feature_linear}.
In demand fairness, we consider normalized ex-ante demand for each data $i$,
\begin{equation}
    d^{\mathrm{EA}}\!\!\left(x^{(i)};\hat a,\hat b\right):=\frac{1}{(\hat{a}^{\mathrm{LS}})^{\top} x^{(i)}} d^{\mathrm{LS}}\!\left(x^{(i)},p^*(\vx^{(i)};\hat a, \hat b);\hat{a}^{\mathrm{LS}},\hat{b}^{\mathrm{LS}}\right).
\label{eq:ex_ante_linear}
\end{equation}
Here, we adopt the least-squares parameters $(\hat a^{\mathrm{LS}}, \hat b^{\mathrm{LS}})$ as the fixed baseline for demand evaluation. Then, the ex-ante demand is evaluated by applying the optimal price $p^*(\vx^{(i)};\hat a, \hat b)$ to this baseline. Then, parity-wise demand fairness replaces the constraint in~\eqref{eq:FEO_price_parity_feature_linear} with
$$\bigg|\,\frac{1}{n_0}\sum_{\{i|g^{(i)}=0\}}d^{\mathrm{EA}}\!\!\left(x^{(i)};\hat a,\hat b\right)-\frac{1}{n_1}\sum_{\{i|g^{(i)}=1\}}d^{\mathrm{EA}}\!\!\left(x^{(i)};\hat a,\hat b\right)\bigg|\le (1-\alpha)\Delta_d,$$
where $\Delta_d$ denotes the initial demand gap between two groups.

On the other hand, the EFO framework operates under the assumption that the estimated demand function, $d^{\text{LS}}(\vx, p; \hat{a}^{\text{LS}}, \hat{b}^{\text{LS}})$, is the ground truth. Given this assumption, it seeks to maximize total profit subject to a set of fairness constraints. 

The $\alpha$-parity-wise price fairness problem can be formulated as
\begin{equation}
\begin{aligned}
    \max_{p_i \ge 0}\quad 
    &\sum_{i=1}^{n} (p_i - c)
    \max\left(0, (\hat{a}^{\mathrm{LS}})^{\top} x^{(i)} 
    + ((\hat{b}^{\mathrm{LS}})^{\top} x^{(i)}) p_i \right) \\
    \text{s.t.}\quad 
    &\bigg|\frac{1}{n_0}\sum_{\{i|g^{(i)}=0\}}p_i-\frac{1}{n_1}\sum_{\{i|g^{(i)}=1\}}p_i\bigg|\le (1-\alpha)\Delta_p.
    %& p_i \le (1-\alpha)\bar{p} + \alpha\underline{p},\quad \forall i \in [n].
\end{aligned}
    \label{eq:EFO_price_parity_feature_linear}
\end{equation}
The demand fairness can be defined analogously by modifying the fairness constraint accordingly.

For parity-wise fairness, the optimization problem is non-convex. Therefore, in our numerical analysis, we employ standard nonlinear optimization solvers to obtain approximate solutions. When the demand function is specified as $d = a^\top x + b_g p$, both parity-wise price and demand fairness for FEO can be efficiently addressed through a one-dimensional parameter search combined with a convex optimization subproblem. The detailed algorithmic procedure is described in Appendix~\ref{appendix:heuristic}.

\paragraph{Rawlsian Fairness.}
Under the FEO framework, the $\alpha$-Rawlsian price fairness formulation adopts the same objective as in \eqref{eq:FEO_price_parity_feature_linear} but replaces the constraint with
\begin{equation}
p^*(\vx^{(i)};\hat a,\hat b) \leq (1-\alpha)\bar{p}+ \alpha\underline{p},~\forall i\in[n],
\label{eq:rawls_price_feature}
\end{equation}
where $\bar{p}$ (resp. $\underline{p}$) represents the upper (resp. lower) bound of price determined by the practitioner. Rawlsian demand fairness ensures that the ex-ante normalized demand exceeds a prescribed minimum level. Mathematically, 
\begin{equation*}
d^{\mathrm{EA}}\!\!\left(x^{(i)};\hat a,\hat b\right) 
\geq \alpha \bar{d} + (1-\alpha)\underline{d},~\forall i\in[n],
\label{eq:rawls_demand_feature}
\end{equation*}
where $\bar{d}$ (resp. $\underline{d}$) is the maximum (resp. minimum) normalized demand.

On the other hand, under the EFO framework, the objective remains the same as in~\eqref{eq:EFO_price_parity_feature_linear}, but the constraint is modified to capture different notions of fairness. Rawlsian price fairness replaces the constraint with
$$p_i \le (1-\alpha)\bar p + \alpha \underline{p}, ~\forall i \in [n].$$
Demand fairness is defined in a similar manner by replacing the constraint with a requirement that the induced demand with respect to $p_i$ exceeds a certain threshold.
%$$\frac{(\hat{a}^{\mathrm{LS}})^{\top} x^{(i)} + ((\hat{b}^{\mathrm{LS}})^{\top} x^{(i)})p_i}{(\hat{a}^{\mathrm{LS}})^{\top} x^{(i)}} \ge \alpha \bar{d} + (1-\alpha)\underline{d},~ \forall i \in [n].$$

One of the key differences between the EFO and FEO frameworks lies in the number of decision variables involved in the optimization problem. In the FEO framework, there are $2m+2$ decision variables -- twice the dimension of the feature space -- whereas in the EFO framework, there are $n$ decision variables corresponding to the number of data points. Consequently, in personalized pricing settings with a large number of consumers (i.e., large $n$), EFO may face an increased computational burden. In contrast, FEO may offer relative scalability advantages since the number of features $m$ is typically much smaller than $n$.

Regarding the nature of the optimization problems, the two frameworks exhibit different convexity properties. In the EFO setting, Rawlsian fairness formulations are non-convex, as the objective in~\eqref{eq:EFO_price_parity_feature_linear} involves a multiplicative interaction between the price $p_i$ and a piecewise-linear demand function, resulting in a non-concave profit function in the decision variables. In the FEO framework, Proposition~\ref{prop:convexity_linear} illustrates that FEO frameworks can be reformulated as convex quadratic programs (QPs). Although the formulation in~\eqref{eq:FEO_price_parity_feature_linear}, under both price and demand fairness, contains a nonlinear constraint, Proposition~\ref{prop:convexity_linear} demonstrates that the problem can nevertheless be rewritten as a convex optimization problem.%In this setting, the number of decision variables directly affects the computational complexity. 
% Note that the problem~\eqref{eq:FEO_price_parity_feature_linear}, under both price and demand fairness, includes a non-linear term in the constraint. 
% Proposition~\ref{prop:convexity_linear} shows that, despite this non-linearity, the problem can be reformulated as a convex optimization problem.

\begin{proposition}
Within the FEO framework, Rawlsian price and demand fairness can be formulated as convex optimization problems under linear demand.
\label{prop:convexity_linear}
\end{proposition}

% Furthermore, another appealing property of Rawlsian fairness is that the loss value changes in a convex manner with respect to small variations in the fairness level. Proposition~\ref{prop:convexity_wrt_error_linear} formalizes this observation.
% \begin{proposition}
% Denote $l(\alpha)$ as the optimal objective value of FEO with respect to fairness level $\alpha$. 
% % Assume $\mathrm{span}\left(\begin{bmatrix}
% %     x^{(i)}\\
% %     p^{(i)}x^{(i)}
% % \end{bmatrix}\right)_{i\in [n]}=\mathbb{R}^{2(m+1)}$. 
% Assume $\mathrm{span}\{ (x^{(i)}, p^{(i)}x^{(i)})^\top \}_{i \in [n]} = \mathbb{R}^{2(m+1)}$. Then there exists a constant $\delta>0$ such that $\ell(\alpha)$ is a convex function in $\alpha \in \left[0,\delta\right]$, i.e., $\lim_{\alpha\rightarrow 0+}\frac{d^2 \ell(\alpha)}{d\alpha^2}>0$, under both price and demand fairness with linear demand. \label{prop:convexity_wrt_error_linear}
% \end{proposition}
% This result implies that the loss function responds in a smooth and convex manner to small increases in the fairness level. %In other words, one can realize substantial fairness improvements while incurring only a modest sacrifice in estimation loss (or accuracy).

\paragraph{Numerical Analysis} 
We conduct experiments under the linear demand model. Specifically, we consider four types of fairness, given by $\{\text{Rawlsian}, \text{Parity-wise}\} \times \{\text{price}, \text{demand}\}$ fairness. 
We generate $1{,}000$ synthetic samples following the procedure described in 
Appendix~\ref{appendix:synthetic_data}. We vary $\alpha$ from $0$ to $1$ on a grid with step size $0.2$. To validate the performance of our frameworks, we split the dataset into training and test data with an $80$:$20$ ratio, and show the results on test data. 

We consider a linear demand function with parameters $(a_0, a_1, b_0, b_1) = (30, 10, -4, -4)$, i.e., $d = (30 + 10x) - (4 + 4x)p$. 
All parameters are assumed unknown, and the cost is $c = 1$. We impose an individual-level price cap $\tilde{p}=10$ to ensure that fairness constraints are not satisfied by changing the price of a single customer. We solve the constrained optimization problem using the \emph{Sequential Least Squares Quadratic Programming} (SLSQP) algorithm in \emph{SciPy} \citep{2020SciPy-NMeth}.

Figure~\ref{fig:parity_price_linear} presents the price distributions under parity-wise price fairness. As the parity-wise price fairness constraint are tightened in the estimation stage, the demand function being learned becomes less personalized to reduce the difference between the two groups. Consequently, the price distributions of group $0$ and group $1$ converge rapidly for FEO, and when $\alpha=1.0$, all the customers receive the same price. In contrast, imposing parity-wise price fairness constraints in the optimization stage uses the same demand function as the model without fairness constraints, preserving the effectiveness of the features. Table~\ref{tab:measure_parity_price_linear} shows that consumer surplus and social welfare under FEO are higher than those under EFO at the same fairness level, consistent with Corollary~\ref{corollary2}.

\begin{figure}[htbp]
\FIGURE{
\begin{minipage}{\textwidth}
\centering
\captionsetup{justification=centering}
\begin{subfigure}{0.35\textwidth}
    \centering
    \includegraphics[width=\textwidth]{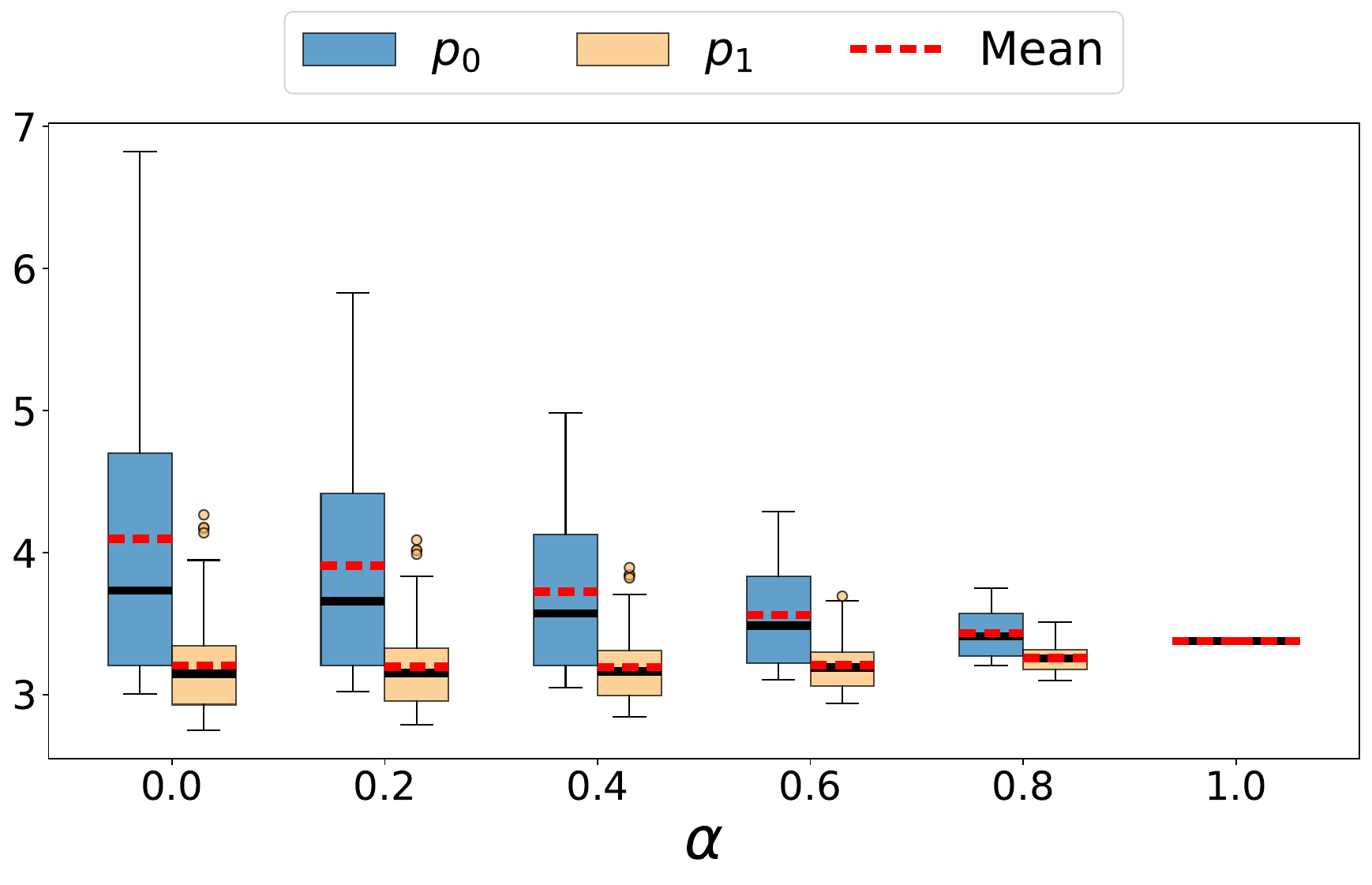}
    \caption{Price (FEO)}
\end{subfigure}
\begin{subfigure}{0.35\textwidth}
    \centering
    \includegraphics[width=\textwidth]{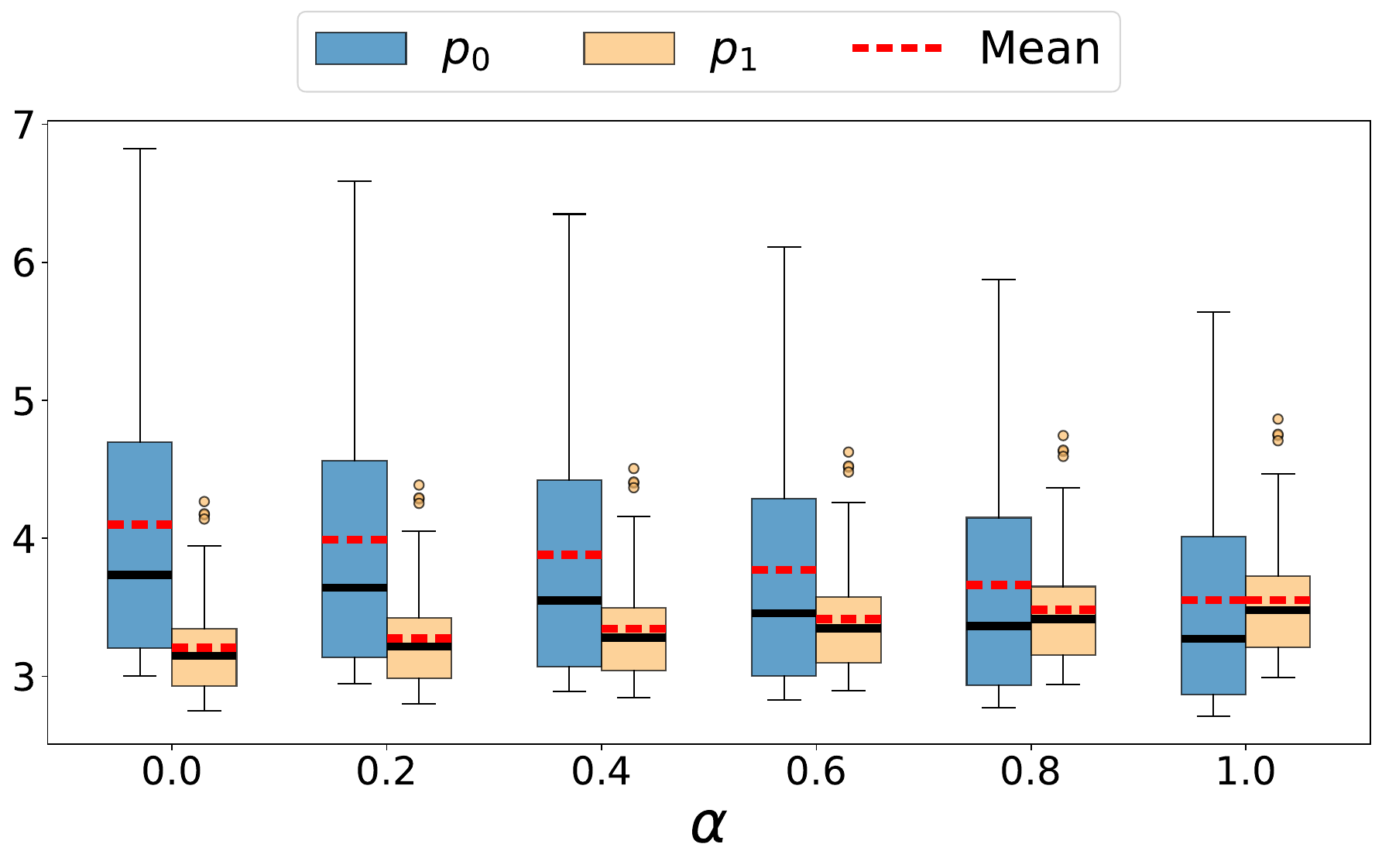}
    \caption{Price (EFO)}
\end{subfigure}
\end{minipage}
}
{Price distributions for FEO and EFO under parity-wise price fairness \label{fig:parity_price_linear}\vspace{1mm}}
{}
\end{figure}

\begin{table}[htbp]
    \centering
    \caption{Normalized performance measures for FEO and EFO under parity-wise price fairness (\%)}
    \label{tab:measure_parity_price_linear}
    \renewcommand{\arraystretch}{1.1}
    \setlength{\tabcolsep}{4pt}
    
    \begin{tabular}{l ccccc c ccccc}
        \toprule
        & \multicolumn{5}{c}{FEO ($\alpha$)} & & \multicolumn{5}{c}{EFO ($\alpha$)} \\
        \cmidrule(lr){2-6} \cmidrule(lr){8-12}
        Measure & 0.2 & 0.4 & 0.6 & 0.8 & 1.0 & & 0.2 & 0.4 & 0.6 & 0.8 & 1.0 \\
        \midrule
        
        $\mathcal{R}(\alpha)/\mathcal{R}(0)$  
        & 99.74  & 98.82  & 97.34 &  95.43  & 93.04 &
        & 99.92 &  99.61 &  99.08 &  98.32  & 97.34 \\ 
        
        $\mathcal{S}(\alpha)/\mathcal{S}(0)$  
        & 106.80 & 113.76 & 119.95 & 123.57 & 122.37 &
        & 101.22 & 102.66 & 104.34 & 106.23 & 108.36 \\ 
        
        $\mathcal{W}(\alpha)/\mathcal{W}(0)$  
        & 102.08 & 103.78 & 104.84 & 104.77 & 102.78 &
        & 100.35 & 100.63 & 100.83 & 100.95 & 101.00 \\ 
        
        \bottomrule
    \end{tabular}
\end{table}

Figure~\ref{fig:parity_demand_linear} presents the demand distributions under parity-wise demand fairness. Here, ex-ante demand is the normalized demand calculated by the estimated demand function using the least square estimator, whereas ex-post demand is calculated by the true demand function. Notice that only ex-ante demand is available to the decision maker. Therefore, the true fairness level is not the same as the fairness level set by the decision maker, and the gap between average ex-post demands of group $0$ and group $1$ does not become $0$ when $\alpha=1.0$. In addition, FEO has a generalization error on test data, resulting in additional gaps than EFO in both ex-ante and ex-post demands. In Table~\ref{tab:measure_parity_demand_linear}, we observe that profit, surplus, and social welfare of FEO are all lower than those of EFO. The opposite results between parity-wise price fairness and parity-wise demand fairness also aligns with Proposition \ref{corollary2} of the non-contextual case.

\begin{figure}[htbp]
\FIGURE{
\begin{minipage}{\textwidth}
\centering
\captionsetup{justification=centering}
\begin{subfigure}{0.35\textwidth}
    \centering
    \includegraphics[width=\textwidth]{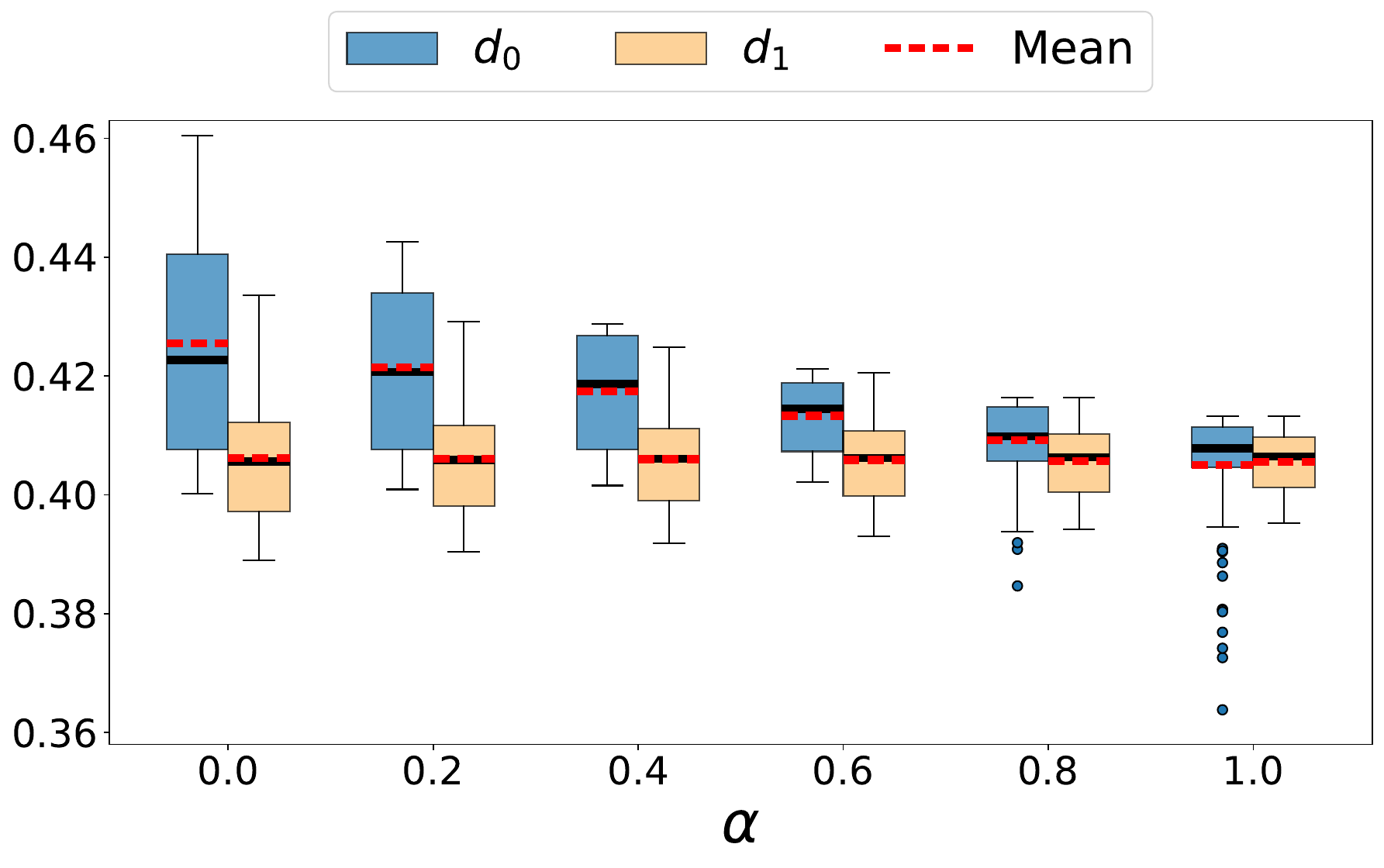}
    \caption{Ex Ante Demand (FEO)}
\end{subfigure}
\begin{subfigure}{0.35\textwidth}
    \centering
    \includegraphics[width=\textwidth]{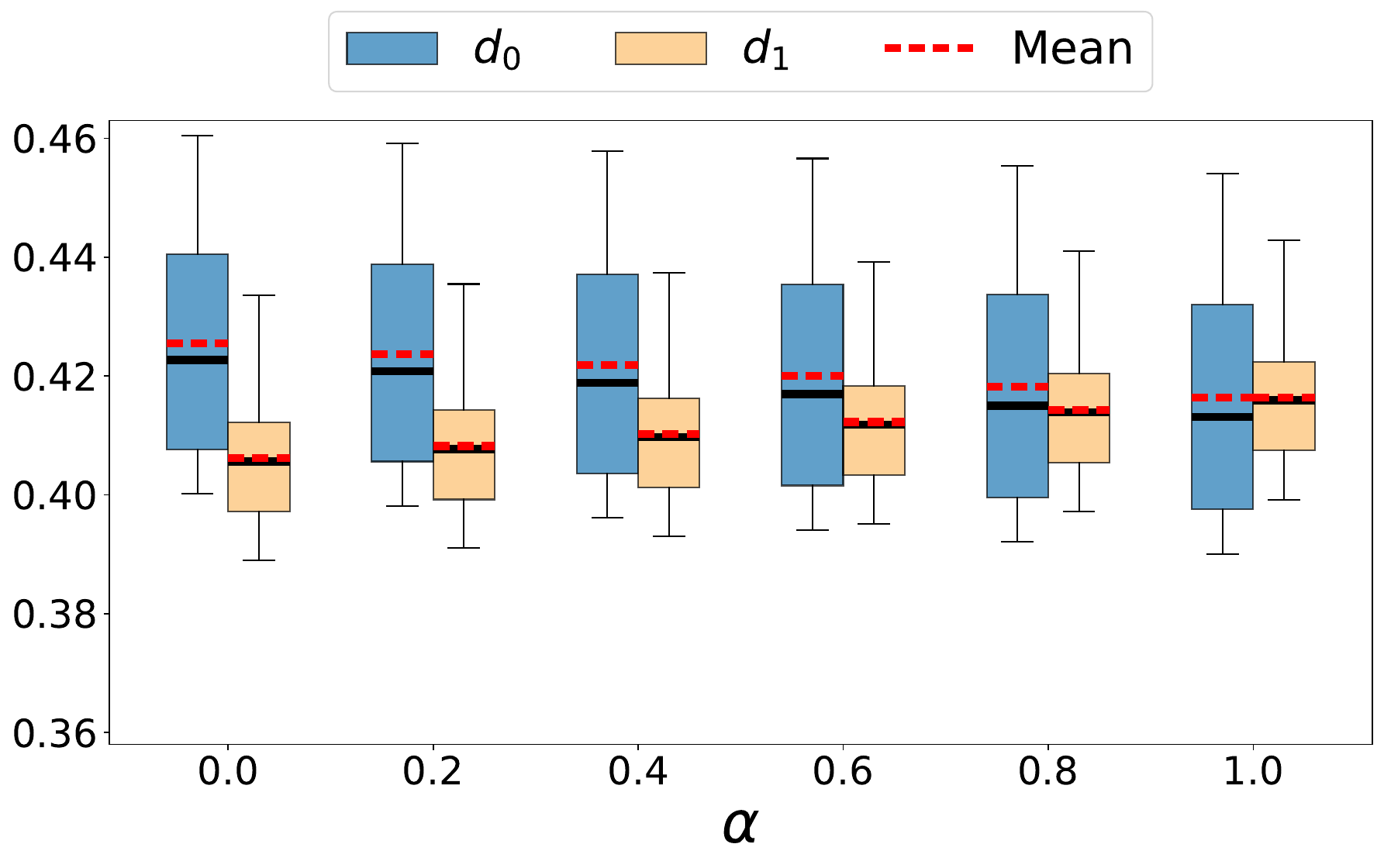}
    \caption{Ex Ante Demand (EFO)}
\end{subfigure}
\begin{subfigure}{0.35\textwidth}
    \centering
    \includegraphics[width=\textwidth]{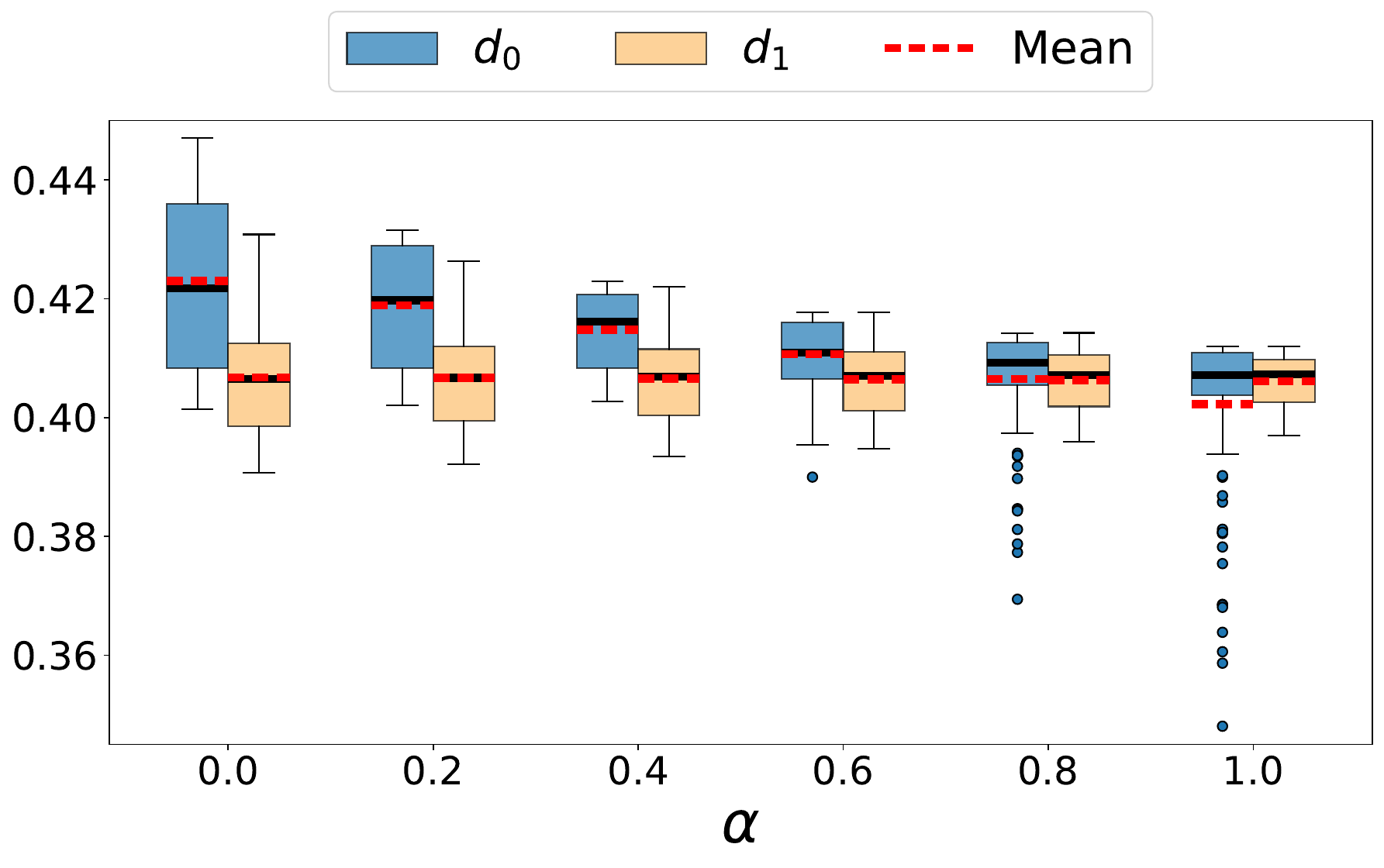}
    \caption{Ex Post Demand (FEO)}
\end{subfigure}
\begin{subfigure}{0.35\textwidth}
    \centering
    \includegraphics[width=\textwidth]{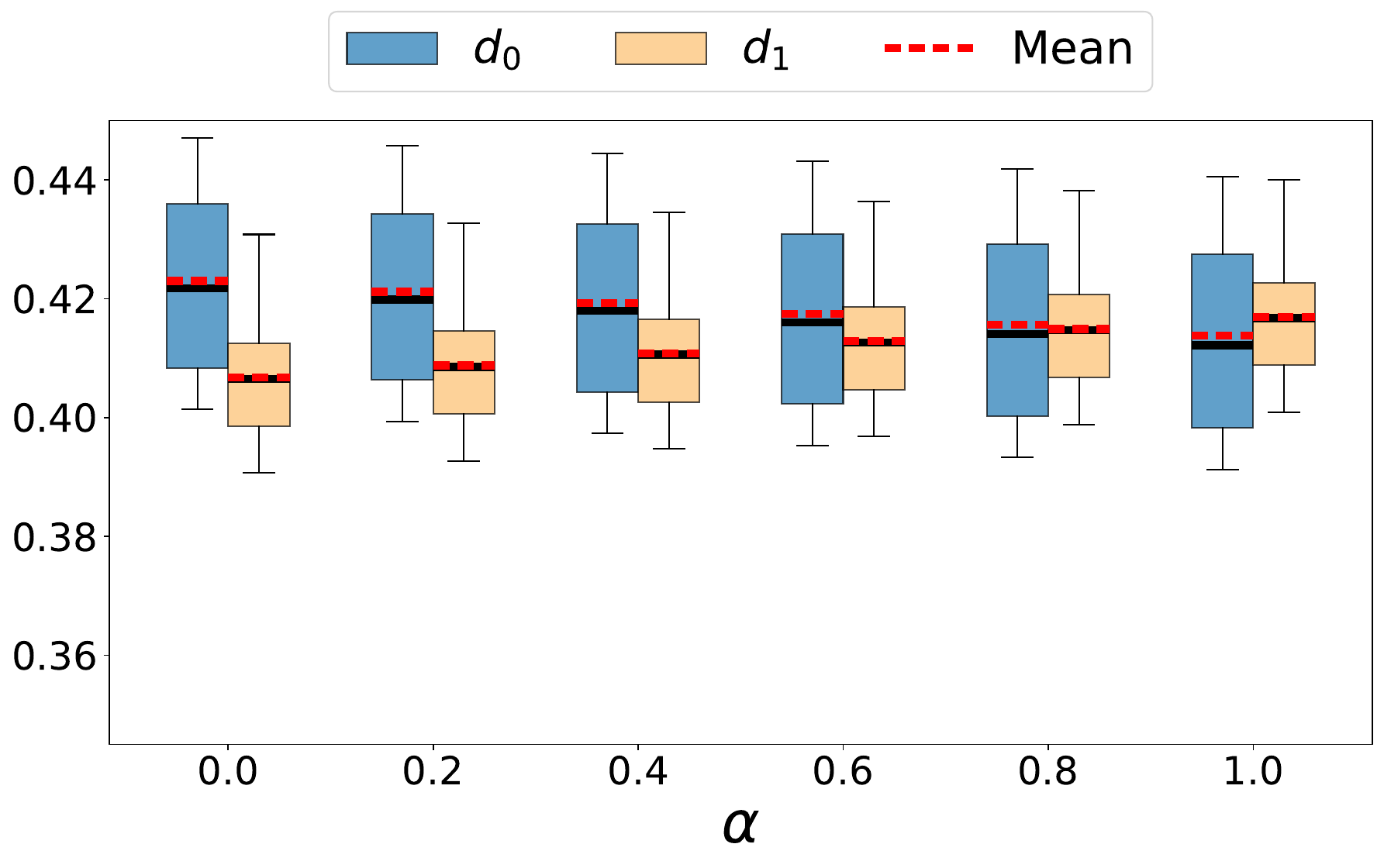}
    \caption{Ex Post Demand (EFO)}
\end{subfigure}
\end{minipage}
}
{Demand distributions for FEO and EFO under parity-wise demand fairness \label{fig:parity_demand_linear}\vspace{2mm}}
{}
\end{figure}

\begin{table}[htbp]
    \centering
    \caption{Normalized performance measures for FEO and EFO under parity-wise demand fairness (\%)}
    \label{tab:measure_parity_demand_linear}
    \renewcommand{\arraystretch}{1.1}
    
    \begin{tabular}{l ccccc c ccccc}
        \toprule
        & \multicolumn{5}{c}{FEO ($\alpha$)} & & \multicolumn{5}{c}{EFO ($\alpha$)} \\
        \cmidrule(lr){2-6} \cmidrule(lr){8-12}
        Measure & 0.2 & 0.4 & 0.6 & 0.8 & 1.0 & & 0.2 & 0.4 & 0.6 & 0.8 & 1.0 \\
        \midrule
        $\mathcal{R}(\alpha)/\mathcal{R}(0)$  
        & 99.96 & 99.89 & 99.79 & 99.65 & 99.48 &
        & 99.99 & 99.99 & 99.97 & 99.95 & 99.93 \\ 
        
        $\mathcal{S}(\alpha)/\mathcal{S}(0)$  
        & 98.72 & 97.45 & 96.19 & 94.96 & 93.74 &
        & 99.99 & 99.99 & 99.97 & 99.95 & 99.93 \\ 
        
        $\mathcal{W}(\alpha)/\mathcal{W}(0)$  
        & 99.55 & 99.08 & 98.60 & 98.09 & 97.57 &
        & 99.99 & 99.98 & 99.97 & 99.96 & 99.94 \\ 
        \bottomrule
    \end{tabular}
\end{table}

Figure~\ref{fig:Rawlsian_price_linear} illustrates the results under Rawlsian price fairness. Unlike Proposition \ref{prop:rawlsian_fairness}, FEO and EFO give different optimal prices under the same $\alpha$. This is because, for feature-based demand models, imposing Rawlsian price fairness constraints in the estimation stage changes the demand function itself, thereby shifting the entire demand distribution. In contrast, imposing Rawlsian price fairness constraints in the optimization stage affects only the subset of the population facing high prices. Table~\ref{tab:measure_rawlsian_price_linear} shows that the consumer surplus and social welfare of FEO are higher than those of EFO as $\alpha$ increases, and eventually become the same at $\alpha=1.0$. 

\begin{figure}[htbp]
\FIGURE{
\begin{minipage}{\textwidth}
\centering
\captionsetup{justification=centering}
\begin{subfigure}{0.35\textwidth}
    \centering
    \includegraphics[width=\textwidth]{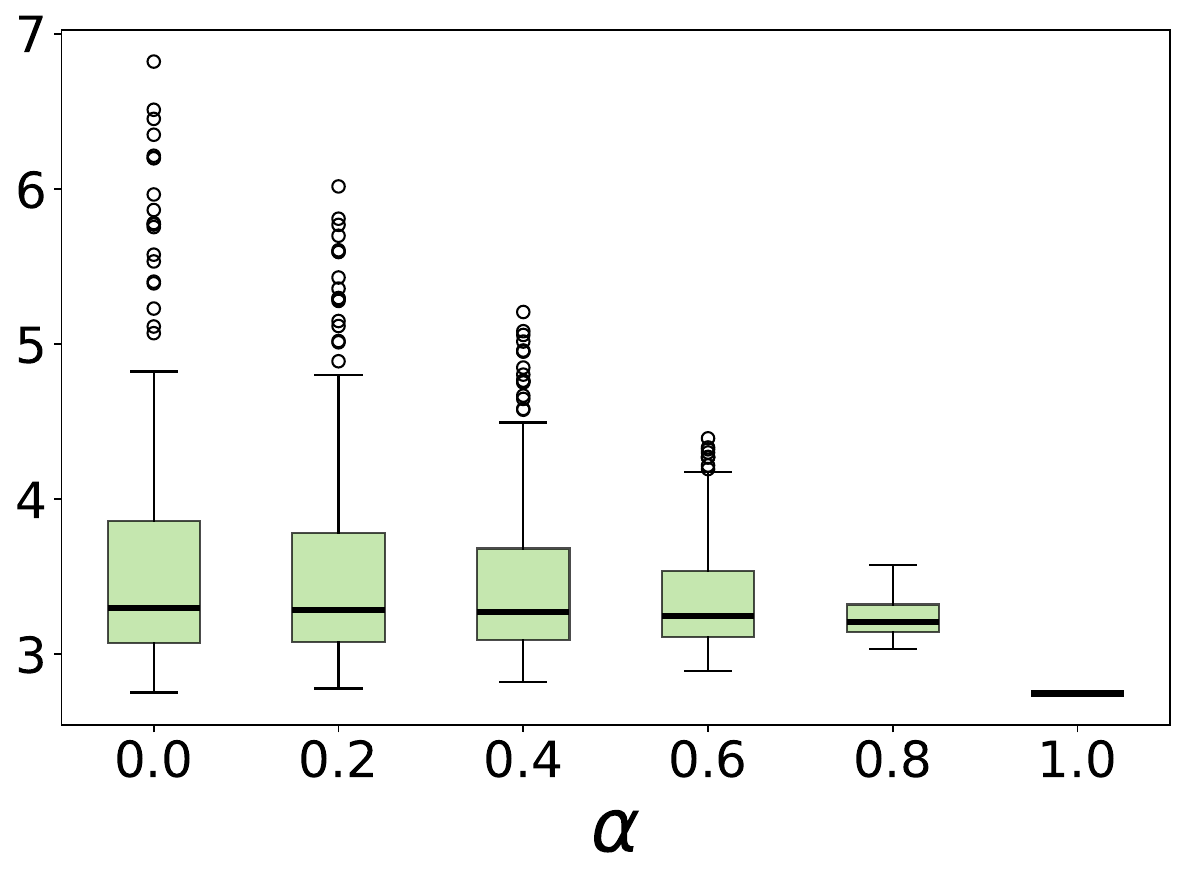}
    \caption{Price under FEO}
    % \label{fig:rawls_price_feo}
\end{subfigure}
\begin{subfigure}{0.35\textwidth}
    \centering
    \includegraphics[width=\textwidth]{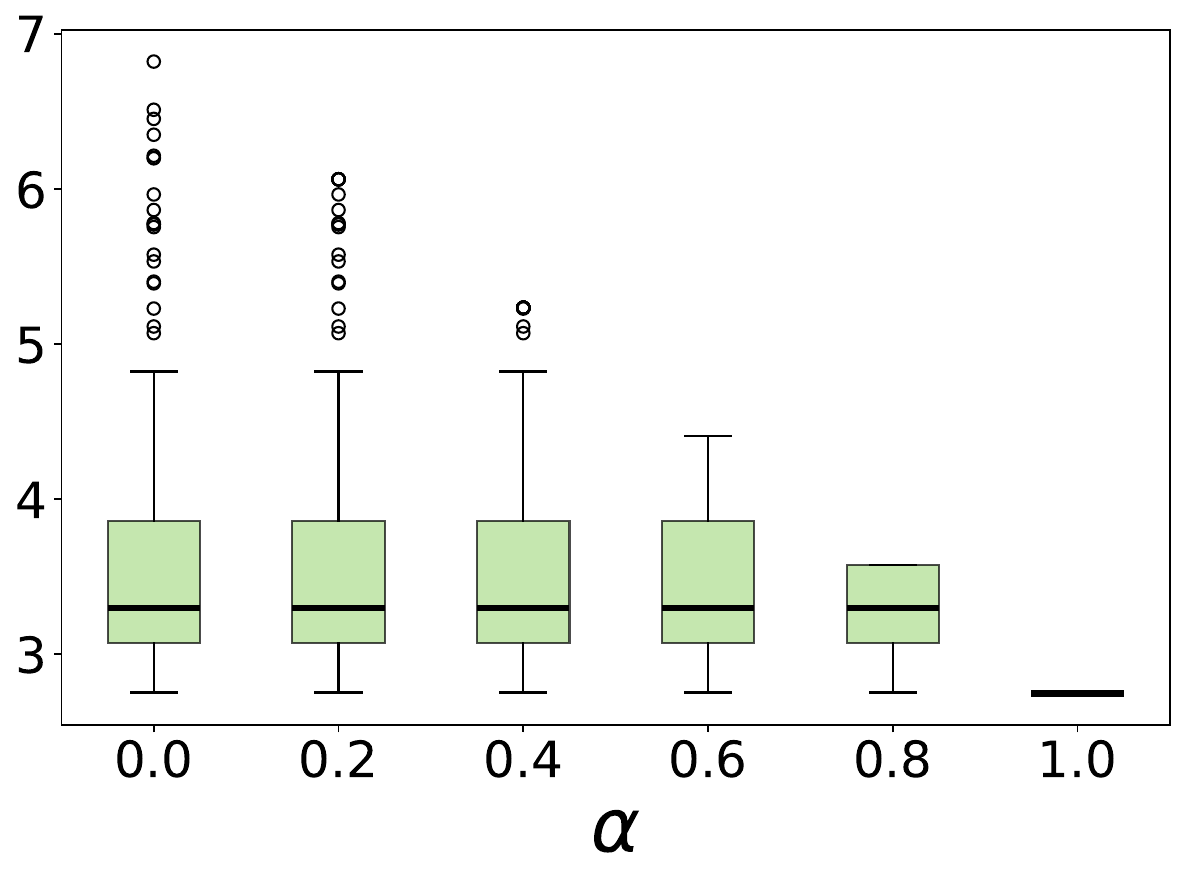}
    \caption{Price under EFO}
    % \label{fig:rawls_price_efo}
\end{subfigure}
\end{minipage}
}
{Price distributions for FEO and EFO under Rawlsian price fairness \label{fig:Rawlsian_price_linear}\vspace{0.5em}}
{}
\end{figure}

\begin{table}[htbp]
    \centering
    \caption{Normalized performance measures for FEO and EFO under Rawlsian price fairness (\%)}
    \label{tab:measure_rawlsian_price_linear}
    \renewcommand{\arraystretch}{1.1}
    % No font size reduction (\small/\footnotesize) as requested
    
    \begin{tabular}{l p{0.2cm} ccccc p{0.4cm} ccccc}
        \toprule
        & & \multicolumn{5}{c}{FEO ($\alpha$)} & & \multicolumn{5}{c}{EFO ($\alpha$)} \\
        \cmidrule(lr){3-7} \cmidrule(lr){9-13}
        Measure & & 0.2 & 0.4 & 0.6 & 0.8 & 1.0 & & 0.2 & 0.4 & 0.6 & 0.8 & 1.0 \\
        \midrule
        $\mathcal{R}(\alpha)/\mathcal{R}(0)$  
        & & 99.85 &  99.13 &  97.58  & 94.62  & 86.35 &
        & 99.99 &  99.72  & 98.69  & 95.78  & 86.35 \\ 
        
        $\mathcal{S}(\alpha)/\mathcal{S}(0)$  
        & & 105.46 & 112.29 & 121.00 & 132.33 & 178.65 &
        & 100.74 & 104.26 & 111.61 & 128.27 & 178.65 \\ 
        
        $\mathcal{W}(\alpha)/\mathcal{W}(0)$  
        & & 101.71  & 103.50 & 105.35 & 107.14 &  116.98 &
        & 100.24 & 101.23 & 102.98 & 106.56 & 116.98 \\ 
        \bottomrule
    \end{tabular}
\end{table}

Figure~\ref{fig:Rawlsian_demand_linear} illustrates the results of Rawlsian demand fairness. In this case, FEO and EFO also lead to different optimal price distributions. In Table~\ref{tab:measure_rawlsian_demand_linear}, consumer surplus and social welfare of FEO are higher than those of EFO. Notably, for both Rawlsian price fairness and Rawlsian demand fairness, the gap between FEO and EFO in terms of profit, surplus and social welfare are relatively small.

\begin{figure}[htbp]
\FIGURE{
\begin{minipage}{\textwidth}
\centering
\captionsetup{justification=centering}
\begin{subfigure}{0.235\textwidth}
    \centering
    \includegraphics[width=\textwidth]{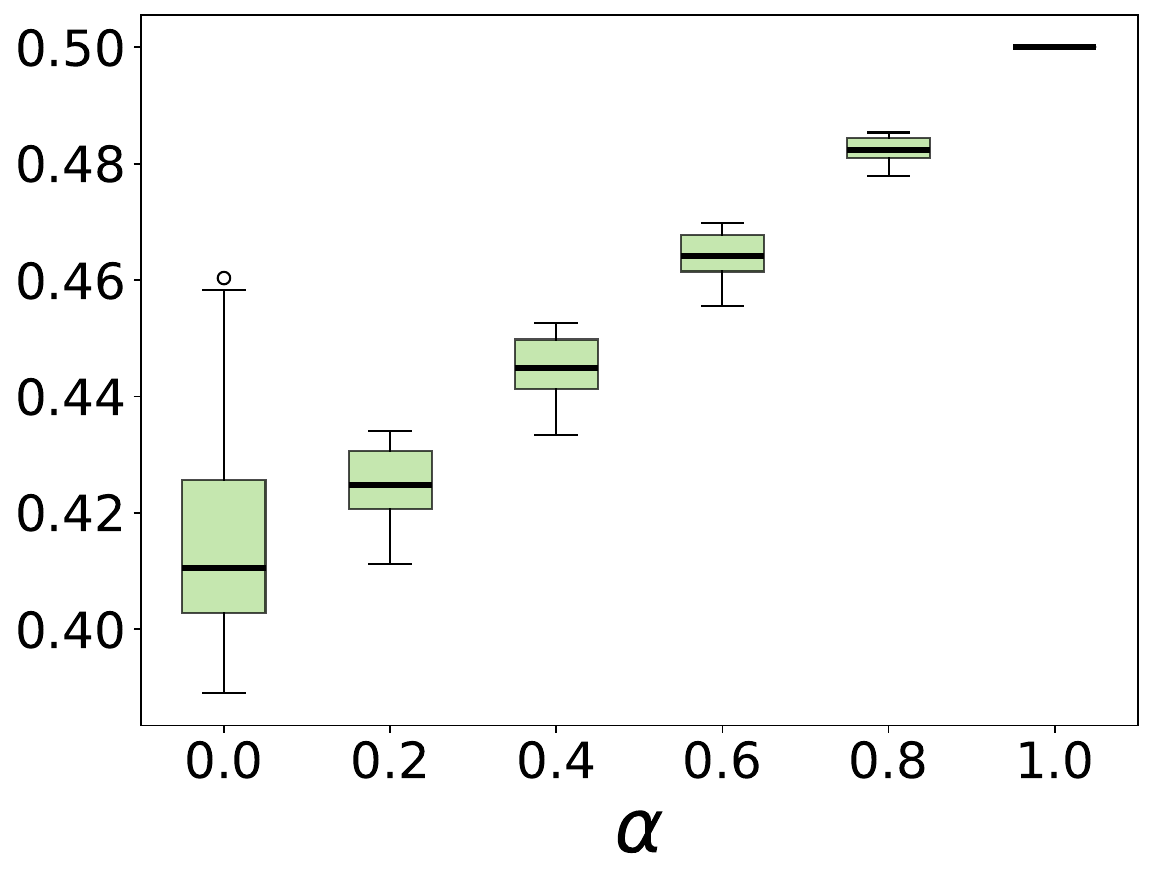}
    \caption{Ex Ante Demand (FEO)}
\end{subfigure}
\begin{subfigure}{0.235\textwidth}
    \centering
    \includegraphics[width=\textwidth]{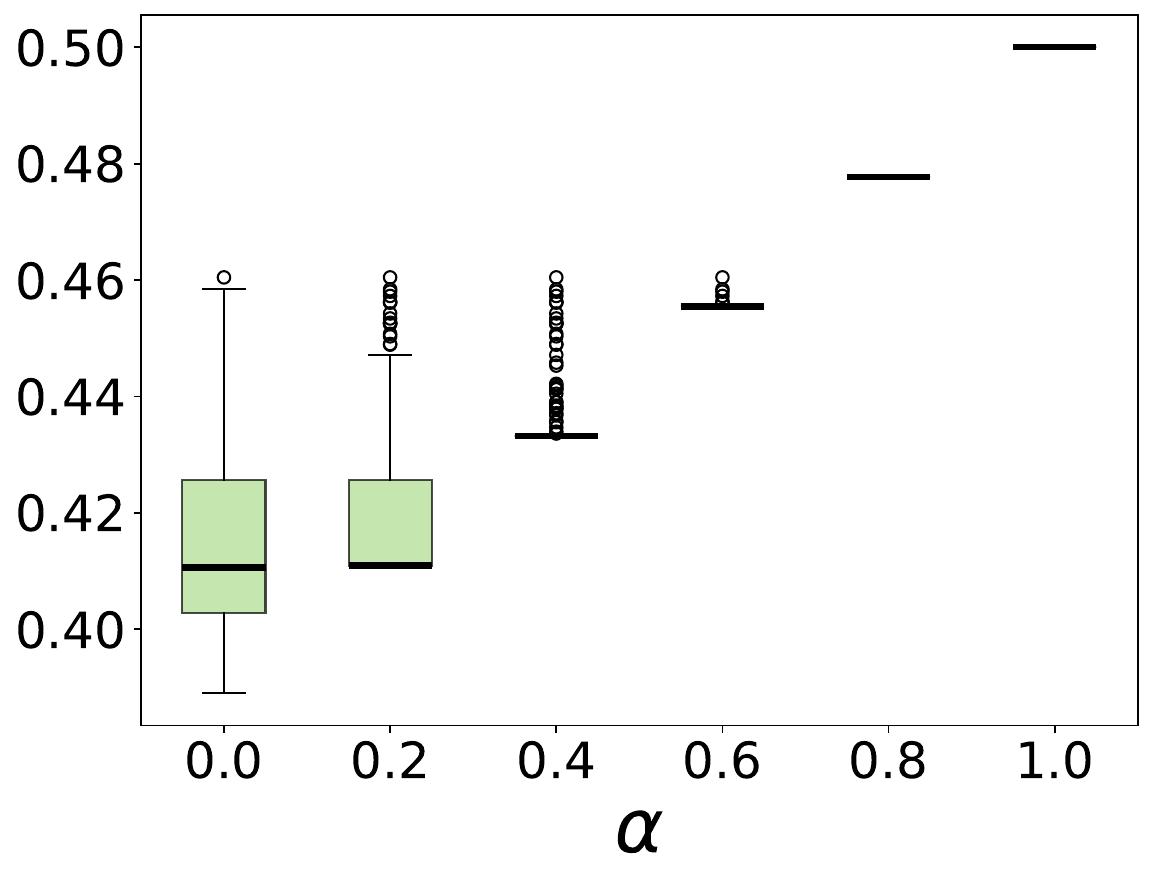}
    \caption{Ex Ante Demand (EFO)}
\end{subfigure}
\begin{subfigure}{0.235\textwidth}
    \centering
    \includegraphics[width=\textwidth]{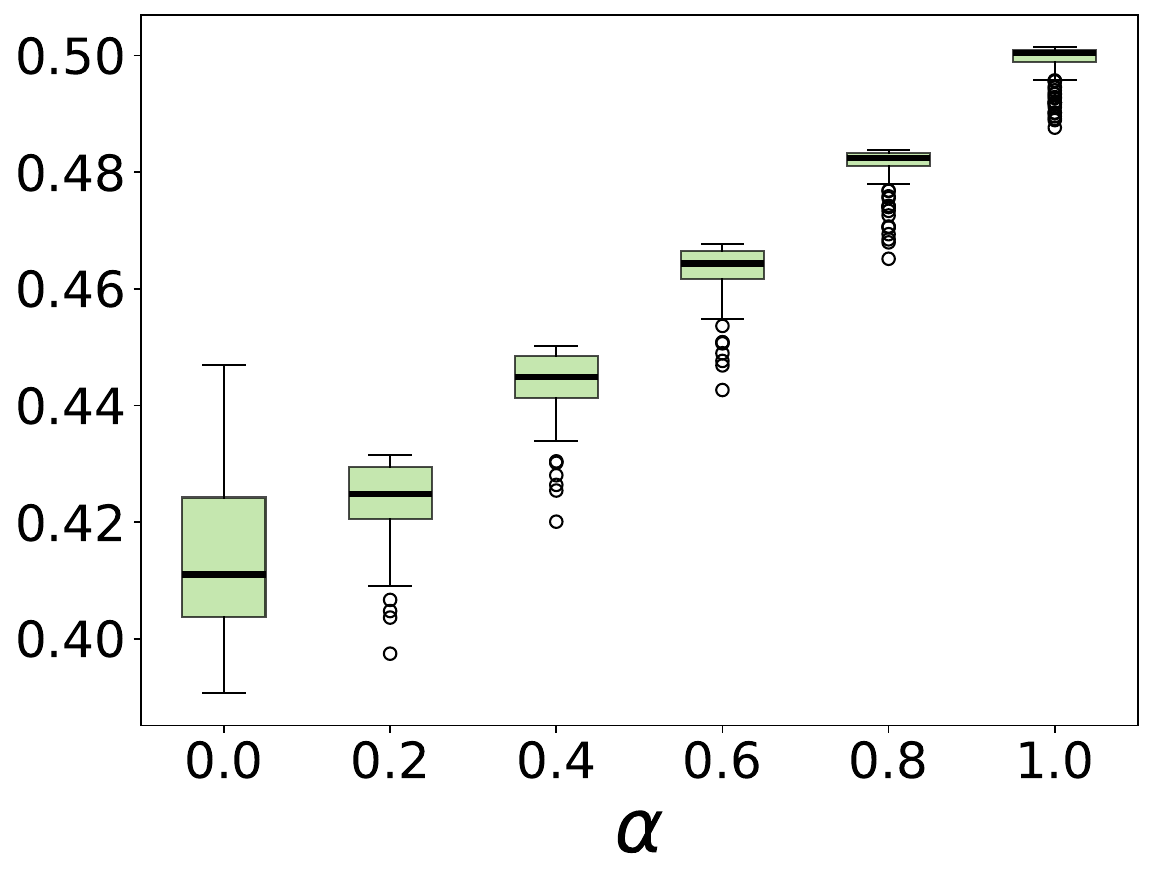}
    \caption{Ex Post Demand (FEO)}
\end{subfigure}
\begin{subfigure}{0.235\textwidth}
    \centering
    \includegraphics[width=\textwidth]{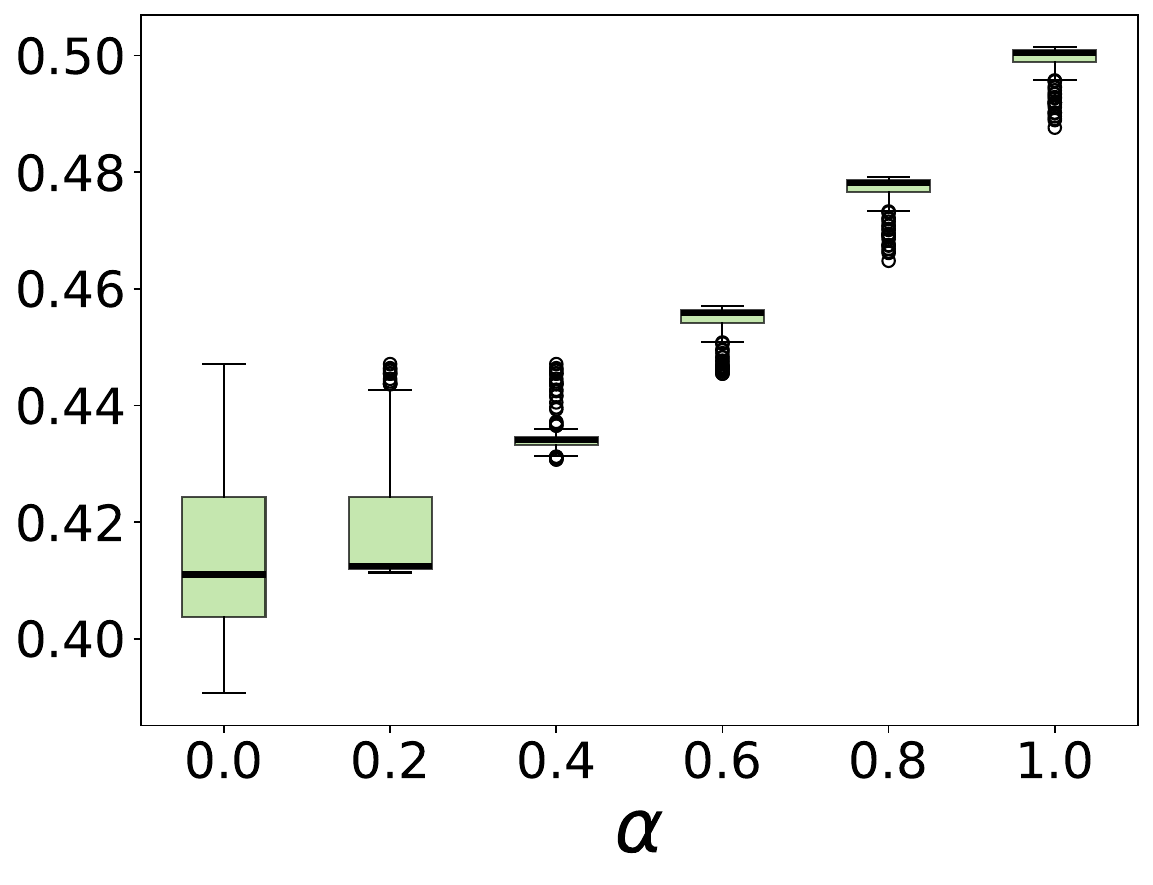}
    \caption{Ex Post Demand (EFO)}
\end{subfigure}
\end{minipage}
}
{Demand distributions for FEO and EFO under Rawlsian demand fairness \label{fig:Rawlsian_demand_linear}\vspace{0.5em}}
{}
\end{figure}

\begin{table}[htbp]
    \centering
    \caption{Normalized performance measures for FEO and EFO under Rawlsian demand fairness (\%)}
    \label{tab:measure_rawlsian_demand_linear}
    \renewcommand{\arraystretch}{1.1}
    \setlength{\tabcolsep}{4pt} % Adjust column spacing for a perfect fit
    
    \begin{tabular}{l ccccc p{0.3cm} ccccc}
    \toprule
    & \multicolumn{5}{c}{FEO ($\alpha$)} & & \multicolumn{5}{c}{EFO ($\alpha$)} \\
    \cmidrule(lr){2-6} \cmidrule(lr){8-12}
    Measure & 0.2 & 0.4 & 0.6 & 0.8 & 1.0 & & 0.2 & 0.4 & 0.6 & 0.8 & 1.0 \\
    \midrule
    
    % Revenue Row
   $\mathcal{R}(\alpha)/\mathcal{R}(0)$  
        & 99.77 &  99.39  & 98.60  & 97.44  & 95.96 &
        & 99.96 &  99.67  & 98.97  & 97.75 &  95.96 \\ 
        
    $\mathcal{S}(\alpha)/\mathcal{S}(0)$  
        & 103.42 & 113.44 & 123.48 & 133.39 & 143.33 &
        & 102.04 & 108.99 & 118.86 & 130.77 & 143.31 \\ 
        
    $\mathcal{W}(\alpha)/\mathcal{W}(0)$  
        & 100.98 & 104.05 & 106.86 & 109.38 & 111.68 &
        & 100.65 & 102.77 & 105.57 & 108.71 & 111.68 \\ 

    \bottomrule
    \end{tabular}
\end{table}

\subsection{Feature-based Logistic Demand Models}
\label{sec:logistic_model}
In this section, we examine an alternative specification of the demand function, namely the logistic model. 
Following the approach in Section~\ref{sec:linear_model}, we formulate the problem, analyze its key properties, 
and investigate whether the results obtained in the stylized model continue to hold in this setting.

We consider the same feature vector 
$\vx^{(i)} = [1, x^{(i)}_1, x^{(i)}_2, \ldots, x^{(i)}_m] \in \mathbb{R}^{m+1}$ 
as in Section~\ref{sec:linear_model}. 
The logistic demand function is defined as
\begin{equation*}
    d^{\mathrm{CE}}(\vx,p;a,b) := \sigma(a^\top \vx + bp),
\end{equation*}
where $\sigma(z) := \frac{1}{1 + \exp(-z)}$ is the sigmoid function. 
$\mathrm{CE}$ denotes the cross-entropy loss, which is equivalent (up to a constant factor) to the deviance of the logistic demand model, given by
\begin{equation*}
\ell^{\mathrm{CE}}(\hat{a}, \hat{b}) 
= - \sum_{g\in\{0,1\}} \sum_{\{ i \mid g^{(i)} = g \}} \frac{1}{n_g}
\Bigl[
    d^{(i)} \log \sigma\!\bigl(\hat{a}^{\top} x^{(i)} + \hat{b} p^{(i)}\bigr)
    + (1 - d^{(i)}) \log\bigl(1 - \sigma(\hat{a}^{\top} x^{(i)} + \hat{b} p^{(i)})\bigr)
\Bigr].
\end{equation*}

Given a feature vector $\vx$ and the estimated parameters $(\hat a, \hat b)$, the optimal price $p^*$ is defined as
\begin{equation*}
\begin{aligned}
p^*(\vx; \hat{a}, \hat{b}) := \arg\max_{p \geq 0} \, (p - c)  d^{\mathrm{CE}}(\vx,p;\hat a,\hat b)=\frac{1 + \mathbf{W}_0\!\bigl(e^{\hat{a}^\top x^{(i)} + \hat{b}c - 1}\bigr)}{-\hat{b}} + c ,
\end{aligned}
\end{equation*}
where $\mathbf{W}_0(\cdot)$ is the Lambert \emph{W} function.

\paragraph{Parity-wise Fairness.} We first discuss how to formulate price and demand fairness under FEO. Following the structure of~\eqref{eq:FEO_price_parity_feature_linear}, we formulate the $\alpha$-parity-wise price fairness criterion as follows.
\begin{equation}
\begin{aligned}
    \min_{\hat{a},\, \hat{b} \le 0} \quad & \ell^{\mathrm{CE}}(\hat{a}, \hat{b}) \\
    \text{s.t.} \quad &\bigg|\frac{1}{n_0}\sum_{\{i|g^{(i)}=0\}}p^*(\vx^{(i)};\hat a, \hat b)-\frac{1}{n_1}\sum_{\{i|g^{(i)}=1\}}p^*(\vx^{(i)};\hat a, \hat b)\bigg|\le (1-\alpha)\Delta_p,\\
    & p^*(\vx^{(i)};\hat a, \hat b)\le \tilde p,\quad \forall i\in[n],
\end{aligned}
\label{eq:FEO_price_parity_feature_logistic}
\end{equation}
where $\Delta_p$ denotes the baseline price gap observed when no fairness constraint is applied (i.e., $\alpha=0$), and $\tilde p$ is the price cap.

Similarly, the formulation of demand fairness follows the same structure as in~\eqref{eq:FEO_price_parity_feature_logistic}, with the fairness constraint replaced by
$$\bigg|\,\frac{1}{n_0}\sum_{\{i|g^{(i)}=0\}}d^{\mathrm{CE}}\left(\vx,p^*(\vx;\hat a,\hat b);\hat a^{\text{CE}},\hat b^{\text{CE}}\right)-\frac{1}{n_1}\sum_{\{i|g^{(i)}=1\}}d^{\mathrm{CE}}\left(\vx,p^*(\vx;\hat a^{\text{CE}},\hat b^{\text{CE}});\hat a,\hat b\right)\bigg|\le (1-\alpha)\Delta_d.$$
Here, $\Delta_d$ denotes the initial demand gap between two groups, and $(\hat a^{\text{CE}}, \hat b^{\text{CE}})$ denotes the optimal solution obtained in the absence of fairness criteria, i.e., $(\hat a^{\text{CE}}, \hat b^{\text{CE}}) := \arg\min_{\hat a, \hat b \le 0} \ell^{\text{CE}}(\hat a, \hat b)$. The term $d^{\text{CE}}(\vx, p^*(\vx;\hat a, \hat b); \hat a^{\text{CE}}, \hat b^{\text{CE}})$ represents the ex-ante demand. Analogous to the linear demand case in \eqref{eq:ex_ante_linear}, this is derived by treating the estimates $(\hat a^{\text{CE}}, \hat b^{\text{CE}})$ as the ground-truth demand parameters. The demand is evaluated based on these fixed parameters, while the price follows the optimal pricing policy $p^*(\vx; \hat a, \hat b)$ dictated by the current decision parameters $(\hat a, \hat b)$.

Under the EFO framework, we assume the estimated parameters $(\hat a^{\text{CE}}, \hat b^{\text{CE}})$ to be the ground truth and maximize total profit based on this fixed model. Specifically, the parity-wise price fairness problem is formulated as follows.
\begin{equation}
\begin{aligned}
    \max_{p_i \ge 0} \quad & \sum_{i\in[n]} (p_i - c)  d^{\mathrm{CE}}(\vx^{(i)},p_i;\hat a^{\text{CE}},\hat b^{\text{CE}}) \\
    \text{s.t.} \quad &\bigg|\frac{1}{n_0}\sum_{\{i|g^{(i)}=0\}}p_i-\frac{1}{n_1}\sum_{\{i|g^{(i)}=1\}}p_i\bigg|\le (1-\alpha)\Delta_p,\\
    & p_i\le \tilde p,\quad \forall i\in[n].
\end{aligned}
\label{eq:EFO_price_parity_feature_logistic}
\end{equation}
Similarly, parity-wise demand fairness shares the same objective function, while the fairness constraint is defined by the difference in the mean predicted demand between the two groups.

\paragraph{Rawlsian Fairness.} First, within the FEO framework, the objective remains the minimization of the loss function, consistent with~\eqref{eq:FEO_price_parity_feature_logistic}. Depending on the fairness criteria employed, the constraints are adapted accordingly. For instance, under Rawlsian price fairness, the estimated parameters $(\hat a, \hat b)$ must ensure that the resulting optimal price for every individual does not exceed a specified threshold, i.e.,
$$p^*(\mathbf{x}^{(i)}; \hat a, \hat b) \le (1 - \alpha)\bar{p} + \alpha \underline{p}, \quad \forall i \in [n].$$
Under Rawlsian demand fairness, the framework ensures that each individual's demand is maintained above a specific threshold, i.e., 
$$d^{\mathrm{CE}}\left(\vx^{(i)},p^*(\mathbf{x}^{(i)};\hat a, \hat b);\hat a^{\text{CE}},\hat b^{\text{CE}}\right) \ge \alpha\bar{d} + (1 - \alpha) \underline{d}, \quad \forall i \in [n].$$

Within the EFO framework, the objective is to maximize the profit as defined in~\eqref{eq:EFO_price_parity_feature_logistic}, with $p_i$ as the decision variable. For instance, Rawlsian price fairness imposes the following constraint while keeping the same objective:
$$p_i \le (1-\alpha)\bar{p} + \alpha \underline{p}, \quad \forall i \in [n].$$
Similarly, Rawlsian demand fairness sets a lower bound for the demand, which is determined by the decision variable $p_i$ and estimates $(\hat a^{\text{CE}}, \hat b^{\text{CE}})$.
% \begin{equation}
% \begin{aligned}
%     \min_{\hat{a},\, \hat{b} \le 0} \quad & \ell^{\mathrm{CE}}(\hat{a}, \hat{b}) \\
%     \text{s.t.} \quad 
%     & \frac{1 + \mathbf{W}_0\!\bigl(e^{\hat{a}^\top x^{(i)} + \hat{b}c - 1}\bigr)}{-\hat{b}} + c 
%     \le (1 - \alpha)\bar{p} + \alpha \underline{p}, 
%     \quad \forall i \in [n],
% \end{aligned}
% \label{eq:FEO_price_rawlsian_feature_logistics}
% \end{equation}
% where $\bar{p}$ (resp.\ $\underline{p}$) denotes the maximum (resp.\ minimum) prices. Note that 
% \[
% \frac{1 + \mathbf{W}_0\!\bigl(e^{\hat{a}^\top x + \hat{b}c - 1}\bigr)}{-\hat{b}} + c
% \]
% is the optimal price that maximizes profit under the demand function $d^{\mathrm{CE}}(x,p;\hat{a},\hat{b})$.

% The formulation of demand fairness shares the same objective structure as that in~\eqref{eq:FEO_price_rawlsian_feature_logistics} with different constraint
% \begin{equation*}
% d^{\mathrm{CE}}\!\left(x^{(i)},\frac{1 + \mathbf{W}_0\!\bigl(e^{\hat{a}^\top x^{(i)} + \hat{b}c - 1}\bigr)}{-\hat{b}} + c;\hat{a}^{\mathrm{CE}},\hat{b}^{\mathrm{CE}}\right)
% \geq \alpha \bar{d} + (1-\alpha)\underline{d},
% \end{equation*}
% where $\bar{d}$ (resp. $\underline{d}$) is the maximum (resp. minimum) demand.% under the unconstrained problem.

Under both price and demand fairness, the FEO problem is not a convex optimization problem. 
However, similar to the linear demand case, Proposition~\ref{prop:convexity_logistic} shows that, 
despite this non-linearity, the problem can be reformulated as a convex optimization problem.

\begin{proposition}
Within the FEO framework, Rawlsian price and demand fairness can be formulated as convex optimization problems under logistic demand.
\label{prop:convexity_logistic}
\end{proposition}

% Moreover, the property established for the linear demand model -- that under Rawlsian fairness, 
% the loss value changes in a convex manner with respect to small variations in the fairness level 
% (as shown in Proposition~\ref{prop:convexity_wrt_error_linear}) -- also holds under the logistic demand model. 
% Corollary~\ref{prop:convexity_wrt_error_logistic} formalizes this result.

% \begin{corollary}
% Assume that the matrix $[X,\mathbf{p}]$ is full-rank, where $X := [\mathbf{x}^{(1)}, \dots, \mathbf{x}^{(n)}]^\top$ and $\mathbf{p} = [p^{(1)}, \dots, p^{(n)}]^\top$. Let $\ell(\alpha)$ denote the optimal value of the FEO problem with respect to the fairness level $\alpha$. 
% Then, there exists a constant $\delta > 0$ such that $\ell(\alpha)$ is a convex function in $\alpha \in [0, \delta]$, 
% i.e., $\lim_{\alpha \to 0^+} \frac{d^2 \ell(\alpha)}{d\alpha^2} > 0$, 
% under either price or demand fairness with logistic demand.
% \label{prop:convexity_wrt_error_logistic}
% \end{corollary}

%The numerical results under the logistic demand are presented in Appendix \ref{appendix:synthetic_logistic}.

\section{Case Study: Experiments on Real-World Data}
%https://www.openicpsr.org/openicpsr/project/113240/version/V1/view
\label{numerical_results}
We conduct a case study based on survey data about vaccination against tick-borne encephalitis (TBE) in Sweden from \citet{slunge2015willingness}, which includes an assessment of individuals' willingness to receive the TBE vaccine at randomly assigned prices. This dataset is particularly suitable for our purposes because vaccination decisions are influenced not only by willingness-to-pay and socioeconomic status, but also by broader considerations of public health and equity. In such policy contexts, pricing may be guided by profit motives as well as concerns about fairness. 

\subsection{Dataset and Setup}
The dataset contains information on price ($p$), willingness to buy ($d\in\{0,1\}$), and a range of covariates, including age, gender, income, education level, area of residence, knowledge about TBE, health status, trust in vaccine recommendations, and risk-related factors. \citet{slunge2015willingness} indicates that a substantial proportion of low-income individuals at risk of TBE are deterred from vaccination due to the prevailing market price of the vaccine. Accordingly, we define the income indicator $g$, where $g=1$ (resp. $g=0$) denotes low (resp. high) income, with individuals classified as low (resp. high) income if their income is below (resp. above) $20{,}000$~SEK, following \citet{slunge2015willingness}. The dataset has also been used in recent fair pricing studies \citep{kallus2021fairness, xu2022regulatory}.

After excluding entries with missing values and restricting to individuals who have not previously been vaccinated, the dataset comprises $1,151$ observations. As covariates, we include \textit{female}, \textit{age}, \textit{income}, \textit{university}, \textit{urban}, \textit{tbeincidence}, \textit{summerhouserisk}, \textit{outTBEarea}, \textit{worktickrisk}, \textit{knowledge}, \textit{tickbiteever}, \textit{diseaseexperience}, \textit{healthrisktickbite}, and \textit{lowtrustvaccinerec} in the feature vector $\mathbf{x}$, following the set of variables identified as significant in \citet{slunge2015willingness}. All feature vectors are normalized, and prices are scaled to lie within the interval $[0,1]$. We set the (normalized) unit cost to $0.05$ (corresponding to $50$ SEK without normalization), which is set below the minimum observed price ($0.1$), and impose a (normalized) price cap of $2.0$ (corresponding to $2000$ SEK), which is set above the maximum observed price ($1.0$). We adopt an $80/20$ split between the training and test sets, and all reported results are based on the test set.

We use a logistic demand model, $d = \sigma(a^\top x + bp)$, for estimation in our case study. 
To evaluate profit, consumer surplus, and social welfare, it is necessary to recover the true demand at each price conditional on the observed features. For this purpose, we consider two scenarios. In the first, the true demand is logistic, which we refer to as the \emph{well-specified} case (Section~\ref{appendix:well_specified_logistic}). In the second, the true demand follows a linear model, $d = 1+\left(b^\top x\right)p$, which we treat as the \emph{mis-specified} case (Appendix~\ref{appendix:mis_specified_logistic}). We also consider estimation using a linaer demand model, and results for the well-specified and mis-specified cases are reported in Appendix~\ref{appendix:linear_well_specified} and Appendix~\ref{appendix:linear_misspecified}, respectively.

To ensure that the data used for training are generated consistently with the assumed true model, we regenerate demand outcomes based on the corresponding demand function. %(either linear or logistic). 
Lastly, we set $\bar{p}$ (resp. $\underline{p}$) in~\eqref{eq:rawls_price_feature} as the highest (resp. lowest) optimal price obtained on the training set using parameters estimated without fairness constraints. We set $\bar{d}$ as the lowest demand in the training set when charging everybody price $p=c$, and $\underline{d}$ as the lowest demand observed in the training set without fairness constraints.

\subsection{Results}
\label{appendix:well_specified_logistic}
Table \ref{tab:measure_loss_fairness_logistic} illustrates the outcomes under parity-wise loss fairness. When $\alpha\geq 0.6$, the profit, consumer surplus and social welfare are all lower than the no-fairness baseline, implying that parity-wise loss fairness may harm the whole system.

\begin{table}[htbp]
    \centering
    \caption{Normalized performance measures across $\alpha$ under parity-wise loss fairness (\%)}
    \label{tab:measure_loss_fairness_logistic}
    \renewcommand{\arraystretch}{1.1} 

    \begin{tabular}{lccccc}
    \toprule
    & \multicolumn{5}{c}{$\alpha$} \\
    \cmidrule(lr){2-6}
    Measure & 0.2 & 0.4 & 0.6 & 0.8 & 1.0 \\
    \midrule
    $\mathcal{R}(\alpha)/\mathcal{R}(0)$ & 100.11 & 100.01 & 99.86 & 99.77 & 99.64 \\
    $\mathcal{S}(\alpha)/\mathcal{S}(0)$ & 95.04 &  92.04 & 90.99 & 91.52 & 91.90 \\
    $\mathcal{W}(\alpha)/\mathcal{W}(0)$ & 97.84 &  96.44 & 95.89 & 96.08 & 96.17 \\
    \bottomrule
    \end{tabular}
\end{table}

Figure~\ref{fig:logistic_parity_price_case} shows the price distributions under parity-wise price fairness. At $\alpha=1$, the average prices charged to the low-income and high-income groups become identical under EFO, whereas this is not the case under FEO. This discrepancy arises because the model is evaluated on test data with a finite sample size, where the empirical distribution of the test set inevitably deviates from that of the training set. %Notably, FEO results in a narrower price distribution. Unlike EFO, which treats each price as a separate decision variable, FEO optimizes a pricing policy tied to estimated demand, which forces prices to move in tandem. 
Table~\ref{tab:measure_parity_price_logistic} shows that FEO yields higher consumer surplus and social welfare, which implies that Corollary~\ref{corollary2} not only holds for linear demand models but also for logistic demand models. For instance, When we aim to reduce the mean price gap by $40\%$ ($\alpha = 0.4$), allowing an additional $0.65\%$ loss in profit leads FEO to achieve $11.39\%$ higher consumer surplus, whereas EFO does not make improvements in consumer surplus. This suggests that imposing parity-wise price fairness in FEO is more beneficial to customers.%More specifically, sacrificing $0.65\%$ of the profit leads to more than $11\%$ increase in consumer surplus. This suggests that imposing parity-wise price fairness in FEO is more beneficial to customers.

\begin{figure}[htbp]
\FIGURE{
\begin{minipage}{\textwidth}
\centering
\captionsetup{justification=centering}
\begin{subfigure}{0.35\textwidth}
    \centering
    \includegraphics[width=\textwidth]{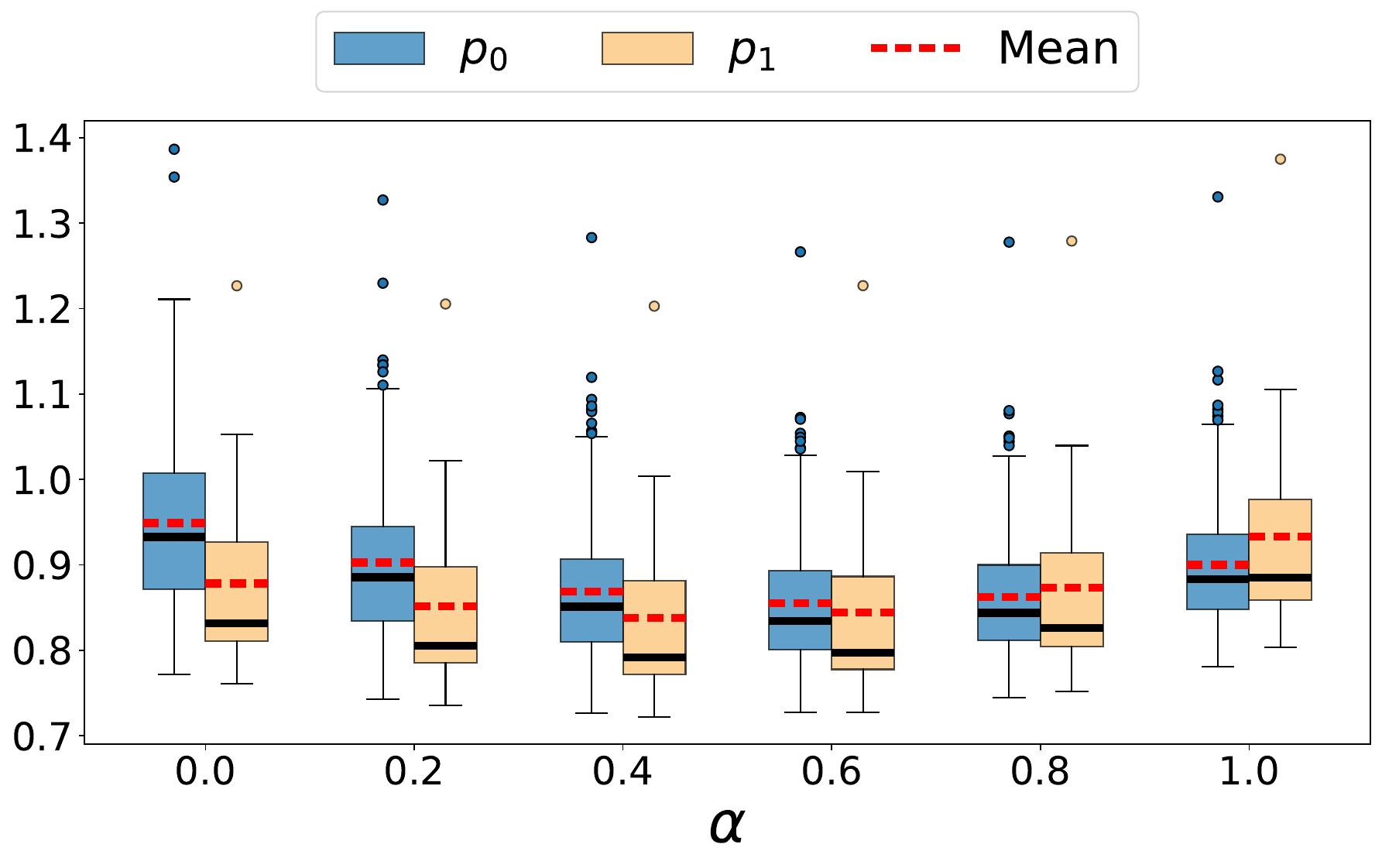}
    \caption{Price (FEO)}
\end{subfigure}
\begin{subfigure}{0.35\textwidth}
    \centering
    \includegraphics[width=\textwidth]{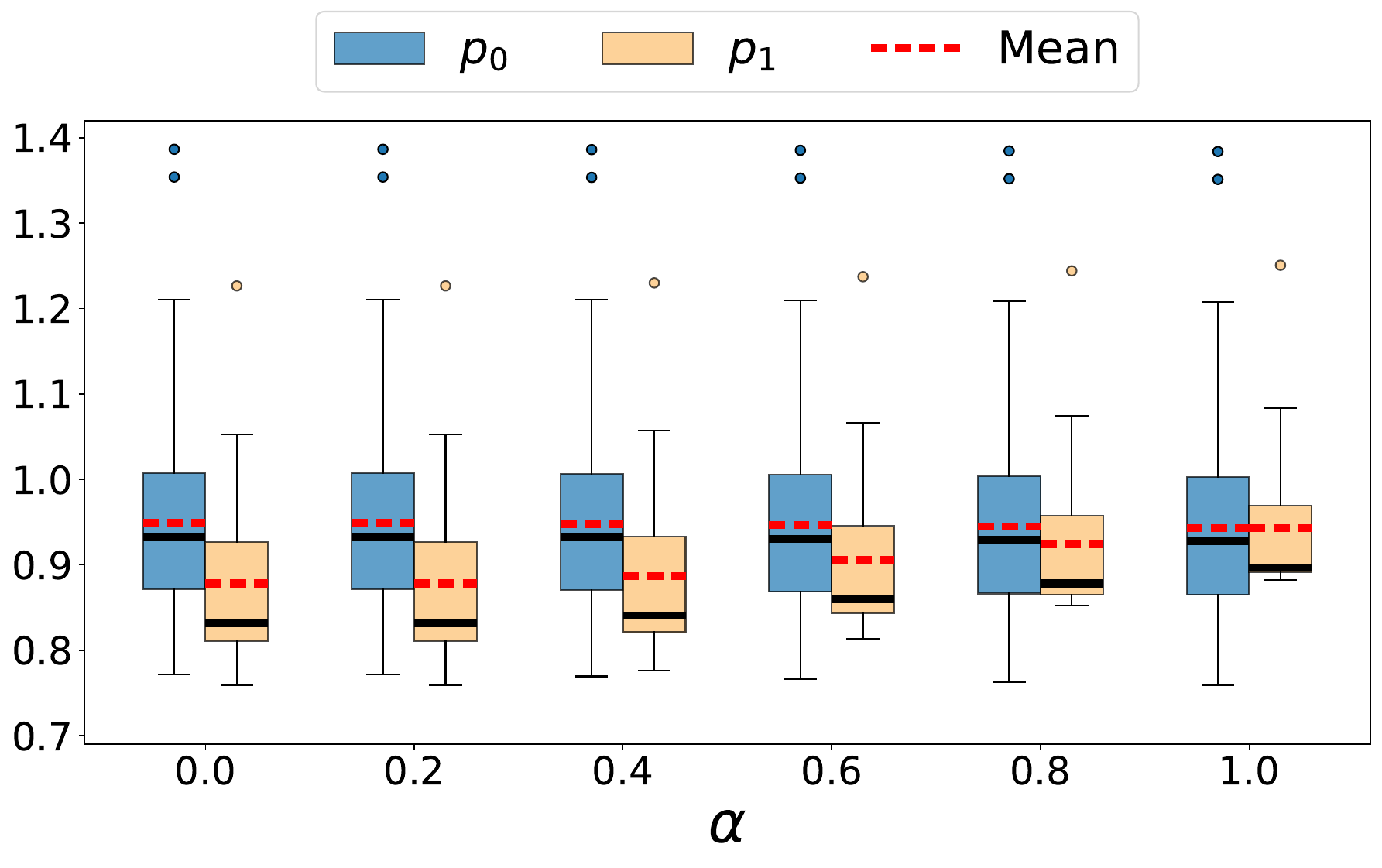}
    \caption{Price (EFO)}
\end{subfigure}
\end{minipage}
}
{Price distributions for FEO and EFO under parity-wise price fairness \label{fig:logistic_parity_price_case}\vspace{1mm}}
{}
\end{figure}

\begin{table}[htbp]
    \centering
    \caption{Normalized performance measures for FEO and EFO under parity-wise price fairness (\%)}
    \label{tab:measure_parity_price_logistic}
    \renewcommand{\arraystretch}{1.1}
    \setlength{\tabcolsep}{4pt}
     
    \begin{tabular}{l ccccc c ccccc}
        \toprule
        & \multicolumn{5}{c}{FEO ($\alpha$)} & & \multicolumn{5}{c}{EFO ($\alpha$)} \\
        \cmidrule(lr){2-6} \cmidrule(lr){8-12}
        Measure & 0.2 & 0.4 & 0.6 & 0.8 & 1.0 & & 0.2 & 0.4 & 0.6 & 0.8 & 1.0 \\
        \midrule
        % Revenue Row
        $\mathcal{R}(\alpha)/\mathcal{R}(0)$  
        & 99.81 &  99.35 &  99.02  & 99.01  & 99.33  & 
        & 100.0 & 100.0 & 99.99 & 99.98 & 99.96  \\ 
        
        % Surplus Row
        $\mathcal{S}(\alpha)/\mathcal{S}(0)$  
        & 106.46 & 111.39 & 113.52 & 112.43 & 106.94 & 
        & 100.0 & 100.0 & 99.99 & 99.97 & 99.95 \\ 
        
        % Welfare Row  
        $\mathcal{W}(\alpha)/\mathcal{W}(0)$  
        & 102.79 & 104.74 & 105.51 & 105.02 & 102.73 & 
        & 100.0 & 100.0 & 99.99 & 99.97 & 99.95 \\ 
        \bottomrule
    \end{tabular}
\end{table}

Figure~\ref{fig:logistic_parity_demand_case} presents the results of parity-wise demand fairness. First, the ex-ante demand under EFO achieves perfect demand fairness, whereas FEO does not. Similar to parity-wise price fairness, this discrepancy arises from differences between the empirical distributions of the training and test data. Due to estimation error, the estimated demand function differs from the true demand model. Nevertheless, the fairness constraints still effectively shift the ex-post demand gap between the two groups in a similar way to the ex-ante demand for both FEO and EFO.

%Due to the estimation error, the estimated demand function differs from the true demand model, but the fairness constraint still effectively shifts the average demand gap between the two groups for both FEO and EFO. 
Table~\ref{tab:measure_parity_demand_logistic} shows that, to achieve the same fairness level $\alpha$, EFO leads to much smaller changes in profit, consumer surplus and social welfare than FEO. Moreover, when $\alpha\leq 0.8$, the profit, consumer surplus and social welfare of EFO are all higher than FEO. This suggests that using EFO to achieve parity-wise demand fairness is a safer option for both customers and the seller.%stable and beneficial option.

\begin{figure}[htbp]
\FIGURE{
\begin{minipage}{\textwidth}
\centering
\captionsetup{justification=centering}
\begin{subfigure}{0.35\textwidth}
    \centering
    \includegraphics[width=\textwidth]{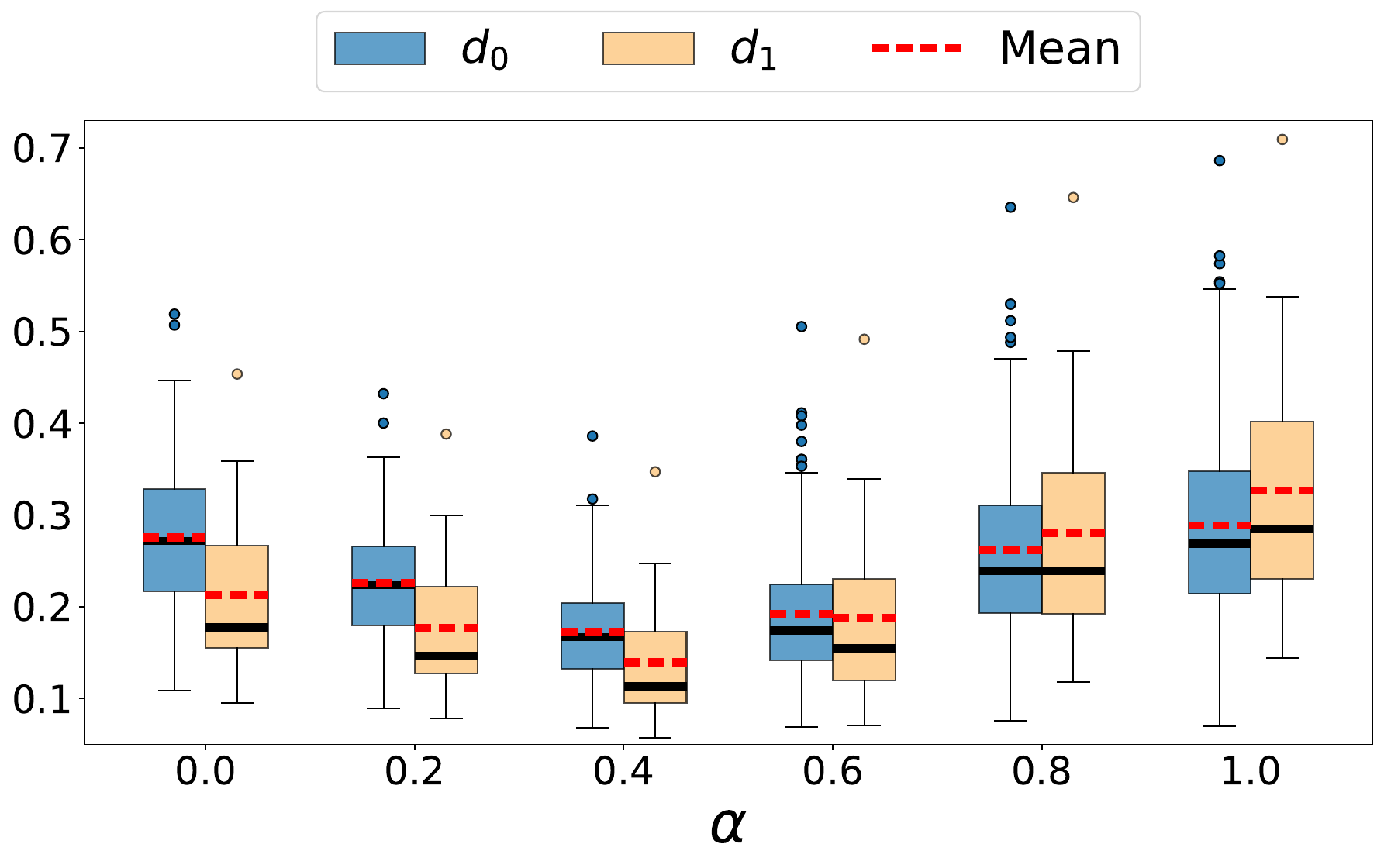}
    \caption{Ex Ante Demand (FEO)}
\end{subfigure}
\begin{subfigure}{0.35\textwidth}
    \centering
    \includegraphics[width=\textwidth]{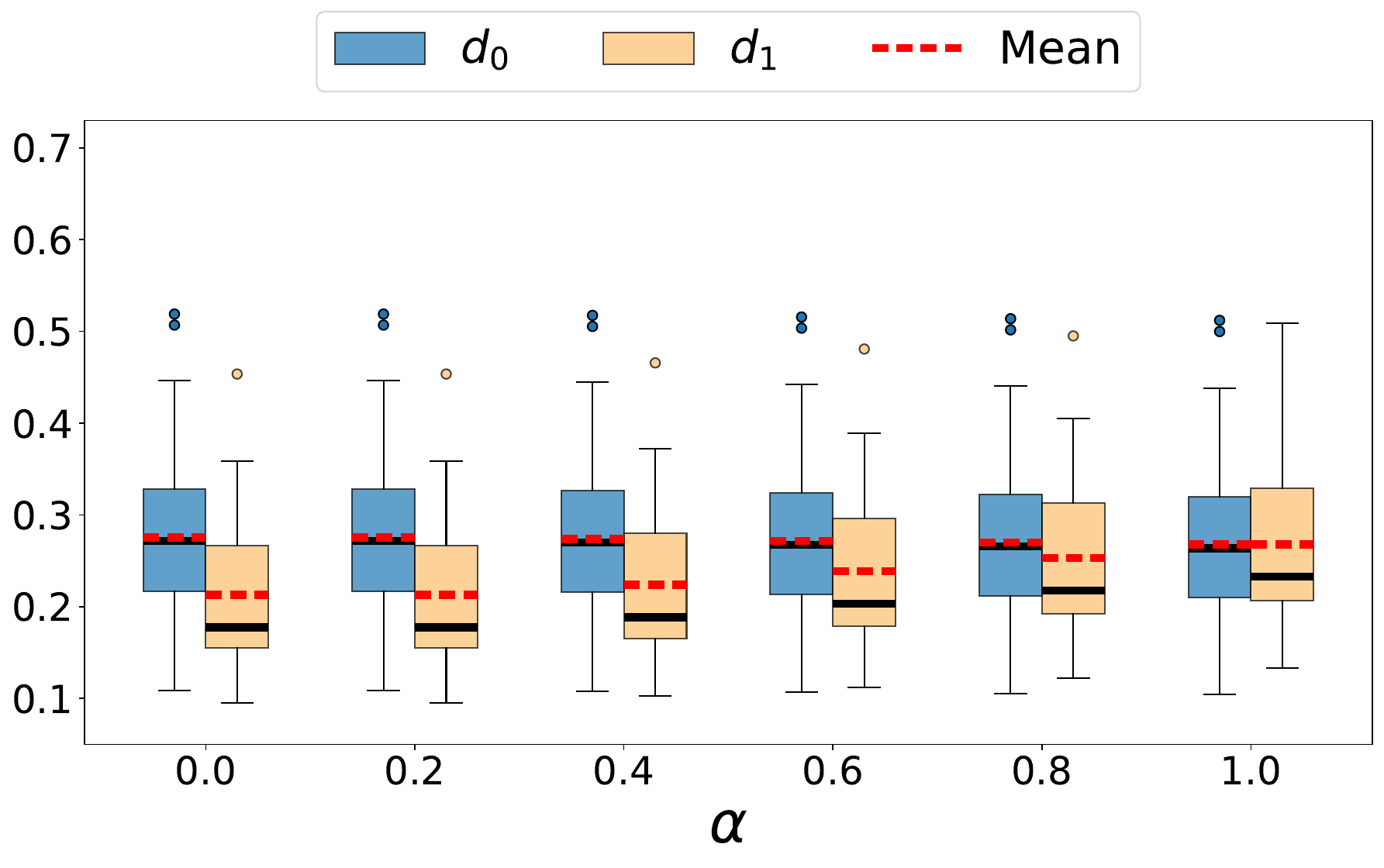}
    \caption{Ex Ante Demand (EFO)}
\end{subfigure}
\begin{subfigure}{0.35\textwidth}
    \centering
    \includegraphics[width=\textwidth]{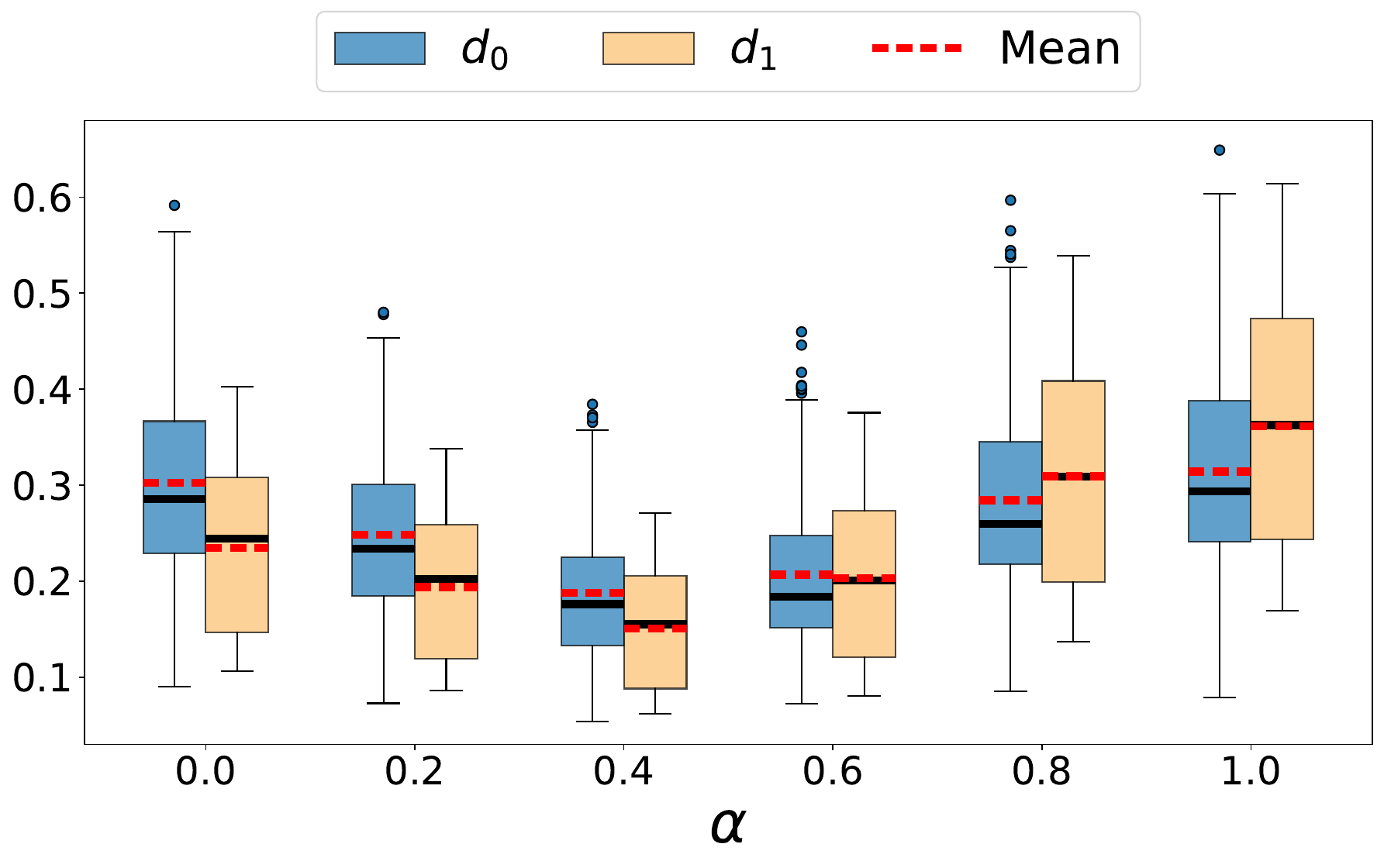}
    \caption{Ex Post Demand (FEO)}
\end{subfigure}
\begin{subfigure}{0.35\textwidth}
    \centering
    \includegraphics[width=\textwidth]{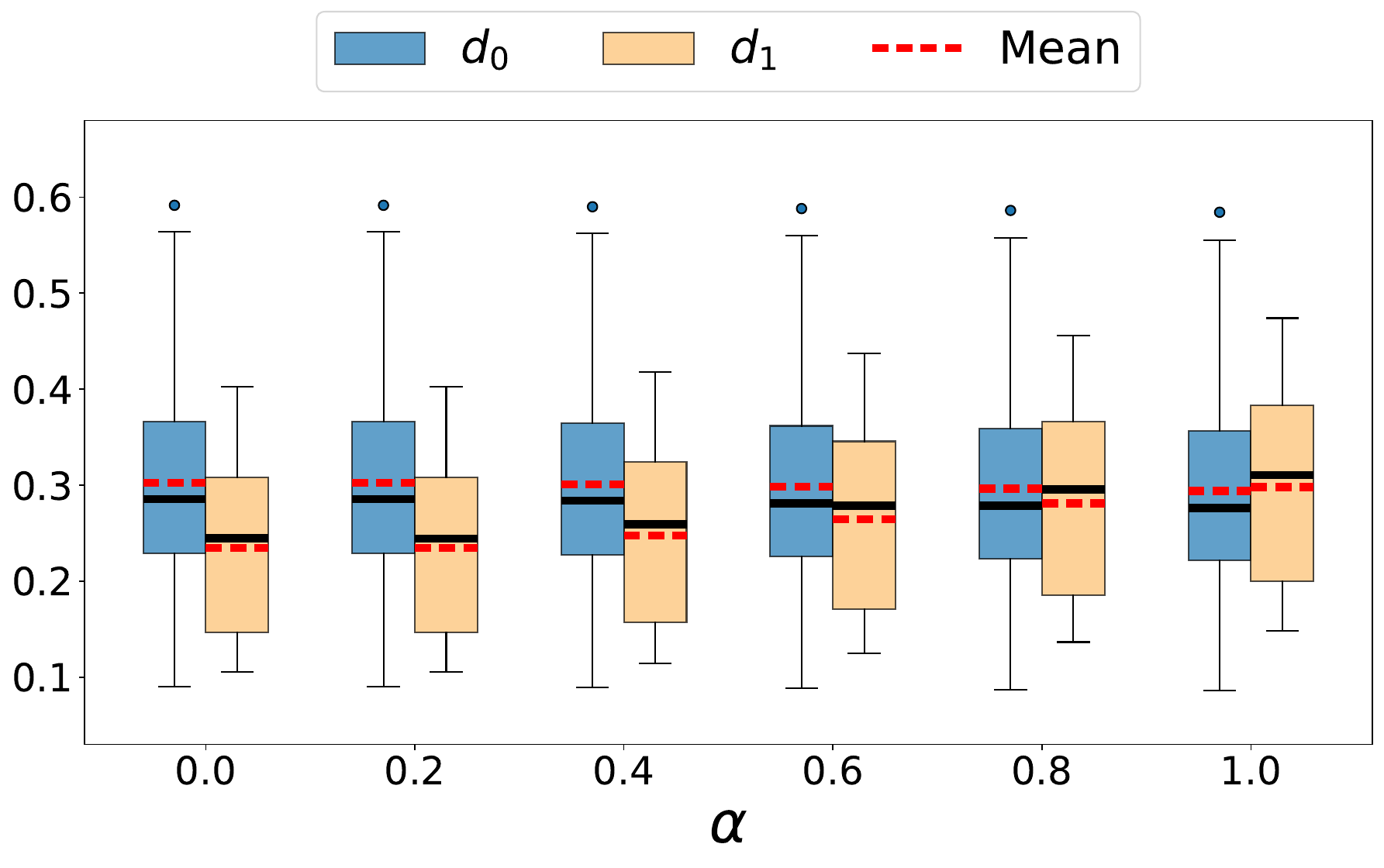}
    \caption{Ex Post Demand (EFO)}
\end{subfigure}
\end{minipage}
}
{Demand distributions for FEO and EFO under parity-wise demand fairness \label{fig:logistic_parity_demand_case}\vspace{2mm}}
{}
\end{figure}

\begin{table}[htbp]
    \centering
    \caption{Normalized performance measures for FEO and EFO under parity-wise demand fairness (\%)}
    \label{tab:measure_parity_demand_logistic}
    \renewcommand{\arraystretch}{1.1}
    
    \begin{tabular}{l ccccc c ccccc}
        \toprule
        & \multicolumn{5}{c}{FEO ($\alpha$)} & & \multicolumn{5}{c}{EFO ($\alpha$)} \\
        \cmidrule(lr){2-6} \cmidrule(lr){8-12}
        Measure & 0.2 & 0.4 & 0.6 & 0.8 & 1.0 & & 0.2 & 0.4 & 0.6 & 0.8 & 1.0 \\
        \midrule
        $\mathcal{R}(\alpha)/\mathcal{R}(0)$ & 97.50 & 89.15 & 91.61 & 95.62 &  93.10 & & 100.0 & 99.99 & 99.92 & 99.79 & 99.60 \\
        
        $\mathcal{S}(\alpha)/\mathcal{S}(0)$ & 78.76 & 57.37 & 65.71 & 97.26 & 111.32  & & 100.0 & 99.90 & 99.81 & 99.77 & 99.76  \\
        
        $\mathcal{W}(\alpha)/\mathcal{W}(0)$ & 89.11 & 74.93 & 80.01 & 96.35 & 101.26  & & 100.0 & 99.95 & 99.87 & 99.78 & 99.67 \\
        \bottomrule
    \end{tabular}
\end{table}

Figure~\ref{fig:logistic_Rawlsian_price_case} illustrates the price distribution under Rawlsian price fairness. For both FEO and EFO, the maximum price decreases as the fairness level $\alpha$ increases, showing that Rawlsian price fairness effectively protects individuals facing excessively high prices. However, the two approaches affect the distribution differently: FEO shifts the entire price distribution downward, whereas EFO only adjusts prices that violate the fairness constraint. %EFO only alters part of the distribution. 
Table~\ref{tab:measure_rawlsian_price_logistic} shows that FEO leads to higher consumer surplus and social welfare than EFO. Notably, at fairness level $\alpha=0.2$, profit, consumer surplus and social welfare for FEO are all higher than the baseline (i.e., $\alpha=0$), showing that a small level of Rawlsian price fairness for FEO can lead to a win-win outcome.

\begin{figure}[htbp]
\FIGURE{
\begin{minipage}{\textwidth}
\centering
\captionsetup{justification=centering}
\begin{subfigure}{0.35\textwidth}
    \centering
    \includegraphics[width=\textwidth]{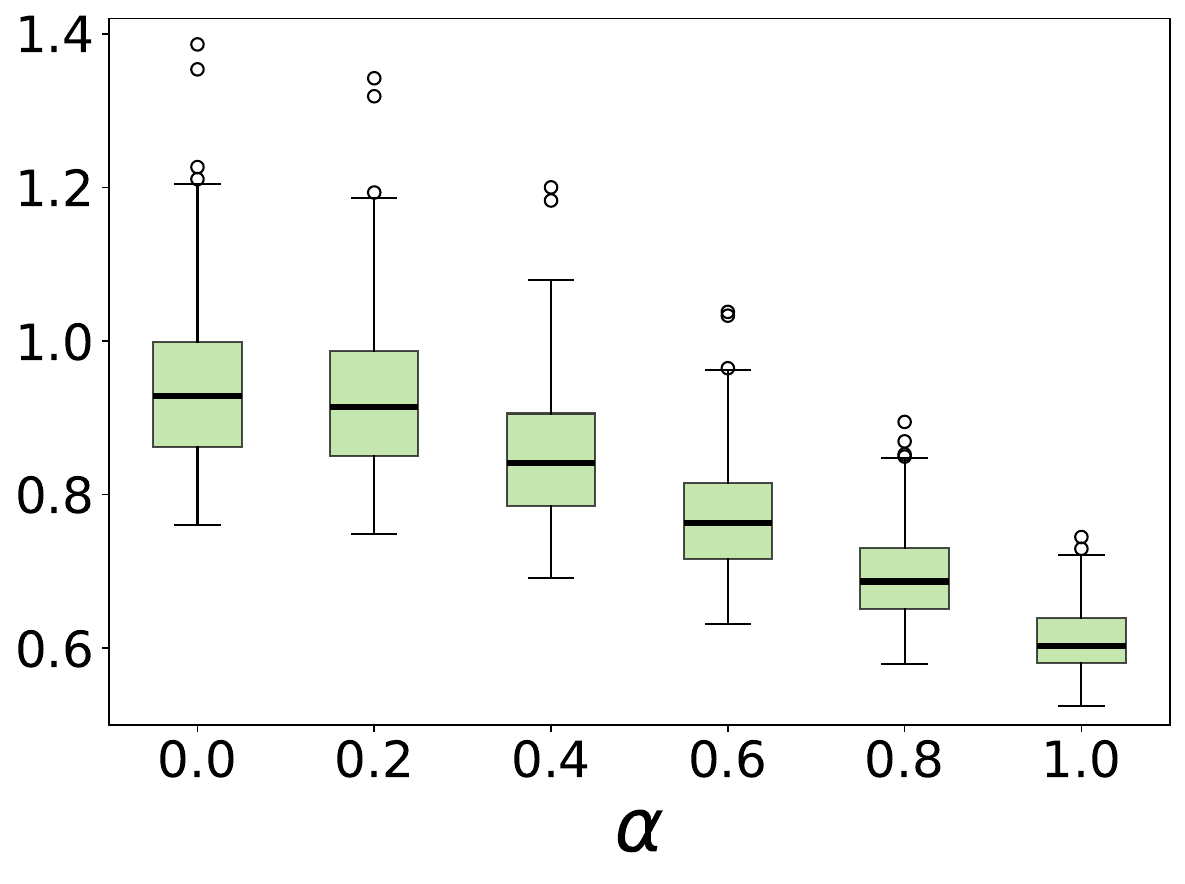}
    \caption{Price under FEO}
    % \label{fig:rawls_price_feo}
\end{subfigure}
\begin{subfigure}{0.35\textwidth}
    \centering
    \includegraphics[width=\textwidth]{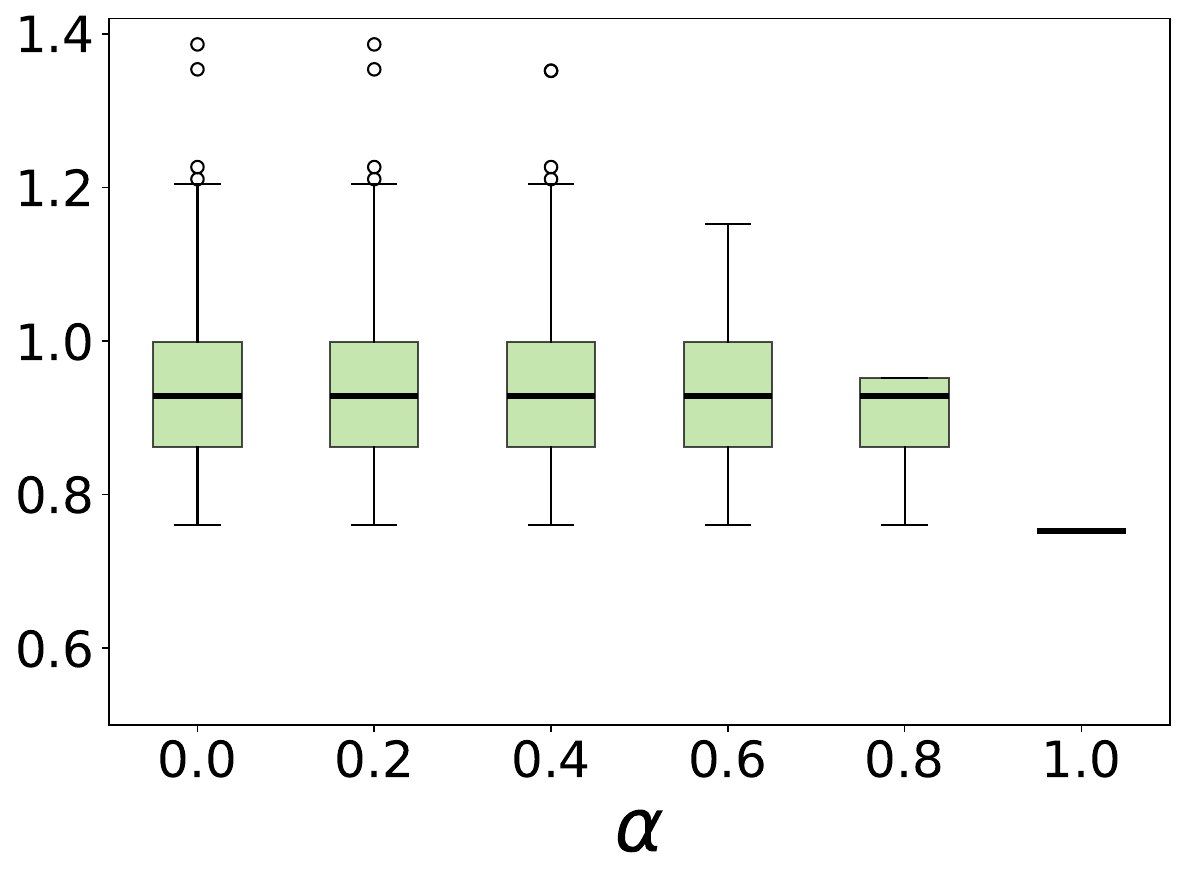}
    \caption{Price under EFO}
    % \label{fig:rawls_price_efo}
\end{subfigure}
\end{minipage}
}
{Price distributions for FEO and EFO under Rawlsian price fairness \label{fig:logistic_Rawlsian_price_case}\vspace{0.5em}}
{}
\end{figure}

Lastly, Figure~\ref{fig:logistic_Rawlsian_demand_case} presents the demand distributions under Rawlsian demand fairness. Similarly, FEO changes the entire demand distribution, whereas EFO only modifies the portion of the population that violates the demand lower bound. %EFO only affects part of the population.
\begin{table}[t]
    \centering
    \caption{Normalized performance measures for FEO and EFO under Rawlsian price fairness (\%)}
    \label{tab:measure_rawlsian_price_logistic}
    \renewcommand{\arraystretch}{1.1}
    % No font size reduction (\small/\footnotesize) as requested
    
    \begin{tabular}{l p{0.2cm} ccccc p{0.4cm} ccccc}
        \toprule
        & & \multicolumn{5}{c}{FEO ($\alpha$)} & & \multicolumn{5}{c}{EFO ($\alpha$)} \\
        \cmidrule(lr){3-7} \cmidrule(lr){9-13}
        Measure & & 0.2 & 0.4 & 0.6 & 0.8 & 1.0 & & 0.2 & 0.4 & 0.6 & 0.8 & 1.0 \\
        \midrule
        $\mathcal{R}(\alpha)/\mathcal{R}(0)$ & & 100.01  & 99.27 &  97.02 &  93.39 &  87.48 & & 100.0 & 100.01 & 100.00  & 99.48  & 94.76\\
        
        $\mathcal{S}(\alpha)/\mathcal{S}(0)$ & & 101.88 & 113.06 & 126.25 & 139.70 & 155.82 & & 100.0 & 100.03 & 100.76 & 107.16 & 132.54\\
        
        $\mathcal{W}(\alpha)/\mathcal{W}(0)$ & & 100.84 & 105.44 & 110.11 & 114.12 & 118.07 & & 100.0 & 100.02 & 100.34 & 102.92 & 111.67\\
        \bottomrule
    \end{tabular}
\end{table}
Table~\ref{tab:measure_rawlsian_demand_logistic} highlights a substantial trade-off: under perfect demand fairness, while FEO delivers a large consumer surplus gain ($113.22\%$ increase), it comes at a significant reduction in profit (a loss of more than $43.45\%$), which may limit its practical feasibility. In contrast, EFO achieves a more modest surplus improvement ($4\%$ increase) with a small impact on profit ($0.91\%$ loss), making it a more implementable policy choice. %shows that FEO leads to more substantial changes in downstream outcomes. When setting $\alpha=1.0$, FEO loses more than $33\%$ of the profit, while the change in profit of EFO is less than $1\%$. This suggests that imposing Rawlsian demand fairness for EFO in practice gives a more predictable solution to the policymaker. 

\begin{figure}[htbp]
\FIGURE{
\begin{minipage}{\textwidth}
\centering
\captionsetup{justification=centering}
\begin{subfigure}{0.235\textwidth}
    \centering
    \includegraphics[width=\textwidth]{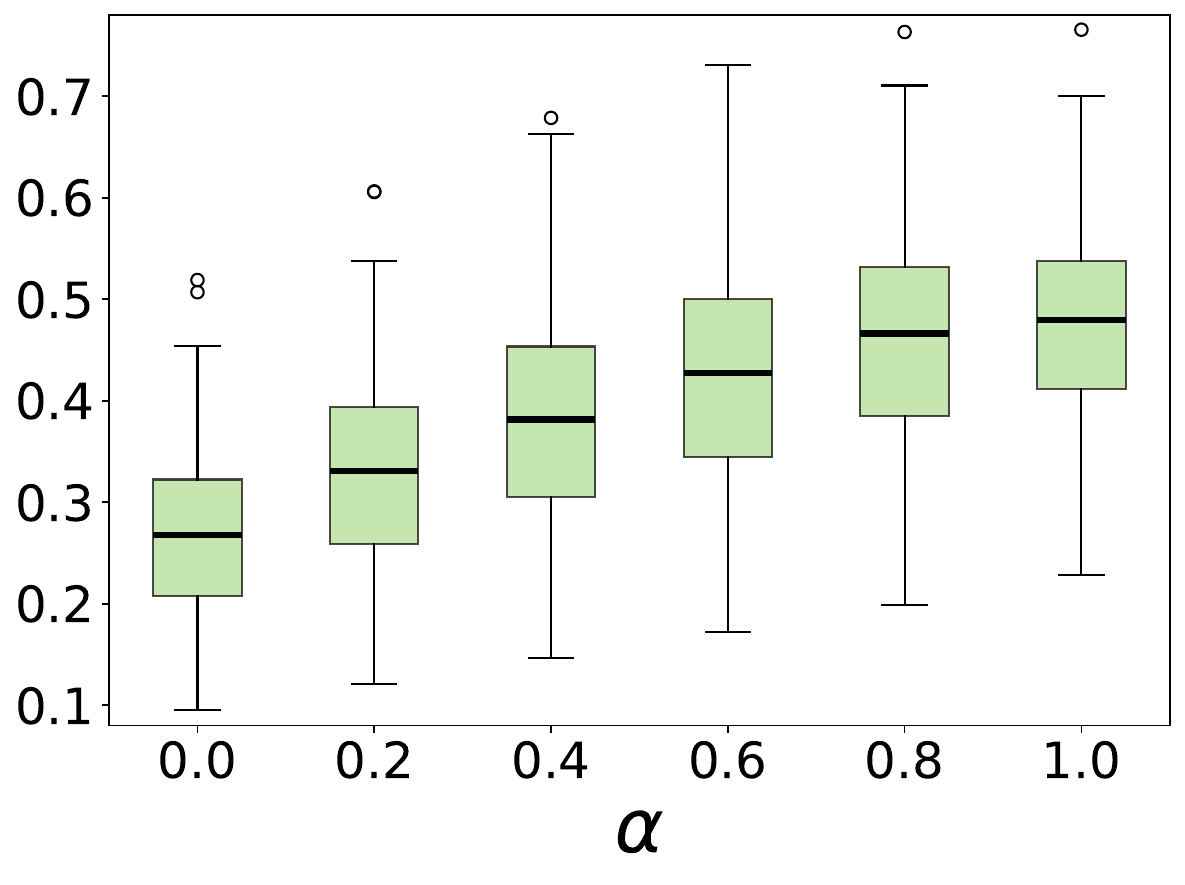}
    \caption{Ex Ante Demand (FEO)}
\end{subfigure}
\begin{subfigure}{0.235\textwidth}
    \centering
    \includegraphics[width=\textwidth]{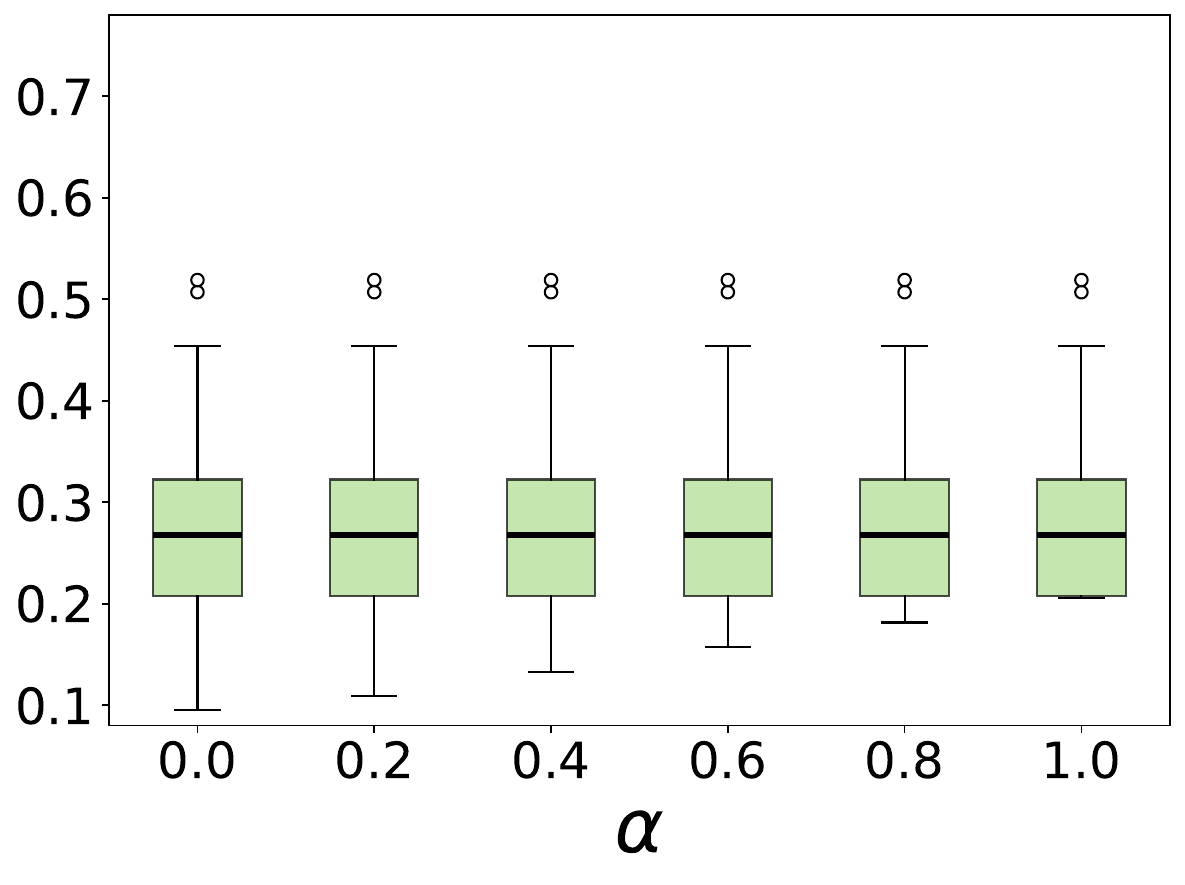}
    \caption{Ex Ante Demand (EFO)}
\end{subfigure}
\begin{subfigure}{0.235\textwidth}
    \centering
    \includegraphics[width=\textwidth]{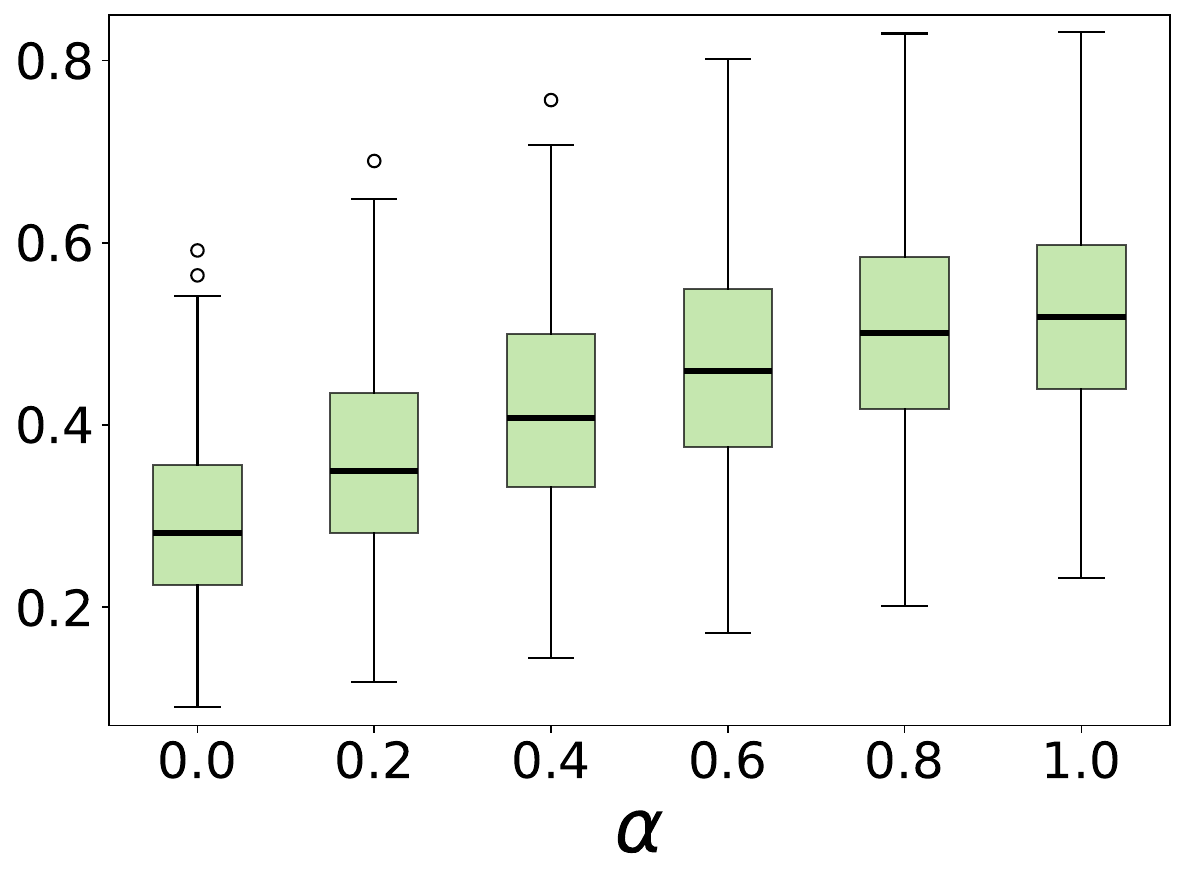}
    \caption{Ex Post Demand (FEO)}
\end{subfigure}
\begin{subfigure}{0.235\textwidth}
    \centering
    \includegraphics[width=\textwidth]{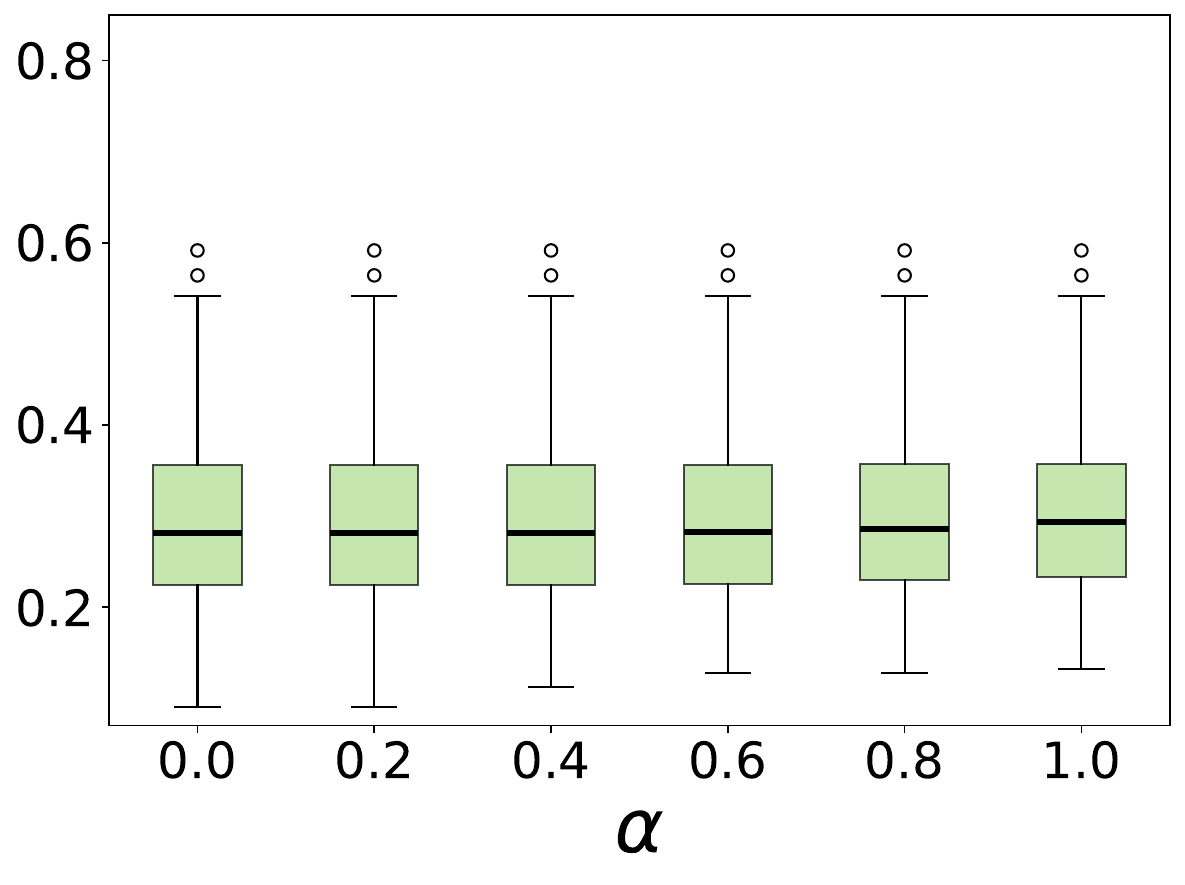}
    \caption{Ex Post Demand (EFO)}
\end{subfigure}
\end{minipage}
}
{Demand distributions for FEO and EFO under Rawlsian demand fairness \label{fig:logistic_Rawlsian_demand_case}\vspace{0.5em}}
{}
\end{figure}

\begin{table}[htbp]
    \centering
    \caption{Normalized performance measures for FEO and EFO under Rawlsian demand fairness (\%)}
    \label{tab:measure_rawlsian_demand_logistic}
    \renewcommand{\arraystretch}{1.1}
    \setlength{\tabcolsep}{4pt} % Adjust column spacing for a perfect fit
    
    \begin{tabular}{l ccccc p{0.3cm} ccccc}
    \toprule
    & \multicolumn{5}{c}{FEO ($\alpha$)} & & \multicolumn{5}{c}{EFO ($\alpha$)} \\
    \cmidrule(lr){2-6} \cmidrule(lr){8-12}
    Measure & 0.2 & 0.4 & 0.6 & 0.8 & 1.0 & & 0.2 & 0.4 & 0.6 & 0.8 & 1.0 \\
    \midrule
    
    % Revenue Row
    $\mathcal{R}(\alpha)/\mathcal{R}(0)$ 
    & 96.24  & 87.18 &  75.39  & 62.97 &  56.55
    & 
    & 99.99  & 99.95  & 99.82  & 99.56 &  99.09 \\

    % Surplus Row
    $\mathcal{S}(\alpha)/\mathcal{S}(0)$ 
    & 129.59 & 156.99 & 181.95 & 203.53 & 213.22
    & 
    & 100.07 & 100.36 & 101.00 & 102.16 & 104.34 \\

    % Welfare Row
    $\mathcal{W}(\alpha)/\mathcal{W}(0)$ 
    & 111.17 & 118.43 & 123.09 & 125.90 & 126.69
    & 
    & 100.03 & 100.14 & 100.35 & 100.73 & 101.44 \\

    \bottomrule
    \end{tabular}
\end{table}

% \textcolor{red}{figure: demand gap decreases, adding ex-ante demand; Table: S,W (FEO) > S, W (EFO). change $\bar{d}$ with $0.5$}

% \textcolor{red}{how to deal with other cases?}

\section{Conclusion}
This paper studies fairness in a data-driven pricing pipeline consisting of a demand estimation stage, where the firm learns a predictive demand model from data, and a price optimization stage, where prices are determined to maximize expected profit based on the estimated demand model. To this end, we develop a comprehensive framework to compare the outcomes that arise when different fairness criteria are enforced at distinct stages of this pipeline.%Our frameworks enable decision makers to implement fairness in such setting.

Focusing on a stylized model with two groups and linear demand, we examine how pricing decisions and societal outcomes are affected by using different fairness criteria and imposing fairness in different stages. %For a loss fairness, we show that it can generate more than one possible solutions, yet have substantially different prices and welfare outcomes. As a result, it can lead to ambiguous pricing decisions and may potentially harm the society. 
Our analysis yields three main insights. First, for loss fairness, which can only be applied in the estimation stage, we show that it can yield multiple feasible solutions that lead to substantially different prices and welfare outcomes. These observations highlight the importance of considering downstream consequences during demand estimation. Second, for parity-wise fairness, we find that when the two groups have similar market sizes and price distributions, where fairness is imposed plays a crucial role. Price fairness at the estimation stage benefits consumers more, while demand fairness at the optimization stage leads to better consumer outcomes. Third, for Rawlsian fairness, imposing fairness in either the estimation or optimization stage leads to identical outcomes. More strikingly, in certain settings, fairness can increase both profit and consumer surplus, yielding a win–win outcome rather than the typical perceived trade-off. %The observations show that considering downstream consequences for demand estimation is an important task. 
%For price and demand fairness, the comparison between Fair-Estimate-then-Optimize (FEO) and Estimate-then-Fair-Optimize (EFO) shows that, under parity-wise perspective, the relative performance of FEO and EFO depends on the data distribution, while under Rawlsian-based perspective, FEO and EFO yield identical outcomes. These results provide guidance on what types of fairness constraints firms should use and where they should be applied. Lastly, we conduct a case study using real data on the willingness-to-pay on vaccination, showing that our frameworks can potentially be used in more general settings.
  % 

To enhance the practical relevance of our findings, we generalize our framework to a more realistic setting that incorporates customer features and considers a logistic demand model. We then evaluate our results under this feature-based pricing using real-world vaccination data. Under parity-wise demand fairness, we find that EFO strictly dominates FEO, yielding higher profit, consumer surplus, and social welfare. Consequently, EFO emerges as the preferred approach. For the other fairness criteria, however, no dominance emerges. Instead, we observe a clear trade-off in which profit declines are accompanied by gains in consumer surplus. Our findings show how a policymaker's regulatory priority -- whether to safeguard profits or maximize consumer surplus -- determines where fairness must be enforced within the algorithmic pricing pipeline.%Our results provide guidance for policymakers deciding which outcomes to regulate and for sellers deciding where fairness should intervene an algorithmic pricing pipeline.}

There are many interesting extensions from our framework for future studies. While our theoretical results are derived under a stylized model, our experiments and case study indicate that the qualitative insights extend to more general demand models. Extending theoretical results to richer demand functions, multiple groups, and alternative market structures would further enhance our understanding of fairness in data-driven pricing systems. In addition, our results highlight the central role that the dataset plays in the performance of fairness-aware demand estimation. A systematic study of the statistical properties of the framework, including sample complexity and generalization error, would help clarify how limitations in data interact with fairness constraints to shape downstream pricing decisions and social outcomes. More broadly, our findings raise important questions about the governance of algorithmic pricing systems. Because fairness interventions can have different effects depending on where they are applied, it is important to examine where accountability and regulatory oversight should reside -- whether in data collection, demand estimation, price optimization, or post-processing. Understanding how auditing and regulation interact with fairness-aware pricing systems remains an important direction for future research.

%Taken together, our findings suggest that designing fair pricing systems requires joint consideration of estimation and optimization, rather than isolated fairness constraints. We believe that our work provides insights for firms to deploy personalized pricing systems and for policymakers to evaluate their societal impacts.

\bibliographystyle{ormsv080}
\bibliography{references}

\newpage
\setcounter{page}{1}
\begin{APPENDIX}{}
\section{Proofs}

\label{appendix_proof}
\begin{proof}{Proof of Proposition \ref{prop:ML_fairness_each}.} For simplicity, denote $\ell_g\big(\hat{b}_g(\alpha)\big)$ by $\ell_g(\alpha)$ under $\alpha$-loss fairness. Without loss of generality, suppose $\ell_1(0)>\ell_0(0)$. Since $\ell_0(0)$ and $\ell_1(0)$ are at their minimum achievable values, they cannot be further reduced. Therefore, to satisfy the fairness constraint with $\alpha$, the only option is to increase $\ell_0(\hat{b}_0)$ while keeping $\ell_1(\hat{b}_1)$ unchanged. Thus, $\ell_1(\alpha) = \ell_1\big(\hat{b}_1\big)$, which implies $b_1(\alpha) = b_1$.  
Moreover,  
\begin{equation}
\ell_0(\alpha) = \ell_1(\alpha) - \big(1-\alpha\big)\big(\ell_1(0) - \ell_0(0)\big) = (1-\alpha)\,\ell_0(0) + \alpha\,\ell_1(0)=(1-\alpha)\sigma_0^2+\alpha\sigma_1^2,
\label{eq1:EL_fairness}
\end{equation}
where the last equality follows from the fact that $\ell_g = \sigma_g^2$ when $n_g \rightarrow \infty$ for any $g \in \{0,1\}$.

Since $n_g\rightarrow\infty$, we can rewrite the loss function as follows,
\begin{equation}
    % \begin{aligned}
    \ell_g(\hat{b}):=\E\left[\Big(a_g+\hat{b}_gp-d\Big)^2\right]
    % &=\E\left[\Big(\left(1-g\right)(b_0-b_0)+g\left(b_1-b_1\right)\cdot p-\left(\epsilon_0+\left(\epsilon_1-\epsilon_0\right)g\right)\Big)^2\right]\\
    =\E\left[\left(\left(\hat{b}_g-b_g\right)p-\epsilon_g\right)^2\right]
    =\left(\hat{b}_g-b_g\right)^2\E\left[p^2|g\right]+\sigma_g^2.
    % \end{aligned}
    \label{eq2:EL_Fairness}
\end{equation}

Then, by \eqref{eq1:EL_fairness} and \eqref{eq2:EL_Fairness},
\begin{equation*}
    \left(b_0\left(\alpha\right)-b_0\right)^2\E\left[p^2|g=0\right]+\sigma_0^2=(1-\alpha)\ell_0+\alpha\ell_1^2=(1-\alpha)\sigma_0^2+\alpha\sigma_1^2.
\end{equation*}
The solution is $b_0(\alpha)=b_0\pm\sqrt{\tfrac{{\sigma_1^2-\sigma_0^2}}{\mathbb{E}\left[p^2|g=0\right]}\alpha}$. Let $b^+_0$ (resp. $b^-_0$) denote $b_0+\sqrt{\tfrac{{\sigma_1^2-\sigma_0^2}}{\mathbb{E}\left[p^2|g=0\right]}\alpha}$ (resp. $b_0-\sqrt{\tfrac{{\sigma_1^2-\sigma_0^2}}{\mathbb{E}\left[p^2|g=0\right]}\alpha}$). Note that when $\alpha > \alpha^{\star}:=\left(b_0\right)^2\tfrac{\E[p^2 \mid g=0]}{\sigma_1^2 - \sigma_0^2}$, $b_0^{+}(\alpha) > 0$, which is infeasible under the constraint $b_g \leq 0$. Therefore, if $\alpha\leq\alpha^{\star}$, there are two optimal solutions, otherwise, only $\big(b_0^{-}(\alpha), b_1\big)$ is the optimal solution.

Given the optimal solution $\big(b_0^{\pm}(\alpha), b_1\big)$, the optimal price of \eqref{eq:profit} is 
\begin{equation*}
p_0^{+}(\alpha) = -\frac{a_0}{2b_0^{+}(\alpha)} + \frac{c}{2}, \quad
p_0^{-}(\alpha) =
\begin{cases}
-\frac{a_0}{2b_0^{-}(\alpha)} + \frac{c}{2}, & \text{if } \alpha \le \alpha^-,\\
c, & \text{otherwise,}
\end{cases}
\quad
p_1(\alpha) = -\frac{a_1}{2b_1} + \frac{c}{2},
\end{equation*}
where $\alpha^-:=\left(1+\frac{a_0}{b_0c}\right)^2\!\alpha^\star$ is the point at which $p_0^{-}(\alpha) = c$.

Differentiating the prices with respect to $\alpha$, for $\alpha \le \alpha^-$,
\begin{equation*}
\frac{d p_0^{\pm}(\alpha)}{d\alpha}
= \frac{a_0}{2\left(b_0^{\pm}(\alpha)\right)^2} \frac{d b_0^{\pm}(\alpha)}{d\alpha}
= \pm \frac{a_0}{4\left(b_0^{\pm}(\alpha)\right)^2}
\sqrt{\frac{\sigma_1^2 - \sigma_0^2}{\mathbb{E}[p^2 \mid g = 0]}} \cdot \frac{1}{\sqrt{\alpha}}.
\end{equation*}
Otherwise, $\frac{d p_0^{+}(\alpha)}{d\alpha}$ remains as above, while $\frac{d p_0^{-}(\alpha)}{d\alpha} = 0$.
Finally, $\frac{d p_1(\alpha)}{d\alpha} = 0$, as it does not depend on $\alpha$.

Denote consumer surplus (respectively, profit and social welfare) under $(p_0^{\pm}(\alpha), p_1(\alpha))$ as $\calS^{\pm}$ (respectively, $\calR^{\pm}$ and $\calW^{\pm}$). Then, the derivative of consumer surplus with respect to $\alpha$ is as follows.
\begin{equation}
    \begin{aligned}
    %\frac{d\calR^{\pm}}{d\alpha}&=\left(a_0+b_0p^{\pm}_0(\alpha)-cb_0\right)\frac{dp^{\pm}_0(\alpha)}{d\alpha}\mathbb{I}\{a_0+b_0p^{\pm}_0(\alpha)\geq0\},\\
    \frac{d\calS^{\pm}}{d\alpha}&=-\max\left(0,a_0+b_0p^{\pm}_0(\alpha)\right)\frac{dp^{\pm}_0(\alpha)}{d\alpha}.
    \end{aligned}
    \label{eq:surplus}
\end{equation}
Note that $\frac{d p_1(\alpha)}{d\alpha} = 0$, so this term does not appear in $\frac{d \mathcal{S}^{\pm}}{d\alpha}$.
Substituting $\frac{d p_0^{\pm}(\alpha)}{d\alpha}$ into~\eqref{eq:surplus}, we obtain the following.
\begin{equation*}
\begin{aligned}
\frac{d \mathcal{S}^{-}}{d\alpha} 
&> 0 ~ \text{for } 0 \le \alpha \le \min\left\{\alpha^-, \alpha^\star\right\}\quad\text{and}\quad
\frac{d \mathcal{S}^{-}}{d\alpha} 
= 0 ~ \text{for } \min\left\{\alpha^-, \alpha^\star\right\} < \alpha \le 1.\\[1em]
\frac{d \mathcal{S}^{+}}{d\alpha} 
&< 0 ~ \text{for } 0 \le \alpha \le \min\left\{\alpha^+,\alpha^\star\right\}\quad\text{and}\quad
\frac{d \mathcal{S}^{+}}{d\alpha} 
= 0 ~ \text{for }  \min\left\{\alpha^+,\alpha^\star\right\} < \alpha \le \min(1,\alpha^\star),
\end{aligned}
\end{equation*}
where $\alpha^+:=\left(\frac{a_0+b_0c}{2a_0 + b_0c}\right)^2{\alpha^\star}^2$ is the point at which $a_0 + b_0p_0^+(\alpha)=0$. $\frac{d \mathcal{S}^{+}}{d\alpha}$ is not defined for $\alpha > \alpha^\star$, since $b_0^{+}(\alpha)$ does not exist.

The derivative of profit with respect to $\alpha$ is,
\begin{equation*}
    \begin{aligned}
    \frac{d\calR^{\pm}}{d\alpha}=\left(a_0+2b_0p^{\pm}_0(\alpha)-b_0c\right)\frac{dp^{\pm}_0(\alpha)}{d\alpha}\mathbb{I}\left(p^{\pm}_0(\alpha)\leq-\frac{a_0}{b_0}\right).
    \end{aligned}
    % \label{eq:profit_derivative}
\end{equation*}
Note that the optimal price with respect to $b_0(0)$ is 
$\max\left(-\tfrac{a_0}{2b_0}+\tfrac{c}{2},c\right)$, which maximizes profit. 
Therefore, for any $b_0^{\pm}(\alpha)$ with $\alpha>0$, the corresponding profit is strictly lower. 
If we compare the ratio of the changes in profit for $0\le\alpha<\min\left\{\alpha^+,\alpha^-,\alpha^\star\right\}$, we obtain
\begin{equation*}
    \frac{\frac{d\calR^{+}(\alpha)}{d\alpha}}{\frac{d\calR^{-}(\alpha)}{d\alpha}}=\frac{-\frac{a_0}{2b_0}+\frac{c}{2}-p_0^+(\alpha)}{-\frac{a_0}{2b_0}+\frac{c}{2}-p_0^-(\alpha)}\frac{\frac{dp_0^{+}(\alpha)}{d\alpha}}{\frac{dp_0^{-}(\alpha)}{d\alpha}}=\left(\frac{b^-_0(\alpha)}{b^+_0(\alpha)}\right)^3>1.
\end{equation*}
Therefore, $\frac{d\calR^{+}(\alpha)}{d\alpha}<\frac{d\calR^{-}(\alpha)}{d\alpha}$.

Lastly, for social welfare, it is clear that $\calW^+$ is decreasing up to $\min\{\alpha^-,\alpha^\star\}$, since both profit and consumer surplus are decreasing with $\alpha$ (note that social welfare is the sum of these two components). For $\min\{\alpha^+,\alpha^\star\}<\alpha<\alpha^\star$, it is $0$. In the case of $\calW^-$, where $\alpha<\min\{\alpha^-,\alpha^\star\}$,
\begin{equation*}
    \begin{aligned}
    \frac{d\calW^{-}}{d\alpha}=-b_0\left(-p^{-}_0(\alpha)+c\right)\frac{dp^{-}_0(\alpha)}{d\alpha}\geq 0.
    \end{aligned}
    % \label{eq:welfare}
\end{equation*}
Otherwise, $\tfrac{d\calW^-}{d\alpha}=0$. \hfill \Halmos
\end{proof}
\medskip

\begin{proof}{Proof of Corollary \ref{prop:ML_fairness_E}.}
% Consider $\alpha \le \alpha^\star$, for which there are two optimal solutions, $b_0^+(\alpha)$ and $b_0^-(\alpha)$. From Proposition \ref{prop:ML_fairness_each}, $\frac{\calS^{+}(\alpha)}{d\alpha}=0$ and $\frac{d\calW^+(\alpha)}{d\alpha}=0$ for $\alpha\geq \alpha^+$. Therefore, in the case of $\alpha^+<\alpha^-<\alpha^*$, for $\alpha^+\le\alpha<\alpha^-$, the expectation of consumer surplus and social welfare are
% \begin{equation*}
% \E\left[\frac{d\calS(\alpha)}{d\alpha}\right]=\frac{1}{2}\frac{d\calS^{-}(\alpha)}{d\alpha}\geq0\quad\text{and}\quad \E\left[\frac{d\calW(\alpha)}{d\alpha}\right]=\frac{1}{2}\frac{d\calW^{-}(\alpha)}{d\alpha}\geq0.
% \end{equation*}
According to the proof of Proposition~\ref{prop:ML_fairness_each}, for $\alpha< \min\{\alpha^+,\alpha^-,\alpha^\star\}$, we know that $\tfrac{d\calS^{-}}{d\alpha}>0$ and $\tfrac{d\calS^{+}}{d\alpha}<0$. Then,
\begin{equation*}
\begin{aligned}
\E\left[\frac{d\calS(\alpha)}{d\alpha}\right]&=\frac{1}{2}\frac{d\calS^{-}(\alpha)}{d\alpha}+\frac{1}{2}\frac{d\calS^{+}(\alpha)}{d\alpha}=\frac{1}{2}\frac{d\calS^{-}(\alpha)}{d\alpha}\left(1+\frac{\frac{d\calS^{+}(\alpha)}{d\alpha}}{\frac{d\calS^{-}(\alpha)}{d\alpha}}\right)\\
&=\frac{1}{2}\underbrace{\frac{d\calS^{-}(\alpha)}{d\alpha}}_{>0}\left(\underbrace{1-\left(\frac{a_0+b_0p_0^+(\alpha)}{a_0+b_0p_0^-(\alpha)}\right)\left(\frac{b_0^-(\alpha)}{b_0^+(\alpha)}\right)^2}_{(*)}\right).
\end{aligned}
\end{equation*}
Therefore, $\E\left[\frac{d\calS(\alpha)}{d\alpha}\right]<0$, if $(*)$ is less than $0$.
\begin{equation*}
\begin{aligned}
1-&\left(\frac{a_0+b_0p_0^+(\alpha)}{a_0+b_0p_0^-(\alpha)}\right)\left(\frac{b^{-}_{0}(\alpha)}{b^{+}_{0}(\alpha)}\right)^2<0 \Leftrightarrow  \left(a_0+b_0p^-(\alpha)\right) (b^{+}_{0}(\alpha))^2<\left(a_0+b_0p^+(\alpha)\right) (b^{-}_{0}(\alpha))^2\\
% &\Leftrightarrow \left(a_0+b_0p^-(\alpha)\right)b^{+}_{0}(\alpha)>\left(a_0+b_0p^+(\alpha)\right)b^{-}_{0}(\alpha)\\
% &\Leftrightarrow a_0\left(b_0^+(\alpha)^2-b_0^-(\alpha)^2\right)>b_0\left(-\frac{a_0}{2b_0^+(\alpha)}b_0^-(\alpha)+\frac{c}{2}b_0^-(\alpha)+\frac{a_0}{2b_0^-(\alpha)}b_0^+(\alpha)-\frac{c}{2}b_0^+(\alpha)\right)\\
&\Leftrightarrow a_0\left(b_0^+(\alpha)^2-b_0^-(\alpha)^2\right)<-b_0\frac{c}{2}\left(b_0^+(\alpha)^2-b_0^-(\alpha)^2\right)+\frac{b_0a_0}{2}\frac{{b^+_0(\alpha)}^3-{b^-_0(\alpha)}^3}{b_0^+(\alpha)b_0^-(\alpha)}\\
&\Leftrightarrow \left(a_0 + b_0\frac{c}{2}\right)\left(b_0^+(\alpha)^2-b_0^-(\alpha)^2\right)<\frac{b_0a_0}{2}\frac{{b^+_0(\alpha)}^3-{b^-_0(\alpha)}^3}{b_0^+(\alpha)b_0^-(\alpha)}\\
&\Leftrightarrow \left(a_0 + b_0\frac{c}{2}\right)\left(b_0^+(\alpha)+b_0^-(\alpha)\right)<\frac{b_0a_0}{2}\frac{{b^+_0(\alpha)}^2+b^+_0(\alpha)b^-_0(\alpha)+{b^-_0(\alpha)}^2}{b_0^+(\alpha)b_0^-(\alpha)}\\
&\Leftrightarrow \left(a_0 + b_0\frac{c}{2}\right)2b_0<\frac{b_0a_0}{2}\frac{3b_0^2-K\alpha}{b_0^2-K\alpha}\\
&\Leftrightarrow \frac{4a_0 + 2b_0c}{a_0}>\frac{3b_0^2-K\alpha}{b_0^2-K\alpha}\\
&\Leftrightarrow \frac{a_0b_0^2+2b_0^3c}{(3a_0+2b_0c)K}>\alpha\\
&\Leftrightarrow \frac{a_0+2b_0c}{3a_0+2b_0c}\alpha^\star=:\bar\alpha>\alpha.
\end{aligned}
\end{equation*}
Here, $K:=\tfrac{{\sigma_1^2-\sigma_0^2}}{\mathbb{E}\left[p^2|g=0\right]}$, which means that $b_0^{\pm}(\alpha)=b_0\pm\sqrt{K\alpha}$ and $\alpha^\star = b_0^2/K$. Given that $\frac{d\calS^+}{d\alpha}<0$, $\frac{d\calS^-}{d\alpha}>0$, and $\E\left[\tfrac{d\calS}{d\alpha}\right] < 0$, we obtain the inequality $-\frac{d\calS^+}{d\alpha} > \frac{d\calS^-}{d\alpha}$. 

In Proposition~\ref{prop:ML_fairness_each}, $\calR^+$ decreases sharply. 
Therefore, combining the facts that 
$-\frac{d\calS^+}{d\alpha} > \frac{d\calS^-}{d\alpha} > 0$ 
and 
$\frac{d\calR^+}{d\alpha} < \frac{d\calR^-}{d\alpha} < 0$, 
we conclude that $-\frac{d\calW^+}{d\alpha} > \frac{d\calW^-}{d\alpha} > 0$, 
which implies that $\mathbb{E}\left[\tfrac{d\calW}{d\alpha}\right] < 0$.\hfill\halmos
\end{proof}
\medskip

\begin{proof}{Proof of Proposition \ref{prop3}.}

Without loss of generality, we can assume $p_0^{\mathrm{LS}}<p_1^{\mathrm{LS}}$ because group $0$ and group $1$ are symmetric. Then we have $\tau_a<\tau_b$. Therefore, we only have to compare $\tau_p$ and $\tau_b$.

We first prove the results of each case for parity-wise price fairness. In order to compare FEO and EFO, we first calculate the first order derivatives of the optimal prices with respect to $\alpha$.

Under parity-wise price fairness, FEO is written as \eqref{FEO_price_pw}:
\begin{equation}
    \label{FEO_price_pw}
    \begin{aligned}
        \min_{\hat{b}_0,\hat{b}_1}&\quad\frac{1}{n_0}\sum_{i:g^{(i)}=0}\left(a_0+\hat{b}_0 p^{(i)}-d^{(i)}\right)^2+\frac{1}{n_1}\sum_{i:g^{(i)}=1}\left(a_1+\hat{b}_1 p^{(i)}-d^{(i)}\right)^2\\
        \text{s.t.}&\quad \left| p_1^*\left(\hat{b}_0, \hat{b}_1\right)-p_0^*\left(\hat{b}_0, \hat{b}_1\right)\right| \leq \left(1-\alpha\right)\left| p_1^*\left(\hat{b}_0^{\mathrm{LS}}, \hat{b}_1^{\mathrm{LS}}\right)-p_0^*\left(\hat{b}_0^{\mathrm{LS}}, \hat{b}_1^{\mathrm{LS}}\right)\right|.
    \end{aligned}
\end{equation}
The objective function can be written as \eqref{trans}
\begin{equation}
    \label{trans}
    \begin{aligned}
        &\frac{1}{n_0}\sum_{i:g^{(i)}=0}\left(a_0+\hat{b}_0 p^{(i)}-d^{(i)}\right)^2+\frac{1}{n_1}\sum_{i:g^{(i)}=1}\left(a_1+\hat{b}_1 p^{(i)}-d^{(i)}\right)^2\\
        =&A\left(\frac{1}{p_0^*\left(\hat{b}_0,\hat{b}_1\right)-\frac{c}{2}}-\frac{1}{p_0^*\left(\hat{b}_0^{\mathrm{LS}},\hat{b}_1^{\mathrm{LS}}\right)-\frac{c}{2}}\right)^2+B\left(\frac{1}{p_1^*\left(\hat{b}_0,\hat{b}_1\right)-\frac{c}{2}}-\frac{1}{p_1^*\left(\hat{b}_0^{\mathrm{LS}},\hat{b}_1^{\mathrm{LS}}\right)-\frac{c}{2}}\right)^2+C,
    \end{aligned}
\end{equation}
where $A=a_0^2\left(\frac{1}{n_0}\sum_{i:g^{(i)}=0}p^{(i)^2}\right)$, $B=a_1^2\left(\frac{1}{n_1}\sum_{i:g^{(i)}=1}p^{(i)^2}\right)$, $C=\frac{1}{n_0}\sum_{i:g^{(i)}=0}\left(a_0+\hat{b}_0^{\mathrm{LS}} p^{(i)}-d^{(i)}\right)^2+\frac{1}{n_1}\sum_{i:g^{(i)}=1}\left(a_1+\hat{b}_1^{\mathrm{LS}} p^{(i)}-d^{(i)}\right)^2$ are constants. So the problem \eqref{FEO_price_pw} is equivalent to
\begin{equation}
    \begin{aligned}
        \min_{\hat{b}_0,\hat{b}_1}&\quad A\left(\frac{1}{p_0^*\left(\hat{b}_0,\hat{b}_1\right)-\frac{c}{2}}-\frac{1}{p_0^*\left(\hat{b}_0^{\mathrm{LS}},\hat{b}_1^{\mathrm{LS}}\right)-\frac{c}{2}}\right)^2+B\left(\frac{1}{p_1^*\left(\hat{b}_0,\hat{b}_1\right)-\frac{c}{2}}-\frac{1}{p_1^*\left(\hat{b}_0^{\mathrm{LS}},\hat{b}_1^{\mathrm{LS}}\right)-\frac{c}{2}}\right)^2\\
        \text{s.t.}&\quad \left| p_1^*\left(\hat{b}_0, \hat{b}_1\right)-p_0^*\left(\hat{b}_0, \hat{b}_1\right)\right| \leq \left(1-\alpha\right)\left| p_1^*\left(\hat{b}_0^{\mathrm{LS}}, \hat{b}_1^{\mathrm{LS}}\right)-p_0^*\left(\hat{b}_0^{\mathrm{LS}}, \hat{b}_1^{\mathrm{LS}}\right)\right|.
    \end{aligned}
    \label{FEO_price_pw_2}
\end{equation}
We can change the decision variables in \eqref{FEO_price_pw_2} to $p_0, p_1$ instead of $\hat{b}_0, \hat{b}_1$ to get
\begin{equation}
    \label{FEO_price_pw_3}
    \begin{aligned}
        \min_{p_0, p_1}&\quad A\left(\frac{1}{p_0-\frac{c}{2}}-\frac{1}{p_0^*\left(\hat{b}_0^{\mathrm{LS}},\hat{b}_1^{\mathrm{LS}}\right)-\frac{c}{2}}\right)^2+B\left(\frac{1}{p_1-\frac{c}{2}}-\frac{1}{p_1^*\left(\hat{b}_0^{\mathrm{LS}},\hat{b}_1^{\mathrm{LS}}\right)-\frac{c}{2}}\right)^2\\
        \text{s.t.}&\quad \left| p_1-p_0 \right| \leq \left(1-\alpha\right)\left| p_1^*\left(\hat{b}_0^{\mathrm{LS}}, \hat{b}_1^{\mathrm{LS}}\right)-p_0^*\left(\hat{b}_0^{\mathrm{LS}}, \hat{b}_1^{\mathrm{LS}}\right)\right|.
    \end{aligned}
\end{equation}

Denote $p_0^{FEO}, p_1^{FEO}$ as the optimal solution to \eqref{FEO_price_pw_3}. First, we show that
\begin{equation}
    \label{FEO_price_pw_constraint}
    p_1^{FEO}-p_0^{FEO}=(1-\alpha)\left(p_1^*\left(\hat{b}_0^{\mathrm{LS}}, \hat{b}_1^{\mathrm{LS}}\right)-p_0^*\left(\hat{b}_0^{\mathrm{LS}}, \hat{b}_1^{\mathrm{LS}}\right)\right).
\end{equation}
Suppose $p_1^{FEO}-p_0^{FEO}<(1-\alpha)\left(p_1^*\left(\hat{b}_0^{\mathrm{LS}}, \hat{b}_1^{\mathrm{LS}}\right)-p_0^*\left(\hat{b}_0^{\mathrm{LS}}, \hat{b}_1^{\mathrm{LS}}\right)\right)$, then either $p_1^{FEO}<p_1^*\left(\hat{b}_0^{\mathrm{LS}}, \hat{b}_1^{\mathrm{LS}}\right)$ or $p_0^{FEO}>p_0^*\left(\hat{b}_0^{\mathrm{LS}}, \hat{b}_1^{\mathrm{LS}}\right)$. Without loss of generality, suppose $p_1^{FEO}<p_1^*\left(\hat{b}_0^{\mathrm{LS}}, \hat{b}_1^{\mathrm{LS}}\right)$. Then $\exists \xi_1>0$, such that $p_1^{FEO}+\xi_1\leq p_1^*\left(\hat{b}_0^{\mathrm{LS}}, \hat{b}_1^{\mathrm{LS}}\right)$ and $p_1^{FEO}+\xi_1-p_0^{FEO}\leq (1-\alpha)\left(p_1^*\left(\hat{b}_0^{\mathrm{LS}}, \hat{b}_1^{\mathrm{LS}}\right)-p_0^*\left(\hat{b}_0^{\mathrm{LS}}, \hat{b}_1^{\mathrm{LS}}\right)\right)$. This implies that $(p_0^{FEO},p_1^{FEO}+\xi_1)$ is a feasible solution to \eqref{FEO_price_pw_3}.
\begin{align*}
    &A\left(\frac{1}{p_0^{FEO}-\frac{c}{2}}-\frac{1}{p_0^*\left(\hat{b}_0^{\mathrm{LS}}, \hat{b}_1^{\mathrm{LS}}\right)-\frac{c}{2}}\right)^2+B\left(\frac{1}{p_1^{FEO}-\frac{c}{2}}-\frac{1}{p_1^*\left(\hat{b}_0^{\mathrm{LS}}, \hat{b}_1^{\mathrm{LS}}\right)-\frac{c}{2}}\right)^2\\
    &>A\left(\frac{1}{p_0^{FEO}-\frac{c}{2}}-\frac{1}{p_0^*\left(\hat{b}_0^{\mathrm{LS}}, \hat{b}_1^{\mathrm{LS}}\right)-\frac{c}{2}}\right)^2+B\left(\frac{1}{p_1^{FEO}+\xi_1-\frac{c}{2}}-\frac{1}{p_1^*\left(\hat{b}_0^{\mathrm{LS}}, \hat{b}_1^{\mathrm{LS}}\right)-\frac{c}{2}}\right)^2,
\end{align*}
which contradicts the definition of $(p_0^{FEO},p_1^{FEO})$.

Therefore, $p_1^{FEO}-p_0^{FEO}=(1-\alpha)\left(p_1^*\left(\hat{b}_0^{\mathrm{LS}}, \hat{b}_1^{\mathrm{LS}}\right)-p_0^*\left(\hat{b}_0^{\mathrm{LS}}, \hat{b}_1^{\mathrm{LS}}\right)\right)$. So $p_0^{FEO}$ is the optimal solution to the unconstrained optimization problem
\begin{align*}
    \min_{p_0}\quad &A\left(\frac{1}{p_0-\frac{c}{2}}-\frac{1}{p_0^*\left(\hat{b}_0^{\mathrm{LS}}, \hat{b}_1^{\mathrm{LS}}\right)-\frac{c}{2}}\right)^2\\
    &+B\left(\frac{1}{p_0+(1-\alpha)\left(p_1^*\left(\hat{b}_0^{\mathrm{LS}}, \hat{b}_1^{\mathrm{LS}}\right)-p_0^*\left(\hat{b}_0^{\mathrm{LS}}, \hat{b}_1^{\mathrm{LS}}\right)\right)-\frac{c}{2}}-\frac{1}{p_1^*\left(\hat{b}_0^{\mathrm{LS}}, \hat{b}_1^{\mathrm{LS}}\right)-\frac{c}{2}}\right)^2.
\end{align*}
If we apply the first-order condition to $p_0^{FEO}$, we have
% \begin{align*}
%     &2A\left(\frac{1}{p_0^{FEO}-\frac{c}{2}}-\frac{1}{p_0^*\left(\hat{b}_0^{\mathrm{LS}}, \hat{b}_1^{\mathrm{LS}}\right)-\frac{c}{2}}\right)\left(-\frac{1}{\left(p_0^{FEO}-\frac{c}{2}\right)^2}\right)+\\
%     &2B\left(\frac{1}{p_0^{FEO}+(1-\alpha)\left(p_1^*\left(\hat{b}_0^{\mathrm{LS}}, \hat{b}_1^{\mathrm{LS}}\right)-p_0^*\left(\hat{b}_0^{\mathrm{LS}}, \hat{b}_1^{\mathrm{LS}}\right)\right)-\frac{c}{2}}-\frac{1}{p_1^*\left(\hat{b}_0^{\mathrm{LS}}, \hat{b}_1^{\mathrm{LS}}\right)-\frac{c}{2}}\right)\cdot\\
%     &\left(-\frac{1}{\left(p_0^{FEO}+(1-\alpha)\left(p_1^*\left(\hat{b}_0^{\mathrm{LS}}, \hat{b}_1^{\mathrm{LS}}\right)-p_0^*\left(\hat{b}_0^{\mathrm{LS}}, \hat{b}_1^{\mathrm{LS}}\right)\right)-\frac{c}{2}\right)^2}\right)=0,
% \end{align*}
% which is equivalent to
\begin{equation*}
    A\left(\frac{1}{p_0^{FEO}-\frac{c}{2}}-\frac{1}{p_0^*\left(\hat{b}_0^{\mathrm{LS}}, \hat{b}_1^{\mathrm{LS}}\right)-\frac{c}{2}}\right)\frac{1}{\left(p_0^{FEO}-\frac{c}{2}\right)^2}+B\left(\frac{1}{p_1^{FEO}-\frac{c}{2}}-\frac{1}{p_1^*\left(\hat{b}_0^{\mathrm{LS}}, \hat{b}_1^{\mathrm{LS}}\right)-\frac{c}{2}}\right)\frac{1}{\left(p_1^{FEO}-\frac{c}{2}\right)^2}=0.
\end{equation*}
Since $p_g^{FEO}-\frac{c}{2}\neq0$ and $p_g^*\left(\hat{b}_0^{\mathrm{LS}}, \hat{b}_1^{\mathrm{LS}}\right)-\frac{c}{2}\neq0$,
\begin{equation}
    \label{paritywise_price_firstorder}
    \begin{aligned}
        &A\left(p_0^*\left(\hat{b}_0^{\mathrm{LS}}, \hat{b}_1^{\mathrm{LS}}\right)-p_0^{FEO}\right)\left(p_1^{FEO}-\frac{c}{2}\right)^3\left(p_1^*\left(\hat{b}_0^{\mathrm{LS}}, \hat{b}_1^{\mathrm{LS}}\right)-\frac{c}{2}\right)\\
        &+B\left(p_1^*\left(\hat{b}_0^{\mathrm{LS}}, \hat{b}_1^{\mathrm{LS}}\right)-p_1^{FEO}\right)\left(p_0^{FEO}-\frac{c}{2}\right)^3\left(p_0^*\left(\hat{b}_0^{\mathrm{LS}}, \hat{b}_1^{\mathrm{LS}}\right)-\frac{c}{2}\right)=0.
    \end{aligned}
\end{equation}
    
Differentiate both sides with respect to $\alpha$, and set $\alpha=0$, we have
\begin{equation}
    \label{dp1}
    A\left(p_1^*\left(\hat{b}_0^{\mathrm{LS}}, \hat{b}_1^{\mathrm{LS}}\right)-\frac{c}{2}\right)^4\left(\frac{dp_0^{FEO}(0)}{d\alpha}\right)+B\left(p_0^*\left(\hat{b}_0^{\mathrm{LS}}, \hat{b}_1^{\mathrm{LS}}\right)-\frac{c}{2}\right)^4\left(\frac{dp_1^{FEO}(0)}{d\alpha}\right)=0.
\end{equation}
From \eqref{FEO_price_pw_constraint} we have
\begin{equation}
    \label{dp2}
    \frac{dp_0^{FEO}(0)}{d\alpha}-\frac{dp_1^{FEO}(0)}{d\alpha}=-\left(p_1^*\left(\hat{b}_0^{\mathrm{LS}}, \hat{b}_1^{\mathrm{LS}}\right)-p_0^*\left(\hat{b}_0^{\mathrm{LS}}, \hat{b}_1^{\mathrm{LS}}\right)\right).
\end{equation}
Combining \eqref{dp1} and \eqref{dp2} we have
\begin{equation*}
    \frac{dp_0^{FEO}(0)}{d\alpha}=\frac{B\left(p_0^*\left(\hat{b}_0^{\mathrm{LS}}, \hat{b}_1^{\mathrm{LS}}\right)-\frac{c}{2}\right)^4}{A\left(p_1^*\left(\hat{b}_0^{\mathrm{LS}}, \hat{b}_1^{\mathrm{LS}}\right)-\frac{c}{2}\right)^4+B\left(p_0^*\left(\hat{b}_0^{\mathrm{LS}}, \hat{b}_1^{\mathrm{LS}}\right)-\frac{c}{2}\right)^4}\left(p_1^*\left(\hat{b}_0^{\mathrm{LS}}, \hat{b}_1^{\mathrm{LS}}\right)-p_0^*\left(\hat{b}_0^{\mathrm{LS}}, \hat{b}_1^{\mathrm{LS}}\right)\right).
\end{equation*}

On the other hand, EFO is written as \eqref{EFO_price_pw}:
\begin{equation}
    \label{EFO_price_pw}
    \begin{aligned}
        \max_{p_0,p_1}&\quad \left(p_0-c\right)\max\left(\hat{b}_0^{\mathrm{LS}} p_0+a_0, 0\right)+\left(p_1-c\right)\max\left(\hat{b}_1^{\mathrm{LS}} p_1+a_1, 0\right)\\
        \text{s.t.}&\quad \left| p_1-p_0\right| \leq \left(1-\alpha\right)\left| p_1^*\left(\hat{b}_0^{\mathrm{LS}}, \hat{b}_1^{\mathrm{LS}}\right)-p_0^*\left(\hat{b}_0^{\mathrm{LS}}, \hat{b}_1^{\mathrm{LS}}\right)\right|.
    \end{aligned}
\end{equation}

Since the estimated demand with no fairness $\hat{b}_0^{\mathrm{LS}} p_0^*\left(\hat{b}_0^{\mathrm{LS}}, \hat{b}_1^{\mathrm{LS}}\right)+a_0, \hat{b}_1^{\mathrm{LS}}p_1^*\left(\hat{b}_0^{\mathrm{LS}}, \hat{b}_1^{\mathrm{LS}}\right)+a_1>0$, there exists $\xi_2>0$, such that when $0\leq \alpha\leq \xi_2$, $\hat{b}_0^{\mathrm{LS}} p_0(\alpha)+a_0, \hat{b}_1^{\mathrm{LS}}p_1(\alpha)+a_1>0$. Therefore, when $0\leq \alpha\leq \xi_2$, \eqref{EFO_price_pw} is equivalent to
\begin{equation}
    \label{EFO_price_pw_2}
    \begin{aligned}
        \max_{p_0,p_1}&\quad \left(p_0-c\right)\left(\hat{b}_0^{\mathrm{LS}} p_0+a_0\right)+\left(p_1-c\right)\left(\hat{b}_1^{\mathrm{LS}} p_1+a_1\right)\\
        \text{s.t.}&\quad \left| p_1-p_0\right| \leq \left(1-\alpha\right)\left| p_1^*\left(\hat{b}_0^{\mathrm{LS}}, \hat{b}_1^{\mathrm{LS}}\right)-p_0^*\left(\hat{b}_0^{\mathrm{LS}}, \hat{b}_1^{\mathrm{LS}}\right)\right|.
    \end{aligned}
\end{equation}

Denote $p_0^{EFO}, p_1^{EFO}$ as the optimal solution to \eqref{EFO_price_pw_2}. We can also prove that
\begin{equation*}
    \label{EFO_price_pw_constraint}
    p_1^{EFO}-p_0^{EFO}=(1-\alpha)\left(p_1^*\left(\hat{b}_0^{\mathrm{LS}}, \hat{b}_1^{\mathrm{LS}}\right)-p_0^*\left(\hat{b}_0^{\mathrm{LS}}, \hat{b}_1^{\mathrm{LS}}\right)\right).
\end{equation*}
Suppose $p_1^{EFO}-p_0^{EFO}<(1-\alpha)\left(p_1^*\left(\hat{b}_0^{\mathrm{LS}}, \hat{b}_1^{\mathrm{LS}}\right)-p_0^*\left(\hat{b}_0^{\mathrm{LS}}, \hat{b}_1^{\mathrm{LS}}\right)\right)$, then either $p_1^{EFO}<p_1^*\left(\hat{b}_0^{\mathrm{LS}}, \hat{b}_1^{\mathrm{LS}}\right)$ or $p_0^{EFO}>p_0^*\left(\hat{b}_0^{\mathrm{LS}}, \hat{b}_1^{\mathrm{LS}}\right)$. Without loss of generality, suppose $p_1^{EFO}<p_1^*\left(\hat{b}_0^{\mathrm{LS}}, \hat{b}_1^{\mathrm{LS}}\right)$. Then $\exists \xi_3>0$, such that $p_1^{EFO}+\xi_3\leq p_1^*\left(\hat{b}_0^{\mathrm{LS}}, \hat{b}_1^{\mathrm{LS}}\right)$ and $p_1^{EFO}+\xi_3-p_0^{EFO}\leq (1-\alpha)\left(p_1^*\left(\hat{b}_0^{\mathrm{LS}}, \hat{b}_1^{\mathrm{LS}}\right)-p_0^*\left(\hat{b}_0^{\mathrm{LS}}, \hat{b}_1^{\mathrm{LS}}\right)\right)$. This implies that $(p_0^{EFO},p_1^{EFO}+\xi_3)$ is a feasible solution to \eqref{EFO_price_pw_2}.

Since $p_1^*\left(\hat{b}_0^{\mathrm{LS}}, \hat{b}_1^{\mathrm{LS}}\right)$ is the maximum point of quadratic function $\left(p_1-c\right)\left(\hat{b}_1^{\mathrm{LS}}p_1+a_1\right)$, and $p_1^{EFO}<p_1^{EFO}+\xi_3<p_1^*\left(\hat{b}_0^{\mathrm{LS}}, \hat{b}_1^{\mathrm{LS}}\right)$, 
\begin{align*}
    &\left(p_0^{EFO}-c\right)\left(\hat{b}_0^{\mathrm{LS}} p_0^{EFO}+a_0\right)+\left(p_1^{EFO}+\xi_3-c\right)\left(\hat{b}_1^{\mathrm{LS}}\left( p_1^{EFO}+\xi_3\right)+a_1\right)\\
    &>\left(p_0^{EFO}-c\right)\left(\hat{b}_0^{\mathrm{LS}} p_0^{EFO}+a_0\right)+\left(p_1^{EFO}-c\right)\left(\hat{b}_1^{\mathrm{LS}} p_1^{EFO}+a_1\right),
\end{align*}
which contradicts to the definition of $(p_0^{EFO},p_1^{EFO})$.

Therefore, $p_1^{EFO}-p_0^{EFO}=(1-\alpha)\left(p_1^*\left(\hat{b}_0^{\mathrm{LS}}, \hat{b}_1^{\mathrm{LS}}\right)-p_0^*\left(\hat{b}_0^{\mathrm{LS}}, \hat{b}_1^{\mathrm{LS}}\right)\right)$. So $p_0^{EFO}$ is the optimal solution to the unconstrained optimization problem
\begin{align*}
    \max_{p_0} \quad &\left(p_0-c\right)\left(\hat{b}_0^{\mathrm{LS}} p_0+a_0\right)+\\
    &\left(p_0+(1-\alpha)\left(p_1^*\left(\hat{b}_0^{\mathrm{LS}}, \hat{b}_1^{\mathrm{LS}}\right)-p_0^*\left(\hat{b}_0^{\mathrm{LS}}, \hat{b}_1^{\mathrm{LS}}\right)\right)-c\right)\cdot\\
    &\left(\hat{b}_1^{\mathrm{LS}} \left(p_0+(1-\alpha)\left(p_1^*\left(\hat{b}_0^{\mathrm{LS}}, \hat{b}_1^{\mathrm{LS}}\right)-p_0^*\left(\hat{b}_0^{\mathrm{LS}}, \hat{b}_1^{\mathrm{LS}}\right)\right)\right)+a_1\right).
\end{align*}
The objective function is a quadratic function, so we can get the closed form solution of $p_0^{EFO}$ and $p_1^{EFO}$ respectively,
\begin{align*}
    &p_0^{EFO}=\left(1-\alpha\right)p_0^*\left(\hat{b}_0^{\mathrm{LS}}, \hat{b}_1^{\mathrm{LS}}\right)+\alpha\left(-\frac{a_0+a_1}{2\left(\hat{b}_0^{\mathrm{LS}}+\hat{b}_1^{\mathrm{LS}}\right)}+\frac{c}{2}\right),\\
    &p_1^{EFO}=\left(1-\alpha\right)p_1^*\left(\hat{b}_0^{\mathrm{LS}}, \hat{b}_1^{\mathrm{LS}}\right)+\alpha\left(-\frac{a_0+a_1}{2\left(\hat{b}_0^{\mathrm{LS}}+\hat{b}_1^{\mathrm{LS}}\right)}+\frac{c}{2}\right).
\end{align*}

Taking derivative for both sides and set $\alpha=0$, we have
\begin{equation*}
    \frac{dp_0^{EFO}(0)}{d\alpha}=\frac{\hat{b}_1^{\mathrm{LS}}}{\hat{b}_0^{\mathrm{LS}}+\hat{b}_1^{\mathrm{LS}}}\left(p_1^*\left(\hat{b}_0^{\mathrm{LS}}, \hat{b}_1^{\mathrm{LS}}\right)-p_0^*\left(\hat{b}_0^{\mathrm{LS}}, \hat{b}_1^{\mathrm{LS}}\right)\right).
\end{equation*}

Next, we show that $\tau_b>\tau_p$ implies $\frac{dp_0^{FEO}(0)}{d\alpha}<\frac{dp_0^{EFO}(0)}{d\alpha}$.
\begin{align*}
    &~\frac{dp_0^{FEO}(0)}{d\alpha}<\frac{dp_0^{EFO}(0)}{d\alpha},\\
    \Leftrightarrow&~\left(\frac{B\left(p_0^*\left(\hat{b}_0^{\mathrm{LS}}, \hat{b}_1^{\mathrm{LS}}\right)-\frac{c}{2}\right)^4}{A\left(p_1^*\left(\hat{b}_0^{\mathrm{LS}}, \hat{b}_1^{\mathrm{LS}}\right)-\frac{c}{2}\right)^4+B\left(p_0^*\left(\hat{b}_0^{\mathrm{LS}}, \hat{b}_1^{\mathrm{LS}}\right)-\frac{c}{2}\right)^4}-\frac{\hat{b}_1^{\mathrm{LS}}}{\hat{b}_0^{\mathrm{LS}}+\hat{b}_1^{\mathrm{LS}}}\right)\left(p_1^*\left(\hat{b}_0^{\mathrm{LS}}, \hat{b}_1^{\mathrm{LS}}\right)-p_0^*\left(\hat{b}_0^{\mathrm{LS}}, \hat{b}_1^{\mathrm{LS}}\right)\right)<0,\\
    \Leftrightarrow&~ \frac{B\left(p_0^*\left(\hat{b}_0^{\mathrm{LS}}, \hat{b}_1^{\mathrm{LS}}\right)-\frac{c}{2}\right)^4}{A\left(p_1^*\left(\hat{b}_0^{\mathrm{LS}}, \hat{b}_1^{\mathrm{LS}}\right)-\frac{c}{2}\right)^4+B\left(p_0^*\left(\hat{b}_0^{\mathrm{LS}}, \hat{b}_1^{\mathrm{LS}}\right)-\frac{c}{2}\right)^4}<\frac{\hat{b}_1^{\mathrm{LS}}}{\hat{b}_0^{\mathrm{LS}}+\hat{b}_1^{\mathrm{LS}}},\\
    \Leftrightarrow&~\frac{A\left(p_1^*\left(\hat{b}_0^{\mathrm{LS}}, \hat{b}_1^{\mathrm{LS}}\right)-\frac{c}{2}\right)^4}{B\left(p_0^*\left(\hat{b}_0^{\mathrm{LS}}, \hat{b}_1^{\mathrm{LS}}\right)-\frac{c}{2}\right)^4}>\frac{\hat{b}_0^{\mathrm{LS}}}{\hat{b}_1^{\mathrm{LS}}},\\
    \Leftrightarrow&~\frac{\left(p_1^*\left(\hat{b}_0^{\mathrm{LS}}, \hat{b}_1^{\mathrm{LS}}\right)-\frac{c}{2}\right)^4}{\left(p_0^*\left(\hat{b}_0^{\mathrm{LS}}, \hat{b}_1^{\mathrm{LS}}\right)-\frac{c}{2}\right)^4}>\frac{a_1\frac{1}{n_1}\sum_{i:g^{(i)}=1}p^{(i)^2}\left(p_1^*\left(\hat{b}_0^{\mathrm{LS}}, \hat{b}_1^{\mathrm{LS}}\right)-\frac{c}{2}\right)}{a_0\frac{1}{n_0}\sum_{i:g^{(i)}=0}p^{(i)^2}\left(p_0^*\left(\hat{b}_0^{\mathrm{LS}}, \hat{b}_1^{\mathrm{LS}}\right)-\frac{c}{2}\right)},\\
    \Leftrightarrow&~\frac{\left(p_1^*\left(\hat{b}_0^{\mathrm{LS}}, \hat{b}_1^{\mathrm{LS}}\right)-\frac{c}{2}\right)}{\left(p_0^*\left(\hat{b}_0^{\mathrm{LS}}, \hat{b}_1^{\mathrm{LS}}\right)-\frac{c}{2}\right)}>\left(\frac{a_1\frac{1}{n_1}\sum_{i:g^{(i)}=1}p^{(i)^2}}{a_0\frac{1}{n_0}\sum_{i:g^{(i)}=0}p^{(i)^2}}\right)^{\frac{1}{3}}\\
    \Leftrightarrow&~\frac{\hat{b}_0^{\mathrm{LS}}}{\hat{b}_1^{\mathrm{LS}}}> \tau_p.
\end{align*}
Therefore, there exists $\xi>0$, when the level of price fairness constraints $\alpha$ satisfies $0<\alpha<\xi$, $p_0^{FEO}<p_0^{EFO}$. This also implies $p_1^{FEO}<p_1^{EFO}$.

Notice that $\calS$ and $\calW$ are decreasing with respect to $p_0$ and $p_1$. Therefore, $\calS^{FEO}>\calS^{EFO}$, $\calW^{FEO}>\calW^{EFO}$. 

Similarly, we can get that
\begin{equation*}
    \frac{dp_0^{FEO}(0)}{d\alpha}>\frac{dp_0^{EFO}(0)}{d\alpha}\Leftrightarrow \tau_b<\tau_p.
\end{equation*}
Therefore, there exists $\xi>0$, when the level of price fairness constraints $\alpha$ satisfies $0<\alpha<\xi$, $p_0^{FEO}>p_0^{EFO}$. This also implies $p_1^{FEO}>p_1^{EFO}$, $\calS^{FEO}<\calS^{EFO}$, and $\calW^{FEO}<\calW^{EFO}$.

%%%%%%%%%%%%%%%%%%%%%%%%%%%%%%%%%%%%%%%%%%%%%%%
If $\tau_b=\tau_p$, $ \frac{dp_0^{FEO}(0)}{d\alpha}=\frac{dp_0^{EFO}(0)}{d\alpha}$. In this case we prove that $\frac{d^2p_0^{FEO}(0)}{d\alpha^2}>0$ and $\frac{d^2p_0^{EFO}(0)}{d\alpha^2}=0$. For FEO, take the second-order derivative with respect to $\alpha$ for equation \eqref{paritywise_price_firstorder} and set $\alpha=0$, we have
\begin{equation*}
    \begin{aligned}
        &A\left(p_1^*\left(\hat{b}_0^{\mathrm{LS}},\hat{b}_1^{\mathrm{LS}}\right)-\frac{c}{2}\right)^4\frac{d^2p_0^{FEO}\left(0\right)}{d\alpha^2}+6A\left(p_1^*\left(\hat{b}_0^{\mathrm{LS}},\hat{b}_1^{\mathrm{LS}}\right)-\frac{c}{2}\right)^3\frac{dp_0^{FEO}\left(0\right)}{d\alpha}\frac{dp_1^{FEO}\left(0\right)}{d\alpha}\\
        &+B\left(p_0^*\left(\hat{b}_0^{\mathrm{LS}},\hat{b}_1^{\mathrm{LS}}\right)-\frac{c}{2}\right)^4\frac{d^2p_1^{FEO}\left(0\right)}{d\alpha^2}+6B\left(p_0^*\left(\hat{b}_0^{\mathrm{LS}},\hat{b}_1^{\mathrm{LS}}\right)-\frac{c}{2}\right)^3\frac{dp_0^{FEO}\left(0\right)}{d\alpha}\frac{dp_1^{FEO}\left(0\right)}{d\alpha}=0.
    \end{aligned}
\end{equation*}
From equation \eqref{FEO_price_pw_constraint} we have $\frac{d^2p_0^{FEO}\left(0\right)}{d\alpha^2}=\frac{d^2p_1^{FEO}\left(0\right)}{d\alpha^2}$. Notice that $\frac{dp_0^{FEO}\left(0\right)}{d\alpha}>0$ and $\frac{dp_1^{FEO}\left(0\right)}{d\alpha}<0$, so we have
\begin{equation*}
    \begin{aligned}
        &A\left(p_1^*\left(\hat{b}_0^{\mathrm{LS}},\hat{b}_1^{\mathrm{LS}}\right)-\frac{c}{2}\right)^4\frac{d^2p_0^{FEO}\left(0\right)}{d\alpha^2}+B\left(p_0^*\left(\hat{b}_0^{\mathrm{LS}},\hat{b}_1^{\mathrm{LS}}\right)-\frac{c}{2}\right)^4\frac{d^2p_0^{FEO}\left(0\right)}{d\alpha^2}\\
        =&-6A\left(p_1^*\left(\hat{b}_0^{\mathrm{LS}},\hat{b}_1^{\mathrm{LS}}\right)-\frac{c}{2}\right)^3\frac{dp_0^{FEO}\left(0\right)}{d\alpha}\frac{dp_1^{FEO}\left(0\right)}{d\alpha}-6B\left(p_0^*\left(\hat{b}_0^{\mathrm{LS}},\hat{b}_1^{\mathrm{LS}}\right)-\frac{c}{2}\right)^3\frac{dp_0^{FEO}\left(0\right)}{d\alpha}\frac{dp_1^{FEO}\left(0\right)}{d\alpha}
        >0.
    \end{aligned}
\end{equation*}
Therefore, $\frac{d^2p_0^{FEO}(0)}{d\alpha^2}>0$. On the other hand, since $p_0^{EFO}=\left(1-\alpha\right)p_0^*\left(\hat{b}_0^{\mathrm{LS}}, \hat{b}_1^{\mathrm{LS}}\right)+\alpha\left(-\frac{a_0+a_1}{2\left(\hat{b}_0^{\mathrm{LS}}+\hat{b}_1^{\mathrm{LS}}\right)}+\frac{c}{2}\right)$ for $0\leq \alpha \leq \xi_2$, we have $\frac{d^2p_0^{EFO}(0)}{d\alpha^2}=0$. Then we have $\frac{d^2p_0^{FEO}(0)}{d\alpha^2}>\frac{d^2p_0^{EFO}(0)}{d\alpha^2}$. By comparing the second-order Taylor expansion of $p_0^{FEO}$ and $p_1^{EFO}$, there exists $\xi>0$, when the level of price fairness constraints $\alpha$ satisfies $0<\alpha<\xi$, $p_0^{FEO}>p_0^{EFO}$. This also implies $p_1^{FEO}>p_1^{EFO}$, $\calS^{FEO}<\calS^{EFO}$, and $\calW^{FEO}<\calW^{EFO}$.

Next, we prove the corresponding results for parity-wise demand fairness. Similar to parity-wise price fairness, we first calculate the first order derivatives of the optimal prices with respect to $\alpha$.

For parity-wise demand fairness, FEO and EFO are written as \eqref{FEO_demand_pw} and \eqref{EFO_demand_pw}:
\begin{equation}
    \label{FEO_demand_pw}
    \begin{aligned}
        \min_{\hat{b}_0,\hat{b}_1}&\quad\frac{1}{n_0}\sum_{i:g^{(i)}=0}\left(a_0+\hat{b}_0 p^{(i)}-d^{(i)}\right)^2+\frac{1}{n_1}\sum_{i:g^{(i)}=1}\left(a_1+\hat{b}_1 p^{(i)}-d^{(i)}\right)^2\\
        \text{s.t.}&\quad \left| \frac{1}{a_1}\left( a_1+ \hat{b}_1^{\mathrm{LS}} p_1^*\left(\hat{b}_0, \hat{b}_1\right)\right)-\frac{1}{a_0}\left( a_0+b_0^{\mathrm{LS}} p_0^*\left(\hat{b}_0, \hat{b}_1\right)\right)\right|
        \\&\qquad\leq \left(1-\alpha\right)\left| \frac{1}{a_1}\left( a_1+ \hat{b}_1^{\mathrm{LS}} p_1^*\left(\hat{b}_0^{\mathrm{LS}}, \hat{b}_1^{\mathrm{LS}}\right)\right)-\frac{1}{a_0}\left( a_0+b_0^{\mathrm{LS}} p_0^*\left(\hat{b}_0^{\mathrm{LS}}, \hat{b}_1^{\mathrm{LS}}\right)\right)\right|.
    \end{aligned}
\end{equation}
and
\begin{equation}
    \label{EFO_demand_pw}
    \begin{aligned}
        \max_{p_0,p_1}&\quad \left(p_0-c\right)\max\left(\hat{b}_0^{\mathrm{LS}} p_0+a_0, 0\right)+\left(p_1-c\right)\max\left(\hat{b}_1^{\mathrm{LS}} p_1+a_1, 0\right)\\
        \text{s.t.}&\quad \left| \frac{1}{a_1}\left( a_1+\hat{b}_1^{\mathrm{LS}}p_1\right)-\frac{1}{a_0}\left( a_0+\hat{b}_0^{\mathrm{LS}}p_0\right)\right| \\&\qquad \leq \left(1-\alpha\right)\left| \frac{1}{a_1}\left( a_1+\hat{b}_1^{\mathrm{LS}}p_1^*\left(\hat{b}_0^{\mathrm{LS}}, \hat{b}_1^{\mathrm{LS}}\right)\right)-\frac{1}{a_0}\left( a_0+\hat{b}_0^{\mathrm{LS}}p_0^*\left(\hat{b}_0^{\mathrm{LS}}, \hat{b}_1^{\mathrm{LS}}\right)\right)\right|.
    \end{aligned}
\end{equation}

For FEO, using \eqref{trans} and changing the decision variables in \eqref{FEO_demand_pw} to $p_0, p_1$ instead of $\hat{b}_0, \hat{b}_1$, we can get the equivalent problem
\begin{equation}
    \label{FEO_demand_pw_2}
    \begin{aligned}
        \min_{p_0, p_1}&\quad A\left(\frac{1}{p_0-\frac{c}{2}}-\frac{1}{p_0^*\left(\hat{b}_0^{\mathrm{LS}},\hat{b}_1^{\mathrm{LS}}\right)-\frac{c}{2}}\right)^2+B\left(\frac{1}{p_1-\frac{c}{2}}-\frac{1}{p_1^*\left(\hat{b}_0^{\mathrm{LS}},\hat{b}_1^{\mathrm{LS}}\right)-\frac{c}{2}}\right)^2\\
        \text{s.t.}&\quad \left| \frac{\hat{b}_1^{\mathrm{LS}}}{a_1}p_1-\frac{\hat{b}_0^{\mathrm{LS}}}{a_0}p_0 \right| \leq \left(1-\alpha\right)\left| \frac{\hat{b}_1^{\mathrm{LS}}}{a_1}p_1^*\left(\hat{b}_0^{\mathrm{LS}}, \hat{b}_1^{\mathrm{LS}}\right)-\frac{\hat{b}_0^{\mathrm{LS}}}{a_0}p_0^*\left(\hat{b}_0^{\mathrm{LS}}, \hat{b}_1^{\mathrm{LS}}\right)\right|.
    \end{aligned}
\end{equation}

Similar to price fairness, denote $p_0^{FEO}, p_1^{FEO}$ as the optimal solution to \eqref{FEO_demand_pw_2}. We can prove that
\begin{equation}
    \label{FEO_demand_pw_constraint}
    \frac{\hat{b}_1^{\mathrm{LS}}}{a_1}p_1^{FEO}-\frac{\hat{b}_0^{\mathrm{LS}}}{a_0}p_0^{FEO}  = \left(1-\alpha\right)\left( \frac{\hat{b}_1^{\mathrm{LS}}}{a_1}p_1^*\left(\hat{b}_0^{\mathrm{LS}}, \hat{b}_1^{\mathrm{LS}}\right)-\frac{\hat{b}_0^{\mathrm{LS}}}{a_0}p_0^*\left(\hat{b}_0^{\mathrm{LS}}, \hat{b}_1^{\mathrm{LS}}\right)\right).
\end{equation}
Suppose $\frac{\hat{b}_1^{\mathrm{LS}}}{a_1}p_1^{FEO}-\frac{\hat{b}_0^{\mathrm{LS}}}{a_0}p_0^{FEO} < \left(1-\alpha\right)\left( \frac{\hat{b}_1^{\mathrm{LS}}}{a_1}p_1^*\left(\hat{b}_0^{\mathrm{LS}}, \hat{b}_1^{\mathrm{LS}}\right)-\frac{\hat{b}_0^{\mathrm{LS}}}{a_0}p_0^*\left(\hat{b}_0^{\mathrm{LS}}, \hat{b}_1^{\mathrm{LS}}\right)\right)$, then either $p_1^{FEO}>p_1^*\left(\hat{b}_0^{\mathrm{LS}}, \hat{b}_1^{\mathrm{LS}}\right)$ or $p_0^{FEO}<p_0^*\left(\hat{b}_0^{\mathrm{LS}}, \hat{b}_1^{\mathrm{LS}}\right)$. Without loss of generality, suppose $p_1^{FEO}>p_1^*\left(\hat{b}_0^{\mathrm{LS}}, \hat{b}_1^{\mathrm{LS}}\right)$. Then $\exists \xi_4>0$, such that $p_1^{EFO}-\xi_4\geq p_1^*\left(\hat{b}_0^{\mathrm{LS}}, \hat{b}_1^{\mathrm{LS}}\right)$ and $\frac{\hat{b}_1^{\mathrm{LS}}}{a_1}\left(p_1^{FEO}-\xi_4\right)-\frac{\hat{b}_0^{\mathrm{LS}}}{a_0}p_0^{FEO} < \left(1-\alpha\right)\left( \frac{\hat{b}_1^{\mathrm{LS}}}{a_1}p_1^*\left(\hat{b}_0^{\mathrm{LS}}, \hat{b}_1^{\mathrm{LS}}\right)-\frac{\hat{b}_0^{\mathrm{LS}}}{a_0}p_0^*\left(\hat{b}_0^{\mathrm{LS}}, \hat{b}_1^{\mathrm{LS}}\right)\right)$. This implies that $(p_0^{FEO},p_1^{FEO}-\xi_4)$ is a feasible solution to \eqref{FEO_demand_pw_2}.

However,
\begin{align*}
    &A\left(\frac{1}{p_0^{FEO}-\frac{c}{2}}-\frac{1}{p_0^*\left(\hat{b}_0^{\mathrm{LS}},\hat{b}_1^{\mathrm{LS}}\right)-\frac{c}{2}}\right)^2+B\left(\frac{1}{p_1^{FEO}-\frac{c}{2}}-\frac{1}{p_1^*\left(\hat{b}_0^{\mathrm{LS}},\hat{b}_1^{\mathrm{LS}}\right)-\frac{c}{2}}\right)^2\\
    &\qquad > A\left(\frac{1}{p_0^{FEO}-\frac{c}{2}}-\frac{1}{p_0^*\left(\hat{b}_0^{\mathrm{LS}},\hat{b}_1^{\mathrm{LS}}\right)-\frac{c}{2}}\right)^2+B\left(\frac{1}{p_1^{FEO}-\xi_4-\frac{c}{2}}-\frac{1}{p_1^*\left(\hat{b}_0^{\mathrm{LS}},\hat{b}_1^{\mathrm{LS}}\right)-\frac{c}{2}}\right)^2,
\end{align*}
which contradicts to the definition of $(p_0^{FEO}, p_1^{FEO})$.

Therefore, $\frac{\hat{b}_1^{\mathrm{LS}}}{a_1}p_1^{FEO}-\frac{\hat{b}_0^{\mathrm{LS}}}{a_0}p_0^{FEO}  = \left(1-\alpha\right)\left( \frac{\hat{b}_1^{\mathrm{LS}}}{a_1}p_1^*\left(\hat{b}_0^{\mathrm{LS}}, \hat{b}_1^{\mathrm{LS}}\right)-\frac{\hat{b}_0^{\mathrm{LS}}}{a_0}p_0^*\left(\hat{b}_0^{\mathrm{LS}}, \hat{b}_1^{\mathrm{LS}}\right)\right)$. So $p_0^{FEO}$ is the optimal solution to the unconstrained optimization problem
\begin{align*}
    \min_{p_0}\quad &A\left(\frac{1}{p_0-\frac{c}{2}}-\frac{1}{p_0^*\left(\hat{b}_0^{\mathrm{LS}},\hat{b}_1^{\mathrm{LS}}\right)-\frac{c}{2}}\right)^2\\&+B\left(\frac{1}{\frac{a_1}{\hat{b}_1^{\mathrm{LS}}}\left(\frac{\hat{b}_0^{\mathrm{LS}}}{a_0}p_0^{FEO}+\left(1-\alpha\right)\left( \frac{\hat{b}_1^{\mathrm{LS}}}{a_1}p_1^*\left(\hat{b}_0^{\mathrm{LS}}, \hat{b}_1^{\mathrm{LS}}\right)-\frac{\hat{b}_0^{\mathrm{LS}}}{a_0}p_0^*\left(\hat{b}_0^{\mathrm{LS}}, \hat{b}_1^{\mathrm{LS}}\right)\right)\right)-\frac{c}{2}}-\frac{1}{p_1^*\left(\hat{b}_0^{\mathrm{LS}},\hat{b}_1^{\mathrm{LS}}\right)-\frac{c}{2}}\right)^2.
\end{align*}

Apply the first-order condition to $p_0^{FEO}$, we have
\begin{align*}
    &2A\left(\frac{1}{p_0^{FEO}-\frac{c}{2}}-\frac{1}{p_0^*\left(\hat{b}_0^{\mathrm{LS}}, \hat{b}_1^{\mathrm{LS}}\right)-\frac{c}{2}}\right)\left(-\frac{1}{\left(p_0^{FEO}-\frac{c}{2}\right)^2}\right)+\\
    &2B\left(\frac{1}{\frac{a_1}{\hat{b}_1^{\mathrm{LS}}}\left(\frac{\hat{b}_0^{\mathrm{LS}}}{a_0}p_0^{FEO}+\left(1-\alpha\right)\left( \frac{\hat{b}_1^{\mathrm{LS}}}{a_1}p_1^*\left(\hat{b}_0^{\mathrm{LS}}, \hat{b}_1^{\mathrm{LS}}\right)-\frac{\hat{b}_0^{\mathrm{LS}}}{a_0}p_0^*\left(\hat{b}_0^{\mathrm{LS}}, \hat{b}_1^{\mathrm{LS}}\right)\right)\right)-\frac{c}{2}}-\frac{1}{p_1^*\left(\hat{b}_0^{\mathrm{LS}},\hat{b}_1^{\mathrm{LS}}\right)-\frac{c}{2}}\right)\cdot\\
    &\left(-\frac{1}{\left(\frac{a_1}{\hat{b}_1^{\mathrm{LS}}}\left(\frac{\hat{b}_0^{\mathrm{LS}}}{a_0}p_0^{FEO}+\left(1-\alpha\right)\left( \frac{\hat{b}_1^{\mathrm{LS}}}{a_1}p_1^*\left(\hat{b}_0^{\mathrm{LS}}, \hat{b}_1^{\mathrm{LS}}\right)-\frac{\hat{b}_0^{\mathrm{LS}}}{a_0}p_0^*\left(\hat{b}_0^{\mathrm{LS}}, \hat{b}_1^{\mathrm{LS}}\right)\right)\right)-\frac{c}{2}\right)^2}\right)\frac{a_1\hat{b}_0^{\mathrm{LS}}}{a_0\hat{b}_1^{\mathrm{LS}}}=0,
\end{align*}
which is equivalent to
\begin{equation*}
    A\left(\frac{1}{p_0^{FEO}-\frac{c}{2}}-\frac{1}{p_0^*\left(\hat{b}_0^{\mathrm{LS}}, \hat{b}_1^{\mathrm{LS}}\right)-\frac{c}{2}}\right)\frac{1}{\left(p_0^{FEO}-\frac{c}{2}\right)^2}+B\left(\frac{1}{p_1^{FEO}-\frac{c}{2}}-\frac{1}{p_1^*\left(\hat{b}_0^{\mathrm{LS}}, \hat{b}_1^{\mathrm{LS}}\right)-\frac{c}{2}}\right)\frac{1}{\left(p_1^{FEO}-\frac{c}{2}\right)^2}\frac{a_1\hat{b}_0^{\mathrm{LS}}}{a_0\hat{b}_1^{\mathrm{LS}}}=0.
\end{equation*}
since $p_g^{FEO}-\frac{c}{2}\neq0$ and $p_g^*\left(\hat{b}_0^{\mathrm{LS}}, \hat{b}_1^{\mathrm{LS}}\right)-\frac{c}{2}\neq0$,
\begin{align}
\label{paritywise_demand_firstorder}
    &A\left(p_0^*\left(\hat{b}_0^{\mathrm{LS}}, \hat{b}_1^{\mathrm{LS}}\right)-p_0^{FEO}\right)\left(p_1^{FEO}-\frac{c}{2}\right)^3
    +B\left(p_1^*\left(\hat{b}_0^{\mathrm{LS}}, \hat{b}_1^{\mathrm{LS}}\right)-p_1^{FEO}\right)\left(p_0^{FEO}-\frac{c}{2}\right)^3=0.
\end{align}
Differentiate both sides with respect to $\alpha$, and set $\alpha=0$, we have
\begin{equation}
    \label{dp3}
    A\left(p_1^*\left(\hat{b}_0^{\mathrm{LS}}, \hat{b}_1^{\mathrm{LS}}\right)-\frac{c}{2}\right)^3\left(\frac{dp_0^{FEO}(0)}{d\alpha}\right)+B\left(p_0^*\left(\hat{b}_0^{\mathrm{LS}}, \hat{b}_1^{\mathrm{LS}}\right)-\frac{c}{2}\right)^3\left(\frac{dp_1^{FEO}(0)}{d\alpha}\right)=0.
\end{equation}
%%%%%%%%%%%%%%%%%%%%%%%%%%%%%%%%%%%%%%%%%%%%%%
From \eqref{FEO_demand_pw_constraint} we have
\begin{equation}
    \label{dp4}
    \frac{\hat{b}_1^{\mathrm{LS}}}{a_1}\frac{dp_1^{FEO}\left(0\right)}{d\alpha}-\frac{\hat{b}_0^{\mathrm{LS}}}{a_0}\frac{dp_0^{FEO}\left(0\right)}{d\alpha}  = -\left( \frac{\hat{b}_1^{\mathrm{LS}}}{a_1}p_1^*\left(\hat{b}_0^{\mathrm{LS}}, \hat{b}_1^{\mathrm{LS}}\right)-\frac{\hat{b}_0^{\mathrm{LS}}}{a_0}p_0^*\left(\hat{b}_0^{\mathrm{LS}}, \hat{b}_1^{\mathrm{LS}}\right)\right),
\end{equation}
so combining \eqref{dp3} and \eqref{dp4} we have
\begin{equation*}
    \frac{dp_0^{FEO}\left(0\right)}{d\alpha}=\frac{B\left(p_0^*\left(\hat{b}_0^{\mathrm{LS}}, \hat{b}_1^{\mathrm{LS}}\right)-\frac{c}{2}\right)^2}{A\left(p_1^*\left(\hat{b}_0^{\mathrm{LS}}, \hat{b}_1^{\mathrm{LS}}\right)-\frac{c}{2}\right)^2+B\left(p_0^*\left(\hat{b}_0^{\mathrm{LS}}, \hat{b}_1^{\mathrm{LS}}\right)-\frac{c}{2}\right)^2}\frac{\left(p_0^*\left(\hat{b}_0^{\mathrm{LS}}, \hat{b}_1^{\mathrm{LS}}\right)-p_1^*\left(\hat{b}_0^{\mathrm{LS}}, \hat{b}_1^{\mathrm{LS}}\right)\right)}{p_1^*\left(\hat{b}_0^{\mathrm{LS}}, \hat{b}_1^{\mathrm{LS}}\right)-\frac{c}{2}}\frac{c}{2}
\end{equation*}
for the FEO problem.

For the EFO problem, since the estimated demand with no fairness $\hat{b}_0^{\mathrm{LS}} p_0^*\left(\hat{b}_0^{\mathrm{LS}}, \hat{b}_1^{\mathrm{LS}}\right)+a_0, \hat{b}_1^{\mathrm{LS}}p_1^*\left(\hat{b}_0^{\mathrm{LS}}, \hat{b}_1^{\mathrm{LS}}\right)+a_1>0$, there exists $\xi_2>0$, such that when $0\leq \alpha\leq \xi_5$, $\hat{b}_0^{\mathrm{LS}} p_0(\alpha)+a_0, \hat{b}_1^{\mathrm{LS}}p_1(\alpha)+a_1>0$. Therefore, when $0\leq \alpha\leq \xi_5$, \eqref{EFO_price_pw} is equivalent to
\begin{equation}
    \label{EFO_demand_pw_2}
    \begin{aligned}
        \max_{p_0,p_1}&\quad \left(p_0-c\right)\left(\hat{b}_0^{\mathrm{LS}} p_0+a_0\right)+\left(p_1-c\right)\left(\hat{b}_1^{\mathrm{LS}} p_1+a_1\right)\\
        \text{s.t.}&\quad \left| \frac{1}{a_1}\left( a_1+\hat{b}_1^{\mathrm{LS}}p_1\right)-\frac{1}{a_0}\left( a_0+\hat{b}_0^{\mathrm{LS}}p_0\right)\right| \\&\qquad \leq \left(1-\alpha\right)\left| \frac{1}{a_1}\left( a_1+\hat{b}_1^{\mathrm{LS}}p_1^*\left(\hat{b}_0^{\mathrm{LS}}, \hat{b}_1^{\mathrm{LS}}\right)\right)-\frac{1}{a_0}\left( a_0+\hat{b}_0^{\mathrm{LS}}p_0^*\left(\hat{b}_0^{\mathrm{LS}}, \hat{b}_1^{\mathrm{LS}}\right)\right)\right|.
    \end{aligned}
\end{equation}

Denote $p_0^{EFO}, p_1^{EFO}$ as the optimal solution to \eqref{EFO_demand_pw_2}. We can also prove that
\begin{equation*}
    \label{EFO_demand_pw_constraint}
    \frac{\hat{b}_1^{\mathrm{LS}}}{a_1}p_1^{EFO}-\frac{\hat{b}_0^{\mathrm{LS}}}{a_0}p_0^{EFO}  = \left(1-\alpha\right)\left( \frac{\hat{b}_1^{\mathrm{LS}}}{a_1}p_1^*\left(\hat{b}_0^{\mathrm{LS}}, \hat{b}_1^{\mathrm{LS}}\right)-\frac{\hat{b}_0^{\mathrm{LS}}}{a_0}p_0^*\left(\hat{b}_0^{\mathrm{LS}}, \hat{b}_1^{\mathrm{LS}}\right)\right).
\end{equation*}
Suppose $\frac{\hat{b}_1^{\mathrm{LS}}}{a_1}p_1^{EFO}-\frac{\hat{b}_0^{\mathrm{LS}}}{a_0}p_0^{EFO} < \left(1-\alpha\right)\left( \frac{\hat{b}_1^{\mathrm{LS}}}{a_1}p_1^*\left(\hat{b}_0^{\mathrm{LS}}, \hat{b}_1^{\mathrm{LS}}\right)-\frac{\hat{b}_0^{\mathrm{LS}}}{a_0}p_0^*\left(\hat{b}_0^{\mathrm{LS}}, \hat{b}_1^{\mathrm{LS}}\right)\right)$, then either $p_1^{EFO}>p_1^*\left(\hat{b}_0^{\mathrm{LS}}, \hat{b}_1^{\mathrm{LS}}\right)$ or $p_0^{EFO}<p_0^*\left(\hat{b}_0^{\mathrm{LS}}, \hat{b}_1^{\mathrm{LS}}\right)$. Without loss of generality, suppose $p_1^{EFO}>p_1^*\left(\hat{b}_0^{\mathrm{LS}}, \hat{b}_1^{\mathrm{LS}}\right)$. Then $\exists \xi_6>0$, such that $p_1^{EFO}-\xi_6\geq p_1^*\left(\hat{b}_0^{\mathrm{LS}}, \hat{b}_1^{\mathrm{LS}}\right)$ and $\frac{\hat{b}_1^{\mathrm{LS}}}{a_1}\left(p_1^{EFO}-\xi_6\right)-\frac{\hat{b}_0^{\mathrm{LS}}}{a_0}p_0^{EFO} < \left(1-\alpha\right)\left( \frac{\hat{b}_1^{\mathrm{LS}}}{a_1}p_1^*\left(\hat{b}_0^{\mathrm{LS}}, \hat{b}_1^{\mathrm{LS}}\right)-\frac{\hat{b}_0^{\mathrm{LS}}}{a_0}p_0^*\left(\hat{b}_0^{\mathrm{LS}}, \hat{b}_1^{\mathrm{LS}}\right)\right)$. This implies that $(p_0^{EFO},p_1^{EFO}-\xi_6)$ is a feasible solution to \eqref{EFO_price_pw_2}.

Since $p_1^*\left(\hat{b}_0^{\mathrm{LS}}, \hat{b}_1^{\mathrm{LS}}\right)$ is the maximum point of quadratic function $\left(p_1-c\right)\left(\hat{b}_1^{\mathrm{LS}}p_1+a_1\right)$, and $p_1^{EFO}>p_1^{EFO}-\xi_6>p_1^*\left(\hat{b}_0^{\mathrm{LS}}, \hat{b}_1^{\mathrm{LS}}\right)$, 
\begin{align*}
    &\left(p_0^{EFO}-c\right)\left(\hat{b}_0^{\mathrm{LS}} p_0^{EFO}+a_0\right)+\left(p_1^{EFO}-\xi_6-c\right)\left(\hat{b}_1^{\mathrm{LS}}\left( p_1^{EFO}-\xi_6\right)+a_1\right)\\
    &>\left(p_0^{EFO}-c\right)\left(\hat{b}_0^{\mathrm{LS}} p_0^{EFO}+a_0\right)+\left(p_1^{EFO}-c\right)\left(\hat{b}_1^{\mathrm{LS}} p_1^{EFO}+a_1\right),
\end{align*}
which contradicts to the definition of $(p_0^{EFO},p_1^{EFO})$.

Therefore, $\frac{\hat{b}_1^{\mathrm{LS}}}{a_1}p_1^{EFO}-\frac{\hat{b}_0^{\mathrm{LS}}}{a_0}p_0^{EFO} = \left(1-\alpha\right)\left( \frac{\hat{b}_1^{\mathrm{LS}}}{a_1}p_1^*\left(\hat{b}_0^{\mathrm{LS}}, \hat{b}_1^{\mathrm{LS}}\right)-\frac{\hat{b}_0^{\mathrm{LS}}}{a_0}p_0^*\left(\hat{b}_0^{\mathrm{LS}}, \hat{b}_1^{\mathrm{LS}}\right)\right)$. So $p_0^{EFO}$ is the optimal solution to the unconstrained optimization problem
\begin{align*}
    \max_{p_0}\quad &\left(p_0-c\right)\left(\hat{b}_0^{\mathrm{LS}}p_0+a_0\right)+\\
    &\left(\frac{a_1}{\hat{b}_1^{\mathrm{LS}}}\left(\frac{\hat{b}_0^{\mathrm{LS}}}{a_0}p_0^{FEO}+\left(1-\alpha\right)\left( \frac{\hat{b}_1^{\mathrm{LS}}}{a_1}p_1^*\left(\hat{b}_0^{\mathrm{LS}}, \hat{b}_1^{\mathrm{LS}}\right)-\frac{\hat{b}_0^{\mathrm{LS}}}{a_0}p_0^*\left(\hat{b}_0^{\mathrm{LS}}, \hat{b}_1^{\mathrm{LS}}\right)\right)\right)-c\right)\cdot\\
    &\left(\hat{b}_1^{\mathrm{LS}}\left(\frac{a_1}{\hat{b}_1^{\mathrm{LS}}}\left(\frac{\hat{b}_0^{\mathrm{LS}}}{a_0}p_0^{FEO}+\left(1-\alpha\right)\left( \frac{\hat{b}_1^{\mathrm{LS}}}{a_1}p_1^*\left(\hat{b}_0^{\mathrm{LS}}, \hat{b}_1^{\mathrm{LS}}\right)-\frac{\hat{b}_0^{\mathrm{LS}}}{a_0}p_0^*\left(\hat{b}_0^{\mathrm{LS}}, \hat{b}_1^{\mathrm{LS}}\right)\right)\right)\right)+a_1\right).
\end{align*}
The objective function is a quadratic function with respect to $p_0$, so we can get the closed form solution of $p_0^{EFO}$ and $p_1^{EFO}$ respectively,
\begin{align*}
    p_0^{EFO}=\left(1-\alpha\right)p_0^*\left(\hat{b}_0^{\mathrm{LS}}, \hat{b}_1^{\mathrm{LS}}\right)+\alpha \frac{a_0}{\hat{b}_0^{\mathrm{LS}}}\left(-\frac{1}{2}+\frac{c}{2}\frac{a_0+a_1}{\frac{a_0^2}{\hat{b}_0^{\mathrm{LS}}}+\frac{a_1^2}{\hat{b}_1^{\mathrm{LS}}}}\right),\\
    p_1^{EFO}=\left(1-\alpha\right)p_1^*\left(\hat{b}_0^{\mathrm{LS}}, \hat{b}_1^{\mathrm{LS}}\right)+\alpha \frac{a_1}{\hat{b}_1^{\mathrm{LS}}}\left(-\frac{1}{2}+\frac{c}{2}\frac{a_0+a_1}{\frac{a_0^2}{\hat{b}_0^{\mathrm{LS}}}+\frac{a_1^2}{\hat{b}_1^{\mathrm{LS}}}}\right).
\end{align*}

Taking the derivative of both sides, we have
\begin{equation*}
    \frac{dp_0^{EFO}}{d\alpha}=\frac{c}{2}\frac{\frac{a_0a_1}{\hat{b}_0^{\mathrm{LS}}}-\frac{a_1^2}{\hat{b}_1^{\mathrm{LS}}}}{\frac{a_0^2}{\hat{b}_0^{\mathrm{LS}}}+\frac{a_1^2}{\hat{b}_1^{\mathrm{LS}}}}.
\end{equation*}

Next, we show that $\tau_b>\tau_p$ implies $\frac{dp_0^{FEO}(0)}{d\alpha}>\frac{dp_0^{EFO}(0)}{d\alpha}$.
\begin{align*}
    &\frac{dp_0^{FEO}\left(0\right)}{d\alpha}=\frac{B\left(p_0^*\left(\hat{b}_0^{\mathrm{LS}}, \hat{b}_1^{\mathrm{LS}}\right)-\frac{c}{2}\right)^2}{A\left(p_1^*\left(\hat{b}_0^{\mathrm{LS}}, \hat{b}_1^{\mathrm{LS}}\right)-\frac{c}{2}\right)^2+B\left(p_0^*\left(\hat{b}_0^{\mathrm{LS}}, \hat{b}_1^{\mathrm{LS}}\right)-\frac{c}{2}\right)^2}\frac{\left(p_0^*\left(\hat{b}_0^{\mathrm{LS}}, \hat{b}_1^{\mathrm{LS}}\right)-p_1^*\left(\hat{b}_0^{\mathrm{LS}}, \hat{b}_1^{\mathrm{LS}}\right)\right)}{p_1^*\left(\hat{b}_0^{\mathrm{LS}}, \hat{b}_1^{\mathrm{LS}}\right)-\frac{c}{2}}\frac{c}{2}>\frac{c}{2}\frac{\frac{a_0a_1}{\hat{b}_0^{\mathrm{LS}}}-\frac{a_1^2}{\hat{b}_1^{\mathrm{LS}}}}{\frac{a_0^2}{\hat{b}_0^{\mathrm{LS}}}+\frac{a_1^2}{\hat{b}_1^{\mathrm{LS}}}}\\
    %\Leftrightarrow&\frac{dp_0^{FEO}\left(0\right)}{d\alpha}=\frac{B\left(p_0^*\left(\hat{b}_0^{\mathrm{LS}}, \hat{b}_1^{\mathrm{LS}}\right)-\frac{c}{2}\right)^2}{A\left(p_1^*\left(\hat{b}_0^{\mathrm{LS}}, \hat{b}_1^{\mathrm{LS}}\right)-\frac{c}{2}\right)^2+B\left(p_0^*\left(\hat{b}_0^{\mathrm{LS}}, \hat{b}_1^{\mathrm{LS}}\right)-\frac{c}{2}\right)^2}\frac{\left(p_0^*\left(\hat{b}_0^{\mathrm{LS}}, \hat{b}_1^{\mathrm{LS}}\right)-p_1^*\left(\hat{b}_0^{\mathrm{LS}}, \hat{b}_1^{\mathrm{LS}}\right)\right)}{p_1^*\left(\hat{b}_0^{\mathrm{LS}}, \hat{b}_1^{\mathrm{LS}}\right)-\frac{c}{2}}\\&\qquad>\frac{a_1\left(p_0^*\left(\hat{b}_0^{\mathrm{LS}}, \hat{b}_1^{\mathrm{LS}}\right)-p_1^*\left(\hat{b}_0^{\mathrm{LS}}, \hat{b}_1^{\mathrm{LS}}\right)\right)}{a_0\left(p_0^*\left(\hat{b}_0^{\mathrm{LS}}, \hat{b}_1^{\mathrm{LS}}\right)-\frac{c}{2}\right)+a_1\left(p_1^*\left(\hat{b}_0^{\mathrm{LS}}, \hat{b}_1^{\mathrm{LS}}\right)-\frac{c}{2}\right)}\\
    \Leftrightarrow&\frac{B\left(p_0^*\left(\hat{b}_0^{\mathrm{LS}}, \hat{b}_1^{\mathrm{LS}}\right)-\frac{c}{2}\right)^2}{A\left(p_1^*\left(\hat{b}_0^{\mathrm{LS}}, \hat{b}_1^{\mathrm{LS}}\right)-\frac{c}{2}\right)^2+B\left(p_0^*\left(\hat{b}_0^{\mathrm{LS}}, \hat{b}_1^{\mathrm{LS}}\right)-\frac{c}{2}\right)^2}<\frac{a_1\left(p_1^*\left(\hat{b}_0^{\mathrm{LS}}, \hat{b}_1^{\mathrm{LS}}\right)-\frac{c}{2}\right)}{a_0\left(p_0^*\left(\hat{b}_0^{\mathrm{LS}}, \hat{b}_1^{\mathrm{LS}}\right)-\frac{c}{2}\right)+a_1\left(p_1^*\left(\hat{b}_0^{\mathrm{LS}}, \hat{b}_1^{\mathrm{LS}}\right)-\frac{c}{2}\right)}\\
    \Leftrightarrow&\frac{A\left(p_1^*\left(\hat{b}_0^{\mathrm{LS}}, \hat{b}_1^{\mathrm{LS}}\right)-\frac{c}{2}\right)^2}{B\left(p_0^*\left(\hat{b}_0^{\mathrm{LS}}, \hat{b}_1^{\mathrm{LS}}\right)-\frac{c}{2}\right)^2}>\frac{a_0\left(p_0^*\left(\hat{b}_0^{\mathrm{LS}}, \hat{b}_1^{\mathrm{LS}}\right)-\frac{c}{2}\right)}{a_1\left(p_1^*\left(\hat{b}_0^{\mathrm{LS}}, \hat{b}_1^{\mathrm{LS}}\right)-\frac{c}{2}\right)}\\
    \Leftrightarrow&\frac{\left(p_1^*\left(\hat{b}_0^{\mathrm{LS}}, \hat{b}_1^{\mathrm{LS}}\right)-\frac{c}{2}\right)}{\left(p_0^*\left(\hat{b}_0^{\mathrm{LS}}, \hat{b}_1^{\mathrm{LS}}\right)-\frac{c}{2}\right)}>\left(\frac{a_1\frac{1}{n_1}\sum_{i:g^{(i)}=1}p^{(i)^2}}{a_0\frac{1}{n_0}\sum_{i:g^{(i)}=0}p^{(i)^2}}\right)^{\frac{1}{3}}\\
    \Leftrightarrow&\frac{\hat{b}_0^{\mathrm{LS}}}{\hat{b_1^{\mathrm{LS}}}}> \tau_p.
\end{align*}
Therefore, there exists $\xi>0$, when the level of price fairness constraints $\alpha$ satisfies $0<\alpha<\xi$, $p_0^{FEO}>p_0^{EFO}$. This also implies $p_1^{FEO}>p_1^{EFO}$.

Notice that $\calS$ and $\calR$ are decreasing with respect to $p_0$ and $p_1$. Therefore, $\calS^{FEO}<\calS^{EFO}$, $\calW^{FEO}<\calW^{EFO}$. 

Similarly, we have
\begin{equation*}
    \frac{dp_0^{FEO}(0)}{d\alpha}<\frac{dp_0^{EFO}(0)}{d\alpha}\Leftrightarrow \tau_b<\tau_p
\end{equation*}
Therefore, there exists $\xi>0$, when the level of price fairness constraints $\alpha$ satisfies $0<\alpha<\xi$, $p_0^{FEO}<p_0^{EFO}$. This also implies $p_1^{FEO}<p_1^{EFO}$, $\calS^{FEO}>\calS^{EFO}$, and $\calW^{FEO}>\calW^{EFO}$.

If $\tau_b=\tau_p$, $ \frac{dp_0^{FEO}(0)}{d\alpha}=\frac{dp_0^{EFO}(0)}{d\alpha}$. Similar to parity-wise price fairness, we take the second-order derivative of $\alpha$ in equation \eqref{paritywise_demand_firstorder} and set $\alpha=0$ to get
\begin{equation*}
    \begin{aligned}
        &A\left(p_1^*\left(\hat{b}_0^{\mathrm{LS}},\hat{b}_1^{\mathrm{LS}}\right)-\frac{c}{2}\right)^3\frac{d^2p_0^{FEO}\left(0\right)}{d\alpha^2}+6A\left(p_1^*\left(\hat{b}_0^{\mathrm{LS}},\hat{b}_1^{\mathrm{LS}}\right)-\frac{c}{2}\right)^2\frac{dp_0^{FEO}\left(0\right)}{d\alpha}\frac{dp_1^{FEO}\left(0\right)}{d\alpha}\\
        &+B\left(p_0^*\left(\hat{b}_0^{\mathrm{LS}},\hat{b}_1^{\mathrm{LS}}\right)-\frac{c}{2}\right)^3\frac{d^2p_1^{FEO}\left(0\right)}{d\alpha^2}+6B\left(p_0^*\left(\hat{b}_0^{\mathrm{LS}},\hat{b}_1^{\mathrm{LS}}\right)-\frac{c}{2}\right)^2\frac{dp_0^{FEO}\left(0\right)}{d\alpha}\frac{dp_1^{FEO}\left(0\right)}{d\alpha}=0.
    \end{aligned}
\end{equation*}

From equation \eqref{FEO_demand_pw_constraint} we have $\frac{\hat{b}_0^{\mathrm{LS}}}{a_0}\frac{d^2p_0^{FEO}\left(0\right)}{d\alpha^2}=\frac{\hat{b}_1^{\mathrm{LS}}}{a_1}\frac{d^2p_1^{FEO}\left(0\right)}{d\alpha^2}$. Notice that $\frac{dp_0^{FEO}\left(0\right)}{d\alpha}<0$ and $\frac{dp_1^{FEO}\left(0\right)}{d\alpha}>0$, so we have
\begin{equation*}
    \begin{aligned}
        &A\left(p_1^*\left(\hat{b}_0^{\mathrm{LS}},\hat{b}_1^{\mathrm{LS}}\right)-\frac{c}{2}\right)^3\frac{d^2p_0^{FEO}\left(0\right)}{d\alpha^2}+B\left(p_0^*\left(\hat{b}_0^{\mathrm{LS}},\hat{b}_1^{\mathrm{LS}}\right)-\frac{c}{2}\right)^3\frac{d^2p_0^{FEO}\left(0\right)}{d\alpha^2}\\
        =&-6A\left(p_1^*\left(\hat{b}_0^{\mathrm{LS}},\hat{b}_1^{\mathrm{LS}}\right)-\frac{c}{2}\right)^2\frac{dp_0^{FEO}\left(0\right)}{d\alpha}\frac{dp_1^{FEO}\left(0\right)}{d\alpha}-6B\left(p_0^*\left(\hat{b}_0^{\mathrm{LS}},\hat{b}_1^{\mathrm{LS}}\right)-\frac{c}{2}\right)^2\frac{dp_0^{FEO}\left(0\right)}{d\alpha}\frac{dp_1^{FEO}\left(0\right)}{d\alpha}
        >0.
    \end{aligned}
\end{equation*}
Therefore, $\frac{d^2p_0^{FEO}}{d\alpha^2}>0$. On the other hand $\frac{d^2p_0^{EFO}}{d\alpha^2}=0$. So $\frac{d^2p_0^{FEO}}{d\alpha^2}>\frac{d^2p_0^{EFO}}{d\alpha^2}$. By comparing the second-order Taylor expansion of $p_0^{FEO}$ and $p_1^{EFO}$, there exists $\xi>0$, when the level of price fairness constraints $\alpha$ satisfies $0<\alpha<\xi$, $p_0^{FEO}>p_0^{EFO}$. This also implies $p_1^{FEO}>p_1^{EFO}$, $\calS^{FEO}<\calS^{EFO}$, and $\calW^{FEO}<\calW^{EFO}$.\hfill{\Halmos}
\end{proof}
\medskip

\begin{proof}{Proof of Proposition \ref{prop4}.}
Without loss of generality, suppose $p_0^{\mathrm{LS}}<p_1^{\mathrm{LS}}$. Then $\tau_b>\tau_a$.
\begin{equation*}
    \calR(p_0,p_1) = \sum_{g \in \{0,1\}} n_g \left(p_g - c\right) \max\left(0,a_g+b_g p_g\right),
\end{equation*}
There exists $\xi_3>0$, such that when $0<\alpha<\xi_3$, $a_0+b_0p_0, a_1+b_1p_1>0$.
\begin{equation*}
    \begin{aligned}
            \calR\left(p_0^{FEO},p_1^{FEO}\right)-\calR\left(p_0^{EFO},p_1^{EFO}\right)=&\sum_{g \in \{0,1\}} \left(p_g^{FEO} - c\right)\left(a_g+b_g p_g^{FEO}\right)-\left(p_g^{EFO} - c\right)\left(a_g+b_g p_g^{EFO}\right)\\
            =&\sum_{g \in \{0,1\}}b_g\left(p_g^{FEO}-\left(-\frac{a_g}{2b_g}+\frac{c}{2}\right)\right)^2-b_g\left(p_g^{EFO}-\left(-\frac{a_g}{2b_g}+\frac{c}{2}\right)\right)^2\\
            =&\sum_{g \in \{0,1\}}b_g\left(p_g^{FEO}-p_g^{EFO}\right)\left(p_g^{FEO}+p_g^{EFO}-\left(-\frac{a_g}{b_g}+c\right)\right).
    \end{aligned}
\end{equation*}
\begin{equation*}
    \begin{aligned}
        \frac{d\left(\calR\left(p_0^{FEO},p_1^{FEO}\right)-\calR\left(p_0^{EFO},p_1^{EFO}\right)\right)}{d\alpha}=\sum_{g \in \{0,1\}}b_g\left(\frac{dp_g^{FEO}}{d\alpha}-\frac{dp_g^{EFO}}{d\alpha}\right)\left(p_g^{FEO}+p_g^{EFO}-\left(-\frac{a_g}{b_g}+c\right)\right)\\
        +b_g\left(p_g^{FEO}-p_g^{EFO}\right)\left(\frac{dp_g^{FEO}}{d\alpha}+\frac{dp_g^{EFO}}{d\alpha}\right).
    \end{aligned}
\end{equation*}
When $\alpha=0$,
\begin{equation*}
    \begin{aligned}
        \frac{d\left(\calR\left(p_0^{FEO},p_1^{FEO}\right)-\calR\left(p_0^{EFO},p_1^{EFO}\right)\right)}{d\alpha}\vrule_{\alpha=0}=&\sum_{g \in \{0,1\}}b_g\left(\frac{dp_g^{FEO}}{d\alpha}\vrule_{\alpha=0}-\frac{dp_g^{EFO}}{d\alpha}\vrule_{\alpha=0}\right)\left(2p_g^{\mathrm{LS}}-\left(-\frac{a_g}{b_g}+c\right)\right)\\
        =&\left(\frac{dp_0^{FEO}}{d\alpha}\vrule_{\alpha=0}-\frac{dp_0^{EFO}}{d\alpha}\vrule_{\alpha=0}\right)\left(a_0\left(1-\frac{b_0}{\hat{b}_0^{\mathrm{LS}}}\right)+a_1\left(1-\frac{b_1}{\hat{b}_1^{\mathrm{LS}}}\right)\right)
    \end{aligned}
\end{equation*}
For parity-wise price fairness, by Proposition \ref{prop3},
\begin{equation*}
    \frac{dp_0^{FEO}}{d\alpha}\vrule_{\alpha=0}-\frac{dp_0^{EFO}}{d\alpha}\vrule_{\alpha=0}<0 \Leftrightarrow \tau_b>\tau_p,
\end{equation*}  
therefore, if $\tau_b>\tau_p$ and
\begin{equation*}
    a_0\left(1-\frac{b_0}{\hat{b}_0^{\mathrm{LS}}}\right)+a_1\left(1-\frac{b_1}{\hat{b}_1^{\mathrm{LS}}}\right)<0,
\end{equation*}
we have
\begin{equation*}
    \frac{d\left(\calR\left(p_0^{FEO},p_1^{FEO}\right)-\calR\left(p_0^{EFO},p_1^{EFO}\right)\right)}{d\alpha}\vrule_{\alpha=0}>0,
\end{equation*}
then there exists $\xi>0$, when the fairness level $\alpha$ satisfies $0<\alpha<\xi$, $\calR^{FEO}>\calR^{EFO}$.
Similarly, if $\tau_b<\tau_p$ and
\begin{equation*}
    a_0\left(1-\frac{b_0}{\hat{b}_0^{\mathrm{LS}}}\right)+a_1\left(1-\frac{b_1}{\hat{b}_1^{\mathrm{LS}}}\right)>0,
\end{equation*}
we have
\begin{equation*}
    \frac{d\left(\calR\left(p_0^{FEO},p_1^{FEO}\right)-\calR\left(p_0^{EFO},p_1^{EFO}\right)\right)}{d\alpha}\vrule_{\alpha=0}>0.
\end{equation*}

For parity-wise demand fairness, \begin{equation*}
    \frac{dp_0^{FEO}}{d\alpha}\vrule_{\alpha=0}-\frac{dp_0^{EFO}}{d\alpha}\vrule_{\alpha=0}<0 \Leftrightarrow \tau_b<\tau_p,
\end{equation*}  
therefore, if $\tau_b>\tau_p$ and
\begin{equation*}
    a_0\left(1-\frac{b_0}{\hat{b}_0^{\mathrm{LS}}}\right)+a_1\left(1-\frac{b_1}{\hat{b}_1^{\mathrm{LS}}}\right)>0,
\end{equation*}
we have
\begin{equation*}
    \frac{d\left(\calR\left(p_0^{FEO},p_1^{FEO}\right)-\calR\left(p_0^{EFO},p_1^{EFO}\right)\right)}{d\alpha}\vrule_{\alpha=0}>0,
\end{equation*}
then there exists $\xi>0$, when the fairness level $\alpha$ satisfies $0<\alpha<\xi$, $\calR^{FEO}>\calR^{EFO}$.
Similarly, if $\tau_b<\tau_p$ and
\begin{equation*}
    a_0\left(1-\frac{b_0}{\hat{b}_0^{\mathrm{LS}}}\right)+a_1\left(1-\frac{b_1}{\hat{b}_1^{\mathrm{LS}}}\right)<0,
\end{equation*}
we have 
\begin{equation*}
    \frac{d\left(\calR\left(p_0^{FEO},p_1^{FEO}\right)-\calR\left(p_0^{EFO},p_1^{EFO}\right)\right)}{d\alpha}\vrule_{\alpha=0}>0.
\end{equation*}

Next we consider the case where $\tau_b=\tau_p$. Then,
\begin{equation*}
    \frac{d\left(\calR\left(p_0^{FEO},p_1^{FEO}\right)-\calR\left(p_0^{EFO},p_1^{EFO}\right)\right)}{d\alpha}\vrule_{\alpha=0}=0.
\end{equation*}
Taking the second order derivative of $\alpha$ and set $\alpha=0$, we have
\begin{equation*}
    \begin{aligned}
        \frac{d^2\left(\calR\left(p_0^{FEO},p_1^{FEO}\right)-\calR\left(p_0^{EFO},p_1^{EFO}\right)\right)}{d\alpha^2}\vrule_{\alpha=0}=&
        \left(\frac{d^2p_0^{FEO}}{d\alpha^2}\vrule_{\alpha=0}-\frac{d^2p_0^{EFO}}{d\alpha^2}\vrule_{\alpha=0}\right)\left(a_0\left(1-\frac{b_0}{\hat{b}_0^{\mathrm{LS}}}\right)+a_1\left(1-\frac{b_1}{\hat{b}_1^{\mathrm{LS}}}\right)\right)\\
        =&\frac{d^2p_0^{FEO}}{d\alpha^2}\vrule_{\alpha=0}\left(a_0\left(1-\frac{b_0}{\hat{b}_0^{\mathrm{LS}}}\right)+a_1\left(1-\frac{b_1}{\hat{b}_1^{\mathrm{LS}}}\right)\right).
    \end{aligned}
\end{equation*}
Notice that $\frac{d^2p_0^{FEO}}{d\alpha^2}\vrule_{\alpha=0}>0$, therefore, when
\begin{equation*}
    a_0\left(1-\frac{b_0}{\hat{b}_0^{\mathrm{LS}}}\right)+a_1\left(1-\frac{b_1}{\hat{b}_1^{\mathrm{LS}}}\right)>0,
\end{equation*}
we have
\begin{equation*}
    \frac{d\left(\calR\left(p_0^{FEO},p_1^{FEO}\right)-\calR\left(p_0^{EFO},p_1^{EFO}\right)\right)}{d\alpha}\vrule_{\alpha=0}>0,
\end{equation*}
then there exists $\xi>0$, when the fairness level $\alpha$ satisfies $0<\alpha<\xi$, $\calR^{FEO}>\calR^{EFO}$.

In all these cases above, if the sign of $a_0\Big(1-\tfrac{b_0}{\hat{b}_0^{\mathrm{LS}}}\Big)+a_1\Big(1-\tfrac{b_1}{\hat{b}_1^{\mathrm{LS}}}\Big)$ changes, the sign of
\begin{equation*}
    \frac{d\left(\calR\left(p_0^{FEO},p_1^{FEO}\right)-\calR\left(p_0^{EFO},p_1^{EFO}\right)\right)}{d\alpha}\vrule_{\alpha=0}
\end{equation*}
also changes. Therefore, we are able to characterize the cases where $\mathcal R^{\mathrm{FEO}}<\mathcal R^{\mathrm{EFO}}$ accordingly.
\hfill\Halmos
\end{proof}
\medskip

\begin{proof}{Proof of Proposition \ref{prop:rawlsian_fairness}.}
Without loss of generality, suppose $p_0^{\mathrm{LS}}<p_1^{\mathrm{LS}}$. For Rawlsian price fairness, FEO is written as \eqref{FEO_price_Rawlsian}
\begin{equation}
    \label{FEO_price_Rawlsian}
    \begin{aligned}
        \min_{\hat{b}_0,\hat{b}_1}&\quad\frac{1}{n_0}\sum_{i:g^{(i)}=0}\left(a_0+\hat{b}_0 p^{(i)}-d^{(i)}\right)^2+\frac{1}{n_1}\sum_{i:g^{(i)}=1}\left(a_1+\hat{b}_1 p^{(i)}-d^{(i)}\right)^2\\
        \text{s.t.}&\quad  p_g^*\left(\hat{b}_0, \hat{b}_1\right)\leq (1-\alpha)\bar{p}+\alpha \underline{p},\quad g=0,1.
    \end{aligned}
\end{equation}
Since $p_0^{\mathrm{LS}}<p_1^{\mathrm{LS}}$, $\bar{p}=p_1^{\mathrm{LS}}$, $\underline{p}=p_0^{\mathrm{LS}}$. Notice that $p_g^{\mathrm{LS}}=-\frac{a_g}{2\hat{b}_g^{\mathrm{LS}}}+\frac{c}{2}$, so problem \eqref{FEO_price_Rawlsian} can be divided into two subproblems
\begin{equation}
    \label{FEO_price_Rawlsian_sub}
    \begin{aligned}
        \min_{\hat{b}_g}&\quad\frac{1}{n_g}\sum_{i:g^{(i)}=g}\left(a_g+\hat{b}_g p^{(i)}-d^{(i)}\right)^2\\
        \text{s.t.}&\quad  -\frac{a_g}{2\hat{b}_g^{\mathrm{LS}}}+\frac{c}{2}\leq (1-\alpha)\bar{p}+\alpha \underline{p}.
    \end{aligned}
\end{equation}
for $g\in\{0,1\}$.

For group $0$, since $p_0^{\mathrm{LS}}=\underline{p}\leq(1-\alpha)\bar{p}+\alpha\underline{p}$, the optimal solution for \eqref{FEO_price_Rawlsian_sub} is still $\hat{b}_0^{\mathrm{LS}}$. For group $1$, \eqref{FEO_price_Rawlsian_sub} can be written as a convex optimization problem
\begin{equation*}
    \begin{aligned}
        \min_{\hat{b}_1}&\quad\frac{1}{n_1}\sum_{i:g^{(i)}=1}\left(a_1+\hat{b}_1 p^{(i)}-d^{(i)}\right)^2\\
        \text{s.t.}&\quad  -\frac{a_1}{2}\geq \left((1-\alpha)\bar{p}+\alpha \underline{p}-\frac{c}{2}\right)\hat{b}_1^{\mathrm{LS}}.
    \end{aligned}
\end{equation*}
with a single variable $\hat{b}_1$. We can get the closed form solution $\hat{b}_1^{\mathrm{FEO}}=-\frac{a_1}{2\left((1-\alpha)\bar{p}+\alpha \underline{p}-\frac{c}{2}\right)}$. Therefore, $p_0^*\left(\hat{b}_0^{\mathrm{FEO}},\hat{b}_1^{\mathrm{FEO}}\right)=p_0^{\mathrm{LS}}$ and $p_1^*\left(\hat{b}_0^{\mathrm{FEO}},\hat{b}_1^{\mathrm{FEO}}\right)=(1-\alpha)p_1^{\mathrm{LS}}+\alpha p_0^{\mathrm{LS}}$.

On the other hand, EFO can be written as \eqref{EFO_price_Rawlsian}
\begin{equation}
    \label{EFO_price_Rawlsian}
    \begin{aligned}
        \max_{p_0,p_1}&\quad \left(p_0-c\right)\max\left(\hat{b}_0^{\mathrm{LS}} p_0+a_0, 0\right)+\left(p_1-c\right)\max\left(\hat{b}_1^{\mathrm{LS}} p_1+a_1, 0\right)\\
        \text{s.t.}&\quad p_g\leq (1-\alpha)\bar{p}+\alpha \underline{p},\quad \forall g\in\{0,1\}.
    \end{aligned}
\end{equation}
which can also be divided into two subproblems
\begin{equation}
    \label{EFO_price_Rawlsian_sub}
    \begin{aligned}
        \max_{p_g}&\quad \left(p_g-c\right)\max\left(\hat{b}_g^{\mathrm{LS}} p_g+a_g, 0\right)\\
        \text{s.t.}&\quad p_g\leq (1-\alpha)\bar{p}+\alpha \underline{p},\quad\forall g\in\{0,1\}.
    \end{aligned}
\end{equation}

For group $0$, since $p_0^{\mathrm{LS}}=\underline{p}\leq(1-\alpha)\bar{p}+\alpha\underline{p}$, the optimal solution for \eqref{FEO_price_Rawlsian_sub} is still $\hat{b}_0^{\mathrm{LS}}$. For group $1$, with the Rawlsian price fairness constraints, $p_1\leq \bar{p}$, so $\hat{b}_1^{\mathrm{LS}} p_1+a_1\geq \hat{b}_1^{\mathrm{LS}} \bar{p}+a_1=\hat{b}_1^{\mathrm{LS}} p_1^{\mathrm{LS}}+a_1>0$. Then \eqref{EFO_price_Rawlsian_sub} is maximizing a quadratic function with boundary constraints, from which we can get the closed form optimal solution $p_1^{\mathrm{EFO}}=p_1^{\mathrm{LS}}$.

Therefore, we have
\begin{equation*}
    p_0^*\left(\hat{b}_0^{\mathrm{FEO}},\hat{b}_1^{\mathrm{FEO}}\right)=p_0^{\mathrm{EFO}}=p_0^{\mathrm{LS}},\quad p_1^*\left(\hat{b}_0^{\mathrm{FEO}},\hat{b}_1^{\mathrm{FEO}}\right)=p_1^{\mathrm{EFO}}=(1-\alpha)p_1^{\mathrm{LS}}+\alpha p_0^{\mathrm{LS}}.
\end{equation*}

For Rawlsian demand fairness, FEO is written as \eqref{FEO_demand_Rawlsian}
\begin{equation}
    \label{FEO_demand_Rawlsian}
    \begin{aligned}
        \min_{\hat{b}_0,\hat{b}_1}&\quad\frac{1}{n_0}\sum_{i:g^{(i)}=0}\left(a_0+\hat{b}_0 p^{(i)}-d^{(i)}\right)^2+\frac{1}{n_1}\sum_{i:g^{(i)}=1}\left(a_1+\hat{b}_1 p^{(i)}-d^{(i)}\right)^2\\
        \text{s.t.}&\quad  \frac{1}{a_g}\left(a_g+\hat{b}_gp_g^*\left(\hat{b}_0, \hat{b}_1\right)\right)\geq (1-\alpha)\underline{d}+\alpha \bar{d},\quad \forall g\in\{0,1\},
    \end{aligned}
\end{equation}
where 
$$\underline{d}=\frac{a_1+\hat b_1^{\mathrm{LS}}p_1^{\mathrm{LS}}}{a_1}=1+\frac{\hat b_1^{\mathrm{LS}}}{a_1}p_1^{\mathrm{LS}}=\frac{1}{2}-\frac{c\hat b_1^{\mathrm{LS}}}{2a_1}\quad\text{and}\quad\bar d = \frac{a_0+\hat b_0^{\mathrm{LS}}p_0^{\mathrm{LS}}}{a_0}=1+\frac{\hat b_0^{\mathrm{LS}}}{a_0}p_0^{\mathrm{LS}}=\frac{1}{2}-\frac{c\hat b_0^{\mathrm{LS}}}{2a_0}.$$
Since $-\frac{a_0}{2\hat b_0^{\mathrm{LS}}} + \frac{c}{2}=p_1^{\mathrm{LS}}>p_1^{\mathrm{LS}}=-\frac{a_1}{2\hat b_1^{\mathrm{LS}}} + \frac{c}{2}$, $\bar d > \underline{d}$.
\eqref{FEO_demand_Rawlsian} can be divided into two subproblems
\begin{equation*}
    \label{FEO_demand_Rawlsian_sub}
    \begin{aligned}
        \min_{\hat{b}_g}&\quad\frac{1}{n_g}\sum_{i:g^{(i)}=g}\left(a_g+\hat{b}_g p^{(i)}-d^{(i)}\right)^2\\
        \text{s.t.}&\quad  1+\frac{\hat{b}_g}{a_g}p_g^*\left(\hat{b}_0, \hat{b}_1\right)\geq (1-\alpha)\underline{d}+\alpha \bar{d},\quad\forall g\in\{0,1\}.
    \end{aligned}
\end{equation*}
Only group~$1$ is binding. Therefore, $p_0^{\mathrm{FEO}}=p_0^{\mathrm{LS}}$ and
$$p_1^{\mathrm{FEO}}=(1-\alpha)\underline{d} + \alpha \bar d -1 = (1-\alpha) p_1^{\mathrm{LS}}+ \alpha \frac{a_1}{a_0}\frac{\hat b_0^{\mathrm{LS}}}{\hat b_1^{\mathrm{LS}}}p_0^{\mathrm{LS}}.$$
% $\frac{1}{a_g}\left(a_g+\hat{b}_gp_g^*\left(\hat{b}_0, \hat{b}_1\right)\right)=\frac{1}{2}+\frac{c\hat{b}_g}{2}$, which is linear with respect to $\hat{b}_g$, so we are able to get the closed form solution $\hat{b}_0^{\mathrm{FEO}}=\frac{2\left((1-\alpha)\underline{d}+\alpha \bar{d}\right)-1}{c}$, $\hat{b}_1^{\mathrm{FEO}}=\hat{b}_1^{\mathrm{LS}}$. Therefore, $p_0^*\left(\hat{b}_0^{  \mathrm{FEO}},\hat{b}_1^{\mathrm{FEO}}\right)=(1-\alpha)p_0^{\mathrm{LS}}+\alpha\frac{a_0\hat{b}_1^{\mathrm{LS}}}{a_1\hat{b}_0^{\mathrm{LS}}}p_1^{\mathrm{LS}}$ and $p_1^*\left(\hat{b}_0^{\mathrm{FEO}},\hat{b}_1^{\mathrm{FEO}}\right)=p_1^{\mathrm{LS}}$.

On the other hand, EFO can be written as \eqref{EFO_demand_Rawlsian}
\begin{equation}
    \label{EFO_demand_Rawlsian}
    \begin{aligned}
        \max_{p_0,p_1}&\quad \left(p_0-c\right)\max\left(\hat{b}_0^{\mathrm{LS}} p_0+a_0, 0\right)+\left(p_1-c\right)\max\left(\hat{b}_1^{\mathrm{LS}} p_1+a_1, 0\right)\\
        \text{s.t.}&\frac{1}{a_g}\left(a_g+\hat{b}_gp_g\right)\geq (1-\alpha)\underline{d}+\alpha \bar{d},\quad\forall g\in\{0,1\},
    \end{aligned}
\end{equation}
which can also be divided into two subproblems
\begin{equation*}
    \label{EFO_demand_Rawlsian_sub}
    \begin{aligned}
        \max_{p_g}&\quad \left(p_g-c\right)\max\left(\hat{b}_g^{\mathrm{LS}} p_g+a_g, 0\right)\\
        \text{s.t.}&\quad \frac{1}{a_g}\left(a_g+\hat{b}_gp_g\right)\geq (1-\alpha)\underline{d}+\alpha \bar{d},\quad\forall g\in\{0,1\},
    \end{aligned}
\end{equation*}
which are also maximizing quadratic functions with linear constraints. Therefore, same as price fairness, the closed form solution satisfies 
\begin{equation*}
        p_0^*\left(\hat{b}_0^{  \mathrm{FEO}},\hat{b}_1^{\mathrm{FEO}}\right)=p_0^{\mathrm{EFO}}=p_0^{\mathrm{LS}},\quad p_1^*\left(\hat{b}_0^{\mathrm{FEO}},\hat{b}_1^{\mathrm{FEO}}\right)=(1-\alpha)p_1^{\mathrm{LS}}+\alpha\frac{a_1\hat{b}_0^{\mathrm{LS}}}{a_0\hat{b}_1^{\mathrm{LS}}}p_0^{\mathrm{LS}}.
    \end{equation*}

Lastly, in the case where $p_0^\ast(\hat b_0^{\mathrm{LS}}, \hat b_1^{\mathrm{LS}}) < p_1^\ast(b_0, b_1) < p_1^\ast(\hat b_0^{\mathrm{LS}}, \hat b_1^{\mathrm{LS}})$, an increase in $\alpha$ from zero causes the prices under both FEO and EFO to decrease until they converge to the true optimal price, $p_1^\ast(b_0, b_1)$. Within this range of $\alpha$, the decrease in $p_1$ directly enhances total surplus. Furthermore, as $p_1$ approaches the true optimal price $p_1^\ast(b_0, b_1)$, the firm's profit also increases. Consequently, social welfare, the sum of profit and surplus, increases as well.
\hfill{\Halmos}
\end{proof}
\medskip

\begin{proof}{Proof of Proposition~\ref{prop7}}
Without loss of generality, suppose $p_0^{\mathrm{LS}}<p_1^{\mathrm{LS}}$. Then we have $\tau_b>\tau_a$. Denote $q_g=\frac{1}{\hat{b}_g^{\mathrm{LS}}}v_g$.

For parity-wise \textbf{price fairness}, FEO is written as \eqref{eq:prop7:FEO_price}
\begin{equation}
    \label{eq:prop7:FEO_price}
    \begin{aligned}
        \min_{\hat{a}_0,\hat{b}_0,\hat{a}_1,\hat{b}_1}&\quad\frac{1}{n_0}\sum_{i:g^{(i)}=0}\left(\hat{a}_0+\hat{b}_0 p^{(i)}-d^{(i)}\right)^2+\frac{1}{n_1}\sum_{i:g^{(i)}=1}\left(\hat{a}_1+\hat{b}_1 p^{(i)}-d^{(i)}\right)^2\\
        \text{s.t.}&\quad \left| p_1^*\left(\hat{a}_1, \hat{b}_1\right)-p_0^*\left(\hat{a}_0,\hat{b}_0\right)\right| \leq \left(1-\alpha\right)\left| p_1^{\mathrm{LS}} - p_0^{\mathrm{LS}}\right|,
    \end{aligned}
\end{equation}
where $p_g^*\left(\hat{a}_g, \hat{b}_g\right)=-\frac{\hat{a}_g}{2\hat{b}_g}+\frac{c}{2}$.
Denote $\theta_g =
\begin{bmatrix}
\hat{a}_g \\
\hat{b}_g
\end{bmatrix}$, we can rewrite problem \eqref{eq:prop7:FEO_price} as \eqref{eq:prop7:FEO_price_2}:  
\begin{equation}
    \label{eq:prop7:FEO_price_2}
    \begin{aligned}
        \min_{\theta_0,\theta_1}&\quad \left(\theta_0-\theta_0^{\mathrm{LS}}\right)^TM_0\left(\theta_0-\theta_0^{\mathrm{LS}}\right)+\left(\theta_1-\theta_1^{\mathrm{LS}}\right)^TM_1\left(\theta_1-\theta_1^{\mathrm{LS}}\right)\\
        \text{s.t.}&\quad \left| p_1^*\left(\theta_1\right)-p_0^*\left(\theta_0\right)\right| \leq \left(1-\alpha\right)\left| p_1^{\mathrm{LS}} - p_0^{\mathrm{LS}}\right|.
    \end{aligned}
\end{equation}
First, we can show that the optimal solution of problem \eqref{eq:prop7:FEO_price_2} $(\theta_0^{\mathrm{FEO}},\theta_1^{\mathrm{FEO}})$ satisfy
\begin{equation}
    \label{eq:prop7:constraint_1}
    p_1^*\left(\theta_1^{\mathrm{FEO}}\right)-p_0^*\left(\theta_0^{\mathrm{FEO}}\right) =\left(1-\alpha\right)\left( p_1^{\mathrm{LS}} - p_0^{\mathrm{LS}}\right).
\end{equation}
If $p_1^*\left(\theta_1\right)-p_0^*\left(\theta_0\right)<\left(1-\alpha\right)\left( p_1^{\mathrm{LS}} - p_0^{\mathrm{LS}}\right)$, either $p_1^*\left(\theta_1\right)<p_1^{\mathrm{LS}}$ or $p_0^*\left(\theta_0\right)>p_0^{\mathrm{LS}}$. Without loss of generality, suppose $p_0^*\left(\theta_0\right)>p_0^{\mathrm{LS}}$. since $p_0^*(\theta)$ is continuous with respect to $\theta$, there exists $\xi_1\in(0,1]$, such that $p_1^*\left(\theta_1\right)-p_0^*\left(\xi_1\theta_0+\left(1-\xi_1\right)\theta_0^{\mathrm{LS}}\right)<\left(1-\alpha\right)\left( p_1^{\mathrm{LS}} - p_0^{\mathrm{LS}}\right)$. Therefore, $(\xi_1\theta_0^{\mathrm{FEO}}+\left(1-\xi_1\right)\theta_0^{\mathrm{LS}},\theta_1^{\mathrm{FEO}})$ is a feasible solution to \eqref{eq:prop7:FEO_price_2} with strictly smaller objective value, which contradicts to the definition of $(\theta_0^{\mathrm{FEO}},\theta_1^{\mathrm{FEO}})$.

Therefore, the optimal solution of problem \eqref{eq:prop7:FEO_price_2} is equivalent to the optimal solution of
\begin{equation}
    \label{eq:prop7:FEO_price_3}
    \begin{aligned}
        \min_{\theta_0,\theta_1}&\quad \left(\theta_0-\theta_0^{\mathrm{LS}}\right)^TM_0\left(\theta_0-\theta_0^{\mathrm{LS}}\right)+\left(\theta_1-\theta_1^{\mathrm{LS}}\right)^TM_1\left(\theta_1-\theta_1^{\mathrm{LS}}\right)\\
        \text{s.t.}&\quad p_1^*\left(\theta_1\right)-p_0^*\left(\theta_0\right) = \left(1-\alpha\right)\left( p_1^{\mathrm{LS}} - p_0^{\mathrm{LS}}\right).
    \end{aligned}
\end{equation}

The KKT conditions of problem \eqref{eq:prop7:FEO_price_3} with Lagrangian $\mathcal{L}(\theta_0,\theta_1, \lambda)=\left(\theta_0-\theta_0^{\mathrm{LS}}\right)^TM_0\left(\theta_0-\theta_0^{\mathrm{LS}}\right)+\left(\theta_1-\theta_1^{\mathrm{LS}}\right)^TM_1\left(\theta_1-\theta_1^{\mathrm{LS}}\right)+\lambda\left(p_1^*\left(\theta_1\right)-p_0^*\left(\theta_0\right) - \left(1-\alpha\right)\left( p_1^{\mathrm{LS}} - p_0^{\mathrm{LS}}\right)\right)$ are
\begin{subequations}
\begin{align}
    \frac{\partial\mathcal{L}}{\partial \theta_0}=2M_0\left(\theta_0-\theta_0^{\mathrm{LS}}\right)-\lambda \begin{bmatrix}
\frac{1}{2\hat{b}_0} \\
-\frac{\hat{a}_0}{2\hat{b}_0^2}
\end{bmatrix}=0,\label{eq:prop7:KKT1}\\
\frac{\partial\mathcal{L}}{\partial \theta_1}=2M_1\left(\theta_1-\theta_1^{\mathrm{LS}}\right)+\lambda \begin{bmatrix}
\frac{1}{2\hat{b}_1} \\
-\frac{\hat{a}_1}{2\hat{b}_1^2}
\end{bmatrix}=0,\label{eq:prop7:KKT2}\\
\lambda\left(p_1^*\left(\theta_1\right)-p_0^*\left(\theta_0\right) - \left(1-\alpha\right)\left( p_1^{\mathrm{LS}} - p_0^{\mathrm{LS}}\right)\right)=0.\nonumber%\label{eq:prop7:KKT3}
\end{align}
\end{subequations}

Differentiate \eqref{eq:prop7:KKT1} and \eqref{eq:prop7:KKT2} with respect to $\alpha$, we have
\begin{align*}
    2M_0\nabla_\alpha\theta_0-\frac{d\lambda}{d\alpha}\begin{bmatrix}
\frac{1}{2\hat{b}_0} \\
-\frac{\hat{a}_0}{2\hat{b}_0^2}
\end{bmatrix}-\lambda\nabla_\alpha\left(\begin{bmatrix}
\frac{1}{2\hat{b}_0} \\
-\frac{\hat{a}_0}{2\hat{b}_0^2}
\end{bmatrix}\right)=0, \\
2M_1\nabla_\alpha\theta_1+\frac{d\lambda}{d\alpha}\begin{bmatrix}
\frac{1}{2\hat{b}_1} \\
-\frac{\hat{a}_1}{2\hat{b}_1^2}
\end{bmatrix}+\lambda\nabla_\alpha\left(\begin{bmatrix}
\frac{1}{2\hat{b}_1} \\
-\frac{\hat{a}_1}{2\hat{b}_1^2}
\end{bmatrix}\right)=0.
\end{align*}

When $\alpha=0$, we have $\lambda=0$. Therefore,
\begin{align*}
    2M_0\nabla_\alpha\theta_0|_{\alpha=0}-q_0\left.\frac{d\lambda}{d\alpha}\right|_{\alpha=0}=0, \\
2M_1\nabla_\alpha\theta_1|_{\alpha=0}+q_1\left.\frac{d\lambda}{d\alpha}\right|_{\alpha=0}=0.
\end{align*}
since $M_0$, $M_1$ are invertible, we have
\begin{align*}
    \nabla_\alpha\theta_0|_{\alpha=0}=\frac{1}{2}M_0^{-1}q_0\left.\frac{d\lambda}{d\alpha}\right|_{\alpha=0}, \\
\nabla_\alpha\theta_1|_{\alpha=0}=-\frac{1}{2}M_1^{-1}q_1\left.\frac{d\lambda}{d\alpha}\right|_{\alpha=0}.
\end{align*}

On the other hand, differentiate both sides of \eqref{eq:prop7:constraint_1} and set $\alpha=0$, we have
\begin{equation*}
    q_1^T\nabla_\alpha\theta_1|_{\alpha=0}-q_0^T\nabla_\alpha\theta_0|_{\alpha=0}=-(p_1^{\mathrm{LS}}-p_0^{\mathrm{LS}}).
\end{equation*}
Replace $\nabla_\alpha\theta_0|_{\alpha=0}$ and $\nabla_\alpha\theta_1|_{\alpha=0}$ by $\frac{1}{2}M_0^{-1}q_0\left.\frac{d\lambda}{d\alpha}\right|_{\alpha=0}$ and $-\frac{1}{2}M_1^{-1}q_1\left.\frac{d\lambda}{d\alpha}\right|_{\alpha=0}$, we have
\begin{equation*}
    \frac{1}{2}(q_1^TM_1^{-1}q_1+q_0^TM_0^{-1}q_0)\left.\frac{d\lambda}{d\alpha}\right|_{\alpha=0}=p_1^{\mathrm{LS}}-p_0^{\mathrm{LS}}.
\end{equation*}
Therefore,
\begin{equation*}
    \left.\frac{d\lambda}{d\alpha}\right|_{\alpha=0}=\frac{2\left(p_1^{\mathrm{LS}}-p_0^{\mathrm{LS}}\right)}{q_1^TM_1^{-1}q_1+q_0^TM_0^{-1}q_0},
\end{equation*}
and
\begin{equation*}
    \begin{aligned}
        \left.\frac{d p_0^*\left(\theta_0^{\mathrm{FEO}}\right)}{d \alpha}\right|_{\alpha=0}=q_0^T\left.\nabla_\alpha\theta_0\right|_{\alpha=0}=\frac{q_0^TM_0^{-1}q_0}{q_1^TM_1^{-1}q_1+q_0^TM_0^{-1}q_0}\left(p_1^{\mathrm{LS}}-p_0^{\mathrm{LS}}\right).
    \end{aligned}
\end{equation*}

On the other hand,
\begin{equation*}
    \left.\frac{dp_0^{\mathrm{EFO}}}{d\alpha}\right|_{\alpha=0}=\frac{\hat{b}_1^{\mathrm{LS}}}{\hat{b}_0^{\mathrm{LS}}+\hat{b}_1^{\mathrm{LS}}}\left(p_1^{\mathrm{LS}}-p_0^{\mathrm{LS}}\right),
\end{equation*}
so
\begin{align*}
    \left.\frac{d p_0^*\left(\theta_0^{\mathrm{FEO}}\right)}{d \alpha}\right|_{\alpha=0}<\left.\frac{dp_0^{\mathrm{EFO}}}{d\alpha}\right|_{\alpha=0} &\Leftrightarrow \frac{q_0^TM_0^{-1}q_0}{q_1^TM_1^{-1}q_1+q_0^TM_0^{-1}q_0}\left(p_1^{\mathrm{LS}}-p_0^{\mathrm{LS}}\right)<\frac{\hat{b}_1^{\mathrm{LS}}}{\hat{b}_0^{\mathrm{LS}}+\hat{b}_1^{\mathrm{LS}}}\left(p_1^{\mathrm{LS}}-p_0^{\mathrm{LS}}\right)\\
    &\Leftrightarrow q_0^TM_0^{-1}q_0\hat{b}_0^{\mathrm{LS}}>q_1^TM_1^{-1}q_1\hat{b}_1^{\mathrm{LS}}.
\end{align*}

Therefore, if $q_0^TM_0^{-1}q_0\hat{b}_0^{\mathrm{LS}}>q_1^TM_1^{-1}q_1\hat{b}_1^{\mathrm{LS}}$, there exists $\xi>0$ such that when $0<\alpha<\xi$, FEO leads to lower prices than EFO and hence FEO gives higher consumer surplus and social welfare. Similarly, if $q_0^TM_0^{-1}q_0\hat{b}_0^{\mathrm{LS}}<q_1^TM_1^{-1}q_1\hat{b}_1^{\mathrm{LS}}$, $ \frac{d p_0^*\left(\theta_0^{\mathrm{FEO}}\right)}{d \alpha}|_{\alpha=0}>\frac{dp_0^{\mathrm{EFO}}}{d\alpha}|_{\alpha=0}$, then FEO leads to higher prices than EFO and gives lower consumer surplus and social welfare.

For parity-wise \textbf{demand fairness}, FEO is written as \eqref{eq:prop7:FEO_demand}
\begin{equation}
    \label{eq:prop7:FEO_demand}
    \begin{aligned}
        \min_{\hat{a}_0,\hat{b}_0,\hat{a}_1\hat{b}_1}&\quad\frac{1}{n_0}\sum_{i:g^{(i)}=0}\left(\hat{a}_0+\hat{b}_0 p^{(i)}-d^{(i)}\right)^2+\frac{1}{n_1}\sum_{i:g^{(i)}=1}\left(\hat{a}_1+\hat{b}_1 p^{(i)}-d^{(i)}\right)^2\\
        \text{s.t.}&\quad \left| \frac{1}{\hat{a}_1^{\mathrm{LS}}}\left(\hat{a}_1^{\mathrm{LS}}+\hat{b}_1^{\mathrm{LS}}p_1^*\left(\hat{a}_1, \hat{b}_1\right)\right)-\frac{1}{\hat{a}_0^{\mathrm{LS}}}\left(\hat{a}_1^{\mathrm{LS}}+\hat{b}_0^{\mathrm{LS}}p_0^*\left(\hat{a}_0,\hat{b}_0\right)\right)\right| \\
        &\quad \leq \left(1-\alpha\right)\left| \frac{1}{\hat{a}_1^{\mathrm{LS}}}\left(\hat{a}_1^{\mathrm{LS}}+\hat{b}_1^{\mathrm{LS}}p_1^{\mathrm{LS}}\right)-\frac{1}{\hat{a}_0^{\mathrm{LS}}}\left(\hat{a}_1^{\mathrm{LS}}+\hat{b}_0^{\mathrm{LS}}p_0^{\mathrm{LS}}\right)\right|,
    \end{aligned}
\end{equation}
which can be rewritten as \eqref{eq:prop7:FEO_demand_2}:  
\begin{equation}
    \label{eq:prop7:FEO_demand_2}
    \begin{aligned}
        \min_{\theta_0,\theta_1}&\quad \left(\theta_0-\theta_0^{\mathrm{LS}}\right)^TM_0\left(\theta_0-\theta_0^{\mathrm{LS}}\right)+\left(\theta_1-\theta_1^{\mathrm{LS}}\right)^TM_1\left(\theta_1-\theta_1^{\mathrm{LS}}\right)\\
        \text{s.t.}&\quad \left| \frac{\hat{b}_1^{\mathrm{LS}}}{\hat{a}_1^{\mathrm{LS}}}p_1^*\left(\theta_1\right)-\frac{\hat{b}_0^{\mathrm{LS}}}{\hat{a}_0^{\mathrm{LS}}}p_0^*\left(\theta_0\right)\right| \leq \left(1-\alpha\right)\left| \frac{\hat{b}_1^{\mathrm{LS}}}{\hat{a}_1^{\mathrm{LS}}}p_1^{\mathrm{LS}} - \frac{\hat{b}_0^{\mathrm{LS}}}{\hat{a}_0^{\mathrm{LS}}}p_0^{\mathrm{LS}}\right|.
    \end{aligned}
\end{equation}

Similar to the case of parity-wise price fairness, we can show that the optimal solution of problem \eqref{eq:prop7:FEO_demand_2} $\left(\theta_0^{\mathrm{FEO}},\theta_1^{\mathrm{FEO}}\right)$ satisfy
\begin{equation}
    \label{eq:prop7:constraint_2}
    \quad  \frac{\hat{b}_1^{\mathrm{LS}}}{\hat{a}_1^{\mathrm{LS}}}p_1^*\left(\theta_1\right)-\frac{\hat{b}_0^{\mathrm{LS}}}{\hat{a}_0^{\mathrm{LS}}}p_0^*\left(\theta_0\right) = \left(1-\alpha\right)\left(\frac{\hat{b}_1^{\mathrm{LS}}}{\hat{a}_1^{\mathrm{LS}}}p_1^{\mathrm{LS}} - \frac{\hat{b}_0^{\mathrm{LS}}}{\hat{a}_0^{\mathrm{LS}}}p_0^{\mathrm{LS}}\right).
\end{equation}
If $\frac{\hat{b}_1^{\mathrm{LS}}}{\hat{a}_1^{\mathrm{LS}}}p_1^*\left(\theta_1\right)-\frac{\hat{b}_0^{\mathrm{LS}}}{\hat{a}_0^{\mathrm{LS}}}p_0^*\left(\theta_0\right) < \left(1-\alpha\right)\left(\frac{\hat{b}_1^{\mathrm{LS}}}{\hat{a}_1^{\mathrm{LS}}}p_1^{\mathrm{LS}} - \frac{\hat{b}_0^{\mathrm{LS}}}{\hat{a}_0^{\mathrm{LS}}}p_0^{\mathrm{LS}}\right)$,  either $p_1^*\left(\theta_1\right)>p_1^{\mathrm{LS}}$ or $p_0^*\left(\theta_0\right)<p_0^{\mathrm{LS}}$. Without loss of generality, suppose $p_0^*\left(\theta_0\right)<p_0^{\mathrm{LS}}$. since $p_0^*(\theta)$ is continuous with respect to theta, there exists $\xi_1\in(0,1]$, such that $\frac{\hat{b}_1^{\mathrm{LS}}}{\hat{a}_1^{\mathrm{LS}}}p_1^*\left(\theta_1\right)-\frac{\hat{b}_0^{\mathrm{LS}}}{\hat{a}_0^{\mathrm{LS}}}p_0^*\left(\xi_1\theta_0^{\mathrm{FEO}}+\left(1-\xi_1\right)\theta_0^{\mathrm{LS}}\right) < \left(1-\alpha\right)\left(\frac{\hat{b}_1^{\mathrm{LS}}}{\hat{a}_1^{\mathrm{LS}}}p_1^{\mathrm{LS}} - \frac{\hat{b}_0^{\mathrm{LS}}}{\hat{a}_0^{\mathrm{LS}}}p_0^{\mathrm{LS}}\right)$. Therefore, $(\xi_1\theta_0^{\mathrm{FEO}}+\left(1-\xi_1\right)\theta_0^{\mathrm{LS}},\theta_1^{\mathrm{FEO}})$ is a feasible solution to \eqref{eq:prop7:FEO_price_2} with strictly smaller objective value, which contradicts to the definition of $(\theta_0^{\mathrm{FEO}},\theta_1^{\mathrm{FEO}})$.

Then the optimal solution to problem \eqref{eq:prop7:FEO_demand_2} is the same as the optimal solution to
\begin{equation}
    \label{eq:prop7:FEO_demand_3}
    \begin{aligned}
        \min_{\theta_0,\theta_1}&\quad \left(\theta_0-\theta_0^{\mathrm{LS}}\right)^TM_0\left(\theta_0-\theta_0^{\mathrm{LS}}\right)+\left(\theta_1-\theta_1^{\mathrm{LS}}\right)^TM_1\left(\theta_1- \theta_1^{\mathrm{LS}}\right)\\
        \text{s.t.}&\quad \frac{\hat{b}_1^{\mathrm{LS}}}{\hat{a}_1^{\mathrm{LS}}}p_1^*\left(\theta_1\right)-\frac{\hat{b}_0^{\mathrm{LS}}}{\hat{a}_0^{\mathrm{LS}}}p_0^*\left(\theta_0\right)= \left(1-\alpha\right)\left( \frac{\hat{b}_1^{\mathrm{LS}}}{\hat{a}_1^{\mathrm{LS}}}p_1^{\mathrm{LS}} - \frac{\hat{b}_0^{\mathrm{LS}}}{\hat{a}_0^{\mathrm{LS}}}p_0^{\mathrm{LS}}\right).
    \end{aligned}
\end{equation}

The KKT conditions of problem \eqref{eq:prop7:FEO_demand_3} with Lagrangian $\mathcal{L}(\theta_0,\theta_1,\lambda)=\left(\theta_0-\theta_0^{\mathrm{LS}}\right)^TM_0\left(\theta_0-\theta_0^{\mathrm{LS}}\right)+\left(\theta_1-\theta_1^{\mathrm{LS}}\right)^TM_1\left(\theta_1- \theta_1^{\mathrm{LS}}\right)+\lambda\left(\frac{\hat{b}_1^{\mathrm{LS}}}{\hat{a}_1^{\mathrm{LS}}}p_1^*\left(\theta_1\right)-\frac{\hat{b}_0^{\mathrm{LS}}}{\hat{a}_0^{\mathrm{LS}}}p_0^*\left(\theta_0\right)- \left(1-\alpha\right)\left( \frac{\hat{b}_1^{\mathrm{LS}}}{\hat{a}_1^{\mathrm{LS}}}p_1^{\mathrm{LS}} - \frac{\hat{b}_0^{\mathrm{LS}}}{\hat{a}_0^{\mathrm{LS}}}p_0^{\mathrm{LS}}\right)\right)$ are
\begin{subequations}
\begin{align}
    \frac{\partial\mathcal{L}}{\partial \theta_0}=2M_0\left(\theta_0-\theta_0^{\mathrm{LS}}\right)-\lambda \frac{\hat{b}_0^{\mathrm{LS}}}{\hat{a}_0^{\mathrm{LS}}}\begin{bmatrix}
\frac{1}{2\hat{b}_0} \\
-\frac{\hat{a}_0}{2\hat{b}_0^2}
\end{bmatrix}=0,\label{eq:prop7:KKT4}\\
\frac{\partial\mathcal{L}}{\partial \theta_1}=2M_1\left(\theta_1-\theta_1^{\mathrm{LS}}\right)+\lambda\frac{\hat{b}_1^{\mathrm{LS}}}{\hat{a}_1^{\mathrm{LS}}} \begin{bmatrix}
\frac{1}{2\hat{b}_1} \\
-\frac{\hat{a}_1}{2\hat{b}_1^2}
\end{bmatrix}=0,\label{eq:prop7:KKT5}\\
\lambda\left(\frac{\hat{b}_1^{\mathrm{LS}}}{\hat{a}_1^{\mathrm{LS}}}p_1^*\left(\theta_1\right)-\frac{\hat{b}_0^{\mathrm{LS}}}{\hat{a}_0^{\mathrm{LS}}}p_0^*\left(\theta_0\right)- \left(1-\alpha\right)\left( \frac{\hat{b}_1^{\mathrm{LS}}}{\hat{a}_1^{\mathrm{LS}}}p_1^{\mathrm{LS}} - \frac{\hat{b}_0^{\mathrm{LS}}}{\hat{a}_0^{\mathrm{LS}}}p_0^{\mathrm{LS}}\right)\right)=0.\nonumber%\label{eq:prop7:KKT6}
\end{align}
\end{subequations}
Differentiate \eqref{eq:prop7:KKT4} and \eqref{eq:prop7:KKT5} with respect to $\alpha$, we have
\begin{align*}
    2M_0\nabla_\alpha\theta_0-\frac{d\lambda}{d\alpha}\frac{\hat{b}_0^{\mathrm{LS}}}{\hat{a}_0^{\mathrm{LS}}}\begin{bmatrix}
\frac{1}{2\hat{b}_0} \\
-\frac{\hat{a}_0}{2\hat{b}_0^2}
\end{bmatrix}-\lambda\frac{\hat{b}_0^{\mathrm{LS}}}{\hat{a}_0^{\mathrm{LS}}}\nabla_\alpha\left(\begin{bmatrix}
\frac{1}{2\hat{b}_0} \\
-\frac{\hat{a}_0}{2\hat{b}_0^2}
\end{bmatrix}\right)=0, \\
2M_1\nabla_\alpha\theta_1+\frac{d\lambda}{d\alpha}\frac{\hat{b}_1^{\mathrm{LS}}}{\hat{a}_1^{\mathrm{LS}}}\begin{bmatrix}
\frac{1}{2\hat{b}_1} \\
-\frac{\hat{a}_1}{2\hat{b}_1^2}
\end{bmatrix}+\lambda\frac{\hat{b}_1^{\mathrm{LS}}}{\hat{a}_1^{\mathrm{LS}}}\nabla_\alpha\left(\begin{bmatrix}
\frac{1}{2\hat{b}_1} \\
-\frac{\hat{a}_1}{2\hat{b}_1^2}
\end{bmatrix}\right)=0.
\end{align*}
When $\alpha=0$, we have $\lambda=0$. Therefore,
\begin{align*}
    \nabla_\alpha\theta_0|_{\alpha=0}=\frac{1}{2}\frac{\hat{b}_0^{\mathrm{LS}}}{\hat{a}_0^{\mathrm{LS}}}M_0^{-1}q_0\left.\frac{d\lambda}{d\alpha}\right|_{\alpha=0}, \\
\nabla_\alpha\theta_1|_{\alpha=0}=-\frac{1}{2}\frac{\hat{b}_1^{\mathrm{LS}}}{\hat{a}_1^{\mathrm{LS}}}M_1^{-1}q_1\left.\frac{d\lambda}{d\alpha}\right|_{\alpha=0}.
\end{align*}

On the other hand, differentiate both sides of \eqref{eq:prop7:constraint_2} with respect to $\alpha$ and set $\alpha=0$, we have
\begin{equation*}
    \frac{\hat{b}_1^{\mathrm{LS}}}{\hat{a}_1^{\mathrm{LS}}}q_1^T\nabla_\alpha\theta_1|_{\alpha=0}-\frac{\hat{b}_0^{\mathrm{LS}}}{\hat{a}_0^{\mathrm{LS}}}q_0^T\nabla_\alpha\theta_0|_{\alpha=0}=-\left(\frac{\hat{b}_1^{\mathrm{LS}}}{\hat{a}_1^{\mathrm{LS}}}p_1^{\mathrm{LS}}-\frac{\hat{b}_0^{\mathrm{LS}}}{\hat{a}_0^{\mathrm{LS}}}p_0^{\mathrm{LS}}\right).
\end{equation*}
Replace $\nabla_\alpha\theta_0|_{\alpha=0}$ and $\nabla_\alpha\theta_1|_{\alpha=0}$ by $\frac{1}{2}\frac{\hat{b}_0^{\mathrm{LS}}}{\hat{a}_0^{\mathrm{LS}}}M_0^{-1}q_0\left.\frac{d\lambda}{d\alpha}\right|_{\alpha=0}$ and $-\frac{1}{2}\frac{\hat{b}_1^{\mathrm{LS}}}{\hat{a}_1^{\mathrm{LS}}}M_1^{-1}q_1\left.\frac{d\lambda}{d\alpha}\right|_{\alpha=0}$, we have
\begin{equation*}
    \frac{1}{2}\left(\left(\frac{\hat{b}_1^{\mathrm{LS}}}{\hat{a}_1^{\mathrm{LS}}}\right)^2q_1^TM_1^Tq_1+\left(\frac{\hat{b}_0^{\mathrm{LS}}}{\hat{a}_0^{\mathrm{LS}}}\right)^2q_0^TM_0q_0\right)\left.\frac{d\lambda}{d\alpha}\right|_{\alpha=0}=\frac{\hat{b}_1^{\mathrm{LS}}}{\hat{a}_1^{\mathrm{LS}}}p_1^{\mathrm{LS}}-\frac{\hat{b}_0^{\mathrm{LS}}}{\hat{a}_0^{\mathrm{LS}}}p_0^{\mathrm{LS}}.
\end{equation*}
Hence,
\begin{equation*}
    \left.\frac{d\lambda}{d\alpha}\right|_{\alpha=0}=\frac{2\left(\frac{\hat{b}_1^{\mathrm{LS}}}{\hat{a}_1^{\mathrm{LS}}}p_1^{\mathrm{LS}}-\frac{\hat{b}_0^{\mathrm{LS}}}{\hat{a}_0^{\mathrm{LS}}}p_0^{\mathrm{LS}}\right)}{\left(\frac{\hat{b}_1^{\mathrm{LS}}}{\hat{a}_1^{\mathrm{LS}}}\right)^2q_1^TM_1^Tq_1+\left(\frac{\hat{b}_0^{\mathrm{LS}}}{\hat{a}_0^{\mathrm{LS}}}\right)^2q_0^TM_0q_0},
\end{equation*}
and
\begin{equation*}
    \begin{aligned}
        \left.\frac{d p_0^*\left(\theta_0^{\mathrm{FEO}}\right)}{d \alpha}\right|_{\alpha=0}=q_0^T\left.\nabla_\alpha\theta_0\right|_{\alpha=0}
        =&\frac{\frac{\hat{b}_0^{\mathrm{LS}}}{\hat{a}_0^{\mathrm{LS}}}q_0^TM_0q_0}{\left(\frac{\hat{b}_1^{\mathrm{LS}}}{\hat{a}_1^{\mathrm{LS}}}\right)^2q_1^TM_1^Tq_1+\left(\frac{\hat{b}_0^{\mathrm{LS}}}{\hat{a}_0^{\mathrm{LS}}}\right)^2q_0^TM_0q_0}\left(\frac{\hat{b}_1^{\mathrm{LS}}}{\hat{a}_1^{\mathrm{LS}}}p_1^{\mathrm{LS}}-\frac{\hat{b}_0^{\mathrm{LS}}}{\hat{a}_0^{\mathrm{LS}}}p_0^{\mathrm{LS}}\right)\\
        =&\frac{\frac{\hat{b}_0^{\mathrm{LS}}}{\hat{a}_0^{\mathrm{LS}}}q_0^TM_0q_0}{\left(\frac{\hat{b}_1^{\mathrm{LS}}}{\hat{a}_1^{\mathrm{LS}}}\right)^2q_1^TM_1^Tq_1+\left(\frac{\hat{b}_0^{\mathrm{LS}}}{\hat{a}_0^{\mathrm{LS}}}\right)^2q_0^TM_0q_0}\frac{c}{2}\left(\frac{\hat{b}_1^{\mathrm{LS}}}{\hat{a}_1^{\mathrm{LS}}}-\frac{\hat{b}_0^{\mathrm{LS}}}{\hat{a}_0^{\mathrm{LS}}}\right).
    \end{aligned}
\end{equation*}

On the other hand, for EFO with parity-wise demand fairness constraints,
\begin{equation*}
    \left.\frac{dp_0^{\mathrm{EFO}}}{d\alpha}\right|_{\alpha=0}=\frac{c}{2}\frac{\frac{\hat{a}_0^{\mathrm{LS}}\hat{a}_1^\mathrm{LS}}{\hat{b}_0^{\mathrm{LS}}}-\frac{\hat{a}_1^{\mathrm{LS}^2}}{\hat{b}_1^{\mathrm{LS}}}}{\frac{\hat{a}_0^{\mathrm{LS}^2}}{\hat{b}_0^{\mathrm{LS}}}+\frac{\hat{a}_1^{\mathrm{LS}^2}}{\hat{b}_1^{\mathrm{LS}}}},
\end{equation*}
so
\begin{align*}
    \left.\frac{d p_0^*\left(\theta_0^{\mathrm{FEO}}\right)}{d \alpha}\right|_{\alpha=0}<\left.\frac{dp_0^{\mathrm{EFO}}}{d\alpha}\right|_{\alpha=0} &\Leftrightarrow \frac{\frac{\hat{b}_0^{\mathrm{LS}}}{\hat{a}_0^{\mathrm{LS}}}q_0^TM_0q_0}{\left(\frac{\hat{b}_1^{\mathrm{LS}}}{\hat{a}_1^{\mathrm{LS}}}\right)^2q_1^TM_1^Tq_1+\left(\frac{\hat{b}_0^{\mathrm{LS}}}{\hat{a}_0^{\mathrm{LS}}}\right)^2q_0^TM_0q_0}\frac{c}{2}\left(\frac{\hat{b}_1^{\mathrm{LS}}}{\hat{a}_1^{\mathrm{LS}}}-\frac{\hat{b}_0^{\mathrm{LS}}}{\hat{a}_0^{\mathrm{LS}}}\right)<\frac{c}{2}\frac{\frac{\hat{a}_0^{\mathrm{LS}}\hat{a}_1^\mathrm{LS}}{\hat{b}_0^{\mathrm{LS}}}-\frac{\hat{a}_1^{\mathrm{LS}^2}}{\hat{b}_1^{\mathrm{LS}}}}{\frac{\hat{a}_0^{\mathrm{LS}^2}}{\hat{b}_0^{\mathrm{LS}}}+\frac{\hat{a}_1^{\mathrm{LS}^2}}{\hat{b}_1^{\mathrm{LS}}}}\\
    &\Leftrightarrow q_0^TM_0^{-1}q_0\hat{b}_0^{\mathrm{LS}}<q_1^TM_1^{-1}q_1\hat{b}_1^{\mathrm{LS}}.
\end{align*}

Therefore, if $q_0^TM_0^{-1}q_0\hat{b}_0^{\mathrm{LS}}<q_1^TM_1^{-1}q_1\hat{b}_1^{\mathrm{LS}}$, there exists $\xi>0$ such that when $0<\alpha<\xi$, FEO leads to lower prices than EFO and hence FEO gives higher consumer surplus and social welfare. Similarly, if $q_0^TM_0^{-1}q_0\hat{b}_0^{\mathrm{LS}}>q_1^TM_1^{-1}q_1\hat{b}_1^{\mathrm{LS}}$, $\frac{d p_0^*\left(\theta_0^{\mathrm{FEO}}\right)}{d \alpha}|_{\alpha=0}>\frac{dp_0^{\mathrm{EFO}}}{d\alpha}|_{\alpha=0}$, then FEO leads to higher prices than EFO and gives lower consumer surplus and social welfare. \hfill \Halmos
\end{proof}
\medskip

\begin{proof}{Proof of Proposition~\ref{prop:convexity_linear}.}
To prove that \eqref{eq:FEO_price_parity_feature_linear} can be solved as a convex optimization problem under the linear demand model, it suffices to transform the constraints into standard form, since the loss, least square error, in the objective is already convex.

For \textbf{price fairness} constraints, we have
\begin{equation*}
    -\frac{\hat a^\top x^{(i)}}{2\hat{b}^\top x^{(i)}}+\frac{c}{2} \leq \tilde{\alpha},
\end{equation*}
where $\tilde{\alpha}$ denotes $(1-\alpha) \bar{p}+\alpha\underline{p}$. This can be rewritten as 
\begin{equation*}
    \hat a^\top x^{(i)} +(2\tilde{\alpha}-c)\left(\hat{b}^\top x^{(i)}\right)\leq 0.
\end{equation*}
The left-hand side is affine in $\hat{a}$ and $\hat{b}$, and therefore this is a convex constraint. Consequently, the overall problem can be written as a convex optimization problem.

For \textbf{demand fairness} constraints, we have
\begin{equation*}
    \left(\hat a^{\mathrm{LS}}\right)^\top x^{(i)} + \hat b^{\mathrm{LS}}\left(-\frac{\hat a^\top x^{(i)}}{2\hat{b}^\top x^{(i)}}+\frac{c}{2}\right) \geq \tilde{\alpha},
\end{equation*}
which is equivalent to,
\begin{equation*}
    \left(2 \left(\hat a^{\mathrm{LS}}\right)^\top x^{(i)}+c\hat{b}^{\mathrm{LS}}-2\tilde{\alpha}\right)\left(\hat{b}^\top x^{(i)}\right) - \hat b^{\mathrm{LS}} \left(\hat a^\top x^{(i)}\right)\leq 0,
\end{equation*}
where the left-hand side is affine in $\hat{a}$ and $\hat{b}$. Therefore, the overall problem is formulated as a convex optimization problem.\hfill \Halmos
\end{proof}
\medskip
\begin{proof}{Proof of Proposition~\ref{prop:convexity_logistic}.}
Similar to the proof of Proposition~\ref{prop:convexity_linear}, the proof of Proposition~\ref{prop:convexity_logistic} relies on reformulating the constraints in standard form. Since the cross-entropy loss in the objective is already convex, it suffices to show that the constraints can be expressed linearly.

For \textbf{price fairness} constraints, we have
\begin{equation}
    1+\mathbf{W}_0\left(\exp\left(\hat{a}^\top \mathbf{x}^{(i)}+\hat{b}c-1\right)\right) + \tilde{\alpha} b \leq 0,
    \label{constraint:pf_log}
\end{equation}
where $\mathbf{W}_0(\cdot)$ denotes the principal branch of the Lambert--W function and $\tilde{\alpha}=(1-\alpha)\bar{p}+\alpha\underline{p}-c$.
More specifically, the (multi-valued) Lambert--W function $\mathbf{W}(z)$ is defined as the solution 
to $\mathbf{W}(z)e^{\mathbf{W}(z)}=z$. 
When $z>0$, the function admits a unique real solution $\mathbf{W}_0(z)>0$. 
In our case, since $z=\exp\left(a^\top \mathbf{x}_i+\hat{b}c-1\right)>0$, 
we focus on the principal branch $\mathbf{W}_0$.

In addition, by implicit differentiation, we have
\begin{equation*}
    z(1+\mathbf{W}_0(z))\frac{d\mathbf{W}_0(z)}{d z}=\mathbf{W}_0(z).
\end{equation*}
As a consequence, this gets the following formula for the derivative of $\mathbf{W}_0(z)$, $\frac{d\mathbf{W}_0(z)}{d z}=\frac{\mathbf{W}_0(z)}{z\left(1+\mathbf{W}_0\left(z\right)\right)}$. Then,
\begin{align*}
    \frac{d\mathbf{W}_0(e^x)}{dx}&= e^x\frac{d\mathbf{W}_0(e^x)}{de^x}=\frac{\mathbf{W}_0(e^x)}{1+\mathbf{W}_0(e^x)},\\
    \frac{d^2\mathbf{W}_0(e^x)}{dx^2}&=\frac{1}{\left(1+\mathbf{W}_0\left(e^x\right)\right)^2}\left(\mathbf{W}_0(e^x)-\frac{\mathbf{W}_0(e^x)^2}{1+\mathbf{W}_0(e^x)}\right)>0.
\end{align*}
Therefore, $\mathbf{W}_0(e^x)$ is a convex function. 
Since $\hat{a}^\top \mathbf{x}^{(i)}+\hat{b}-1$ is affine in $(\hat{a},\hat{b})$, 
it follows from the composition rule that 
$\mathbf{W}_0\!\left(\exp\!\left(\hat{a}^\top \mathbf{x}^{(i)}+\hat{b}c-1\right)\right)$ 
is also convex. Furthermore, since $1+\tilde{\alpha} \hat{b}$ is affine (and hence convex), 
the sum of two convex functions is convex. 
Therefore, the left-hand side of \eqref{constraint:pf_log} is convex.

In conclusion, the Rawlsian price fairness constraint can be expressed in convex form, 
and thus the overall problem can be reformulated as a convex optimization problem.

Similarly, we only need to verify that the \textbf{demand fairness} constraints can be cast into standard convex form. In particular, the demand fairness constraints take the form
\begin{equation*}
    \frac{1}{1+\exp\left(-\left((\hat{a}^{CE})^\top x^{(i)}+\hat{b}^{\mathrm{CE}}p^{\mathrm{CE}*}(x^{(i)};\hat{a},\hat{b})\right)\right)}\geq \tilde{\alpha} \Leftrightarrow \exp\left(-\left((\hat{a}^{CE})^\top x^{(i)}+\hat{b}^{\mathrm{CE}}p^*\right)\right) \leq \frac{1-\tilde{\alpha}}{\tilde{\alpha}},
\end{equation*}
where $\tilde{\alpha}=(1-\alpha)\bar{d}+\alpha\underline{d}$. Take the logarithm of both sides, we have
\begin{equation*}
     -(\hat{a}^{CE})^\top x^{(i)}-\hat{b}^{\mathrm{CE}} \left(\frac{1+\mathbf{W}_0\left(\exp\left(\hat{a}^\top x^{(i)}+\hat{b} c-1\right)\right)}{-\hat b}+c\right)\leq \ln\left(\frac{1-\tilde{\alpha}}{\tilde{\alpha}}\right),
\end{equation*}
which is equivalent to
\begin{equation*}
     1 + \mathbf{W}_0\left(\exp\left(\hat{a}^\top x^{(i)}+\hat{b} c-1\right)\right)  
     -\frac{1}{\hat{b}^{\mathrm{CE}}}\left((\hat{a}^{\mathrm{CE}})^\top x^{(i)}+\ln\left(\frac{1-\tilde{\alpha}}{\tilde{\alpha}}\right)\right)\hat{b}-c\hat{b} \leq 0.
\end{equation*}
From price fairness, we know that $\mathbf{W}_0\!\left(\exp\left(\hat{a}^\top x^{(i)}+\hat{b} c-1\right)\right)$ 
is a convex function. The remaining term, 
$1-\frac{1}{\hat{b}^{\mathrm{CE}}}\Big((\hat{a}^{\mathrm{CE}})^\top x^{(i)}+\ln\!\left(\tfrac{1-\tilde{\alpha}}{\tilde{\alpha}}\right)\Big)\hat{b}$
is affine in $\hat{b}$. Since the sum of a convex function and an affine function is convex, the left-hand side of the constraint is convex. In conclusion, the optimization from with Rawlsian access fairness can be reformulated as a convex optimization problem.\hfill \Halmos
\end{proof}

\section{Numerical Results}
\label{appendix:numerical}
\subsection{Synthetic Data Generation}
\label{appendix:synthetic_data}
To evaluate our framework in a controlled environment, we construct a synthetic dataset in which both the correlation between features and group labels, as well as the demand function, can be explicitly parameterized.

We consider a binary group label, i.e., $g \in \{0,1\}$, drawn independently from a Bernoulli distribution with parameter $\theta_g$. 
For the numerical experiments, we set $m=1$ and, for simplicity, denote $\vx_{i1}$ as $x$ and $\va$ (resp. $\vb$) as $(a_0,a_1)$ (resp. $(b_0,b_1)$). 
We then generate a continuous feature $x$ that is correlated with the group label $g$ at a pre-specified point-biserial correlation coefficient $\rho \in [0,1]$. 
Formally, define the standardized group variable
\begin{equation*}
    z = \frac{g-\theta_g}{\sqrt{\theta_g (1-\theta_g)}}.
    \label{eq:theta_g}
\end{equation*}
Then we construct
\begin{equation*}
    x = \sigma_X \big( \rho\, z + \sqrt{1-\rho^2}\,\varepsilon \big)+\kappa,
\end{equation*}
where $\varepsilon$ is independent of $g$, with $\E[\varepsilon]=0$ and $Var(\varepsilon)=1$. The constant $\kappa$ is a location parameter that can be chosen freely to shift the support of $x$ (e.g., to ensure $x>0$). By construction, $Var(x) = \sigma_X^2$ and $Corr(x,g) = \rho$ (Refer to Lemma~\ref{lemma:corr} provided at the end of this section.). For instance, if we set $\rho=1$, $\sigma_X=\sqrt{\theta_g(1-\theta_g)}$, and $\kappa=\theta_g$, then $x=g$. 

This simulation design allows us to flexibly control (1) the marginal distribution of the group label, (2) the correlation between features and group identity, and (3) the sensitivity of valuations to observed characteristics. In this paper, we set $\rho = 0.5$, which represents a practically relevant case: features are typically correlated with demographic attributes that are sometimes legally protected (e.g., gender), but the correlation is neither perfect nor zero.%In particular, by varying $\rho$, we can study environments ranging from perfectly segregated groups ($\rho=1$) to statistically indistinguishable ones ($\rho=0$). %Price and demand generation depends on the form of the true demand model.

\begin{lemma}
Let $g \in \{0,1\}$ be drawn from a Bernoulli distribution with parameter $\theta_g$, and define $z = \tfrac{g-\theta_g}{\sqrt{\theta_g(1-\theta_g)}}$. Let $\varepsilon$ is independent of $g$, with $\E[\varepsilon]=0$ and $\operatorname{Var}(\varepsilon)=1$, and define $x = \sigma_X \big( \rho z + \sqrt{1-\rho^2}\,\varepsilon \big)+\kappa$. Then $\operatorname{Var}(x) = \sigma_X^2$ and $\operatorname{Corr}(x,g) = \rho$.
\label{lemma:corr}
\end{lemma}
\begin{proof}{Proof.}
Since $g \sim \text{Bernoulli}(\theta_g)$, we have $\E[g] = \theta_g$ and $\operatorname{Var}(g) = \theta_g(1-\theta_g)$. 
Therefore, $\operatorname{Var}(z) = 1$.  
By independence of $z$ and $\varepsilon$, 
$$\operatorname{Var}(x) = \sigma_X^2\big(\rho^2\operatorname{Var}(z) + (1-\rho^2)\operatorname{Var}(\varepsilon)\big) 
= \sigma_X^2(\rho^2 + 1-\rho^2) = \sigma_X^2.$$
%Next, 
%$$Cov(x,z) = \sigma_X\E\left[\left(\rho z + \sqrt{1-\rho^2}\varepsilon+\kappa\right)z\right] 
%= \sigma_X \rho \E[z^2] = \sigma_X \rho \E[z^2]=\sigma_X \rho\left(Var(z)-\E[z]^2\right) =\sigma_X \rho.$$
Next, since $\mathbb{E}[z] = 0$ and $\mathbb{E}[\varepsilon z] = 0$, we have:
\begin{align*}
\operatorname{Cov}(x,z) &= \mathbb{E}\left[\left(\sigma_X(\rho z + \sqrt{1-\rho^2}\varepsilon) + \kappa\right)z\right] - \mathbb{E}[x]\mathbb{E}[z] \\
&= \sigma_X \rho \mathbb{E}[z^2] + \sigma_X \sqrt{1-\rho^2}\mathbb{E}[\varepsilon z] + \kappa \mathbb{E}[z] \\
&= \sigma_X \rho \mathbb{E}[z^2] = \sigma_X \rho \left(\operatorname{Var}(z) + \mathbb{E}[z]^2\right) = \sigma_X \rho.
\end{align*}
Thus, 
$$\operatorname{Corr}(x,z) = \frac{\operatorname{Cov}(x,z)}{\sqrt{\operatorname{Var}(x)}\sqrt{\operatorname{Var}(z)}} 
= \frac{\sigma_X \rho}{\sigma_X \cdot 1} = \rho.$$
Since $z$ is a standardized linear transformation of $g$, we also have $Corr(x,g)=\rho$.\hfill\Halmos
\end{proof}
\subsection{Heuristic Solution for Parity-wise Fairness}
\label{appendix:heuristic}
If the demand function is $d = a^\top x_i + b_g p_i$, then parity-wise price and demand fairness can be solved by a one-dimensional parameter search with convex optimization subproblems. The original problem is
\begin{equation*}
\begin{aligned}
    \ell(w,\alpha) := \min_{\hat{a}, \, \hat{b}_0 \leq 0, \, \hat{b}_1 \leq 0} \quad & 
        \sum_{g \in \{0,1\}} \frac{1}{n_g} \sum_{\{i \mid g_i = g\}} 
        \big(\hat a^\top x_i + \hat b_g p_i \big)^2 \\
    \text{s.t.} \quad &
        \left|\frac{1}{n_1}\sum_{\{i \mid g_i=1\}} \frac{\hat a^\top x_i}{(-2\hat b_1)}
        - \frac{1}{n_0}\sum_{\{i \mid g_i=0\}} \frac{\hat a^\top x_i}{(-2\hat b_0)}\right|
        \leq (1-\alpha)\Delta.
\end{aligned}
\end{equation*}

The terms $\tfrac{\hat a^\top x_i}{2\hat b_0}$ and $\tfrac{\hat a^\top x_i}{2\hat b_1}$ induce non-convexity.  
We can rewrite the constraint as
\begin{equation*}
\begin{aligned}
    &\left|\frac{1}{n_1}\sum_{\{i \mid g_i=1\}} \frac{\hat a^\top x_i}{(-2\hat b_1)}
    - \frac{1}{n_0}\sum_{\{i \mid g_i=0\}} \frac{\hat a^\top x_i}{(-2\hat b_0)}\right|
    \leq (1-\alpha)\Delta \\
    \Leftrightarrow \; &
    \frac{\hat a^{\top}\bar{x}_1}{(-2\hat b_1)} - \frac{\hat a^{\top}\bar{x}_0}{(-2\hat b_0)}
    \leq (1-\alpha)\Delta 
    \quad \text{and} \quad
    \frac{\hat a^{\top}\bar{x}_0}{(-2\hat b_0)} - \frac{\hat a^{\top}\bar{x}_1}{(-2\hat b_1)}
    \leq (1-\alpha)\Delta \\
    \Leftrightarrow \; &
    \hat a^{\top}\bar{x}_1 - (-2\hat b_1)\frac{\hat a^{\top}\bar{x}_0}{(-2\hat b_0)}
    \leq (1-\alpha)\Delta(-2\hat b_1), \\
    &(-2\hat b_1)\frac{\hat a^{\top}\bar{x}_0}{(-2\hat b_0)} - \hat a^{\top}\bar{x}_1
    \leq (1-\alpha)\Delta(-2\hat b_1),
\end{aligned}
\end{equation*}
where $\bar{x}_g = \tfrac{1}{n_g}\sum_{\{i \mid g_i=g\}} x_i$.  
Now, introduce the parameter $w = \frac{\hat a^{\top}\bar{x}_0}{(-2\hat b_0)}$. Then the subproblems take the form
\begin{equation}
\begin{aligned}
    \ell(w,\alpha) := \min_{\hat{a}, \, \hat{b}_0 \leq 0, \, \hat{b}_1 \leq 0} \quad & 
        \sum_{g \in \{0,1\}} \frac{1}{n_g} \sum_{\{i \mid g_i=g\}}
        \big(\hat a^\top x_i + \hat b_g p_i \big)^2 \\
    \text{s.t.} \quad & \hat a^{\top}\bar{x}_1 - (-2\hat b_1)w \leq (1-\alpha)\Delta(-2\hat b_1), \\
    & (-2\hat b_1)w - \hat a^{\top}\bar{x}_1 \leq (1-\alpha)\Delta(-2\hat b_1), \\
    & \hat a^{\top}\bar{x}_0 \leq w(-2\hat b_0), \\
    & \hat a^{\top}\bar{x}_0 \geq w(-2\hat b_0).
\end{aligned}
\label{sub_problem}
\end{equation}

This problem is convex for each fixed $w$ and $\alpha$. Since we replace 
$\tfrac{\hat a^{\top}\bar{x}_0}{(-2\hat b_0)}$ (which corresponds to the optimal price minus $\tfrac{c}{2}$) by $w$, the feasible range of $w$ lies between $\tfrac{c}{2}$ and the price cap minus $\tfrac{c}{2}$.

\begin{algorithm}[ht]
\caption{Grid Search over $w$ with Convex Optimization Subproblems}
\begin{algorithmic}[1]
\REQUIRE Step size $\delta$, fairness parameter $\alpha$
\STATE Initialize $\ell^* \gets +\infty$
\FOR{each $w \in \left[\tfrac{c}{2}, \tfrac{c}{2}+\delta, \dots, p_{\text{cap}} - \tfrac{c}{2}\right]$}
    \STATE Solve \eqref{sub_problem}%$\min_{\hat{a},\hat{b}_0,\hat{b}_1} \ell(w,\alpha)$
    \IF{$\ell(w,\alpha) < \ell^*$}
        \STATE $\ell^* \gets \ell(w,\alpha)$
        %\STATE $w^* \gets w$
        \STATE $(\hat{a},\hat{b}_0,\hat{b}_1) \gets \left(\hat{a}(w,\alpha),\hat{b}_0(w,\alpha),\hat{b}_1(w,\alpha)\right)$
    \ENDIF
\ENDFOR
\RETURN{$\hat{a},\hat{b}_0,\hat{b}_1$}
\end{algorithmic}
\end{algorithm}

\section{Case Study}
\label{appendix:case_study}

In Appendix~\ref{appendix:linear_well_specified} and Appendix~\ref{appendix:linear_misspecified}, we consider estimation using a linear demand model, when the true demand model is well-specified and mis-specified. Similar to Section \ref{appendix:well_specified_logistic}, we set a normalized unit cost to $0.05$ (corresponding to $50$ SEK without normalization), and a normalized price cap to $1.2$ (corresponding to $1200$ SEK without normalization). For Rawlsian fairness, we set $\bar{p}$ (resp. $\underline{p}$) in~\eqref{eq:rawls_price_feature} as the highest (resp. lowest) optimal price obtained on the training set using parameters estimated without fairness constraints. We set $\bar{d}$ as $0.5$ and $\underline{d}$ as the lowest demand observed in the training set without fairness constraints.

\subsection{Linear Demand Estimation with Linear True Demand}
\label{appendix:linear_well_specified}
We examine a linear demand setting under the assumption that the true demand model is also linear. First, Table~\ref{tab:measure_loss_fairness} illustrates the variations in profit, consumer surplus, and social welfare as $\alpha$ increases. We observe that when $\alpha=0.4$ and $0.6$, the profit, consumer surplus, and social welfare are all lower than those achieved without fairness constraints, suggesting that imposing parity-wise loss fairness may lead to negative outcomes.

\begin{table}[htbp]
    \centering
    \caption{Normalized performance measures across $\alpha$ under parity-wise loss fairness (\%)}
    \label{tab:measure_loss_fairness}
    \renewcommand{\arraystretch}{1.1} 

    \begin{tabular}{lccccc}
    \toprule
    & \multicolumn{5}{c}{$\alpha$} \\
    \cmidrule(lr){2-6}
    Measure & 0.2 & 0.4 & 0.6 & 0.8 & 1.0 \\
    \midrule
    $\mathcal{R}(\alpha)/\mathcal{R}(0)$ & 100.19 & 99.68 & 99.45 & 99.31 & 99.20 \\
    $\mathcal{S}(\alpha)/\mathcal{S}(0)$ & 98.00  & 98.48 & 99.38 & 100.29 & 101.15 \\
    $\mathcal{W}(\alpha)/\mathcal{W}(0)$ & 99.41  & 99.26 & 99.43 & 99.66 & 99.90 \\
    \bottomrule
    \end{tabular}
\end{table}

%\textcolor{red}{dot above the red line (1) the plot: the range of ditb. gets to be large. (2) 0.4-0.6, every value $<$ 100 (3) non-monotonic r/s/w (3) plot 0.0, 0.2, 0.4, 0.6, 0.8, 1.0; Figure 8(b) legend p $\rightarrow$ d, Note: explain the legend? For each table/figure, we should have punch line(s)... For instance, in this case, it is beter to use FEO than EFO.}

Figure~\ref{fig:linear_parity_price_case} illustrates the outcomes under parity-wise price fairness. Similar to Section~\ref{appendix:well_specified_logistic}, when $\alpha=1$, the average prices of two groups become the same only for EFO due to the generalization error of FEO. Consistent with the implication of Corollary~\ref{corollary2}, Table~\ref{tab:measure_parity_price} shows that FEO yields higher consumer surplus and social welfare. More specifically, accepting $1.13\%$ greater profit loss under FEO unlocks an $8.34\%$ higher gain in consumer surplus than EFO under perfect fairness (i.e., $\alpha=1$). Thus, using FEO under parity-wise price fairness will be more beneficial to customers.

\begin{figure}[htbp]
\FIGURE{
\begin{minipage}{\textwidth}
\centering
\captionsetup{justification=centering}
\begin{subfigure}{0.35\textwidth}
    \centering
    \includegraphics[width=\textwidth]{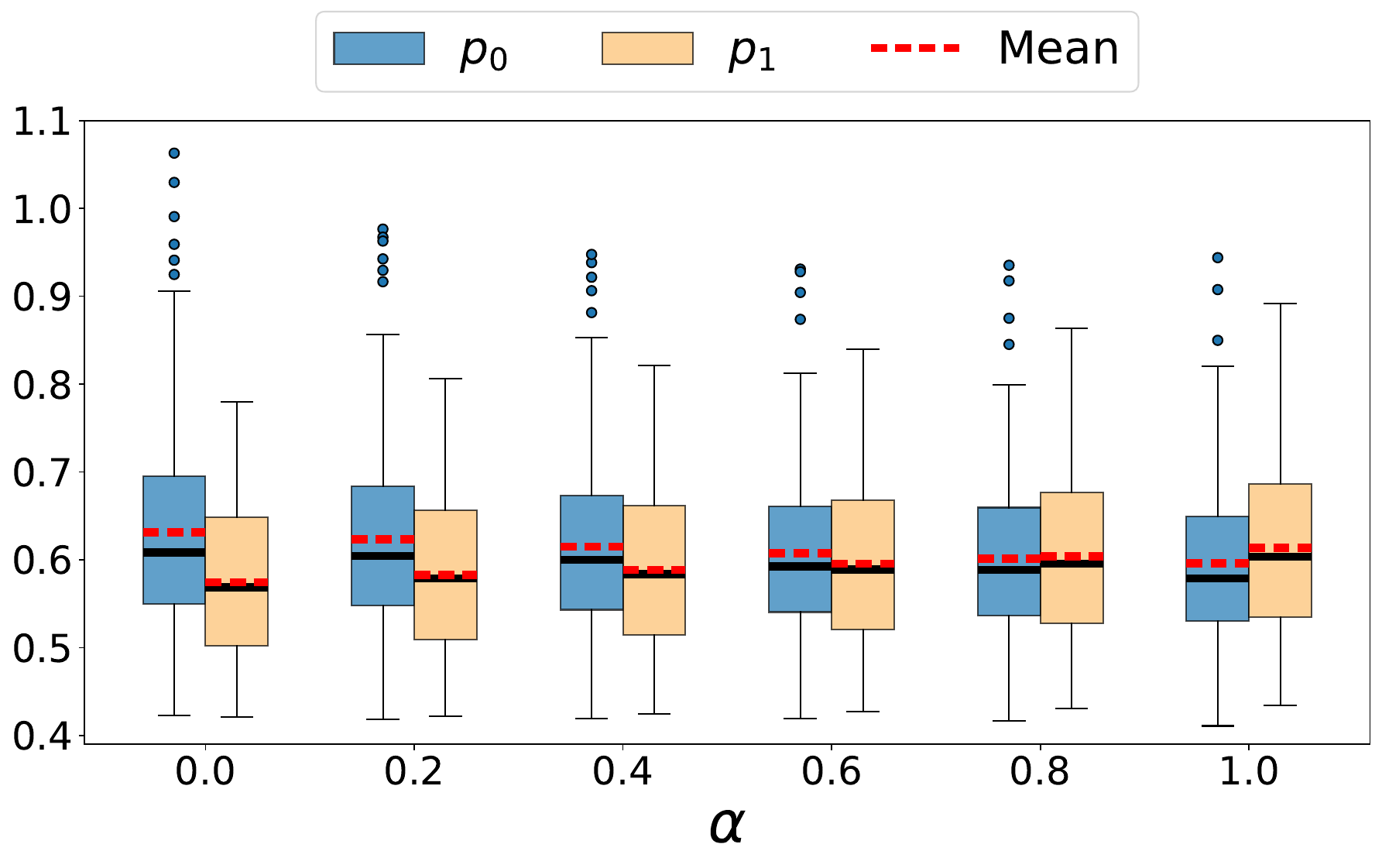}
    \caption{Price (FEO)}
\end{subfigure}
\begin{subfigure}{0.35\textwidth}
    \centering
    \includegraphics[width=\textwidth]{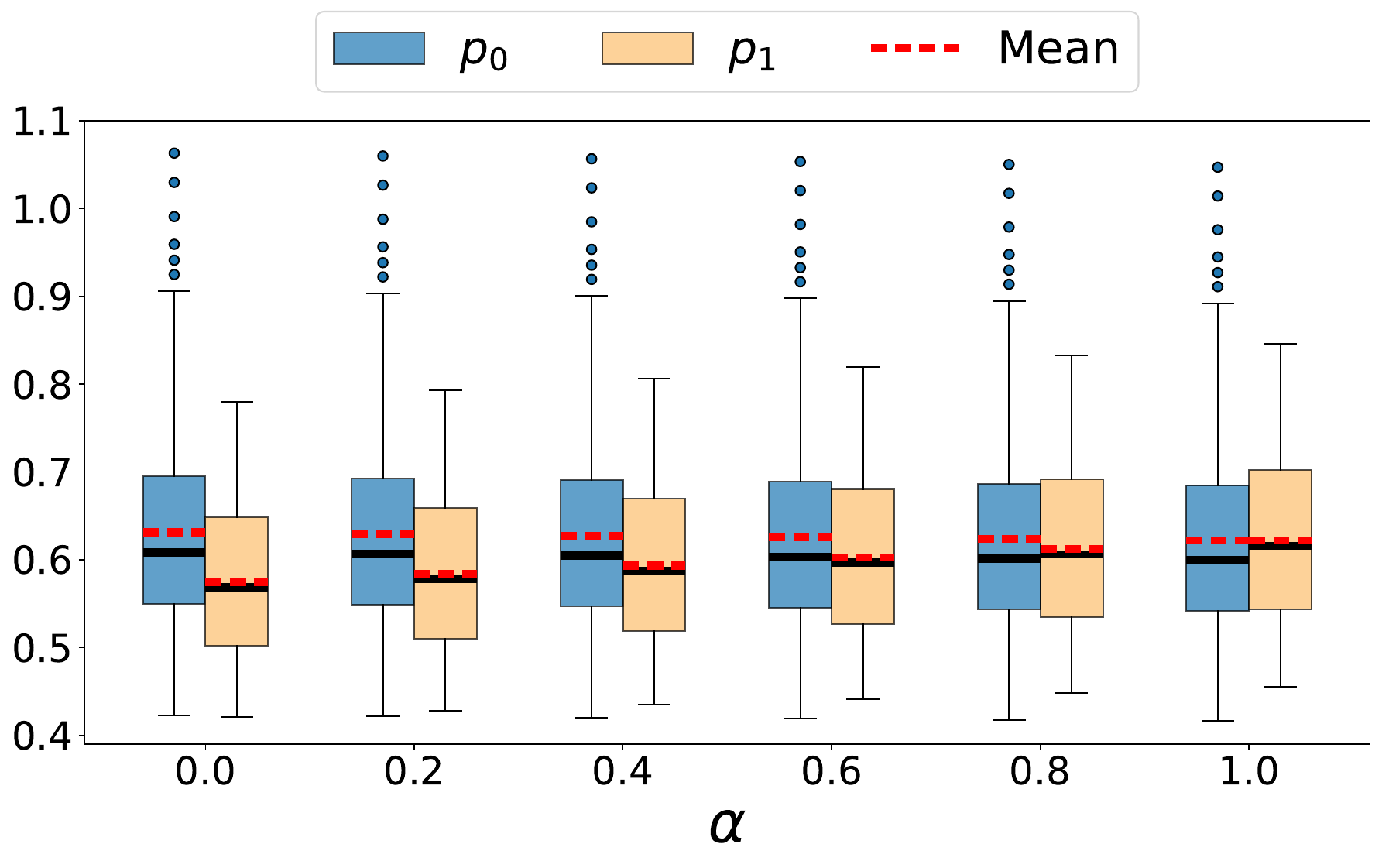}
    \caption{Price (EFO)}
\end{subfigure}
\end{minipage}
}
{Price distributions for FEO and EFO under parity-wise price fairness \label{fig:linear_parity_price_case}\vspace{1mm}}
{}
\end{figure}

\begin{table}[htbp]
    \centering
    \caption{Normalized performance measures for FEO and EFO under parity-wise price fairness (\%)}
    \label{tab:measure_parity_price}
    \renewcommand{\arraystretch}{1.1}
    \setlength{\tabcolsep}{4pt}
    
    \begin{tabular}{l ccccc c ccccc}
        \toprule
        & \multicolumn{5}{c}{FEO ($\alpha$)} & & \multicolumn{5}{c}{EFO ($\alpha$)} \\
        \cmidrule(lr){2-6} \cmidrule(lr){8-12}
        Measure & 0.2 & 0.4 & 0.6 & 0.8 & 1.0 & & 0.2 & 0.4 & 0.6 & 0.8 & 1.0 \\
        \midrule
        % Revenue Row
        $\mathcal{R}(\alpha)/\mathcal{R}(0)$  
        & 99.73  & 99.53  & 99.26  & 98.94  & 98.58  & 
        & 99.96  & 99.92  & 99.86  & 99.79  & 99.71  \\ 
        
        % Surplus Row
        $\mathcal{S}(\alpha)/\mathcal{S}(0)$  
        & 101.97 & 104.14 & 106.02 & 107.59 & 108.83 & 
        & 100.08 & 100.17 & 100.28 & 100.39 & 100.51 \\ 
        
        % Welfare Row
        $\mathcal{W}(\alpha)/\mathcal{W}(0)$  
        & 100.53 & 101.17 & 101.66 & 102.01 & 102.22 & 
        & 100.01 & 100.01 & 100.01 & 100.00 & 100.00 \\ 
        \bottomrule
    \end{tabular}
\end{table}

%\textcolor{red}{partiy-wise price fairness: (1) EFO at alpha=1, same mean of price -- we need to explain why means at alpha=1 under FEO are differ (2) maybe we can talk about the distribution of price (FEO -- range goes down, EFO -- keep the range (or distribution)) ; Table: FEO gives better surplus/welfare (if we }

Figure~\ref{fig:linear_parity_demand_case} presents the results of parity-wise demand fairness. Similar to Section~\ref{appendix:well_specified_logistic}, due to the estimation error, the estimated demand function differs from the true demand model. As a result, there exists a discrepancy between the ex-ante and ex-post demand under both EFO and FEO, and a small ex-ante demand gap may not accurately reflect the true demand gap, which is different from logistic demand estimation. Consequently, the ex-post demand does not change significantly for both FEO and EFO, and Table~\ref{tab:measure_parity_demand} shows that imposing parity-wise demand fairness on ex-ante demand does not lead to substantial changes in profit, consumer surplus, or social welfare. This suggests that implementing parity-wise demand fairness in practice may be challenging.
\begin{figure}[htbp]
\FIGURE{
\begin{minipage}{\textwidth}
\centering
\captionsetup{justification=centering}
\begin{subfigure}{0.35\textwidth}
    \centering
    \includegraphics[width=\textwidth]{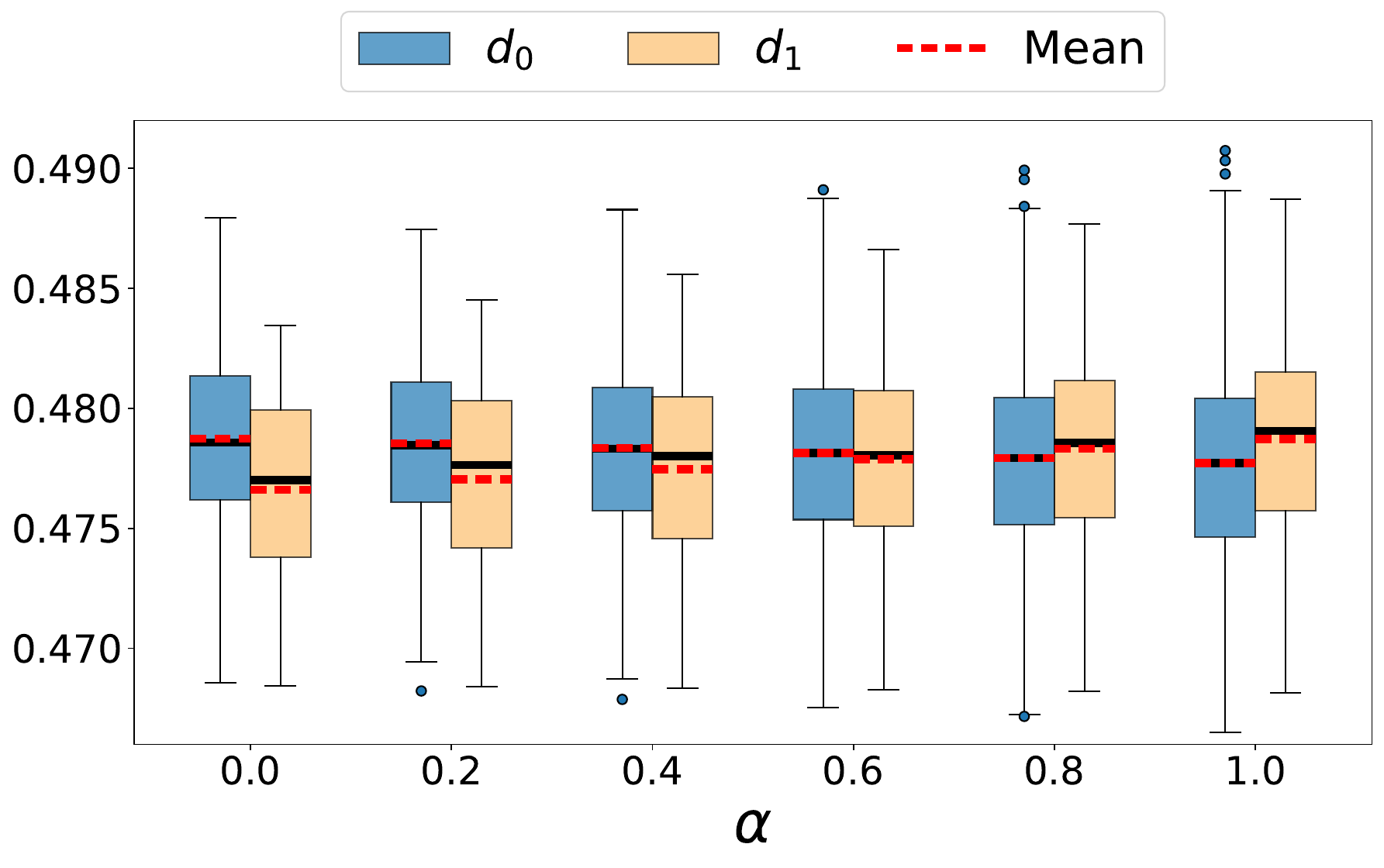}
    \caption{Ex Ante Demand (FEO)}
\end{subfigure}
\begin{subfigure}{0.35\textwidth}
    \centering
    \includegraphics[width=\textwidth]{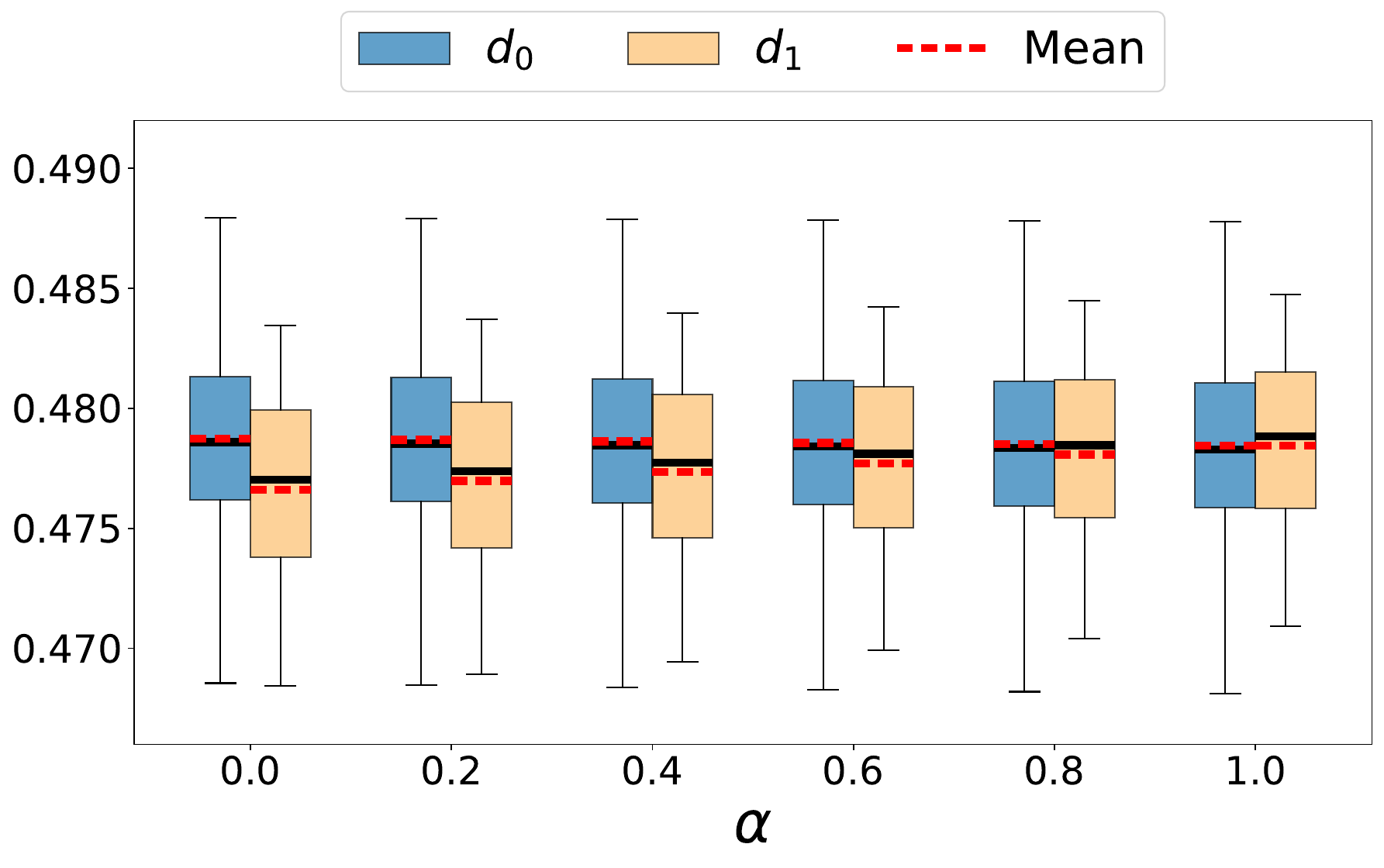}
    \caption{Ex Ante Demand (EFO)}
\end{subfigure}
\begin{subfigure}{0.35\textwidth}
    \centering
    \includegraphics[width=\textwidth]{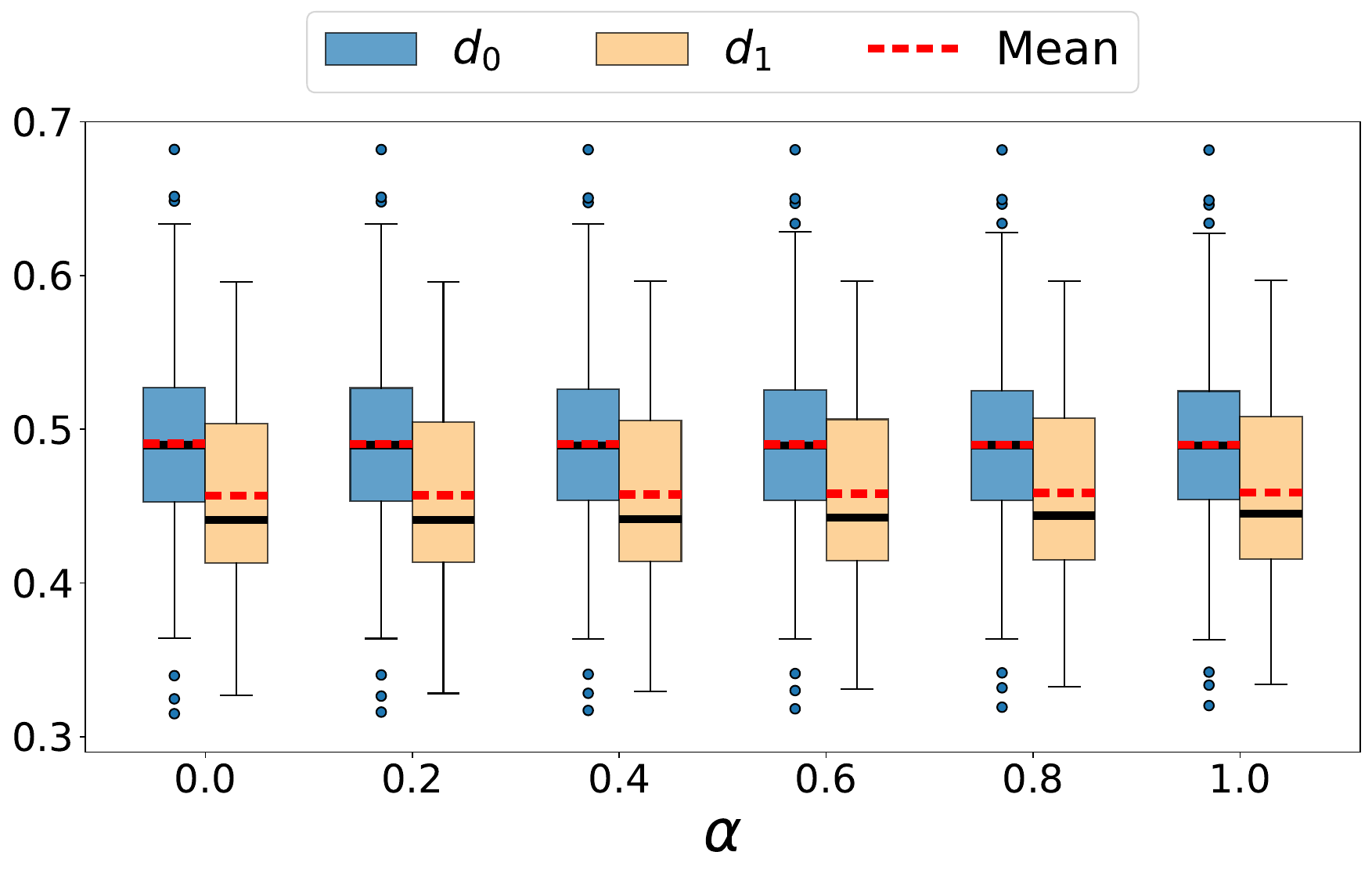}
    \caption{Ex Post Demand (FEO)}
\end{subfigure}
\begin{subfigure}{0.35\textwidth}
    \centering
    \includegraphics[width=\textwidth]{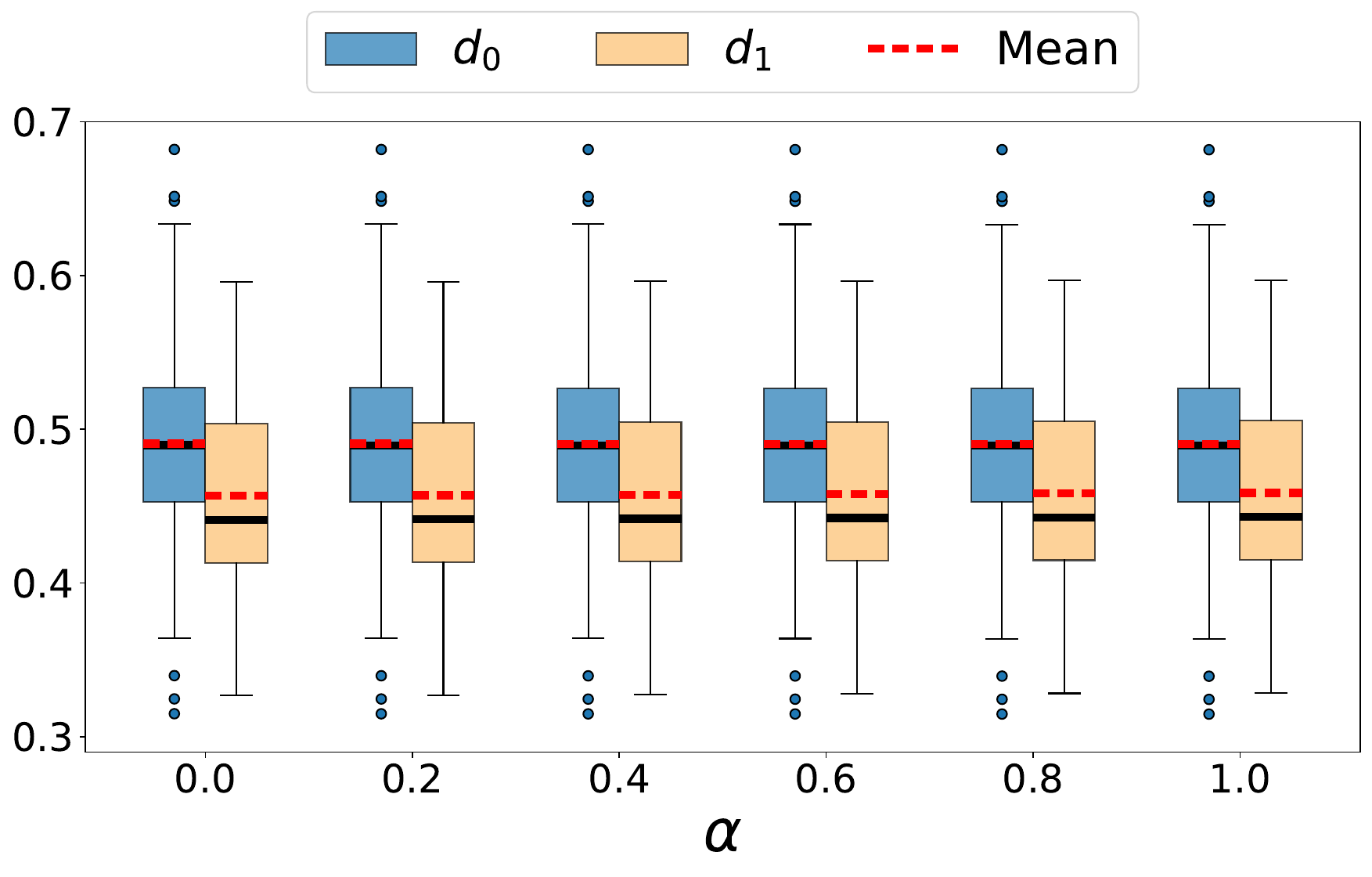}
    \caption{Ex Post Demand (EFO)}
\end{subfigure}
\end{minipage}
}
{Demand distributions for FEO and EFO under parity-wise demand fairness \label{fig:linear_parity_demand_case}\vspace{2mm}}
{}
\end{figure}

\begin{table}[htbp]
    \centering
    \caption{Normalized performance measures for FEO and EFO under parity-wise demand fairness (\%)}
    \label{tab:measure_parity_demand}
    \renewcommand{\arraystretch}{1.1}
    
    \begin{tabular}{l ccccc c ccccc}
        \toprule
        & \multicolumn{5}{c}{FEO ($\alpha$)} & & \multicolumn{5}{c}{EFO ($\alpha$)} \\
        \cmidrule(lr){2-6} \cmidrule(lr){8-12}
        Measure & 0.2 & 0.4 & 0.6 & 0.8 & 1.0 & & 0.2 & 0.4 & 0.6 & 0.8 & 1.0 \\
        \midrule
        $\mathcal{R}(\alpha)/\mathcal{R}(0)$ & 100.01 & 100.02 & 100.04 & 100.05 & 100.06 & & 100.0 & 100.0 & 100.0 & 100.01 & 100.01 \\
        
        $\mathcal{S}(\alpha)/\mathcal{S}(0)$ & 99.95  & 99.89  & 99.84  & 99.78  & 99.73  & & 100.0 & 100.0 & 100.0 & 99.99  & 99.99  \\
        
        $\mathcal{W}(\alpha)/\mathcal{W}(0)$ & 99.99  & 99.98  & 99.97  & 99.95  & 99.94  & & 100.0 & 100.0 & 100.0 & 100.00 & 100.00 \\
        \bottomrule
    \end{tabular}
\end{table}

%\textcolor{red}{dot above the red line; different color (not gray, blue, orange), figures (a) \& (b) mean at $\alpha=1$, figure (c) and (d) small changes; Table: small changes because the demand gap is already small at $\alpha=0$. Similar tendency for linear demand estimator, but not for logistic demand estimator.}

Figure~\ref{fig:linear_Rawlsian_price_case} illustrates the price distribution under Rawlsian price fairness. In both FEO and EFO settings, the maximum price declines monotonically with $\alpha$, effectively curbing extreme price points. In Table~\ref{tab:measure_rawlsian_price}, while the profit difference between the two remains marginal, FEO consistently delivers better consumer surplus and social welfare, particularly up to $\alpha=0.6$, the threshold before the distribution collapses into a few discrete points (See Figure~\ref{fig:linear_Rawlsian_price_case}). This performance gap stems from FEO’s ability to reshape the entire price distribution, whereas EFO merely imposes a cap on high-end prices. Therefore, FEO is an effective path for firms prioritizing consumer surplus with a slight profit trade-off. This is especially true at moderate $\alpha$ levels, where the approach curbs extreme pricing while retaining the benefits of personalization. Similar to Section~\ref{appendix:well_specified_logistic}, when we impose Rawlsian price fairness for FEO at fairness level $\alpha=0.2$, profit, consumer surplus and social welfare are all higher than the baseline (i.e., $\alpha=0$), showing that a small level of Rawlsian price fairness for FEO can lead to a win-win outcome.
%\textcolor{red}{Rawlsian Price Fairness: not gray? the span of distribution of FEO goes down; at $0.2$, FEO gives reulsts $>100$ when alpha is small, FEO gives better consumer surplus/welfare compared to EFO. Rawlsian fairness are harder constraints -- rawlsian fairness (several constraints); parity-wise (one constraint)}

\begin{figure}[htbp]
\FIGURE{
\begin{minipage}{\textwidth}
\centering
\captionsetup{justification=centering}
\begin{subfigure}{0.35\textwidth}
    \centering
    \includegraphics[width=\textwidth]{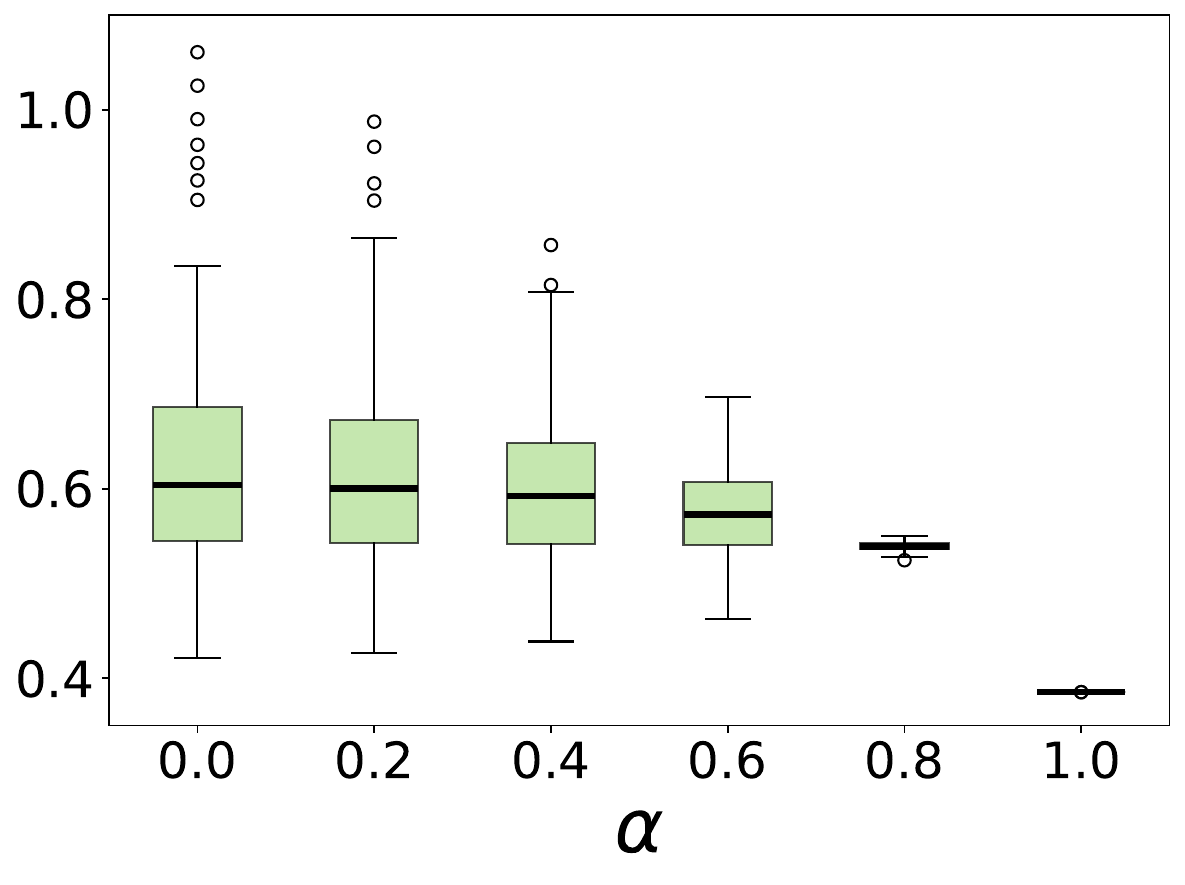}
    \caption{Price under FEO}
    % \label{fig:rawls_price_feo}
\end{subfigure}
\begin{subfigure}{0.35\textwidth}
    \centering
    \includegraphics[width=\textwidth]{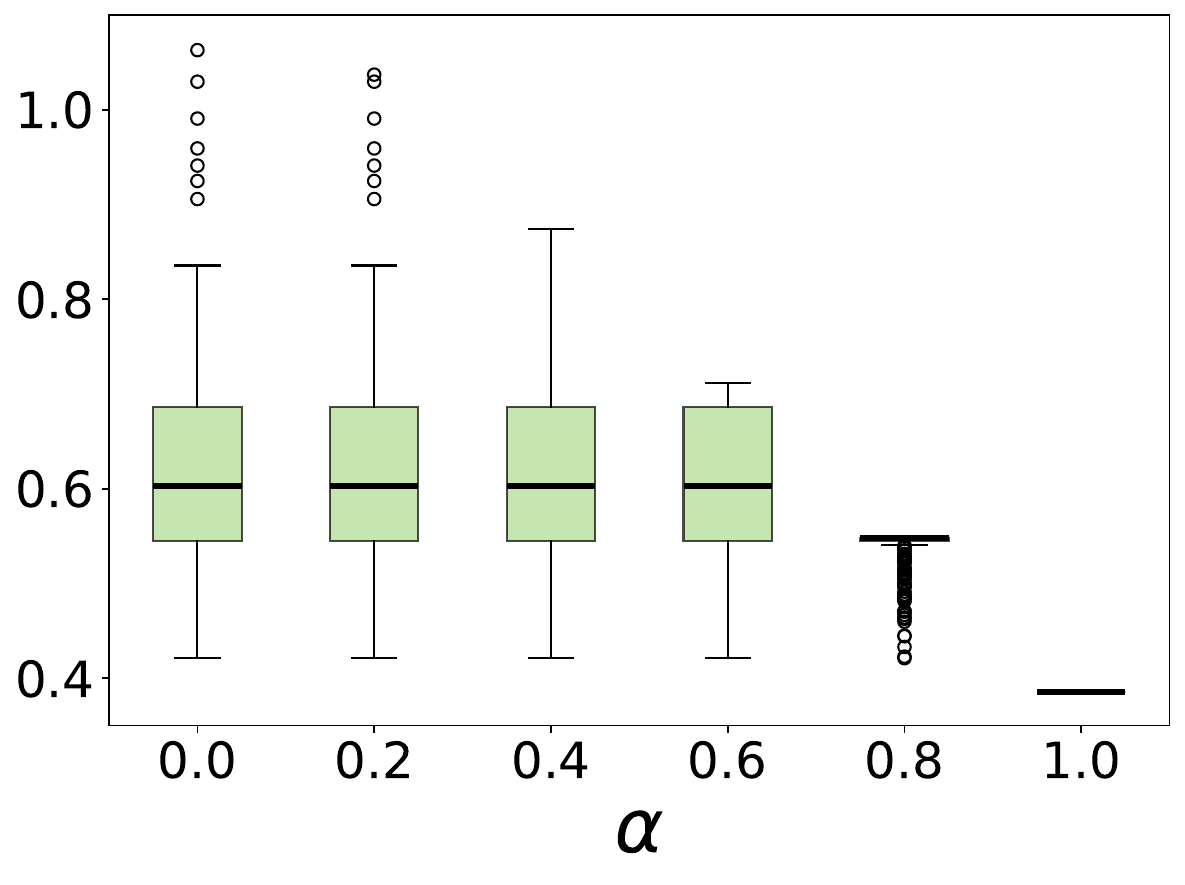}
    \caption{Price under EFO}
    % \label{fig:rawls_price_efo}
\end{subfigure}
\end{minipage}
}
{Price distributions for FEO and EFO under Rawlsian price fairness \label{fig:linear_Rawlsian_price_case}\vspace{0.5em}}
{}
\end{figure}

\begin{table}[htbp]
    \centering
    \caption{Normalized performance measures for FEO and EFO under Rawlsian price fairness (\%)}
    \label{tab:measure_rawlsian_price}
    \renewcommand{\arraystretch}{1.1}
    % No font size reduction (\small/\footnotesize) as requested
    
    \begin{tabular}{l p{0.2cm} ccccc p{0.4cm} ccccc}
        \toprule
        & & \multicolumn{5}{c}{FEO ($\alpha$)} & & \multicolumn{5}{c}{EFO ($\alpha$)} \\
        \cmidrule(lr){3-7} \cmidrule(lr){9-13}
        Measure & & 0.2 & 0.4 & 0.6 & 0.8 & 1.0 & & 0.2 & 0.4 & 0.6 & 0.8 & 1.0 \\
        \midrule
        $\mathcal{R}(\alpha)/\mathcal{R}(0)$ & & 100.01 & 99.70 & 98.61 & 95.61 & 81.14 & & 99.99  & 99.93 & 99.61 & 96.01 & 81.13\\
        
        $\mathcal{S}(\alpha)/\mathcal{S}(0)$ & & 102.70 & 108.27 & 116.41 & 130.20 & 191.43 & & 100.05 & 101.02 & 105.55 & 130.61 & 191.57\\
        
        $\mathcal{W}(\alpha)/\mathcal{W}(0)$ & & 100.97 & 102.74 & 104.94 & 107.91 & 120.35 & & 100.01 & 100.32 & 101.72 & 108.31 & 120.38\\
        \bottomrule
    \end{tabular}
\end{table}

Lastly, Figure~\ref{fig:linear_Rawlsian_demand_case} illustrates the resulting demand distribution under Rawlsian fairness. As $\alpha$ increases, the minimum ex-ante and ex-post demand rises under both FEO and EFO, ensuring that no individual faces a demand level below a certain threshold. Moreover, the minimum ex-post demand also increases for both FEO and EFO, showing that imposing Rawlsian demand fairness indeed benefits individuals with low demand.   Table~\ref{tab:measure_rawlsian_demand} shows that profit, consumer surplus, and social welfare are closely aligned between the two settings. For instance, at $\alpha = 1$, both settings yield similar trade-offs: FEO (resp. EFO) incurs a marginal profit loss of $0.53\%$ (resp. $0.31\%$) to achieve a relatively substantial increase in consumer surplus of $9.95\%$ (resp. $8.10\%$).
\begin{figure}[htbp]
\FIGURE{
\begin{minipage}{\textwidth}
\centering
\captionsetup{justification=centering}
\begin{subfigure}{0.235\textwidth}
    \centering
    \includegraphics[width=\textwidth]{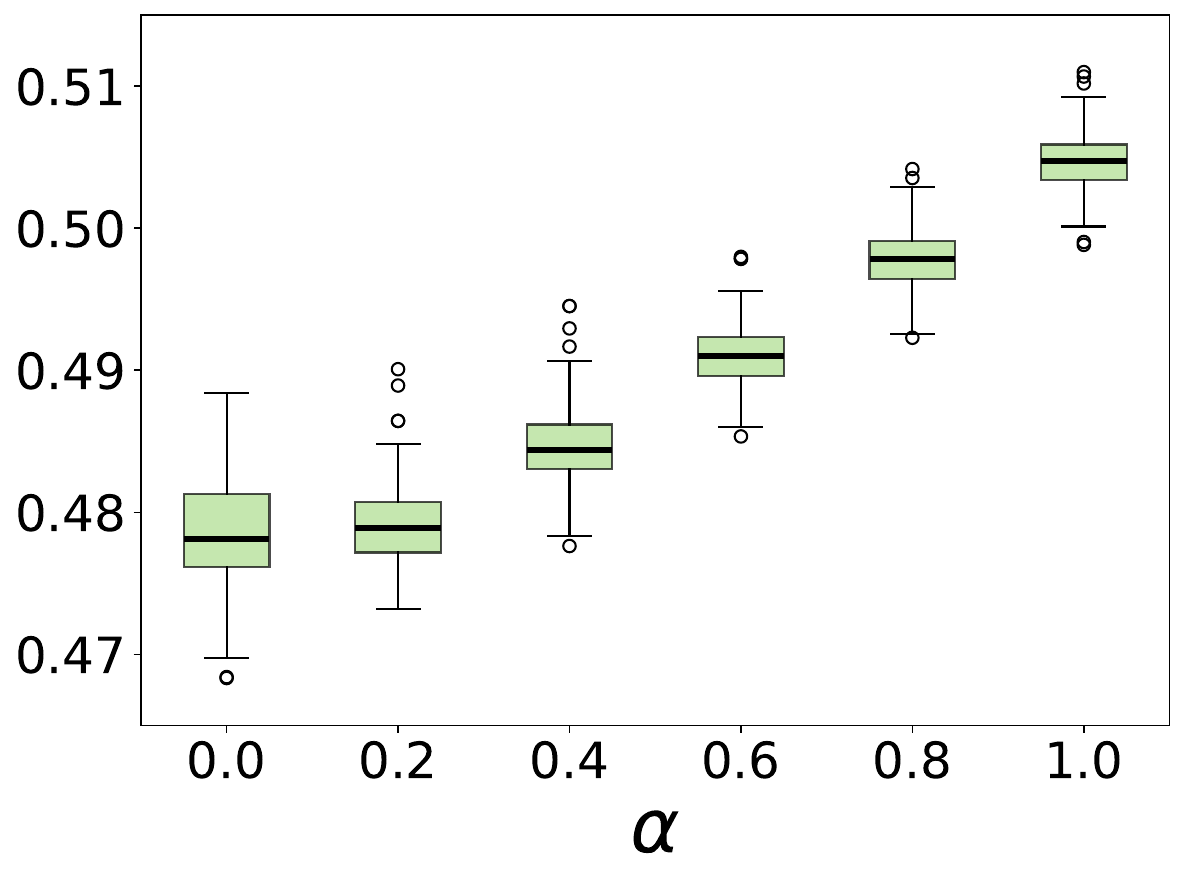}
    \caption{Ex Ante Demand (FEO)}
\end{subfigure}
\begin{subfigure}{0.235\textwidth}
    \centering
    \includegraphics[width=\textwidth]{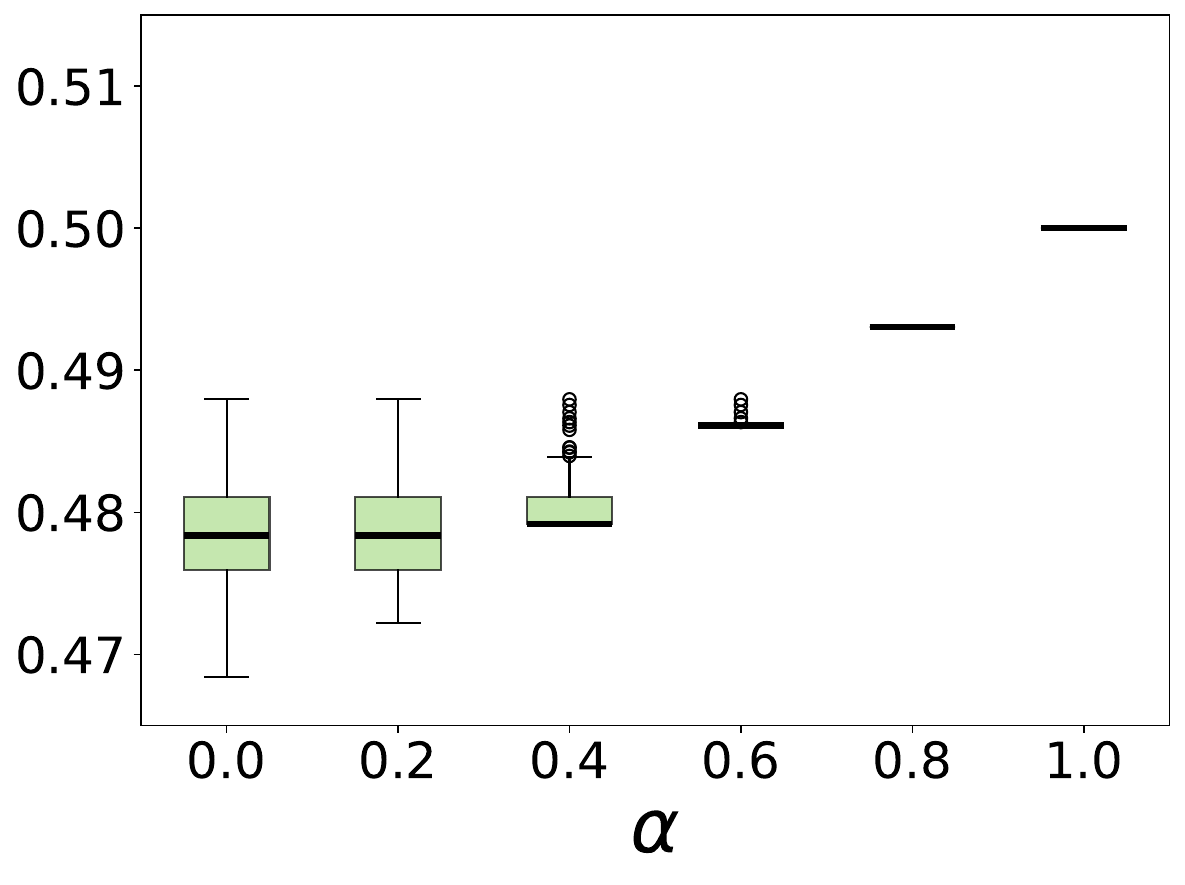}
    \caption{Ex Ante Demand (EFO)}
\end{subfigure}
\begin{subfigure}{0.235\textwidth}
    \centering
    \includegraphics[width=\textwidth]{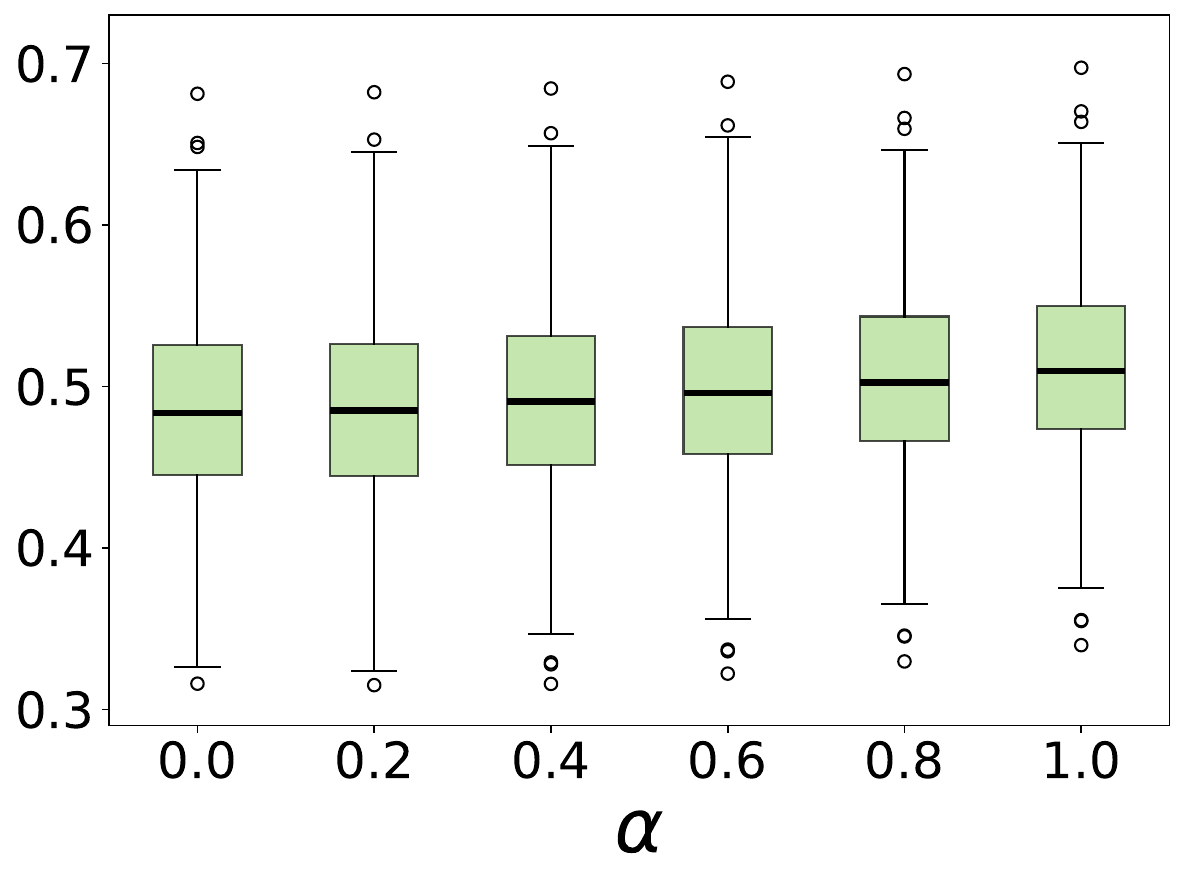}
    \caption{Ex Post Demand (FEO)}
\end{subfigure}
\begin{subfigure}{0.235\textwidth}
    \centering
    \includegraphics[width=\textwidth]{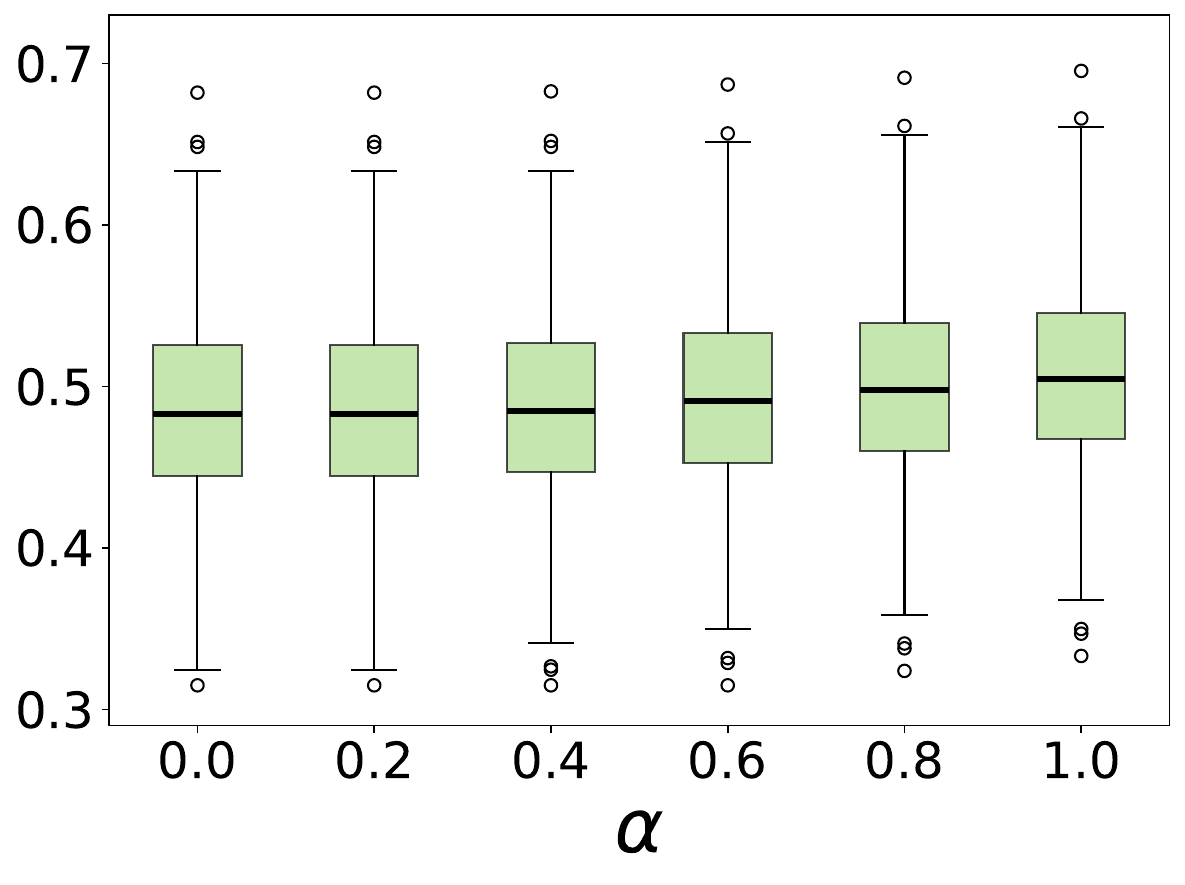}
    \caption{Ex Post Demand (EFO)}
\end{subfigure}
\end{minipage}
}
{Demand distributions for FEO and EFO under Rawlsian demand fairness \label{fig:linear_Rawlsian_demand_case}\vspace{0.5em}}
{}
\end{figure}

\begin{table}[htbp]
    \centering
    \caption{Normalized performance measures for FEO and EFO under Rawlsian demand fairness (\%)}
    \label{tab:measure_rawlsian_demand}
    \renewcommand{\arraystretch}{1.1}
    \setlength{\tabcolsep}{4pt} % Adjust column spacing for a perfect fit
    
    \begin{tabular}{l ccccc p{0.3cm} ccccc}
    \toprule
    & \multicolumn{5}{c}{FEO ($\alpha$)} & & \multicolumn{5}{c}{EFO ($\alpha$)} \\
    \cmidrule(lr){2-6} \cmidrule(lr){8-12}
    Measure & 0.2 & 0.4 & 0.6 & 0.8 & 1.0 & & 0.2 & 0.4 & 0.6 & 0.8 & 1.0 \\
    \midrule
    
    % Revenue Row
    $\mathcal{R}(\alpha)/\mathcal{R}(0)$ 
    & 99.97 & 99.94 & 99.85 & 99.73 & 99.57
    & 
    & 100.00 & 99.99 & 99.92 & 99.82 & 99.69 \\

    % Surplus Row
    $\mathcal{S}(\alpha)/\mathcal{S}(0)$ 
    & 100.09 & 102.13 & 104.60 & 107.22 & 109.95
    & 
    & 100.02 & 100.58 & 102.71 & 105.38 & 108.10 \\

    % Welfare Row
    $\mathcal{W}(\alpha)/\mathcal{W}(0)$ 
    & 100.01 & 100.72 & 101.54 & 102.39 & 103.26
    & 
    & 100.01 & 100.20 & 100.91 & 101.80 & 102.68 \\

    \bottomrule
    \end{tabular}
\end{table}

% \textcolor{red}{figure: demand gap decreases, adding ex-ante demand; Table: S,W (FEO) > S, W (EFO). change $\bar{d}$ with $0.5$}

% \textcolor{red}{how to deal with other cases?}

\subsection{Linear Demand Estimation with Logistic True Demand} %H
\label{appendix:linear_misspecified}

This case considers a linear demand setting while assuming that the true underlying demand model follows a logistic form. Under parity-wise loss fairness, Table~\ref{tab:measure_loss_fairness_1} indicates that firm profit increases, whereas both consumer surplus and social welfare decrease. These results suggest that, in this setting, the imposition of loss fairness may lead to outcomes that are detrimental to consumers, thereby deviating from the normative objective typically associated with fairness considerations. Moreover, comparison with Table~\ref{tab:measure_loss_fairness} highlights that the effects of loss fairness on performance measures are setting-dependent and cannot be characterized uniformly across different model specifications.

\begin{table}[htbp]
    \centering
    \caption{Normalized performance measures across $\alpha$ under parity-wise loss fairness (\%)}
    \label{tab:measure_loss_fairness_1}
    \renewcommand{\arraystretch}{1.1} 

    \begin{tabular}{lccccc}
    \toprule
    & \multicolumn{5}{c}{$\alpha$} \\
    \cmidrule(lr){2-6}
    Measure & 0.2 & 0.4 & 0.6 & 0.8 & 1.0 \\
    \midrule
    $\hat{\mathcal{R}}(\alpha)/\hat{\mathcal{R}}(0)$ & 100.65 & 100.99 & 101.24 & 101.45 & 101.65 \\
    $\hat{\mathcal{S}}(\alpha)/\hat{\mathcal{S}}(0)$ & 99.16  & 98.69  & 98.35  & 98.07  & 97.76 \\
    $\hat{\mathcal{W}}(\alpha)/\hat{\mathcal{W}}(0)$ & 99.78  & 99.64  & 99.55  & 99.47  & 99.37 \\
    \bottomrule
    \end{tabular}
\end{table}

Figure~\ref{fig:linear_parity_price_case_mis} and Table~\ref{tab:measure_parity_price_1} report the results under parity-wise price fairness. Although the true demand model is logistic, the qualitative insights remain consistent with those obtained under linear demand. In particular, FEO achieves higher consumer surplus and social welfare than EFO. Furthermore, under perfect fairness (i.e., $\alpha = 1$), accepting an additional $3.08\%$ reduction in profit under FEO results in a $3.97\%$ greater increase in consumer surplus relative to EFO. These findings suggest that, under parity-wise price fairness, FEO provides more favorable outcomes for consumers.

\begin{figure}[htbp]
\FIGURE{
\begin{minipage}{\textwidth}
\centering
\captionsetup{justification=centering}
\begin{subfigure}{0.35\textwidth}
    \centering
    \includegraphics[width=\textwidth]{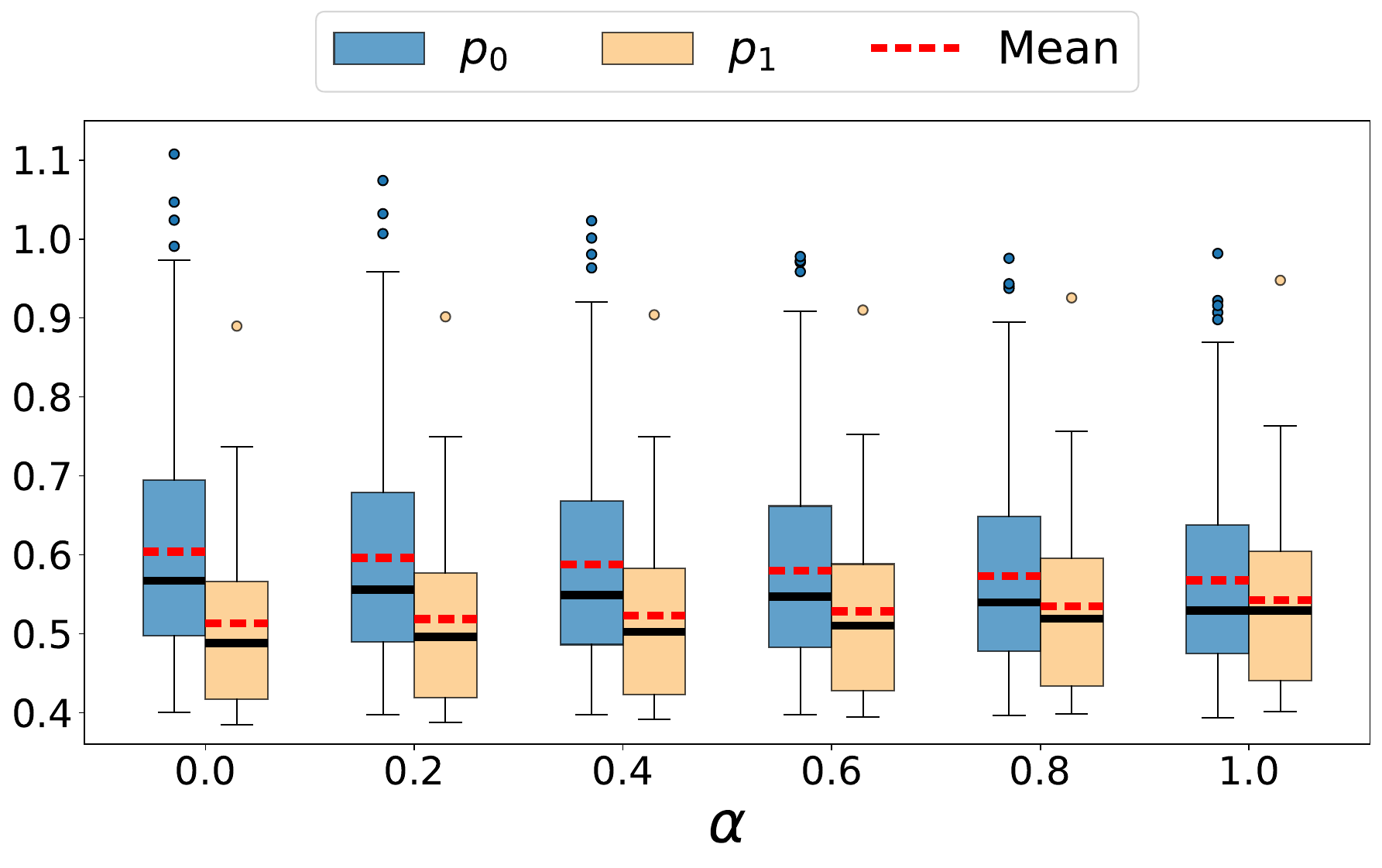}
    \caption{Price (FEO)}
\end{subfigure}
\begin{subfigure}{0.35\textwidth}
    \centering
    \includegraphics[width=\textwidth]{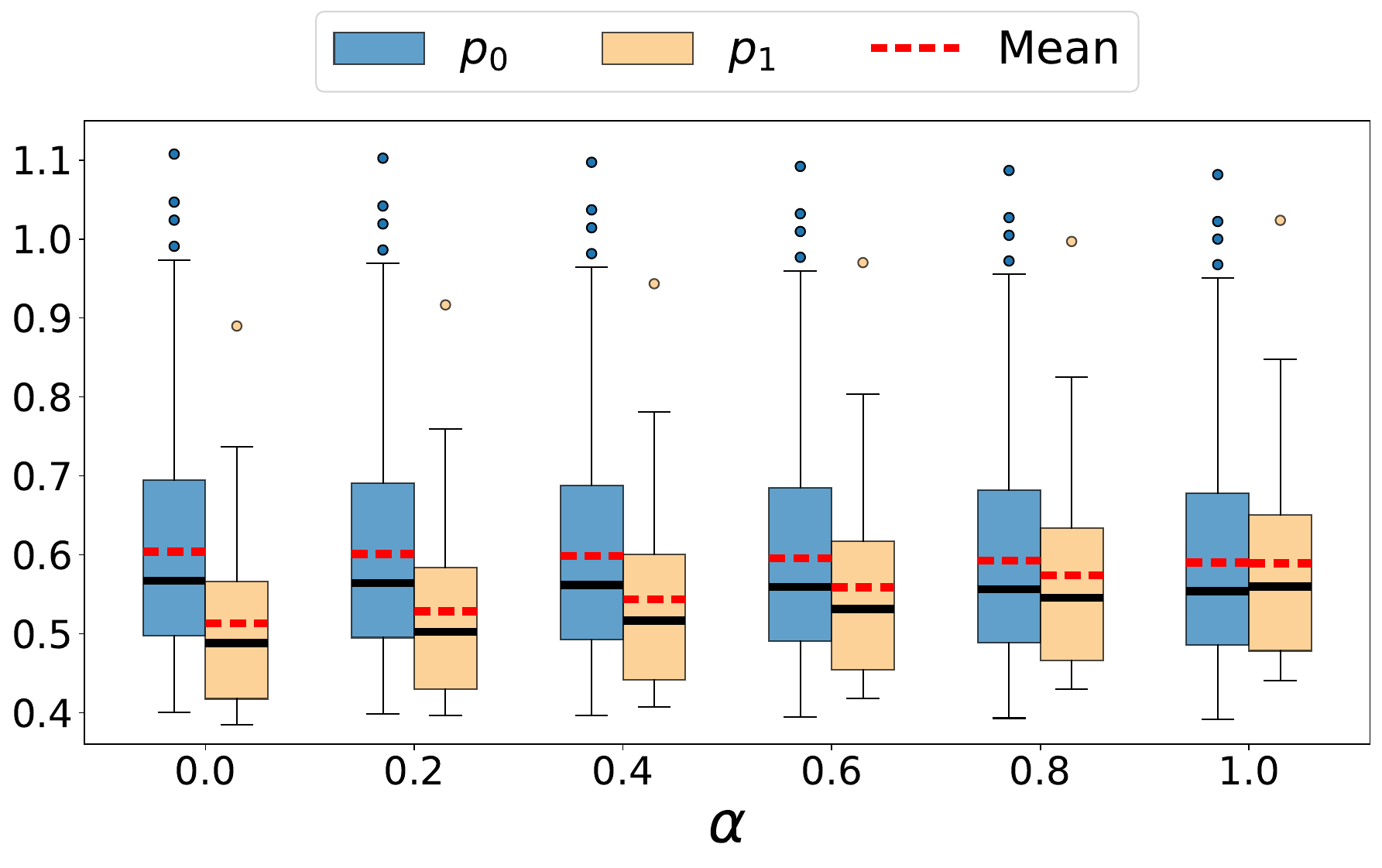}
    \caption{Price (EFO)}
\end{subfigure}
\end{minipage}
}
{Price distributions for FEO and EFO under parity-wise price fairness \label{fig:linear_parity_price_case_mis}\vspace{1mm}}
{}
\end{figure}

\begin{table}[htbp]
    \centering
    \caption{Normalized performance measures for FEO and EFO under parity-wise price fairness (\%)}
    \label{tab:measure_parity_price_1}
    \renewcommand{\arraystretch}{1.1}
    \setlength{\tabcolsep}{4pt}
    
    \begin{tabular}{l ccccc c ccccc}
        \toprule
        & \multicolumn{5}{c}{FEO ($\alpha$)} & & \multicolumn{5}{c}{EFO ($\alpha$)} \\
        \cmidrule(lr){2-6} \cmidrule(lr){8-12}
        Measure & 0.2 & 0.4 & 0.6 & 0.8 & 1.0 & & 0.2 & 0.4 & 0.6 & 0.8 & 1.0 \\
        \midrule
        
        $\mathcal{R}(\alpha)/\mathcal{R}(0)$  
        & 99.25 & 98.46 & 97.69 & 97.00 & 96.42 &
        & 99.92 & 99.83 & 99.73 & 99.62 & 99.50 \\ 
        
        $\mathcal{S}(\alpha)/\mathcal{S}(0)$  
        & 100.98 & 102.08 & 103.07 & 103.93 & 104.57 &
        & 100.11 & 100.23 & 100.35 & 100.47 & 100.60 \\ 
        
        $\mathcal{W}(\alpha)/\mathcal{W}(0)$  
        & 100.27 & 100.58 & 100.85 & 101.06 & 101.20 &
        & 100.03 & 100.07 & 100.09 & 100.12 & 100.15 \\ 
        
        \bottomrule
    \end{tabular}
\end{table}

Figure~\ref{fig:linear_parity_demand_case_mis} illustrates the demand distributions under parity-wise demand fairness. Here, only the ex-ante demand from EFO satisfies demand fairness, while the others do not. This is mainly because the true linear demand function differs from the estimated demand model. Table~\ref{tab:measure_parity_demand_1} summarizes the outcomes under parity-wise demand fairness.
%Compared to the true demand model case, FEO under perfect fairness does not achive perfect demand fairness this is because (. 
Consistent with the case in which the true demand model is linear, the changes in profit, consumer surplus, and social welfare are relatively small. An interesting observation, however, is that both FEO and EFO exhibit slight increases in profit. This arises from the mismatch between the true and estimated demand functions, where the fairness constraint acts as a form of regularization.

\begin{figure}[htbp]
\FIGURE{
\begin{minipage}{\textwidth}
\centering
\captionsetup{justification=centering}
\begin{subfigure}{0.35\textwidth}
    \centering
    \includegraphics[width=\textwidth]{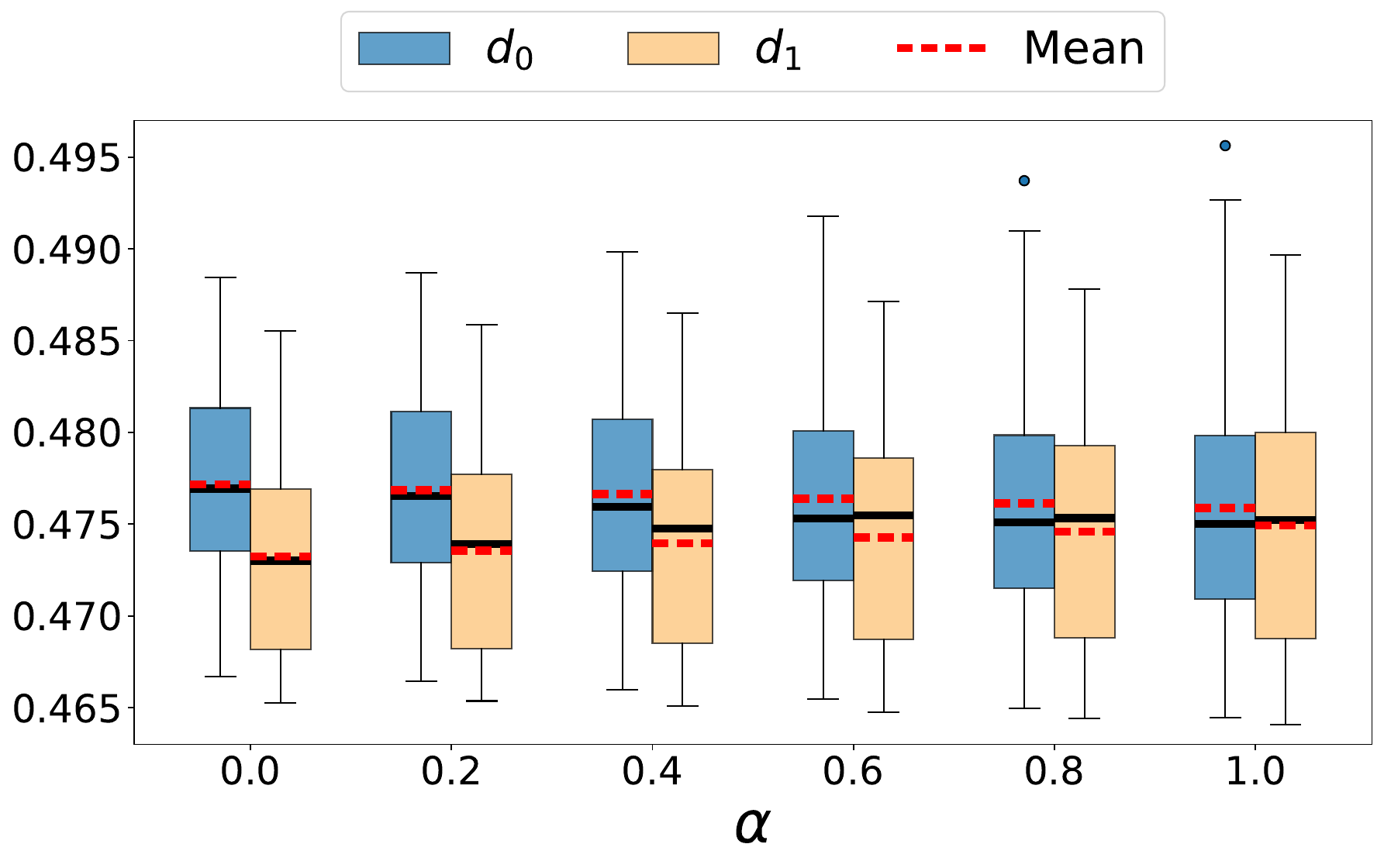}
    \caption{Ex Ante Demand (FEO)}
\end{subfigure}
\begin{subfigure}{0.35\textwidth}
    \centering
    \includegraphics[width=\textwidth]{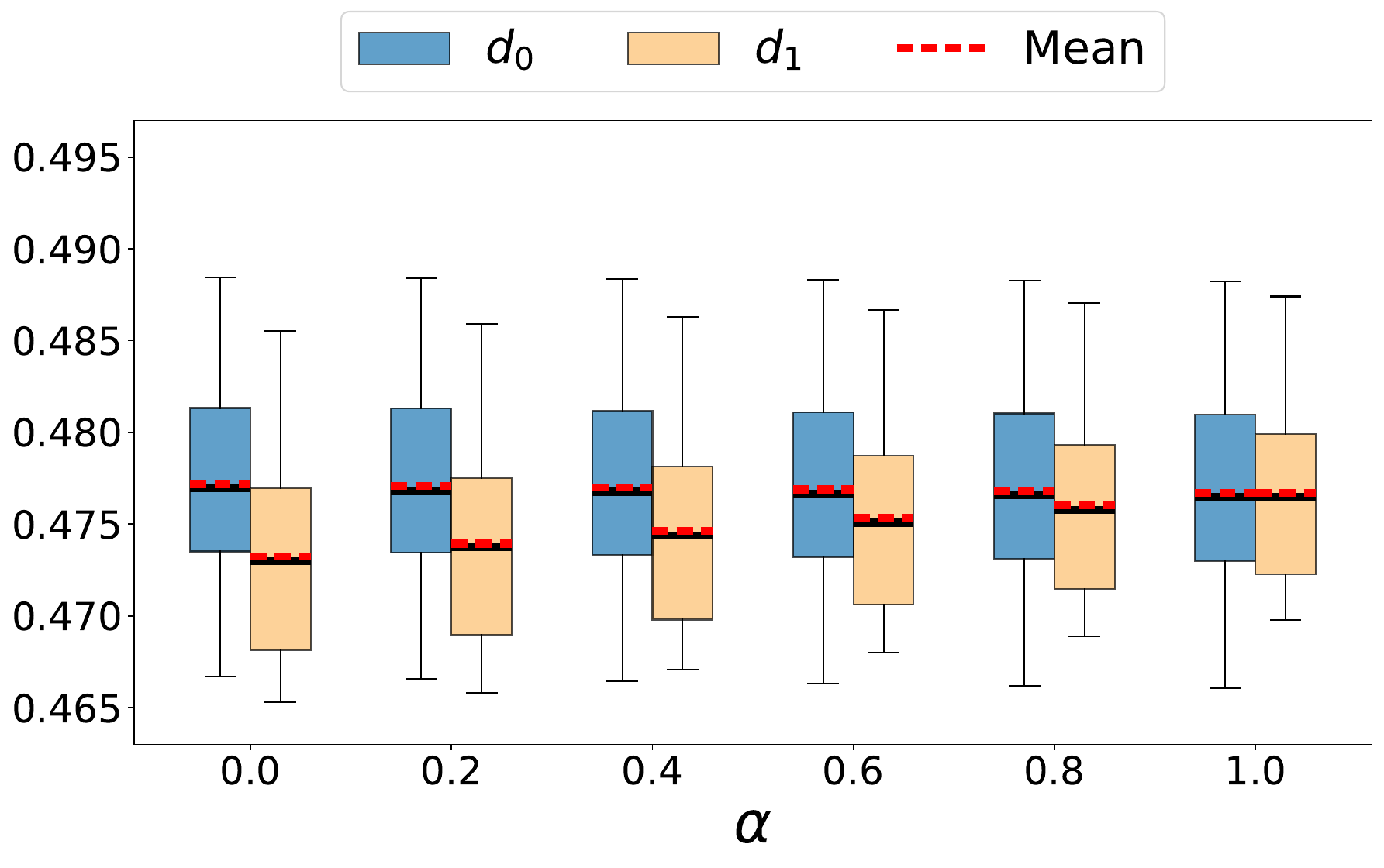}
    \caption{Ex Ante Demand (EFO)}
\end{subfigure}
\begin{subfigure}{0.35\textwidth}
    \centering
    \includegraphics[width=\textwidth]{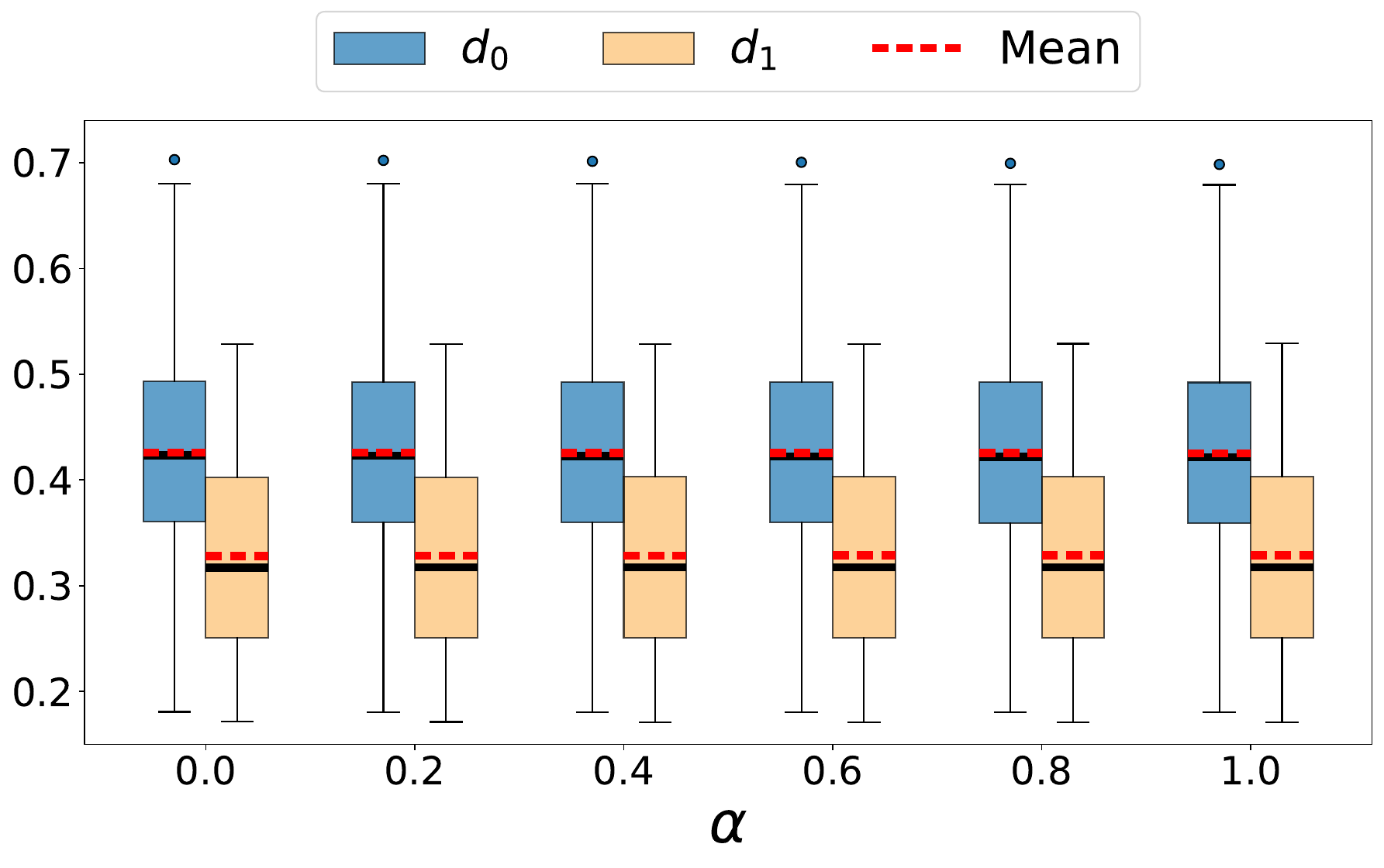}
    \caption{Ex Post Demand (FEO)}
\end{subfigure}
\begin{subfigure}{0.35\textwidth}
    \centering
    \includegraphics[width=\textwidth]{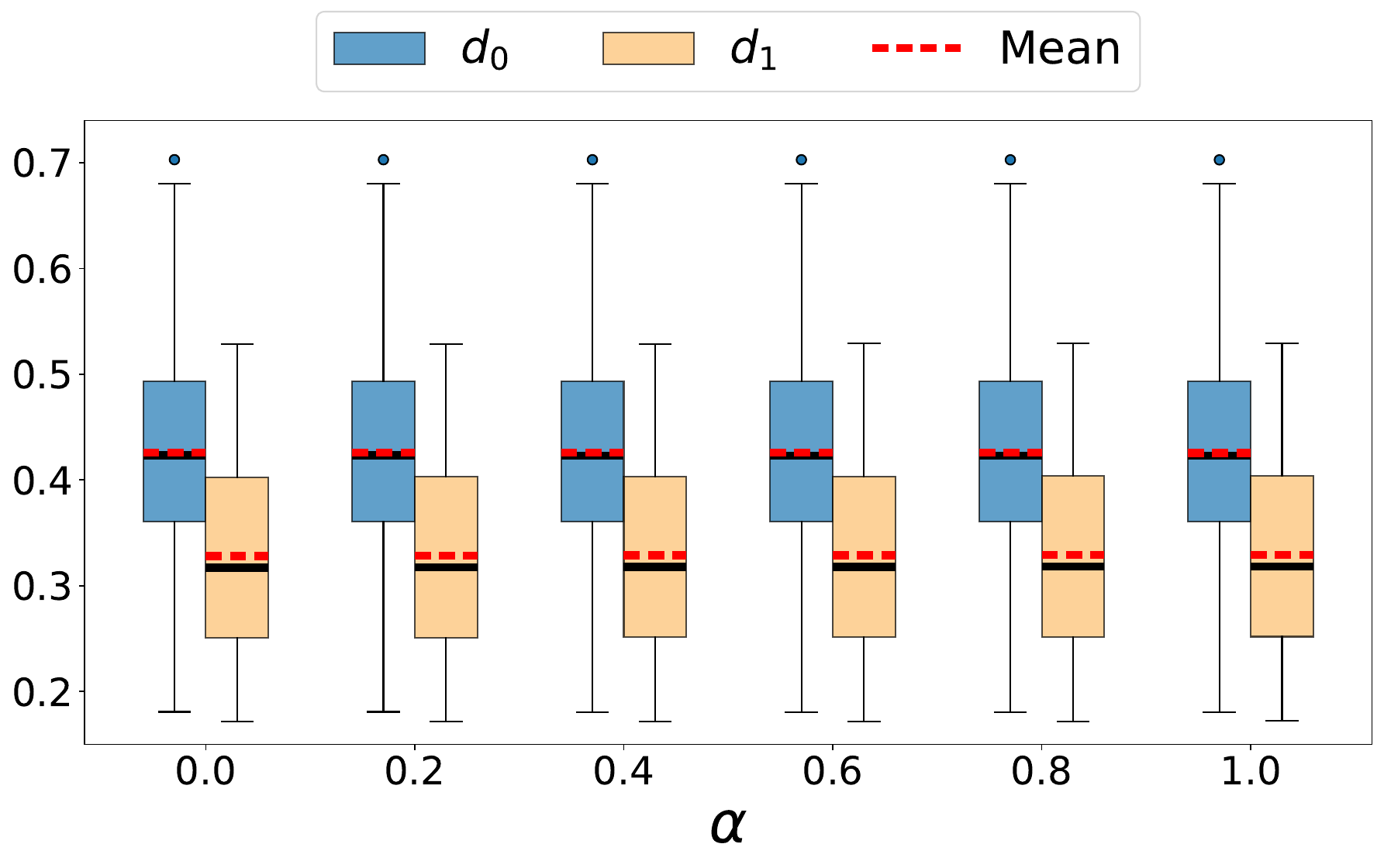}
    \caption{Ex Post Demand (EFO)}
\end{subfigure}
\end{minipage}
}
{Demand distributions for FEO and EFO under parity-wise demand fairness \label{fig:linear_parity_demand_case_mis}\vspace{2mm}}
{}
\end{figure}

\begin{table}[htbp]
    \centering
    \caption{Normalized performance measures for FEO and EFO under parity-wise demand fairness (\%)}
    \label{tab:measure_parity_demand_1}
    \renewcommand{\arraystretch}{1.1}
    
    \begin{tabular}{l ccccc c ccccc}
        \toprule
        & \multicolumn{5}{c}{FEO ($\alpha$)} & & \multicolumn{5}{c}{EFO ($\alpha$)} \\
        \cmidrule(lr){2-6} \cmidrule(lr){8-12}
        Measure & 0.2 & 0.4 & 0.6 & 0.8 & 1.0 & & 0.2 & 0.4 & 0.6 & 0.8 & 1.0 \\
        \midrule
        
        $\mathcal{R}(\alpha)/\mathcal{R}(0)$  
        & 100.03 & 100.04 & 100.07 & 100.09 & 100.11 &
        & 100.00 & 100.00 & 100.01 & 100.01 & 100.01 \\ 
        
        $\mathcal{S}(\alpha)/\mathcal{S}(0)$  
        & 99.97 & 99.95 & 99.92 & 99.89 & 99.87 &
        & 100.00 & 99.99 & 99.99 & 99.99 & 99.99 \\ 
        
        $\mathcal{W}(\alpha)/\mathcal{W}(0)$  
        & 99.99 & 99.99 & 99.98 & 99.98 & 99.97 &
        & 100.00 & 100.00 & 100.00 & 100.00 & 100.00 \\ 
        \bottomrule
    \end{tabular}
\end{table}

Figure~\ref{fig:linear_Rawlsian_price_mis} shows the price distributions under Rawlsian price fairness. Similar to the other cases in Section~\ref{appendix:well_specified_logistic} and Section~\ref{appendix:linear_well_specified}, FEO reduces the spread of the price distribution. Table~\ref{tab:measure_rawlsian_price_1} shows that, consistent with the case of a true linear demand model, FEO achieves higher consumer surplus and social welfare. An interesting observation is that, as in the setting with a linear estimated demand function, the differences in profit, consumer surplus, and social welfare are most pronounced at intermediate values of $\alpha$, rather than at extreme cases such as $\alpha = 1$. This is because, when $\alpha = 1$, prices become nearly static, thereby limiting variation in outcomes across framework.

\begin{figure}[htbp]
\FIGURE{
\begin{minipage}{\textwidth}
\centering
\captionsetup{justification=centering}
\begin{subfigure}{0.35\textwidth}
    \centering
    \includegraphics[width=\textwidth]{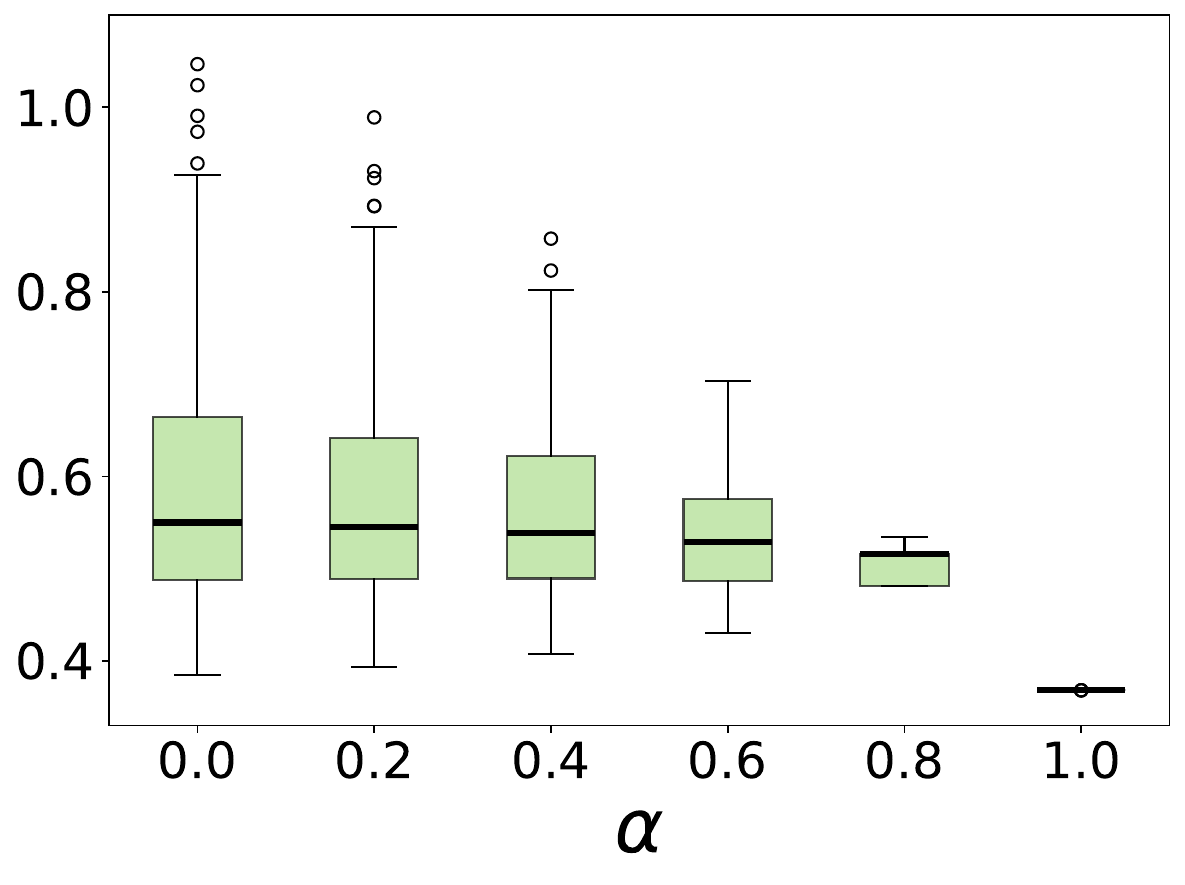}
    \caption{Price under FEO}
    % \label{fig:rawls_price_feo}
\end{subfigure}
\begin{subfigure}{0.35\textwidth}
    \centering
    \includegraphics[width=\textwidth]{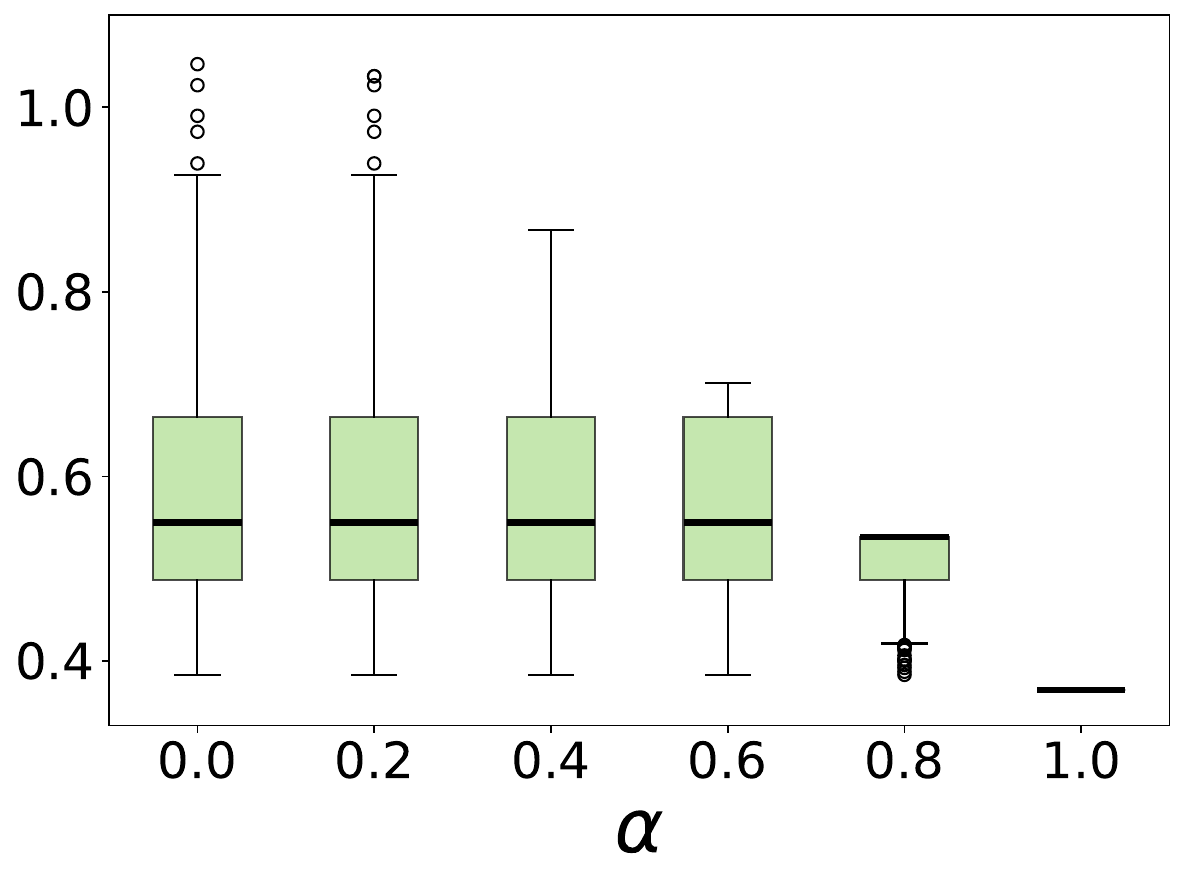}
    \caption{Price under EFO}
    % \label{fig:rawls_price_efo}
\end{subfigure}
\end{minipage}
}
{Price distributions for FEO and EFO under Rawlsian price fairness \label{fig:linear_Rawlsian_price_mis}\vspace{0.5em}}
{}
\end{figure}

\begin{table}[htbp]
    \centering
    \caption{Normalized performance measures for FEO and EFO under Rawlsian price fairness (\%)}
    \label{tab:measure_rawlsian_price_1}
    \renewcommand{\arraystretch}{1.1}
    
    \begin{tabular}{l cccccc c cccccc}
        \toprule
        & \multicolumn{5}{c}{FEO ($\alpha$)} & & \multicolumn{5}{c}{EFO ($\alpha$)} \\
        \cmidrule(lr){2-6} \cmidrule(lr){8-12}
        Measure & 0.2 & 0.4 & 0.6 & 0.8 & 1.0 & & 0.2 & 0.4 & 0.6 & 0.8 & 1.0 \\
        \midrule
        
        $\mathcal{R}(\alpha)/\mathcal{R}(0)$  
        & 98.59 & 96.54 & 93.30 & 87.60 & 67.65 &
        & 99.98 & 99.54 & 97.02 & 88.53 & 67.65 \\ 
        
        $\mathcal{S}(\alpha)/\mathcal{S}(0)$  
        & 102.22 & 105.05 & 108.96 & 115.11 & 135.26 &
        & 100.06 & 100.86 & 104.42 & 114.25 & 135.26 \\ 
        
        $\mathcal{W}(\alpha)/\mathcal{W}(0)$  
        & 100.72 & 101.53 & 102.48 & 103.73 & 107.30 &
        & 100.03 & 100.31 & 101.36 & 103.61 & 107.30 \\ 
        
        \bottomrule
    \end{tabular}
\end{table}

Lastly, Figure~\ref{fig:linear_Rawlsian_demand_mis} illustrates how ex-ante and ex-post demand change under FEO and EFO. Similar to the case in Section~\ref{appendix:linear_well_specified}, the ex post demand changes less under EFO than under FEO. Table~\ref{tab:measure_rawlsian_demand_1} shows the effects of Rawlsian demand fairness on profit, consumer surplus, and social welfare. Consistent with the case of a true linear demand model, FEO and EFO yield broadly similar results. When comparing the two, FEO entails a slightly larger reduction in profit, but achieves higher consumer surplus and social welfare. This pattern is also generally observed under Rawlsian price fairness and in the true linear demand setting. Overall, these results suggest that Rawlsian fairness with feature-based implementation favors consumers when using FEO, whereas EFO may be preferred when the objective is to mitigate profit losses.

\begin{figure}[htbp]
\FIGURE{
\begin{minipage}{\textwidth}
\centering
\captionsetup{justification=centering}
\begin{subfigure}{0.235\textwidth}
    \centering
    \includegraphics[width=\textwidth]{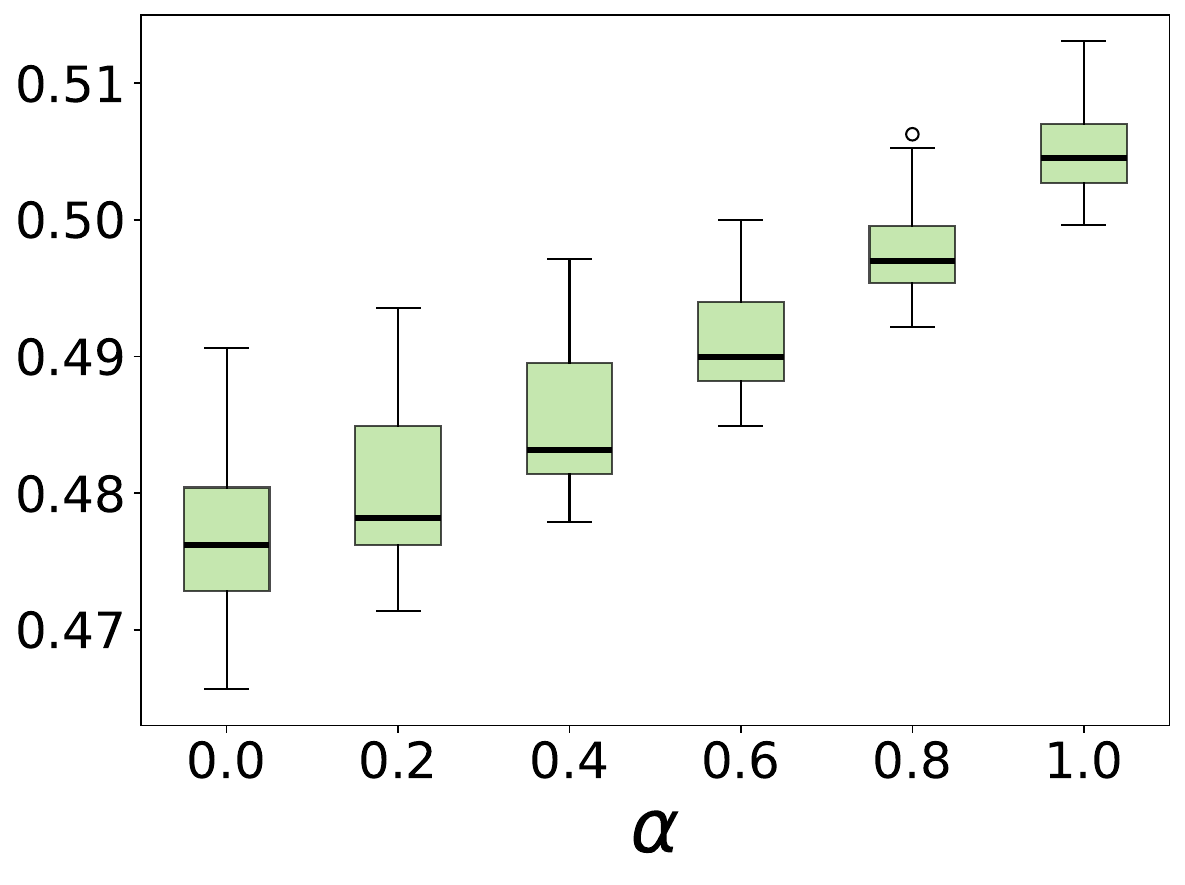}
    \caption{Ex Ante Demand (FEO)}
\end{subfigure}
\begin{subfigure}{0.235\textwidth}
    \centering
    \includegraphics[width=\textwidth]{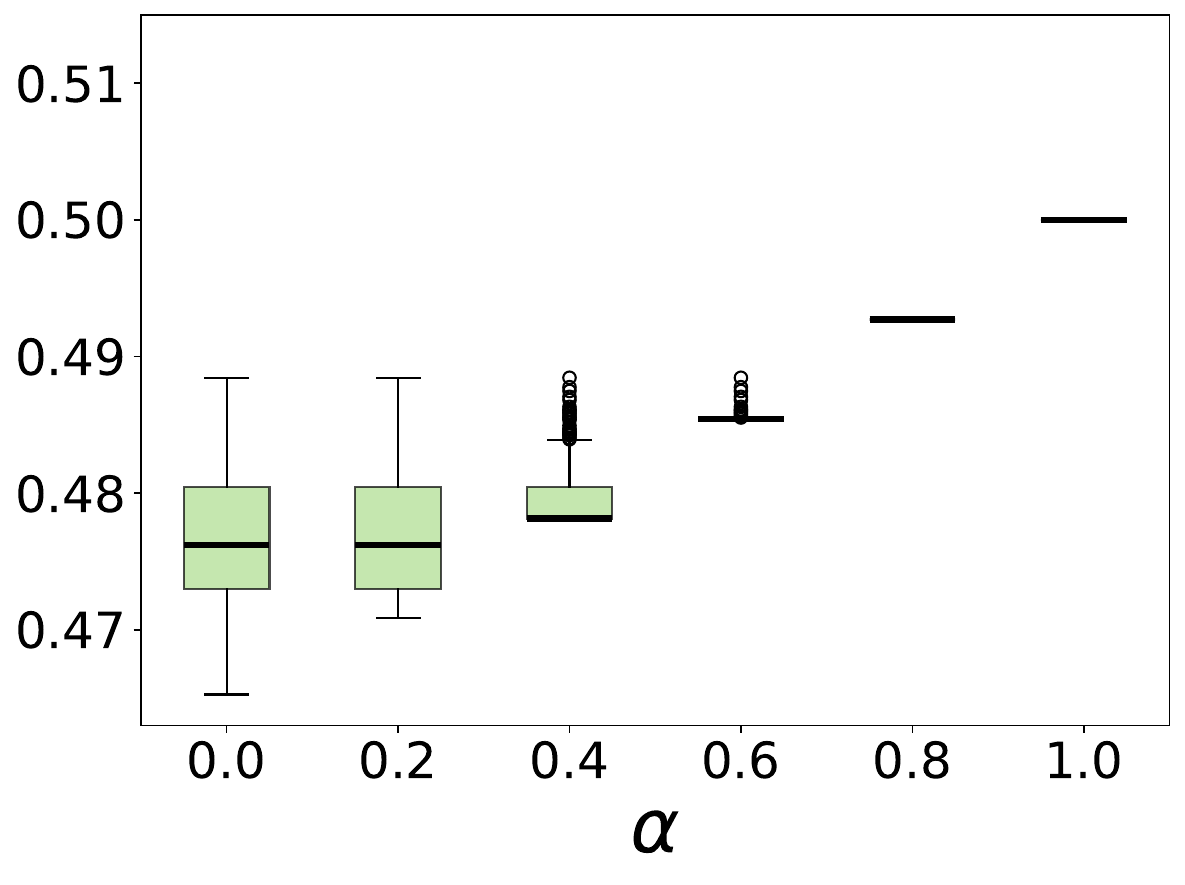}
    \caption{Ex Ante Demand (EFO)}
\end{subfigure}
\begin{subfigure}{0.235\textwidth}
    \centering
    \includegraphics[width=\textwidth]{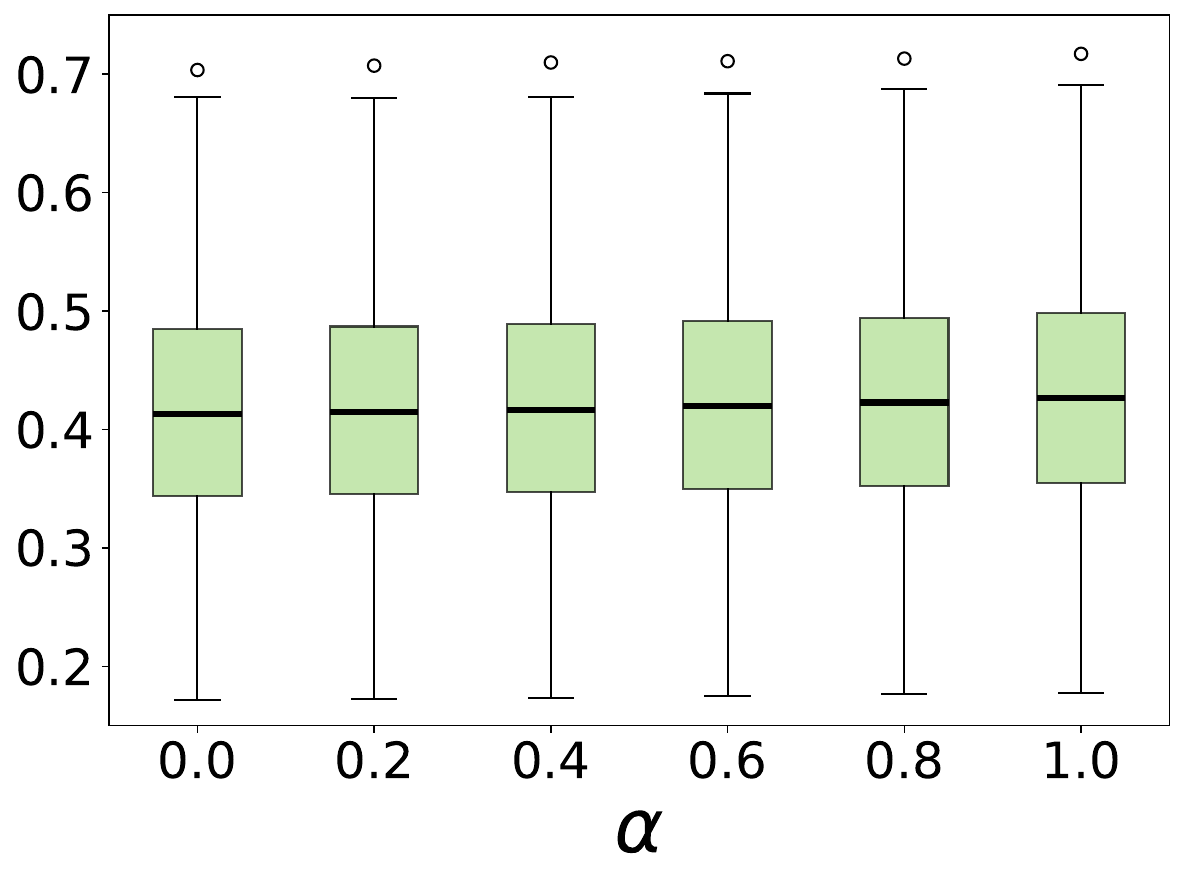}
    \caption{Ex Post Demand (FEO)}
\end{subfigure}
\begin{subfigure}{0.235\textwidth}
    \centering
    \includegraphics[width=\textwidth]{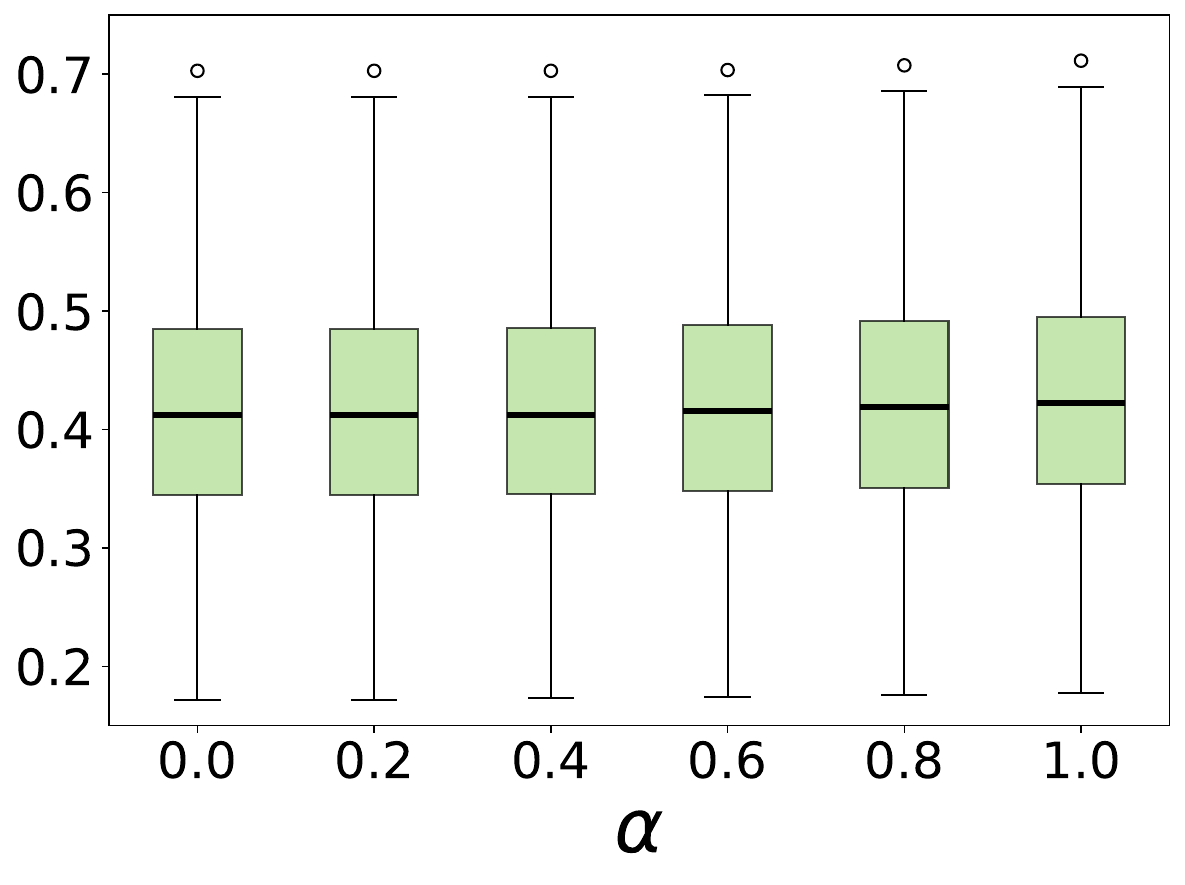}
    \caption{Ex Post Demand (EFO)}
\end{subfigure}
\end{minipage}
}
{Demand distributions for FEO and EFO under Rawlsian demand fairness \label{fig:linear_Rawlsian_demand_mis}\vspace{0.5em}}
{}
\end{figure}

\begin{table}[htbp]
    \centering
    \caption{Normalized performance measures for FEO and EFO under Rawlsian demand fairness (\%)}
    \label{tab:measure_rawlsian_demand_1}
    \renewcommand{\arraystretch}{1.1}
    
    \begin{tabular}{l cccccc c cccccc}
        \toprule
        & \multicolumn{5}{c}{FEO ($\alpha$)} & & \multicolumn{5}{c}{EFO ($\alpha$)} \\
        \cmidrule(lr){2-6} \cmidrule(lr){8-12}
        Measure & 0.2 & 0.4 & 0.6 & 0.8 & 1.0 & & 0.2 & 0.4 & 0.6 & 0.8 & 1.0 \\
        \midrule
        
        $\mathcal{R}(\alpha)/\mathcal{R}(0)$  
        & 99.66 & 99.17 & 98.63 & 97.93 & 97.14 &
        & 99.98 & 99.80 & 99.26 & 98.52 & 97.74 \\ 
        
        $\mathcal{S}(\alpha)/\mathcal{S}(0)$  
        & 100.50 & 101.19 & 101.94 & 102.90 & 104.00 &
        & 100.02 & 100.26 & 100.98 & 102.04 & 103.14 \\ 
        
        $\mathcal{W}(\alpha)/\mathcal{W}(0)$  
        & 100.15 & 100.35 & 100.57 & 100.84 & 101.16 &
        & 100.01 & 100.07 & 100.27 & 100.59 & 100.91 \\ 
        
        \bottomrule
    \end{tabular}
\end{table}

\subsection{Logistic Demand Estimation with Linear True Demand}
\label{appendix:mis_specified_logistic}
We consider a mis-specified setting where we assume the true demand function is linear but the decision maker uses a logistic demand model for estimation. In this case, both FEO and EFO suffer from the error in estimation.

Table~\ref{tab:measure_loss_fairness_logistic_mis} presents the results under parity-wise loss fairness. When increasing $\alpha$, profit first increases then decreases, whereas consumer surplus and social welfare keep decreasing, which is different from Table~\ref{tab:measure_loss_fairness_logistic},  Table~\ref{tab:measure_loss_fairness} and Table~\ref{tab:measure_loss_fairness_1}. Taken together, these findings indicate that parity-wise loss fairness does not necessarily lead to desirable downstream social outcomes and may not align with the fairness objectives of policymakers, sellers, and customers.

\begin{table}[htbp]
    \centering
    \caption{Normalized performance measures across $\alpha$ under parity-wise loss fairness (\%)}
    \label{tab:measure_loss_fairness_logistic_mis}
    \renewcommand{\arraystretch}{1.1} 

    \begin{tabular}{lccccc}
    \toprule
    & \multicolumn{5}{c}{$\alpha$} \\
    \cmidrule(lr){2-6}
    Measure & 0.2 & 0.4 & 0.6 & 0.8 & 1.0 \\
    \midrule
    $\mathcal{R}(\alpha)/\mathcal{R}(0)$ & 101.12 & 101.95 & 102.39 & 102.47 & 102.30 \\
    $\mathcal{S}(\alpha)/\mathcal{S}(0)$ & 95.98  & 92.35 &  89.73 &  88.30  & 88.04 \\
    $\mathcal{W}(\alpha)/\mathcal{W}(0)$ & 98.95 &  97.89 &  97.04 &  96.49  & 96.26 \\
    \bottomrule
    \end{tabular}
\end{table}

Figure~\ref{fig:logistic_parity_price_case_mis} presents the price distributions under parity-wise price fairness. In the setting where the true demand model is not included in the hypothesis class, parity-wise price fairness remains achievable because price is the directly controllable decision variable. However, due to model mis-specification, the resulting pricing decisions lead to worse downstream outcomes. As shown in Table~\ref{tab:measure_parity_price_logistic_mis}, profit, consumer surplus, and social welfare all decrease under EFO. This differs from the results in \cite{cohen2022price}, where consumer surplus and social welfare initially increase when the true demand model is known. 

\begin{figure}[htbp]
\FIGURE{
\begin{minipage}{\textwidth}
\centering
\captionsetup{justification=centering}
\begin{subfigure}{0.35\textwidth}
    \centering
    \includegraphics[width=\textwidth]{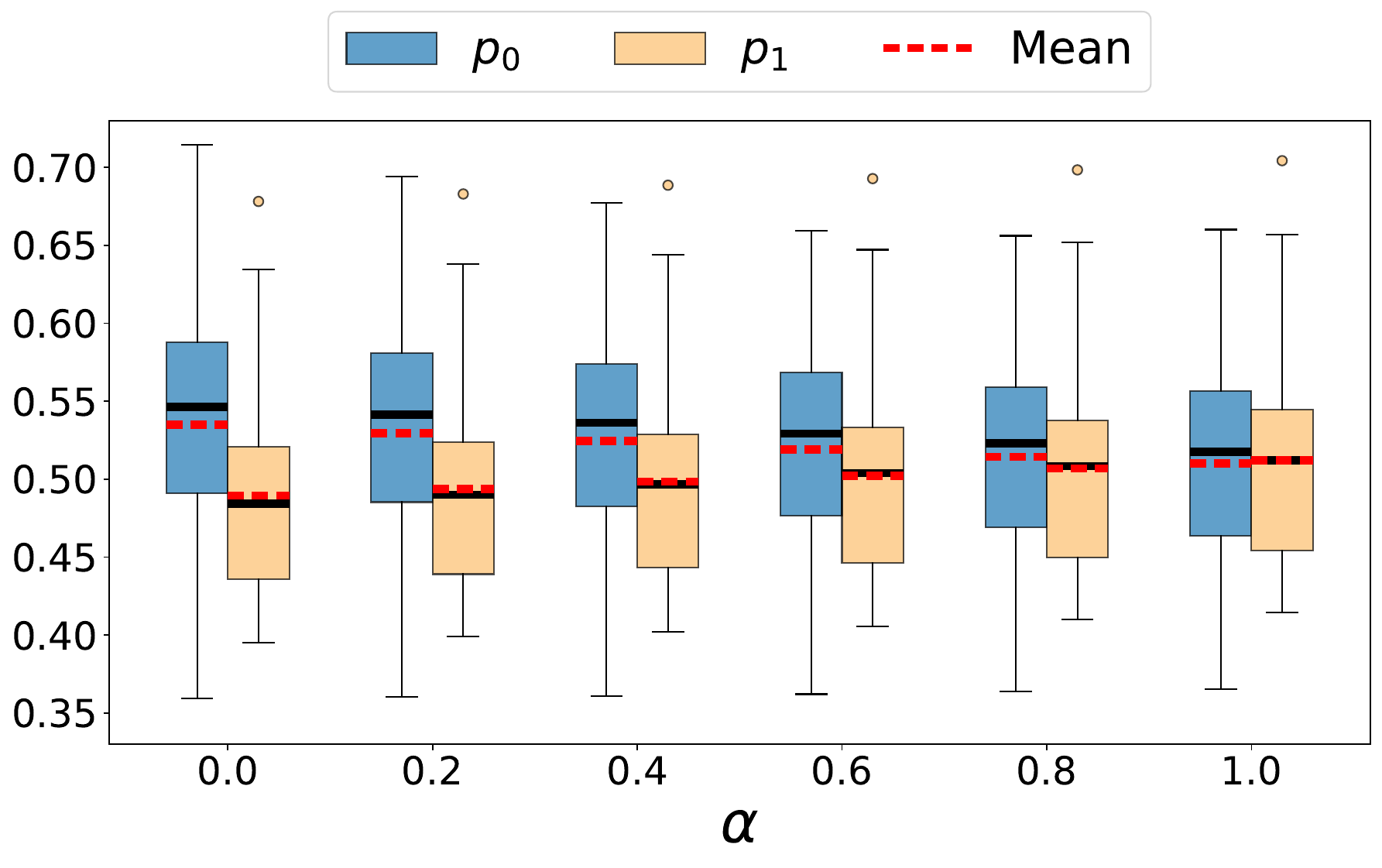}
    \caption{Price (FEO)}
\end{subfigure}
\begin{subfigure}{0.35\textwidth}
    \centering
    \includegraphics[width=\textwidth]{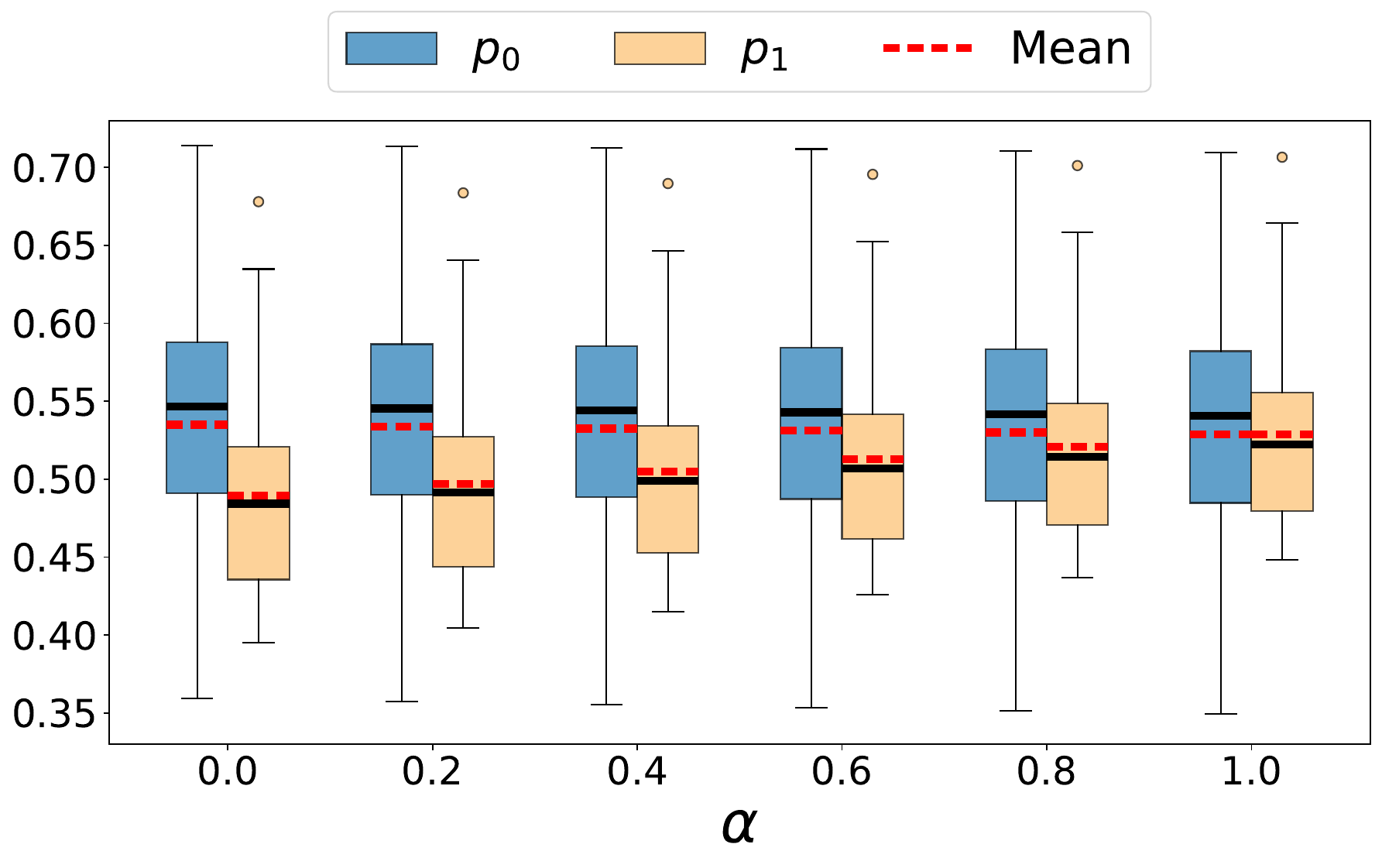}
    \caption{Price (EFO)}
\end{subfigure}
\end{minipage}
}
{Price distributions for FEO and EFO under parity-wise price fairness \label{fig:logistic_parity_price_case_mis}\vspace{1mm}}
{}
\end{figure}

\begin{table}[htbp]
    \centering
    \caption{Normalized performance measures for FEO and EFO under parity-wise price fairness (\%)}
    \label{tab:measure_parity_price_logistic_mis}
    \renewcommand{\arraystretch}{1.1}
    \setlength{\tabcolsep}{4pt}
    
    \begin{tabular}{l ccccc c ccccc}
        \toprule
        & \multicolumn{5}{c}{FEO ($\alpha$)} & & \multicolumn{5}{c}{EFO ($\alpha$)} \\
        \cmidrule(lr){2-6} \cmidrule(lr){8-12}
        Measure & 0.2 & 0.4 & 0.6 & 0.8 & 1.0 & & 0.2 & 0.4 & 0.6 & 0.8 & 1.0 \\
        \midrule
        % Revenue Row
        $\mathcal{R}(\alpha)/\mathcal{R}(0)$  
        & 99.72  & 99.44  & 99.07  & 98.76  & 98.43  & 
        & 99.99 & 99.98 & 99.95 & 99.93 & 99.89  \\ 
        
        % Surplus Row
        $\mathcal{S}(\alpha)/\mathcal{S}(0)$  
        & 101.18 & 102.21 & 103.49 & 104.49 & 105.41 & 
        & 99.98 & 99.97 & 99.95 & 99.92 & 99.89 \\ 
        
        % Welfare Row
        $\mathcal{W}(\alpha)/\mathcal{W}(0)$  
        & 100.34 & 100.61 & 100.94 & 101.18 & 101.38 & 
        & 99.99 & 99.97 & 99.95 & 99.93 & 99.89 \\ 
        \bottomrule
    \end{tabular}
\end{table}

Figure~\ref{fig:logistic_parity_demand_case_mis} illustrates the demand distributions under parity-wise demand fairness. The decision maker only has access to the estimated demand function, so they can only control the ex-ante demand gaps between two groups. However, the ex-post demand gaps even become larger for FEO and EFO when $\alpha=1.0$. Table~\ref{tab:measure_parity_demand_logistic_mis} shows that for FEO, profit increases whereas consumer surplus and social welfare decrease, and for EFO there is no significant change in this metrics. The increase in profit under FEO arises because the fairness constraints act as a regularization term, which is consistent with Proposition~\ref{prop4}.

\begin{figure}[htbp]
\FIGURE{
\begin{minipage}{\textwidth}
\centering
\captionsetup{justification=centering}
\begin{subfigure}{0.35\textwidth}
    \centering
    \includegraphics[width=\textwidth]{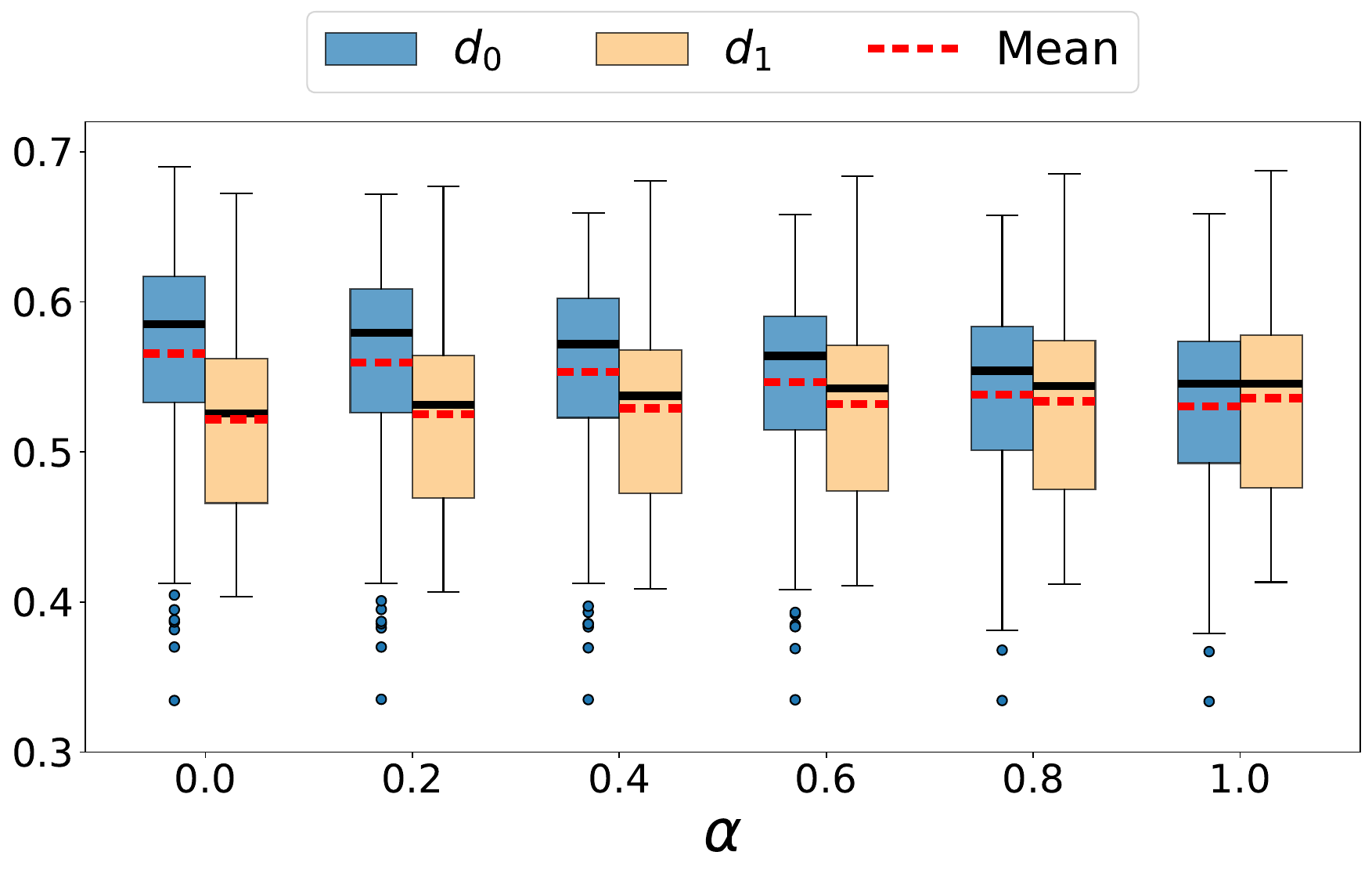}
    \caption{Ex Ante Demand (FEO)}
\end{subfigure}
\begin{subfigure}{0.35\textwidth}
    \centering
    \includegraphics[width=\textwidth]{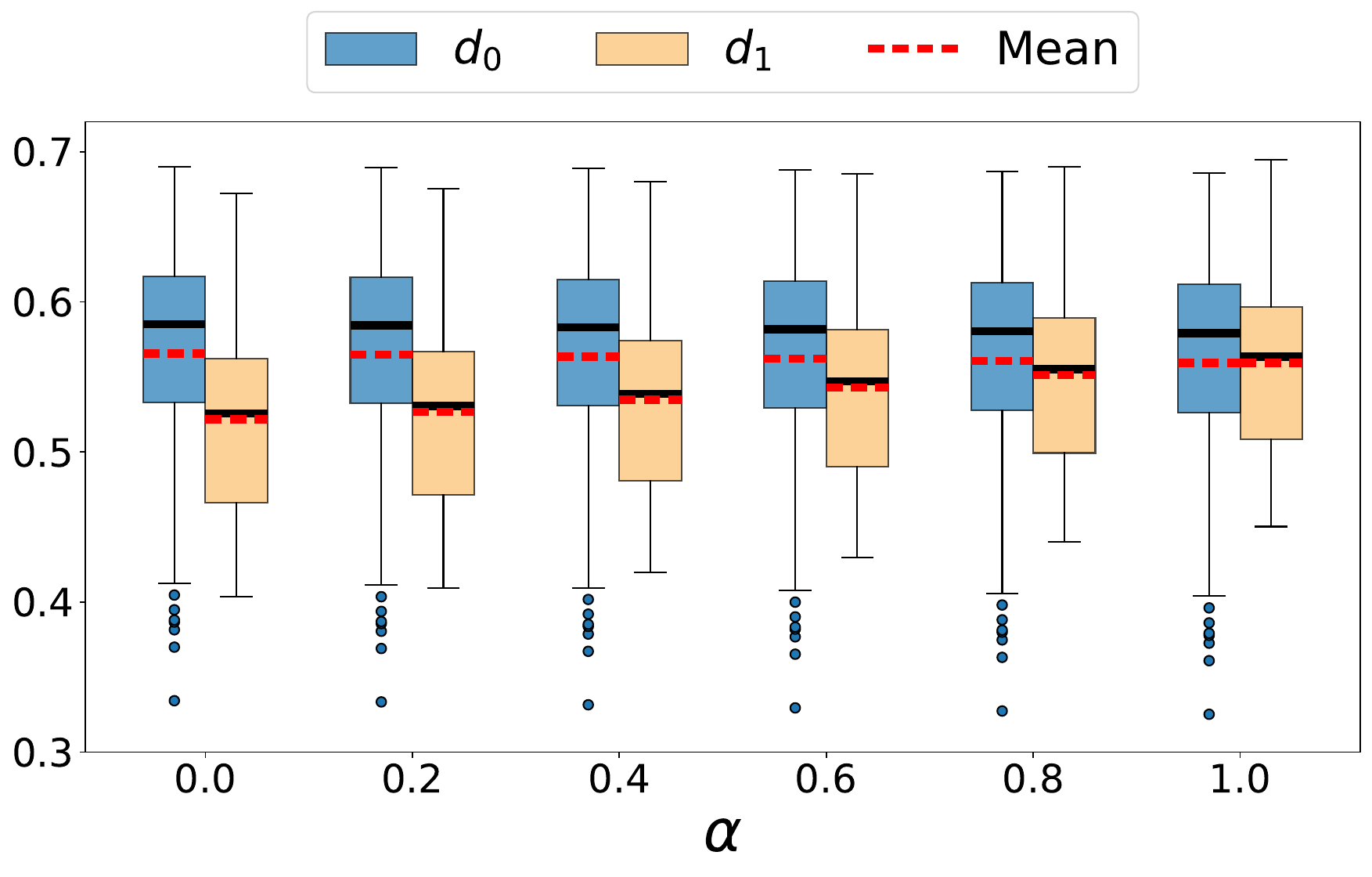}
    \caption{Ex Ante Demand (EFO)}
\end{subfigure}
\begin{subfigure}{0.35\textwidth}
    \centering
    \includegraphics[width=\textwidth]{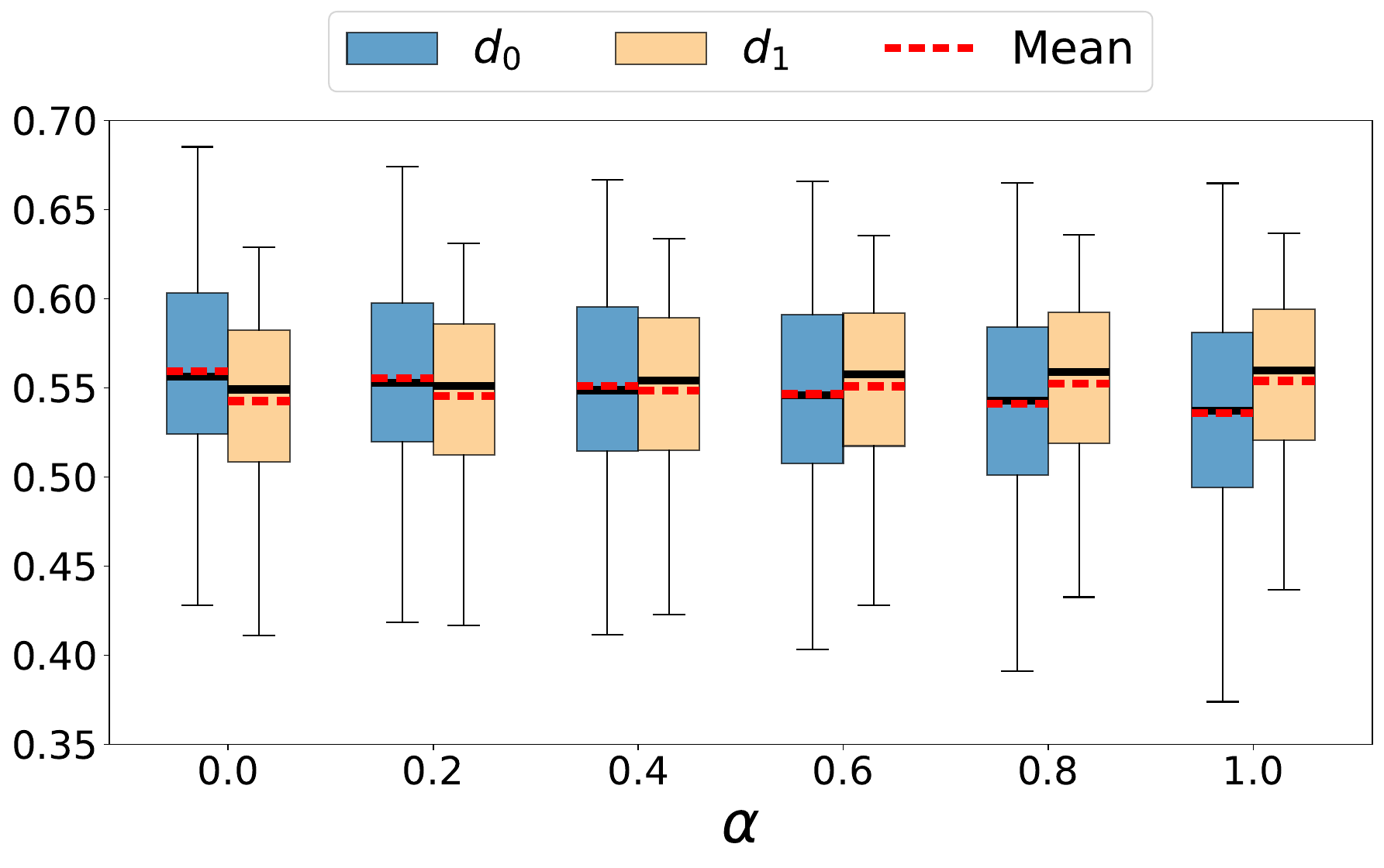}
    \caption{Ex Post Demand (FEO)}
\end{subfigure}
\begin{subfigure}{0.35\textwidth}
    \centering
    \includegraphics[width=\textwidth]{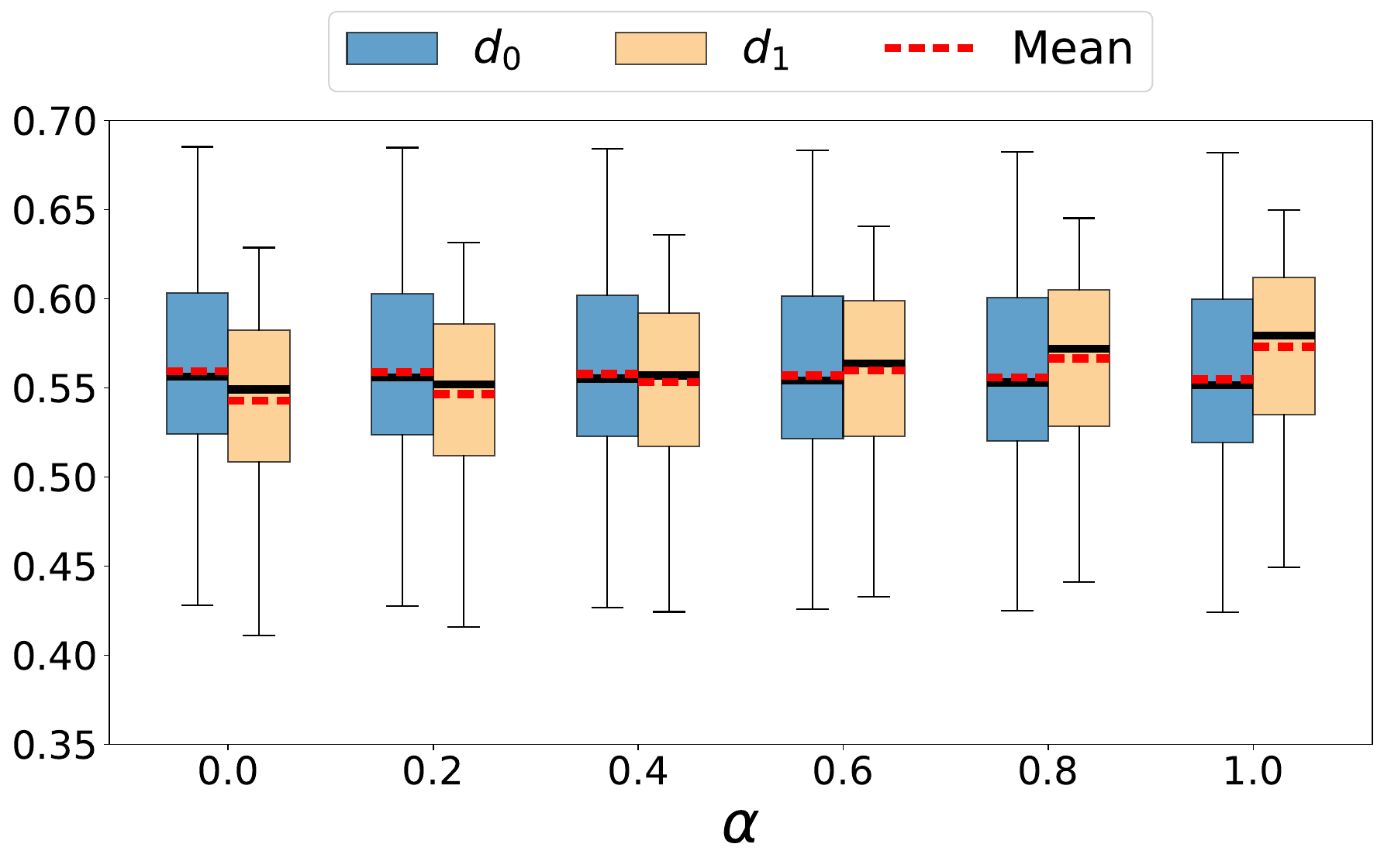}
    \caption{Ex Post Demand (EFO)}
\end{subfigure}
\end{minipage}
}
{Demand distributions for FEO and EFO under parity-wise demand fairness \label{fig:logistic_parity_demand_case_mis}\vspace{2mm}}
{}
\end{figure}

\begin{table}[htbp]
    \centering
    \caption{Normalized performance measures for FEO and EFO under parity-wise demand fairness (\%)}
    \label{tab:measure_parity_demand_logistic_mis}
    \renewcommand{\arraystretch}{1.1}
    
    \begin{tabular}{l ccccc c ccccc}
        \toprule
        & \multicolumn{5}{c}{FEO ($\alpha$)} & & \multicolumn{5}{c}{EFO ($\alpha$)} \\
        \cmidrule(lr){2-6} \cmidrule(lr){8-12}
        Measure & 0.2 & 0.4 & 0.6 & 0.8 & 1.0 & & 0.2 & 0.4 & 0.6 & 0.8 & 1.0 \\
        \midrule
        $\mathcal{R}(\alpha)/\mathcal{R}(0)$ & 100.25 & 100.48 & 100.69 & 100.95 & 101.15 & & 100.00 & 100.00  & 99.99 &  99.98 &  99.97 \\
        
        $\mathcal{S}(\alpha)/\mathcal{S}(0)$ & 98.86 &  97.69 &  96.45 &  94.89 &  93.46  & & 100.01 & 100.03 & 100.04 & 100.06 & 100.08  \\
        
        $\mathcal{W}(\alpha)/\mathcal{W}(0)$ & 99.66 &  99.30 &  98.90 &  98.39  & 97.90  & & 100.01 & 100.01 & 100.02 & 100.02 & 100.01 \\
        \bottomrule
    \end{tabular}
\end{table}

Figure~\ref{fig:logistic_Rawlsian_price_case_mis} and Table~\ref{tab:measure_rawlsian_price_logistic_mis} present the results for Rawlsian price fairness. The change in price distributions and downstream outcomes are similar to well-specified cases. This suggests that Rawlsian price fairness can remain effective even when the estimated demand model is mis-specified.

\begin{figure}[htbp]
\FIGURE{
\begin{minipage}{\textwidth}
\centering
\captionsetup{justification=centering}
\begin{subfigure}{0.35\textwidth}
    \centering
    \includegraphics[width=\textwidth]{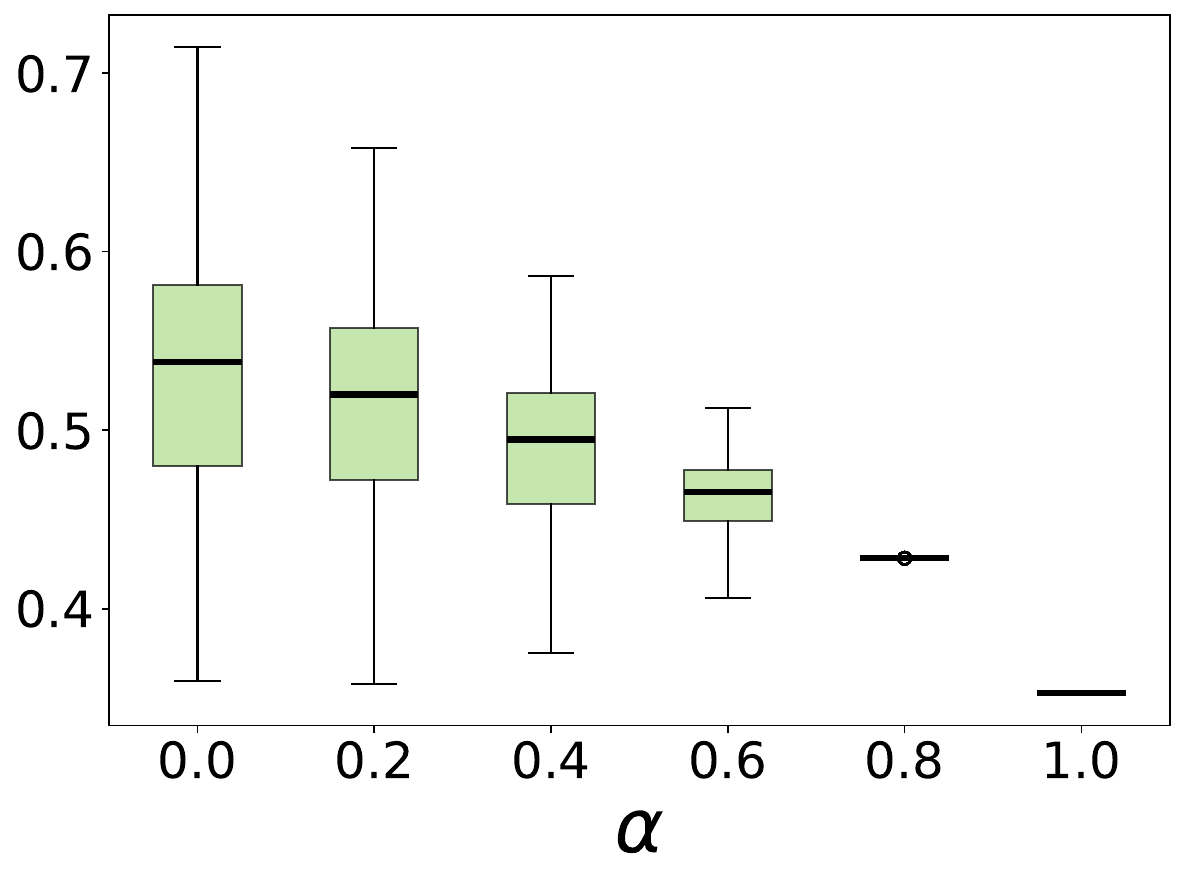}
    \caption{Price under FEO}
    % \label{fig:rawls_price_feo}
\end{subfigure}
\begin{subfigure}{0.35\textwidth}
    \centering
    \includegraphics[width=\textwidth]{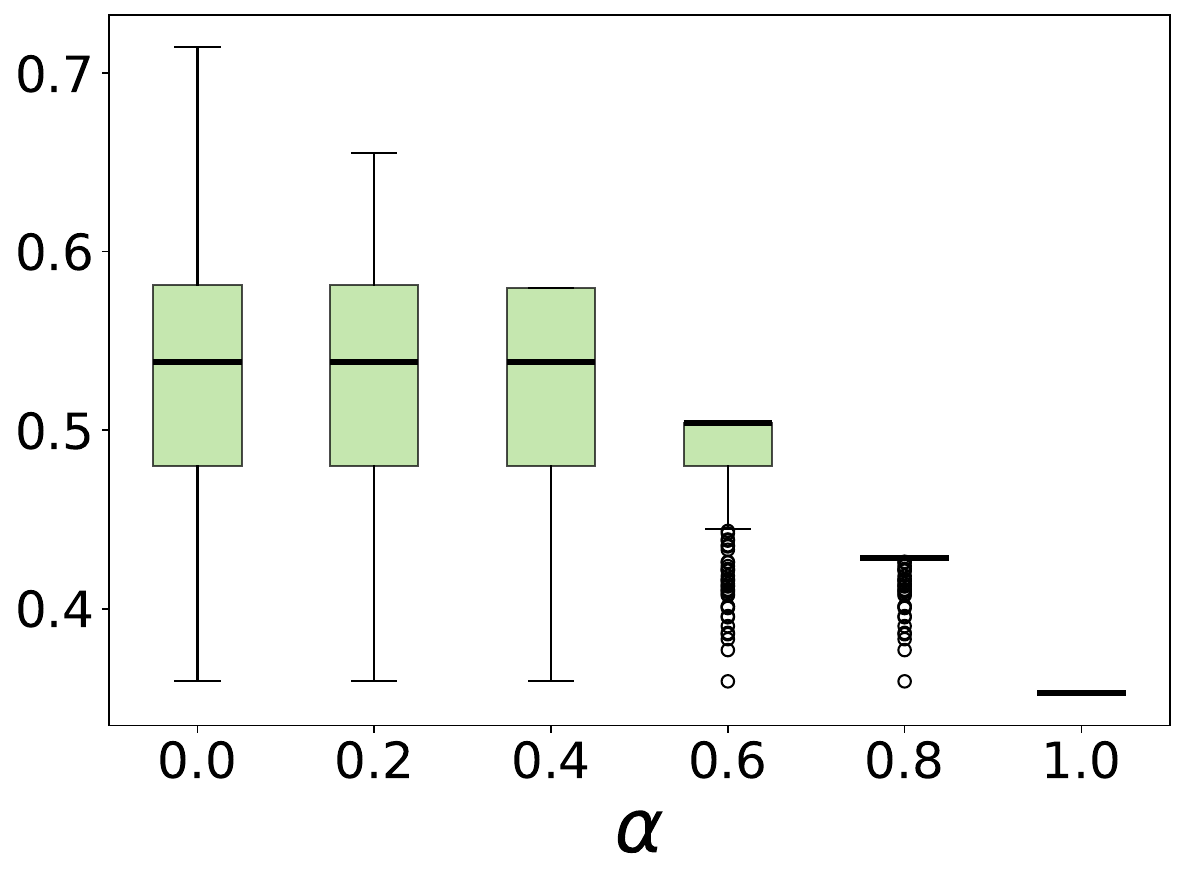}
    \caption{Price under EFO}
    % \label{fig:rawls_price_efo}
\end{subfigure}
\end{minipage}
}
{Price distributions for FEO and EFO under Rawlsian price fairness \label{fig:logistic_Rawlsian_price_case_mis}\vspace{0.5em}}
{}
\end{figure}

\begin{table}[htbp]
    \centering
    \caption{Normalized performance measures for FEO and EFO under Rawlsian price fairness (\%)}
    \label{tab:measure_rawlsian_price_logistic_mis}
    \renewcommand{\arraystretch}{1.1}
    % No font size reduction (\small/\footnotesize) as requested
    
    \begin{tabular}{l p{0.2cm} ccccc p{0.4cm} ccccc}
        \toprule
        & & \multicolumn{5}{c}{FEO ($\alpha$)} & & \multicolumn{5}{c}{EFO ($\alpha$)} \\
        \cmidrule(lr){3-7} \cmidrule(lr){9-13}
        Measure & & 0.2 & 0.4 & 0.6 & 0.8 & 1.0 & & 0.2 & 0.4 & 0.6 & 0.8 & 1.0 \\
        \midrule
        $\mathcal{R}(\alpha)/\mathcal{R}(0)$ & & 98.92 &  96.88 &  94.13 &  89.91 &  79.30 & & 99.94  & 99.31 &  96.31 &  89.78 &  79.30\\
        
        $\mathcal{S}(\alpha)/\mathcal{S}(0)$ & & 104.63 & 112.02 & 120.24 & 131.44 & 157.21 & & 100.16 & 102.62 & 113.17 & 132.21 & 157.21\\
        
        $\mathcal{W}(\alpha)/\mathcal{W}(0)$ & & 101.34 & 103.28 & 105.17 & 107.47 & 112.25 & & 100.03 & 100.71 & 103.44 & 107.72 & 112.25\\
        \bottomrule
    \end{tabular}
\end{table}

Figure~\ref{fig:logistic_Rawlsian_demand_case_mis} shows the results under Rawlsian demand fairness. Similar to parity-wise demand fairness, imposing Rawlsian demand fairness constraints has a much smaller effect on ex-post demand than on ex-ante demand. Under both FEO and EFO, the minimum ex-post demand increases only slightly, or even decreases for some $\alpha$. Moreover, the maximum ex-post demand increases substantially. This suggests that Rawlsian demand fairness constraints may be imposed on the wrong population. Therefore, even though in Table~\ref{tab:measure_rawlsian_demand_logistic_mis} consumer surplus and social welfare increase for both FEO and EFO, the intended fairness objective may still not be achieved. This leaves an open question of how to ensure demand fairness when the demand estimation is insufficiently accurate.

\begin{figure}[htbp]
\FIGURE{
\begin{minipage}{\textwidth}
\centering
\captionsetup{justification=centering}
\begin{subfigure}{0.235\textwidth}
    \centering
    \includegraphics[width=\textwidth]{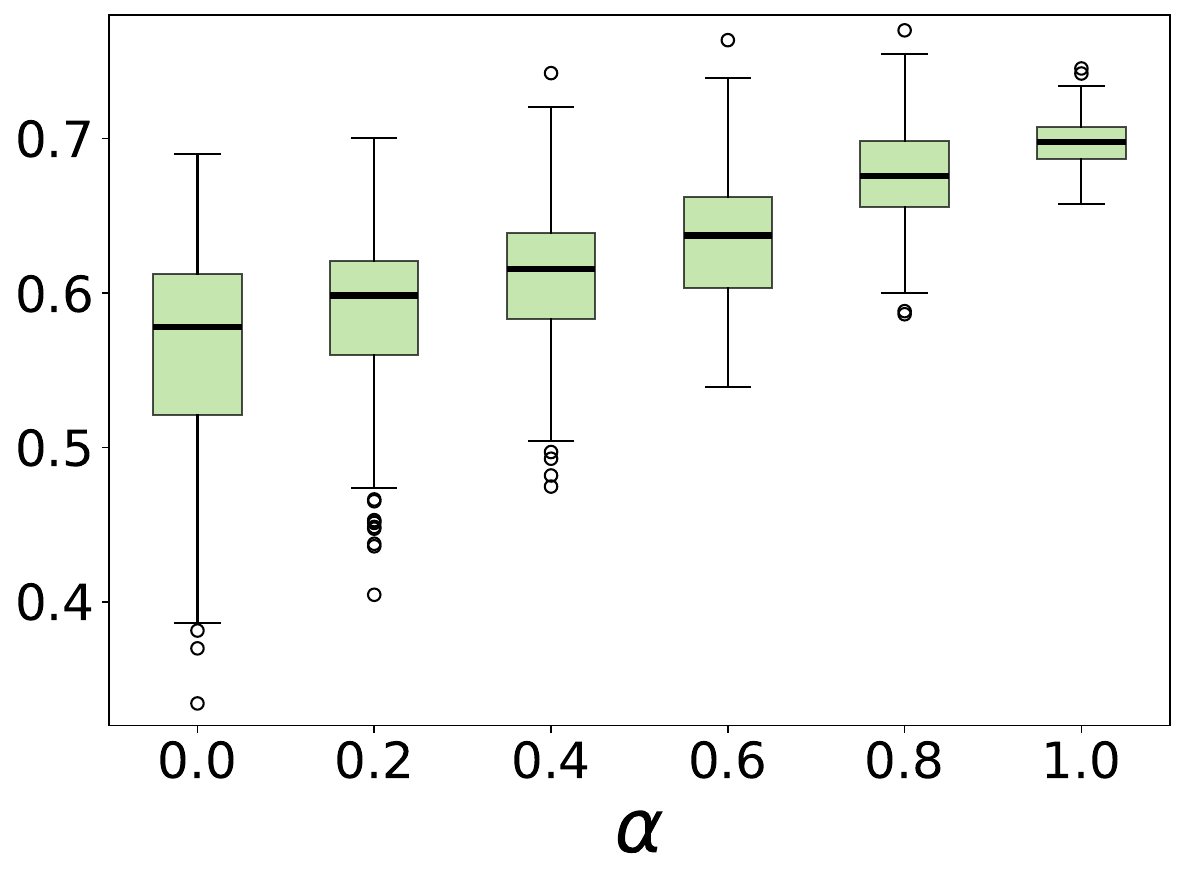}
    \caption{Ex Ante Demand (FEO)}
\end{subfigure}
\begin{subfigure}{0.235\textwidth}
    \centering
    \includegraphics[width=\textwidth]{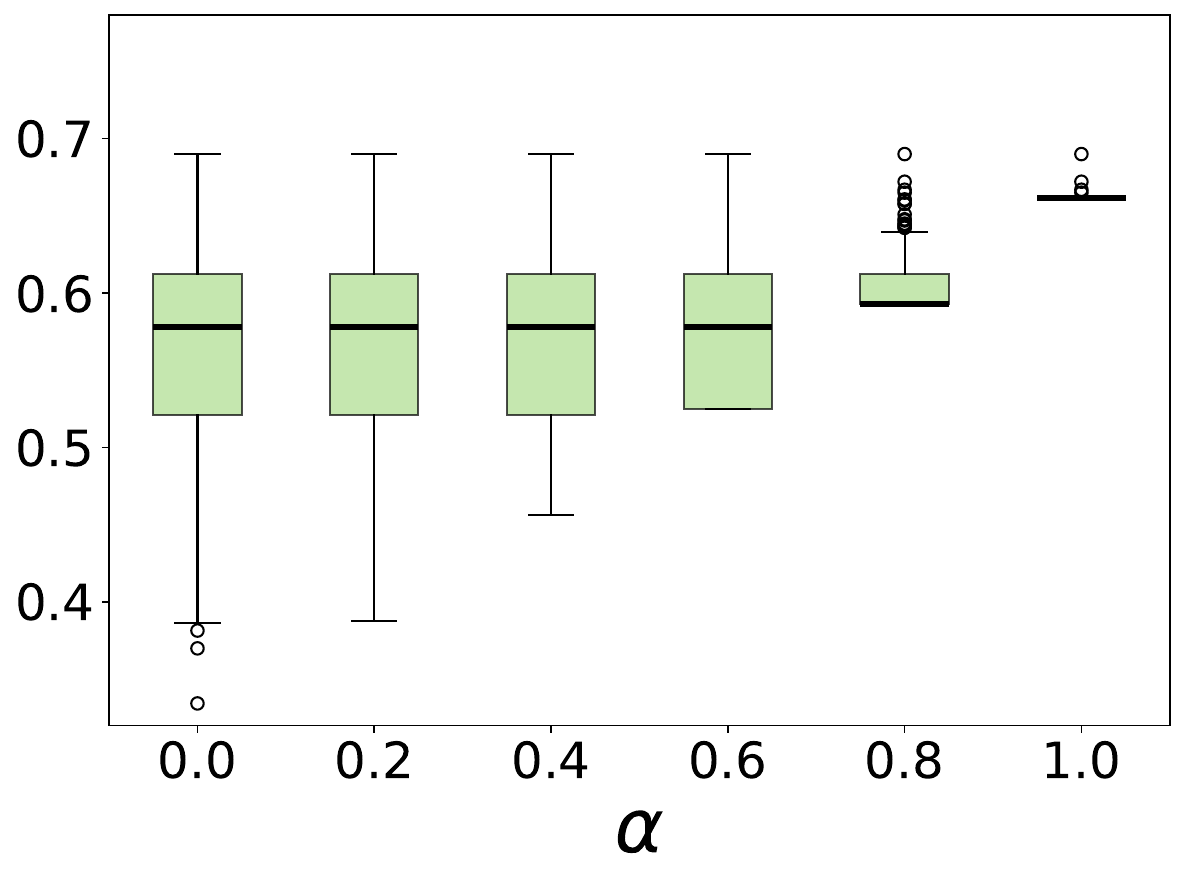}
    \caption{Ex Ante Demand (EFO)}
\end{subfigure}
\begin{subfigure}{0.235\textwidth}
    \centering
    \includegraphics[width=\textwidth]{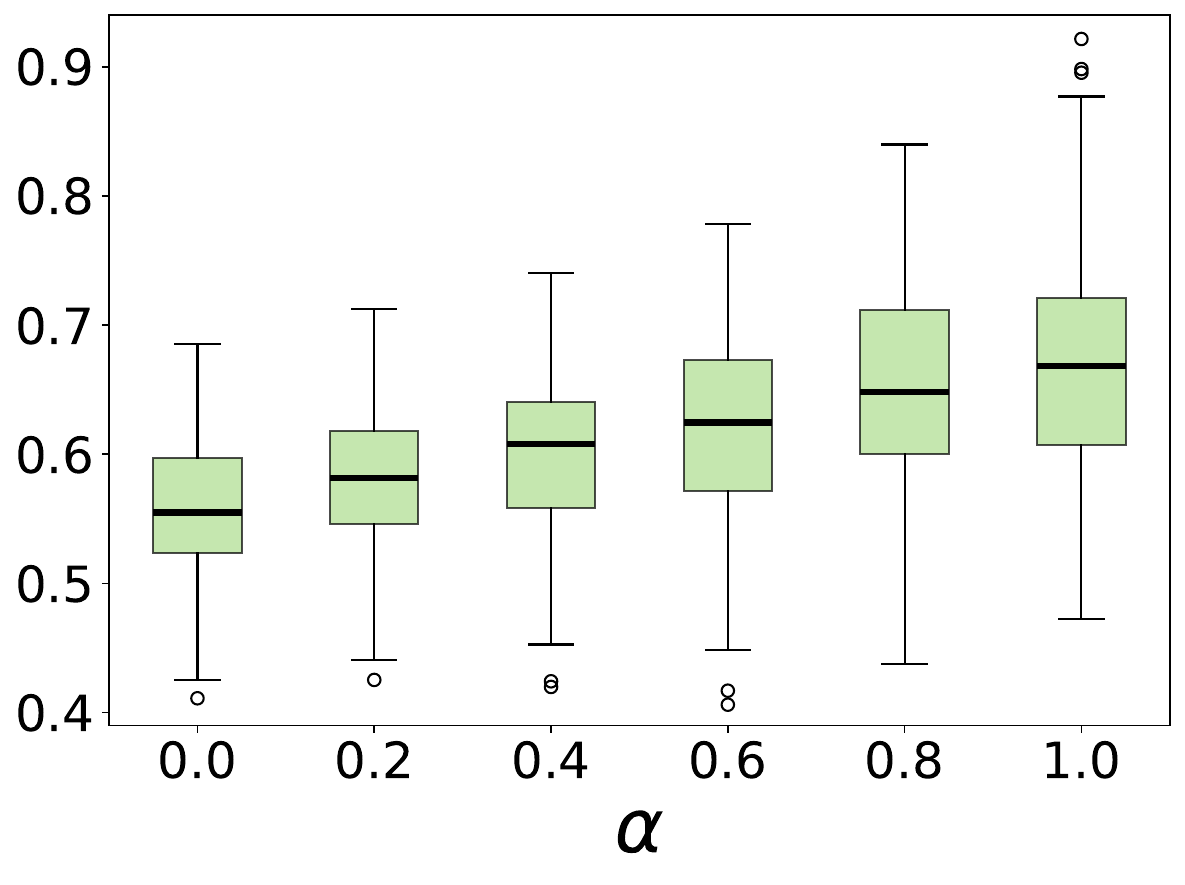}
    \caption{Ex Post Demand (FEO)}
\end{subfigure}
\begin{subfigure}{0.235\textwidth}
    \centering
    \includegraphics[width=\textwidth]{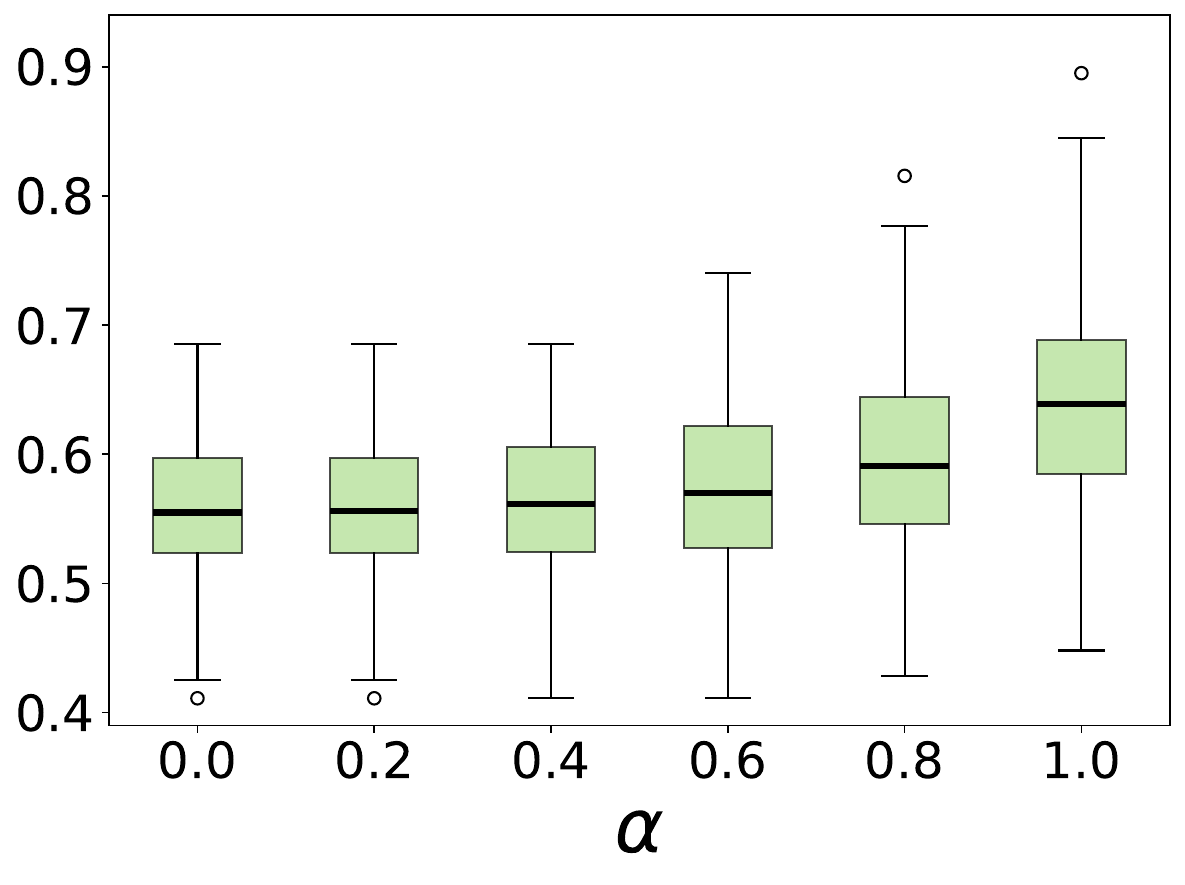}
    \caption{Ex Post Demand (EFO)}
\end{subfigure}
\end{minipage}
}
{Demand distributions for FEO and EFO under Rawlsian demand fairness \label{fig:logistic_Rawlsian_demand_case_mis}\vspace{0.5em}}
{}
\end{figure}

\begin{table}[htbp]
    \centering
    \caption{Normalized performance measures for FEO and EFO under Rawlsian demand fairness (\%)}
    \label{tab:measure_rawlsian_demand_logistic_mis}
    \renewcommand{\arraystretch}{1.1}
    \setlength{\tabcolsep}{4pt} % Adjust column spacing for a perfect fit
    
    \begin{tabular}{l ccccc p{0.3cm} ccccc}
    \toprule
    & \multicolumn{5}{c}{FEO ($\alpha$)} & & \multicolumn{5}{c}{EFO ($\alpha$)} \\
    \cmidrule(lr){2-6} \cmidrule(lr){8-12}
    Measure & 0.2 & 0.4 & 0.6 & 0.8 & 1.0 & & 0.2 & 0.4 & 0.6 & 0.8 & 1.0 \\
    \midrule
    
    % Revenue Row
    $\mathcal{R}(\alpha)/\mathcal{R}(0)$ 
    & 98.51 &  96.37 &  93.35 &  88.38 &  85.95
    & 
    & 99.98 &  99.68  & 98.66 &  96.15 &  90.51 \\

    % Surplus Row
    $\mathcal{S}(\alpha)/\mathcal{S}(0)$ 
    & 106.29 & 113.00 & 120.61 & 133.08 & 138.68
    & 
    & 100.09 & 101.15 & 104.29 & 111.91 & 128.11 \\

    % Welfare Row
    $\mathcal{W}(\alpha)/\mathcal{W}(0)$ 
    & 101.80 & 103.40 & 104.88 & 107.28 & 108.25
    & 
    & 100.03 & 100.30 & 101.04 & 102.81 & 106.41 \\

    \bottomrule
    \end{tabular}
\end{table}

\end{APPENDIX}

%%%%%%%%%%%%%%%%%
\end{document}